\documentclass[11pt,          
               phd,           
               onehalfspacing 
               ]{ncsuthesis}

\usepackage{booktabs}  
\usepackage{amsmath}
\usepackage{textcomp}  
\usepackage{subfig}    

\RequirePackage[
			bibstyle=ieee,
			citestyle=numeric-comp,
			sorting=none,					
			hyperref=true,
			firstinits=true,
			backend=bibtex,
			url=false,
			isbn=false,
			sortcites=true,
			doi=false]{biblatex}	

\makeatletter

\newrobustcmd*{\parentexttrack}[1]{%
  \begingroup
  \blx@blxinit
  \blx@setsfcodes
  \blx@bibopenparen#1\blx@bibcloseparen
  \endgroup}

\AtEveryCite{%
  }

\makeatother
\renewcommand{\cite}[1]{\parencite{#1}}

\renewbibmacro{in:}{%
  \ifentrytype{article}{}{%
  \printtext{\bibstring{in}\intitlepunct}}}
  
\AtEveryBibitem{\clearfield{language}}

\def\blx@bblfile@bibtex{
  \blx@secinit
  \begingroup
  \blx@bblstart

\begingroup
\makeatletter
\@ifundefined{ver@biblatex.sty}
  {\@latex@error
     {Missing 'biblatex' package}
     {The bibliography requires the 'biblatex' package.}
      \aftergroup }
  {}
\endgroup

\sortlist[entry]{none/global/}
  \entry{Tanaka2002}{article}{}
    \name{author}{1}{}{%
      {{hash=TT}{%
         family={Tanaka},
         familyi={T\bibinitperiod},
         given={T.},
         giveni={T\bibinitperiod},
      }}%
    }
    \strng{namehash}{TT1}
    \strng{fullhash}{TT1}
    \field{labelnamesource}{author}
    \field{labeltitlesource}{title}
    \field{number}{11}
    \field{pages}{2888\bibrangedash 2910}
    \field{title}{{A statistical-mechanics approach to large-system analysis of
  CDMA multiuser detectors}}
    \field{volume}{48}
    \field{journaltitle}{IEEE Trans. Inf. Theory}
    \field{month}{11}
    \field{year}{2002}
  \endentry

  \entry{GuoVerdu2005}{article}{}
    \name{author}{2}{}{%
      {{hash=GD}{%
         family={Guo},
         familyi={G\bibinitperiod},
         given={D.},
         giveni={D\bibinitperiod},
      }}%
      {{hash=VS}{%
         family={Verd{\'u}},
         familyi={V\bibinitperiod},
         given={S.},
         giveni={S\bibinitperiod},
      }}%
    }
    \strng{namehash}{GDVS1}
    \strng{fullhash}{GDVS1}
    \field{labelnamesource}{author}
    \field{labeltitlesource}{title}
    \field{number}{6}
    \field{pages}{1983\bibrangedash 2010}
    \field{title}{Randomly spread {CDMA}: {A}symptotics via statistical
  physics}
    \field{volume}{51}
    \field{journaltitle}{IEEE Trans. Inf. Theory}
    \field{month}{06}
    \field{year}{2005}
  \endentry

  \entry{Krzakala2012probabilistic}{article}{}
    \name{author}{5}{}{%
      {{hash=KF}{%
         family={Krzakala},
         familyi={K\bibinitperiod},
         given={F.},
         giveni={F\bibinitperiod},
      }}%
      {{hash=MM}{%
         family={M{\'e}zard},
         familyi={M\bibinitperiod},
         given={M.},
         giveni={M\bibinitperiod},
      }}%
      {{hash=SF}{%
         family={Sausset},
         familyi={S\bibinitperiod},
         given={F.},
         giveni={F\bibinitperiod},
      }}%
      {{hash=SY}{%
         family={Sun},
         familyi={S\bibinitperiod},
         given={Y.},
         giveni={Y\bibinitperiod},
      }}%
      {{hash=ZL}{%
         family={Zdeborov{\'a}},
         familyi={Z\bibinitperiod},
         given={L.},
         giveni={L\bibinitperiod},
      }}%
    }
    \strng{namehash}{KFMMSFSYZL1}
    \strng{fullhash}{KFMMSFSYZL1}
    \field{labelnamesource}{author}
    \field{labeltitlesource}{title}
    \field{number}{08}
    \field{pages}{P08009}
    \field{title}{Probabilistic reconstruction in compressed sensing:
  {A}lgorithms, phase diagrams, and threshold achieving matrices}
    \field{volume}{2012}
    \field{journaltitle}{J. Stat. Mech. -- Theory E.}
    \field{month}{08}
    \field{year}{2012}
  \endentry

  \entry{krzakala2012statistical}{article}{}
    \name{author}{5}{}{%
      {{hash=KF}{%
         family={Krzakala},
         familyi={K\bibinitperiod},
         given={F.},
         giveni={F\bibinitperiod},
      }}%
      {{hash=MM}{%
         family={M{\'e}zard},
         familyi={M\bibinitperiod},
         given={M.},
         giveni={M\bibinitperiod},
      }}%
      {{hash=SF}{%
         family={Sausset},
         familyi={S\bibinitperiod},
         given={F.},
         giveni={F\bibinitperiod},
      }}%
      {{hash=SY}{%
         family={Sun},
         familyi={S\bibinitperiod},
         given={Y.},
         giveni={Y\bibinitperiod},
      }}%
      {{hash=ZL}{%
         family={Zdeborov{\'a}},
         familyi={Z\bibinitperiod},
         given={L.},
         giveni={L\bibinitperiod},
      }}%
    }
    \strng{namehash}{KFMMSFSYZL1}
    \strng{fullhash}{KFMMSFSYZL1}
    \field{labelnamesource}{author}
    \field{labeltitlesource}{title}
    \field{number}{2}
    \field{pages}{021005}
    \field{title}{Statistical-physics-based reconstruction in compressed
  sensing}
    \field{volume}{2}
    \field{journaltitle}{Phys. Rev. X}
    \field{month}{05}
    \field{year}{2012}
  \endentry

  \entry{MezardMontanariBook}{book}{}
    \name{author}{2}{}{%
      {{hash=MM}{%
         family={M{\'e}zard},
         familyi={M\bibinitperiod},
         given={M.},
         giveni={M\bibinitperiod},
      }}%
      {{hash=MA}{%
         family={Montanari},
         familyi={M\bibinitperiod},
         given={A.},
         giveni={A\bibinitperiod},
      }}%
    }
    \list{publisher}{1}{%
      {Oxford University press}%
    }
    \strng{namehash}{MMMA1}
    \strng{fullhash}{MMMA1}
    \field{labelnamesource}{author}
    \field{labeltitlesource}{title}
    \field{title}{Information, {P}hysics, and {C}omputation}
    \field{year}{2009}
  \endentry

  \entry{Barbier2015}{article}{}
    \name{author}{2}{}{%
      {{hash=BJ}{%
         family={Barbier},
         familyi={B\bibinitperiod},
         given={J.},
         giveni={J\bibinitperiod},
      }}%
      {{hash=KF}{%
         family={Krzakala},
         familyi={K\bibinitperiod},
         given={F.},
         giveni={F\bibinitperiod},
      }}%
    }
    \strng{namehash}{BJKF1}
    \strng{fullhash}{BJKF1}
    \field{labelnamesource}{author}
    \field{labeltitlesource}{title}
    \field{title}{Approximate message-passing decoder and capacity-achieving
  sparse superposition codes}
    \field{journaltitle}{Arxiv preprint arXiv:1503.08040}
    \field{month}{03}
    \field{year}{2015}
  \endentry

  \entry{Cover06}{book}{}
    \name{author}{2}{}{%
      {{hash=CTM}{%
         family={Cover},
         familyi={C\bibinitperiod},
         given={T.\bibnamedelima M.},
         giveni={T\bibinitperiod\bibinitdelim M\bibinitperiod},
      }}%
      {{hash=TJA}{%
         family={Thomas},
         familyi={T\bibinitperiod},
         given={J.\bibnamedelima A.},
         giveni={J\bibinitperiod\bibinitdelim A\bibinitperiod},
      }}%
    }
    \list{publisher}{1}{%
      {New York, NY, USA: Wiley-Interscience}%
    }
    \strng{namehash}{CTMTJA1}
    \strng{fullhash}{CTMTJA1}
    \field{labelnamesource}{author}
    \field{labeltitlesource}{title}
    \field{title}{Elements of Information Theory}
    \field{year}{2006}
  \endentry

  \entry{RFG2012}{article}{}
    \name{author}{3}{}{%
      {{hash=RS}{%
         family={Rangan},
         familyi={R\bibinitperiod},
         given={S.},
         giveni={S\bibinitperiod},
      }}%
      {{hash=FAK}{%
         family={Fletcher},
         familyi={F\bibinitperiod},
         given={A.\bibnamedelima K.},
         giveni={A\bibinitperiod\bibinitdelim K\bibinitperiod},
      }}%
      {{hash=GVK}{%
         family={Goyal},
         familyi={G\bibinitperiod},
         given={V.\bibnamedelima K.},
         giveni={V\bibinitperiod\bibinitdelim K\bibinitperiod},
      }}%
    }
    \strng{namehash}{RSFAKGVK1}
    \strng{fullhash}{RSFAKGVK1}
    \field{labelnamesource}{author}
    \field{labeltitlesource}{title}
    \field{number}{3}
    \field{pages}{1902\bibrangedash 1923}
    \field{title}{Asymptotic analysis of {MAP} estimation via the replica
  method and applications to compressed sensing}
    \field{volume}{58}
    \field{journaltitle}{IEEE Trans. Inf. Theory}
    \field{month}{03}
    \field{year}{2012}
  \endentry

  \entry{Montanari2006}{inproceedings}{}
    \name{author}{2}{}{%
      {{hash=MA}{%
         family={Montanari},
         familyi={M\bibinitperiod},
         given={A.},
         giveni={A\bibinitperiod},
      }}%
      {{hash=TD}{%
         family={Tse},
         familyi={T\bibinitperiod},
         given={D.},
         giveni={D\bibinitperiod},
      }}%
    }
    \strng{namehash}{MATD1}
    \strng{fullhash}{MATD1}
    \field{labelnamesource}{author}
    \field{labeltitlesource}{title}
    \field{booktitle}{IEEE Inf. Theory Workshop}
    \field{pages}{160\bibrangedash 164}
    \field{title}{Analysis of belief propagation for non-linear problems: The
  example of {CDMA} (or: How to prove {T}anaka's formula)}
    \field{month}{03}
    \field{year}{2006}
  \endentry

  \entry{Lesieur2015}{inproceedings}{}
    \name{author}{3}{}{%
      {{hash=LT}{%
         family={Lesieur},
         familyi={L\bibinitperiod},
         given={T.},
         giveni={T\bibinitperiod},
      }}%
      {{hash=KF}{%
         family={Krzakala},
         familyi={K\bibinitperiod},
         given={F.},
         giveni={F\bibinitperiod},
      }}%
      {{hash=ZL}{%
         family={Zdeborov{\'a}},
         familyi={Z\bibinitperiod},
         given={L.},
         giveni={L\bibinitperiod},
      }}%
    }
    \strng{namehash}{LTKFZL1}
    \strng{fullhash}{LTKFZL1}
    \field{labelnamesource}{author}
    \field{labeltitlesource}{title}
    \field{booktitle}{Proc. IEEE Int. Symp. Inf. Theory (ISIT)}
    \field{pages}{1635\bibrangedash 1639}
    \field{title}{Phase transitions in sparse {PCA}}
    \list{location}{1}{%
      {Hong Kong, China}%
    }
    \field{month}{07}
    \field{year}{2015}
  \endentry

  \entry{Han2014}{inproceedings}{}
    \name{author}{4}{}{%
      {{hash=HP}{%
         family={Han},
         familyi={H\bibinitperiod},
         given={P.},
         giveni={P\bibinitperiod},
      }}%
      {{hash=NR}{%
         family={Niu},
         familyi={N\bibinitperiod},
         given={R.},
         giveni={R\bibinitperiod},
      }}%
      {{hash=RM}{%
         family={Ren},
         familyi={R\bibinitperiod},
         given={M.},
         giveni={M\bibinitperiod},
      }}%
      {{hash=EYC}{%
         family={Eldar},
         familyi={E\bibinitperiod},
         given={Y.\bibnamedelima C.},
         giveni={Y\bibinitperiod\bibinitdelim C\bibinitperiod},
      }}%
    }
    \strng{namehash}{HPNRRMEYC1}
    \strng{fullhash}{HPNRRMEYC1}
    \field{labelnamesource}{author}
    \field{labeltitlesource}{title}
    \field{booktitle}{Proc. IEEE Global Conf. Signal Inf. Process. (GlobalSIP)}
    \field{pages}{497\bibrangedash 501}
    \field{title}{Distributed approximate message passing for sparse signal
  recovery}
    \list{location}{1}{%
      {Atlanta, GA}%
    }
    \field{month}{12}
    \field{year}{2014}
  \endentry

  \entry{MaBaronNeedell2014}{article}{}
    \name{author}{3}{}{%
      {{hash=MY}{%
         family={Ma},
         familyi={M\bibinitperiod},
         given={Y.},
         giveni={Y\bibinitperiod},
      }}%
      {{hash=BD}{%
         family={Baron},
         familyi={B\bibinitperiod},
         given={D.},
         giveni={D\bibinitperiod},
      }}%
      {{hash=ND}{%
         family={Needell},
         familyi={N\bibinitperiod},
         given={D.},
         giveni={D\bibinitperiod},
      }}%
    }
    \strng{namehash}{MYBDND1}
    \strng{fullhash}{MYBDND1}
    \field{labelnamesource}{author}
    \field{labeltitlesource}{title}
    \field{number}{23}
    \field{pages}{6323\bibrangedash 6334}
    \field{title}{Two-part reconstruction with noisy-sudocodes}
    \field{volume}{62}
    \field{journaltitle}{IEEE Trans. Signal Process.}
    \field{month}{12}
    \field{year}{2014}
  \endentry

  \entry{clickPrediction_MS2007}{inproceedings}{}
    \name{author}{3}{}{%
      {{hash=RM}{%
         family={Richardson},
         familyi={R\bibinitperiod},
         given={M.},
         giveni={M\bibinitperiod},
      }}%
      {{hash=DE}{%
         family={Dominowska},
         familyi={D\bibinitperiod},
         given={E.},
         giveni={E\bibinitperiod},
      }}%
      {{hash=RR}{%
         family={Ragno},
         familyi={R\bibinitperiod},
         given={R.},
         giveni={R\bibinitperiod},
      }}%
    }
    \strng{namehash}{RMDERR1}
    \strng{fullhash}{RMDERR1}
    \field{labelnamesource}{author}
    \field{labeltitlesource}{title}
    \field{booktitle}{Proc. Int. World Wide Web Conf. (WWW)}
    \field{pages}{521\bibrangedash 530}
    \field{title}{Predicting clicks: {E}stimating the click-through rate for
  new ads}
    \list{location}{1}{%
      {Banff, Alberta, Canada}%
    }
    \field{month}{05}
    \field{year}{2007}
  \endentry

  \entry{clickPrediction_Google2013}{inproceedings}{}
    \name{author}{16}{}{%
      {{hash=MHB}{%
         family={McMahan},
         familyi={M\bibinitperiod},
         given={H.\bibnamedelima B.},
         giveni={H\bibinitperiod\bibinitdelim B\bibinitperiod},
      }}%
      {{hash=HG}{%
         family={Holt},
         familyi={H\bibinitperiod},
         given={G.},
         giveni={G\bibinitperiod},
      }}%
      {{hash=SD}{%
         family={Sculley},
         familyi={S\bibinitperiod},
         given={D.},
         giveni={D\bibinitperiod},
      }}%
      {{hash=YM}{%
         family={Young},
         familyi={Y\bibinitperiod},
         given={M.},
         giveni={M\bibinitperiod},
      }}%
      {{hash=ED}{%
         family={Ebner},
         familyi={E\bibinitperiod},
         given={D.},
         giveni={D\bibinitperiod},
      }}%
      {{hash=GJ}{%
         family={Grady},
         familyi={G\bibinitperiod},
         given={J.},
         giveni={J\bibinitperiod},
      }}%
      {{hash=NL}{%
         family={Nie},
         familyi={N\bibinitperiod},
         given={L.},
         giveni={L\bibinitperiod},
      }}%
      {{hash=PT}{%
         family={Phillips},
         familyi={P\bibinitperiod},
         given={T.},
         giveni={T\bibinitperiod},
      }}%
      {{hash=DE}{%
         family={Davydov},
         familyi={D\bibinitperiod},
         given={E.},
         giveni={E\bibinitperiod},
      }}%
      {{hash=GD}{%
         family={Golovin},
         familyi={G\bibinitperiod},
         given={D.},
         giveni={D\bibinitperiod},
      }}%
      {{hash=CS}{%
         family={Chikkerur},
         familyi={C\bibinitperiod},
         given={S.},
         giveni={S\bibinitperiod},
      }}%
      {{hash=LD}{%
         family={Liu},
         familyi={L\bibinitperiod},
         given={D.},
         giveni={D\bibinitperiod},
      }}%
      {{hash=WM}{%
         family={Wattenberg},
         familyi={W\bibinitperiod},
         given={M.},
         giveni={M\bibinitperiod},
      }}%
      {{hash=HAM}{%
         family={Hrafnkelsson},
         familyi={H\bibinitperiod},
         given={A.\bibnamedelima M.},
         giveni={A\bibinitperiod\bibinitdelim M\bibinitperiod},
      }}%
      {{hash=BT}{%
         family={Boulos},
         familyi={B\bibinitperiod},
         given={T.},
         giveni={T\bibinitperiod},
      }}%
      {{hash=KJ}{%
         family={Kubica},
         familyi={K\bibinitperiod},
         given={J.},
         giveni={J\bibinitperiod},
      }}%
    }
    \strng{namehash}{MHBHGSDYMEDGJNLPTDEGDCSLDWMHAMBTKJ1}
    \strng{fullhash}{MHBHGSDYMEDGJNLPTDEGDCSLDWMHAMBTKJ1}
    \field{labelnamesource}{author}
    \field{labeltitlesource}{title}
    \field{booktitle}{Proc. ACM SIGKDD Int. Conf. Knowledge Discovery and Data
  Mining (KDD)}
    \field{pages}{1222\bibrangedash 1230}
    \field{title}{Ad click prediction: {A} view from the trenches}
    \list{location}{1}{%
      {Chicago, IL}%
    }
    \field{month}{08}
    \field{year}{2013}
  \endentry

  \entry{DonohoCS}{article}{}
    \name{author}{1}{}{%
      {{hash=DD}{%
         family={Donoho},
         familyi={D\bibinitperiod},
         given={D.},
         giveni={D\bibinitperiod},
      }}%
    }
    \strng{namehash}{DD1}
    \strng{fullhash}{DD1}
    \field{labelnamesource}{author}
    \field{labeltitlesource}{title}
    \field{number}{4}
    \field{pages}{1289\bibrangedash 1306}
    \field{title}{Compressed sensing}
    \field{volume}{52}
    \field{journaltitle}{IEEE Trans. Inf. Theory}
    \field{month}{04}
    \field{year}{2006}
  \endentry

  \entry{DeanGhemawat2008}{article}{}
    \name{author}{2}{}{%
      {{hash=DJ}{%
         family={Dean},
         familyi={D\bibinitperiod},
         given={J.},
         giveni={J\bibinitperiod},
      }}%
      {{hash=GS}{%
         family={Ghemawat},
         familyi={G\bibinitperiod},
         given={S.},
         giveni={S\bibinitperiod},
      }}%
    }
    \strng{namehash}{DJGS1}
    \strng{fullhash}{DJGS1}
    \field{labelnamesource}{author}
    \field{labeltitlesource}{title}
    \field{number}{1}
    \field{pages}{107\bibrangedash 113}
    \field{title}{Map{R}educe: Simplified data processing on large clusters}
    \field{volume}{51}
    \list{location}{1}{%
      {New York, NY}%
    }
    \field{journaltitle}{Commun. ACM}
    \field{month}{01}
    \field{year}{2008}
  \endentry

  \entry{Mota2012}{article}{}
    \name{author}{3}{}{%
      {{hash=MJ}{%
         family={Mota},
         familyi={M\bibinitperiod},
         given={J.},
         giveni={J\bibinitperiod},
      }}%
      {{hash=XJ}{%
         family={Xavier},
         familyi={X\bibinitperiod},
         given={J.},
         giveni={J\bibinitperiod},
      }}%
      {{hash=AP}{%
         family={Aguiar},
         familyi={A\bibinitperiod},
         given={P.},
         giveni={P\bibinitperiod},
      }}%
    }
    \strng{namehash}{MJXJAP1}
    \strng{fullhash}{MJXJAP1}
    \field{labelnamesource}{author}
    \field{labeltitlesource}{title}
    \field{number}{4}
    \field{pages}{1942\bibrangedash 1956}
    \field{title}{Distributed basis pursuit}
    \field{volume}{60}
    \field{journaltitle}{IEEE Trans. Signal Process.}
    \field{month}{04}
    \field{year}{2012}
  \endentry

  \entry{Patterson2013}{inproceedings}{}
    \name{author}{3}{}{%
      {{hash=PS}{%
         family={Patterson},
         familyi={P\bibinitperiod},
         given={S.},
         giveni={S\bibinitperiod},
      }}%
      {{hash=EYC}{%
         family={Eldar},
         familyi={E\bibinitperiod},
         given={Y.\bibnamedelima C.},
         giveni={Y\bibinitperiod\bibinitdelim C\bibinitperiod},
      }}%
      {{hash=KI}{%
         family={Keidar},
         familyi={K\bibinitperiod},
         given={I.},
         giveni={I\bibinitperiod},
      }}%
    }
    \strng{namehash}{PSEYCKI1}
    \strng{fullhash}{PSEYCKI1}
    \field{labelnamesource}{author}
    \field{labeltitlesource}{title}
    \field{booktitle}{Proc. IEEE Int. Conf. Acoustics, Speech, Signal Process.
  (ICASSP)}
    \field{pages}{4494\bibrangedash 4498}
    \field{title}{Distributed sparse signal recovery for sensor networks}
    \list{location}{1}{%
      {Vancouver, B.C., Canada}%
    }
    \field{month}{05}
    \field{year}{2013}
  \endentry

  \entry{Patterson2014}{article}{}
    \name{author}{3}{}{%
      {{hash=PS}{%
         family={Patterson},
         familyi={P\bibinitperiod},
         given={S.},
         giveni={S\bibinitperiod},
      }}%
      {{hash=EYC}{%
         family={Eldar},
         familyi={E\bibinitperiod},
         given={Y.\bibnamedelima C.},
         giveni={Y\bibinitperiod\bibinitdelim C\bibinitperiod},
      }}%
      {{hash=KI}{%
         family={Keidar},
         familyi={K\bibinitperiod},
         given={I.},
         giveni={I\bibinitperiod},
      }}%
    }
    \list{publisher}{1}{%
      {IEEE}%
    }
    \strng{namehash}{PSEYCKI1}
    \strng{fullhash}{PSEYCKI1}
    \field{labelnamesource}{author}
    \field{labeltitlesource}{title}
    \field{number}{19}
    \field{pages}{4931\bibrangedash 4946}
    \field{title}{Distributed compressed sensing for static and time-varying
  networks}
    \field{volume}{62}
    \field{journaltitle}{IEEE Trans. Signal Process.}
    \field{month}{10}
    \field{year}{2014}
  \endentry

  \entry{Han2015ICASSP}{inproceedings}{}
    \name{author}{3}{}{%
      {{hash=HP}{%
         family={Han},
         familyi={H\bibinitperiod},
         given={P.},
         giveni={P\bibinitperiod},
      }}%
      {{hash=NR}{%
         family={Niu},
         familyi={N\bibinitperiod},
         given={R.},
         giveni={R\bibinitperiod},
      }}%
      {{hash=EYC}{%
         family={Eldar},
         familyi={E\bibinitperiod},
         given={Y.\bibnamedelima C.},
         giveni={Y\bibinitperiod\bibinitdelim C\bibinitperiod},
      }}%
    }
    \strng{namehash}{HPNREYC1}
    \strng{fullhash}{HPNREYC1}
    \field{labelnamesource}{author}
    \field{labeltitlesource}{title}
    \field{booktitle}{Proc. IEEE Int. Conf. Acoustics, Speech, Signal Process.
  (ICASSP)}
    \field{pages}{3766\bibrangedash 3770}
    \field{title}{Modified distributed iterative hard thresholding}
    \list{location}{1}{%
      {Brisbane, Australia}%
    }
    \field{month}{04}
    \field{year}{2015}
  \endentry

  \entry{Ravazzi2015}{article}{}
    \name{author}{3}{}{%
      {{hash=RC}{%
         family={Ravazzi},
         familyi={R\bibinitperiod},
         given={C.},
         giveni={C\bibinitperiod},
      }}%
      {{hash=FSM}{%
         family={Fosson},
         familyi={F\bibinitperiod},
         given={S.\bibnamedelima M.},
         giveni={S\bibinitperiod\bibinitdelim M\bibinitperiod},
      }}%
      {{hash=ME}{%
         family={Magli},
         familyi={M\bibinitperiod},
         given={E.},
         giveni={E\bibinitperiod},
      }}%
    }
    \strng{namehash}{RCFSMME1}
    \strng{fullhash}{RCFSMME1}
    \field{labelnamesource}{author}
    \field{labeltitlesource}{title}
    \field{number}{4}
    \field{pages}{2081\bibrangedash 2100}
    \field{title}{Distributed iterative thresholding for {$\ell _{0}/\ell
  _{1}$} -regularized linear inverse problems}
    \field{volume}{61}
    \field{journaltitle}{IEEE Trans. Inf. Theory}
    \field{month}{04}
    \field{year}{2015}
  \endentry

  \entry{Han2015SPARS}{inproceedings}{}
    \name{author}{3}{}{%
      {{hash=HP}{%
         family={Han},
         familyi={H\bibinitperiod},
         given={P.},
         giveni={P\bibinitperiod},
      }}%
      {{hash=NR}{%
         family={Niu},
         familyi={N\bibinitperiod},
         given={R.},
         giveni={R\bibinitperiod},
      }}%
      {{hash=EYC}{%
         family={Eldar},
         familyi={E\bibinitperiod},
         given={Y.\bibnamedelima C.},
         giveni={Y\bibinitperiod\bibinitdelim C\bibinitperiod},
      }}%
    }
    \strng{namehash}{HPNREYC1}
    \strng{fullhash}{HPNREYC1}
    \field{labelnamesource}{author}
    \field{labeltitlesource}{title}
    \field{booktitle}{Proc. Signal Process. with Adaptive Sparse Structured
  Representations Workshop (SPARS)}
    \field{title}{Communication-efficient distributed {IHT}}
    \list{location}{1}{%
      {Cambridge, United Kingdom}%
    }
    \field{month}{07}
    \field{year}{2015}
  \endentry

  \entry{HanZhuNiuBaron2016ICASSP}{inproceedings}{}
    \name{author}{4}{}{%
      {{hash=HP}{%
         family={Han},
         familyi={H\bibinitperiod},
         given={P.},
         giveni={P\bibinitperiod},
      }}%
      {{hash=ZJ}{%
         family={Zhu},
         familyi={Z\bibinitperiod},
         given={J.},
         giveni={J\bibinitperiod},
      }}%
      {{hash=NR}{%
         family={Niu},
         familyi={N\bibinitperiod},
         given={R.},
         giveni={R\bibinitperiod},
      }}%
      {{hash=BD}{%
         family={Baron},
         familyi={B\bibinitperiod},
         given={D.},
         giveni={D\bibinitperiod},
      }}%
    }
    \strng{namehash}{HPZJNRBD1}
    \strng{fullhash}{HPZJNRBD1}
    \field{labelnamesource}{author}
    \field{labeltitlesource}{title}
    \field{booktitle}{Proc. IEEE Int. Conf. Acoustics, Speech, Signal Process.
  (ICASSP)}
    \field{pages}{6240\bibrangedash 6244}
    \field{title}{Multi-processor approximate message passing using lossy
  compression}
    \list{location}{1}{%
      {Shanghai, China}%
    }
    \field{month}{03}
    \field{year}{2016}
  \endentry

  \entry{Duarte2006IPSN}{inproceedings}{}
    \name{author}{4}{}{%
      {{hash=DMF}{%
         family={Duarte},
         familyi={D\bibinitperiod},
         given={M.\bibnamedelima F.},
         giveni={M\bibinitperiod\bibinitdelim F\bibinitperiod},
      }}%
      {{hash=WMB}{%
         family={Wakin},
         familyi={W\bibinitperiod},
         given={M.\bibnamedelima B.},
         giveni={M\bibinitperiod\bibinitdelim B\bibinitperiod},
      }}%
      {{hash=BD}{%
         family={Baron},
         familyi={B\bibinitperiod},
         given={D.},
         giveni={D\bibinitperiod},
      }}%
      {{hash=BRG}{%
         family={Baraniuk},
         familyi={B\bibinitperiod},
         given={R.\bibnamedelima G.},
         giveni={R\bibinitperiod\bibinitdelim G\bibinitperiod},
      }}%
    }
    \strng{namehash}{DMFWMBBDBRG1}
    \strng{fullhash}{DMFWMBBDBRG1}
    \field{labelnamesource}{author}
    \field{labeltitlesource}{title}
    \field{booktitle}{Proc. IEEE Int. Conf. Inf. Process. Sensor Networks
  (IPSN)}
    \field{pages}{177\bibrangedash 185}
    \field{title}{{Universal distributed sensing via random projections}}
    \list{location}{1}{%
      {Nashville, TN}%
    }
    \field{month}{04}
    \field{year}{2006}
  \endentry

  \entry{HN05}{article}{}
    \name{author}{2}{}{%
      {{hash=HJ}{%
         family={Haupt},
         familyi={H\bibinitperiod},
         given={J.},
         giveni={J\bibinitperiod},
      }}%
      {{hash=NR}{%
         family={Nowak},
         familyi={N\bibinitperiod},
         given={R.},
         giveni={R\bibinitperiod},
      }}%
    }
    \strng{namehash}{HJNR1}
    \strng{fullhash}{HJNR1}
    \field{labelnamesource}{author}
    \field{labeltitlesource}{title}
    \field{number}{9}
    \field{pages}{4036\bibrangedash 4048}
    \field{title}{{S}ignal reconstruction from noisy random projections}
    \field{volume}{52}
    \field{journaltitle}{IEEE Trans. Inf. Theory}
    \field{month}{09}
    \field{year}{2006}
  \endentry

  \entry{BaronDCStech}{report}{}
    \name{author}{5}{}{%
      {{hash=BD}{%
         family={Baron},
         familyi={B\bibinitperiod},
         given={D.},
         giveni={D\bibinitperiod},
      }}%
      {{hash=WMB}{%
         family={Wakin},
         familyi={W\bibinitperiod},
         given={M.\bibnamedelima B.},
         giveni={M\bibinitperiod\bibinitdelim B\bibinitperiod},
      }}%
      {{hash=DMF}{%
         family={Duarte},
         familyi={D\bibinitperiod},
         given={M.\bibnamedelima F.},
         giveni={M\bibinitperiod\bibinitdelim F\bibinitperiod},
      }}%
      {{hash=SS}{%
         family={Sarvotham},
         familyi={S\bibinitperiod},
         given={S.},
         giveni={S\bibinitperiod},
      }}%
      {{hash=BRG}{%
         family={Baraniuk},
         familyi={B\bibinitperiod},
         given={R.\bibnamedelima G.},
         giveni={R\bibinitperiod\bibinitdelim G\bibinitperiod},
      }}%
    }
    \strng{namehash}{BDWMBDMFSSBRG1}
    \strng{fullhash}{BDWMBDMFSSBRG1}
    \field{labelnamesource}{author}
    \field{labeltitlesource}{title}
    \field{number}{ECE-0612}
    \field{title}{Distributed compressed sensing}
    \list{institution}{1}{%
      {Rice University}%
    }
    \field{type}{techreport}
    \field{month}{12}
    \field{year}{2006}
  \endentry

  \entry{chen2006trs}{article}{}
    \name{author}{2}{}{%
      {{hash=CJ}{%
         family={Chen},
         familyi={C\bibinitperiod},
         given={J.},
         giveni={J\bibinitperiod},
      }}%
      {{hash=HX}{%
         family={Huo},
         familyi={H\bibinitperiod},
         given={X.},
         giveni={X\bibinitperiod},
      }}%
    }
    \strng{namehash}{CJHX1}
    \strng{fullhash}{CJHX1}
    \field{labelnamesource}{author}
    \field{labeltitlesource}{title}
    \field{number}{12}
    \field{pages}{4634\bibrangedash 4643}
    \field{title}{Theoretical results on sparse representations of multiple
  measurement vectors}
    \field{volume}{54}
    \field{journaltitle}{IEEE Trans. Signal Process.}
    \field{month}{12}
    \field{year}{2006}
  \endentry

  \entry{cotter2005ssl}{article}{}
    \name{author}{4}{}{%
      {{hash=CSF}{%
         family={Cotter},
         familyi={C\bibinitperiod},
         given={S.\bibnamedelima F.},
         giveni={S\bibinitperiod\bibinitdelim F\bibinitperiod},
      }}%
      {{hash=RBD}{%
         family={Rao},
         familyi={R\bibinitperiod},
         given={B.\bibnamedelima D.},
         giveni={B\bibinitperiod\bibinitdelim D\bibinitperiod},
      }}%
      {{hash=EK}{%
         family={Engan},
         familyi={E\bibinitperiod},
         given={K.},
         giveni={K\bibinitperiod},
      }}%
      {{hash=KDK}{%
         family={Kreutz-Delgado},
         familyi={K\bibinitperiod-D\bibinitperiod},
         given={K.},
         giveni={K\bibinitperiod},
      }}%
    }
    \strng{namehash}{CSFRBDEKKDK1}
    \strng{fullhash}{CSFRBDEKKDK1}
    \field{labelnamesource}{author}
    \field{labeltitlesource}{title}
    \field{number}{7}
    \field{pages}{2477\bibrangedash 2488}
    \field{title}{Sparse solutions to linear inverse problems with multiple
  measurement vectors}
    \field{volume}{53}
    \field{journaltitle}{IEEE Trans. Signal Process.}
    \field{month}{07}
    \field{year}{2005}
  \endentry

  \entry{Mishali08rembo}{article}{}
    \name{author}{2}{}{%
      {{hash=MM}{%
         family={Mishali},
         familyi={M\bibinitperiod},
         given={M.},
         giveni={M\bibinitperiod},
      }}%
      {{hash=EYC}{%
         family={Eldar},
         familyi={E\bibinitperiod},
         given={Y.\bibnamedelima C.},
         giveni={Y\bibinitperiod\bibinitdelim C\bibinitperiod},
      }}%
    }
    \strng{namehash}{MMEYC1}
    \strng{fullhash}{MMEYC1}
    \field{labelnamesource}{author}
    \field{labeltitlesource}{title}
    \field{number}{10}
    \field{pages}{4692\bibrangedash 4702}
    \field{title}{Reduce and boost: {R}ecovering arbitrary sets of jointly
  sparse vectors}
    \field{volume}{56}
    \field{journaltitle}{IEEE Trans. Signal Process.}
    \field{month}{10}
    \field{year}{2009}
  \endentry

  \entry{Berg09jrmm}{article}{}
    \name{author}{2}{}{%
      {{hash=BE}{%
         family={Berg},
         familyi={B\bibinitperiod},
         given={E.},
         giveni={E\bibinitperiod},
      }}%
      {{hash=FMP}{%
         family={Friedlander},
         familyi={F\bibinitperiod},
         given={M.\bibnamedelima P.},
         giveni={M\bibinitperiod\bibinitdelim P\bibinitperiod},
      }}%
    }
    \strng{namehash}{BEFMP1}
    \strng{fullhash}{BEFMP1}
    \field{labelnamesource}{author}
    \field{labeltitlesource}{title}
    \field{title}{Joint-sparse recovery from multiple measurements}
    \field{journaltitle}{Arxiv preprint arXiv:0904.2051}
    \field{month}{04}
    \field{year}{2009}
  \endentry

  \entry{JuYeKi07}{article}{}
    \name{author}{3}{}{%
      {{hash=JH}{%
         family={Jung},
         familyi={J\bibinitperiod},
         given={H.},
         giveni={H\bibinitperiod},
      }}%
      {{hash=YJC}{%
         family={Ye},
         familyi={Y\bibinitperiod},
         given={J.\bibnamedelima C.},
         giveni={J\bibinitperiod\bibinitdelim C\bibinitperiod},
      }}%
      {{hash=KEY}{%
         family={Kim},
         familyi={K\bibinitperiod},
         given={E.\bibnamedelima Y.},
         giveni={E\bibinitperiod\bibinitdelim Y\bibinitperiod},
      }}%
    }
    \list{publisher}{1}{%
      {IOP PUBLISHING LTD}%
    }
    \strng{namehash}{JHYJCKEY1}
    \strng{fullhash}{JHYJCKEY1}
    \field{labelnamesource}{author}
    \field{labeltitlesource}{title}
    \field{number}{11}
    \field{pages}{3201\bibrangedash 3226}
    \field{title}{{Improved k-t {BLAST} and k-t {SENSE} using {FOCUSS}}}
    \field{volume}{52}
    \field{journaltitle}{Physics in Medicine and Biology}
    \field{month}{05}
    \field{year}{2007}
  \endentry

  \entry{JuSuNaKiYe09}{article}{}
    \name{author}{5}{}{%
      {{hash=JH}{%
         family={Jung},
         familyi={J\bibinitperiod},
         given={H.},
         giveni={H\bibinitperiod},
      }}%
      {{hash=SK}{%
         family={Sung},
         familyi={S\bibinitperiod},
         given={K.},
         giveni={K\bibinitperiod},
      }}%
      {{hash=NKS}{%
         family={Nayak},
         familyi={N\bibinitperiod},
         given={K.\bibnamedelima S.},
         giveni={K\bibinitperiod\bibinitdelim S\bibinitperiod},
      }}%
      {{hash=KEY}{%
         family={Kim},
         familyi={K\bibinitperiod},
         given={E.\bibnamedelima Y.},
         giveni={E\bibinitperiod\bibinitdelim Y\bibinitperiod},
      }}%
      {{hash=YJC}{%
         family={Ye},
         familyi={Y\bibinitperiod},
         given={J.\bibnamedelima C.},
         giveni={J\bibinitperiod\bibinitdelim C\bibinitperiod},
      }}%
    }
    \strng{namehash}{JHSKNKSKEYYJC1}
    \strng{fullhash}{JHSKNKSKEYYJC1}
    \field{labelnamesource}{author}
    \field{labeltitlesource}{title}
    \field{number}{1}
    \field{pages}{103\bibrangedash 116}
    \field{title}{{k-t {FOCUSS}: {A} general compressed sensing framework for
  high resolution dynamic {MRI}}}
    \field{volume}{61}
    \field{journaltitle}{J. Magnetic Resonance in Medicine}
    \field{month}{01}
    \field{year}{2009}
  \endentry

  \entry{LeeKimBreslerYe2011}{article}{}
    \name{author}{4}{}{%
      {{hash=LO}{%
         family={Lee},
         familyi={L\bibinitperiod},
         given={O.},
         giveni={O\bibinitperiod},
      }}%
      {{hash=KJM}{%
         family={Kim},
         familyi={K\bibinitperiod},
         given={J.\bibnamedelima M.},
         giveni={J\bibinitperiod\bibinitdelim M\bibinitperiod},
      }}%
      {{hash=BY}{%
         family={Bresler},
         familyi={B\bibinitperiod},
         given={Y.},
         giveni={Y\bibinitperiod},
      }}%
      {{hash=YJC}{%
         family={Ye},
         familyi={Y\bibinitperiod},
         given={J.\bibnamedelima C.},
         giveni={J\bibinitperiod\bibinitdelim C\bibinitperiod},
      }}%
    }
    \strng{namehash}{LOKJMBYYJC1}
    \strng{fullhash}{LOKJMBYYJC1}
    \field{labelnamesource}{author}
    \field{labeltitlesource}{title}
    \field{number}{5}
    \field{pages}{1129\bibrangedash 1142}
    \field{title}{Compressive diffuse optical tomography: {N}oniterative exact
  reconstruction using joint sparsity}
    \field{volume}{30}
    \field{journaltitle}{IEEE Trans. Medical Imaging}
    \field{month}{05}
    \field{year}{2011}
  \endentry

  \entry{GuoWang2008}{article}{}
    \name{author}{2}{}{%
      {{hash=GD}{%
         family={Guo},
         familyi={G\bibinitperiod},
         given={D.},
         giveni={D\bibinitperiod},
      }}%
      {{hash=WCC}{%
         family={Wang},
         familyi={W\bibinitperiod},
         given={C.\bibnamedelima C.},
         giveni={C\bibinitperiod\bibinitdelim C\bibinitperiod},
      }}%
    }
    \strng{namehash}{GDWCC1}
    \strng{fullhash}{GDWCC1}
    \field{labelnamesource}{author}
    \field{labeltitlesource}{title}
    \field{number}{3}
    \field{pages}{421\bibrangedash 431}
    \field{title}{Multiuser detection of sparsely spread {CDMA}}
    \field{volume}{26}
    \field{journaltitle}{IEEE J. Sel. Areas Commun.}
    \field{month}{04}
    \field{year}{2008}
  \endentry

  \entry{CandesRUP}{article}{}
    \name{author}{3}{}{%
      {{hash=CE}{%
         family={Cand\`{e}s},
         familyi={C\bibinitperiod},
         given={E.},
         giveni={E\bibinitperiod},
      }}%
      {{hash=RJ}{%
         family={Romberg},
         familyi={R\bibinitperiod},
         given={J.},
         giveni={J\bibinitperiod},
      }}%
      {{hash=TT}{%
         family={Tao},
         familyi={T\bibinitperiod},
         given={T.},
         giveni={T\bibinitperiod},
      }}%
    }
    \strng{namehash}{CERJTT1}
    \strng{fullhash}{CERJTT1}
    \field{labelnamesource}{author}
    \field{labeltitlesource}{title}
    \field{number}{2}
    \field{pages}{489\bibrangedash 509}
    \field{title}{Robust uncertainty principles: {E}xact signal reconstruction
  from highly incomplete frequency information}
    \field{volume}{52}
    \field{journaltitle}{IEEE Trans. Inf. Theory}
    \field{month}{02}
    \field{year}{2006}
  \endentry

  \entry{GPSR2007}{article}{}
    \name{author}{3}{}{%
      {{hash=FM}{%
         family={Figueiredo},
         familyi={F\bibinitperiod},
         given={M.},
         giveni={M\bibinitperiod},
      }}%
      {{hash=NR}{%
         family={Nowak},
         familyi={N\bibinitperiod},
         given={R.},
         giveni={R\bibinitperiod},
      }}%
      {{hash=WSJ}{%
         family={Wright},
         familyi={W\bibinitperiod},
         given={S.\bibnamedelima J.},
         giveni={S\bibinitperiod\bibinitdelim J\bibinitperiod},
      }}%
    }
    \strng{namehash}{FMNRWSJ1}
    \strng{fullhash}{FMNRWSJ1}
    \field{labelnamesource}{author}
    \field{labeltitlesource}{title}
    \field{number}{4}
    \field{pages}{586\bibrangedash 597}
    \field{title}{Gradient projection for sparse reconstruction: Application to
  compressed sensing and other inverse problems}
    \field{volume}{1}
    \field{journaltitle}{IEEE J. Sel. Topics Signal Proces.}
    \field{month}{12}
    \field{year}{2007}
  \endentry

  \entry{DMM2010ITW1}{inproceedings}{}
    \name{author}{3}{}{%
      {{hash=DDL}{%
         family={Donoho},
         familyi={D\bibinitperiod},
         given={D.\bibnamedelima L.},
         giveni={D\bibinitperiod\bibinitdelim L\bibinitperiod},
      }}%
      {{hash=MA}{%
         family={Maleki},
         familyi={M\bibinitperiod},
         given={A.},
         giveni={A\bibinitperiod},
      }}%
      {{hash=MA}{%
         family={Montanari},
         familyi={M\bibinitperiod},
         given={A.},
         giveni={A\bibinitperiod},
      }}%
    }
    \strng{namehash}{DDLMAMA1}
    \strng{fullhash}{DDLMAMA1}
    \field{labelnamesource}{author}
    \field{labeltitlesource}{title}
    \field{booktitle}{IEEE Inf. Theory Workshop}
    \field{title}{{Message passing algorithms for compressed sensing: I.
  {M}otivation and construction}}
    \field{month}{01}
    \field{year}{2010}
  \endentry

  \entry{RanganGAMP2011ISIT}{inproceedings}{}
    \name{author}{1}{}{%
      {{hash=RS}{%
         family={Rangan},
         familyi={R\bibinitperiod},
         given={S.},
         giveni={S\bibinitperiod},
      }}%
    }
    \strng{namehash}{RS1}
    \strng{fullhash}{RS1}
    \field{labelnamesource}{author}
    \field{labeltitlesource}{title}
    \field{booktitle}{Proc. IEEE Int. Symp. Inf. Theory (ISIT)}
    \field{pages}{2168\bibrangedash 2172}
    \field{title}{Generalized approximate message passing for estimation with
  random linear mixing}
    \list{location}{1}{%
      {St. Petersburg, Russia}%
    }
    \field{month}{07}
    \field{year}{2011}
  \endentry

  \entry{BCS2008}{article}{}
    \name{author}{3}{}{%
      {{hash=JS}{%
         family={Ji},
         familyi={J\bibinitperiod},
         given={S.},
         giveni={S\bibinitperiod},
      }}%
      {{hash=XY}{%
         family={Xue},
         familyi={X\bibinitperiod},
         given={Y.},
         giveni={Y\bibinitperiod},
      }}%
      {{hash=CL}{%
         family={Carin},
         familyi={C\bibinitperiod},
         given={L.},
         giveni={L\bibinitperiod},
      }}%
    }
    \strng{namehash}{JSXYCL1}
    \strng{fullhash}{JSXYCL1}
    \field{labelnamesource}{author}
    \field{labeltitlesource}{title}
    \field{number}{6}
    \field{pages}{2346\bibrangedash 2356}
    \field{title}{Bayesian compressive sensing}
    \field{volume}{56}
    \field{journaltitle}{{IEEE Trans. Signal Process.}}
    \field{month}{06}
    \field{year}{2008}
  \endentry

  \entry{BCSEx2008}{inproceedings}{}
    \name{author}{2}{}{%
      {{hash=SMW}{%
         family={Seeger},
         familyi={S\bibinitperiod},
         given={M.\bibnamedelima W.},
         giveni={M\bibinitperiod\bibinitdelim W\bibinitperiod},
      }}%
      {{hash=NH}{%
         family={Nickisch},
         familyi={N\bibinitperiod},
         given={H.},
         giveni={H\bibinitperiod},
      }}%
    }
    \strng{namehash}{SMWNH1}
    \strng{fullhash}{SMWNH1}
    \field{labelnamesource}{author}
    \field{labeltitlesource}{title}
    \field{booktitle}{Proc. Int. Conf. Mach. Learning}
    \field{pages}{912\bibrangedash 919}
    \field{title}{Compressed sensing and {B}ayesian experimental design}
    \list{location}{1}{%
      {Helsinki, Finland}%
    }
    \field{month}{08}
    \field{year}{2008}
  \endentry

  \entry{CSBP2010}{article}{}
    \name{author}{3}{}{%
      {{hash=BD}{%
         family={Baron},
         familyi={B\bibinitperiod},
         given={D.},
         giveni={D\bibinitperiod},
      }}%
      {{hash=SS}{%
         family={Sarvotham},
         familyi={S\bibinitperiod},
         given={S.},
         giveni={S\bibinitperiod},
      }}%
      {{hash=BRG}{%
         family={Baraniuk},
         familyi={B\bibinitperiod},
         given={R.\bibnamedelima G.},
         giveni={R\bibinitperiod\bibinitdelim G\bibinitperiod},
      }}%
    }
    \strng{namehash}{BDSSBRG1}
    \strng{fullhash}{BDSSBRG1}
    \field{labelnamesource}{author}
    \field{labeltitlesource}{title}
    \field{number}{1}
    \field{pages}{269\bibrangedash 280}
    \field{title}{Bayesian compressive sensing via belief propagation}
    \field{volume}{58}
    \field{journaltitle}{IEEE Trans. Signal Process.}
    \field{month}{01}
    \field{year}{2010}
  \endentry

  \entry{EMGMTSP}{article}{}
    \name{author}{2}{}{%
      {{hash=VJ}{%
         family={Vila},
         familyi={V\bibinitperiod},
         given={J.},
         giveni={J\bibinitperiod},
      }}%
      {{hash=SP}{%
         family={Schniter},
         familyi={S\bibinitperiod},
         given={P.},
         giveni={P\bibinitperiod},
      }}%
    }
    \strng{namehash}{VJSP1}
    \strng{fullhash}{VJSP1}
    \field{labelnamesource}{author}
    \field{labeltitlesource}{title}
    \field{number}{19}
    \field{pages}{4658\bibrangedash 4672}
    \field{title}{Expectation-maximization {G}aussian-mixture approximate
  message passing}
    \field{volume}{61}
    \field{journaltitle}{IEEE Trans. Signal Process.}
    \field{month}{10}
    \field{year}{2013}
  \endentry

  \entry{turboGAMP}{inproceedings}{}
    \name{author}{3}{}{%
      {{hash=ZJ}{%
         family={Ziniel},
         familyi={Z\bibinitperiod},
         given={J.},
         giveni={J\bibinitperiod},
      }}%
      {{hash=RS}{%
         family={Rangan},
         familyi={R\bibinitperiod},
         given={S.},
         giveni={S\bibinitperiod},
      }}%
      {{hash=SP}{%
         family={Schniter},
         familyi={S\bibinitperiod},
         given={P.},
         giveni={P\bibinitperiod},
      }}%
    }
    \strng{namehash}{ZJRSSP1}
    \strng{fullhash}{ZJRSSP1}
    \field{labelnamesource}{author}
    \field{labeltitlesource}{title}
    \field{booktitle}{Proc. IEEE Stat. Signal Process. Workshop (SSP)}
    \field{pages}{325\bibrangedash 328}
    \field{title}{A generalized framework for learning and recovery of
  structured sparse signals}
    \list{location}{1}{%
      {Ann Arbor, MI}%
    }
    \field{month}{08}
    \field{year}{2012}
  \endentry

  \entry{MTKB2014ITA}{article}{}
    \name{author}{4}{}{%
      {{hash=MY}{%
         family={Ma},
         familyi={M\bibinitperiod},
         given={Y.},
         giveni={Y\bibinitperiod},
      }}%
      {{hash=TJ}{%
         family={Tan},
         familyi={T\bibinitperiod},
         given={J.},
         giveni={J\bibinitperiod},
      }}%
      {{hash=KN}{%
         family={Krishnan},
         familyi={K\bibinitperiod},
         given={N.},
         giveni={N\bibinitperiod},
      }}%
      {{hash=BD}{%
         family={Baron},
         familyi={B\bibinitperiod},
         given={D.},
         giveni={D\bibinitperiod},
      }}%
    }
    \strng{namehash}{MYTJKNBD1}
    \strng{fullhash}{MYTJKNBD1}
    \field{labelnamesource}{author}
    \field{labeltitlesource}{title}
    \field{title}{Empirical {B}ayes and full {B}ayes for signal estimation}
    \field{journaltitle}{Arxiv preprint arxiv:1405.2113v1}
    \field{month}{05}
    \field{year}{2014}
  \endentry

  \entry{Figueiredo2003}{article}{}
    \name{author}{2}{}{%
      {{hash=FMAT}{%
         family={Figueiredo},
         familyi={F\bibinitperiod},
         given={M.\bibnamedelima A.\bibnamedelima T.},
         giveni={M\bibinitperiod\bibinitdelim A\bibinitperiod\bibinitdelim
  T\bibinitperiod},
      }}%
      {{hash=NRD}{%
         family={Nowak},
         familyi={N\bibinitperiod},
         given={R.\bibnamedelima D.},
         giveni={R\bibinitperiod\bibinitdelim D\bibinitperiod},
      }}%
    }
    \list{publisher}{1}{%
      {IEEE}%
    }
    \strng{namehash}{FMATNRD1}
    \strng{fullhash}{FMATNRD1}
    \field{labelnamesource}{author}
    \field{labeltitlesource}{title}
    \field{number}{8}
    \field{pages}{906\bibrangedash 916}
    \field{title}{An {EM} algorithm for wavelet-based image restoration}
    \field{volume}{12}
    \field{journaltitle}{IEEE Trans. Image Process.}
    \field{month}{08}
    \field{year}{2003}
  \endentry

  \entry{DonohoKolmogorovCS2006}{inproceedings}{}
    \name{author}{3}{}{%
      {{hash=DD}{%
         family={Donoho},
         familyi={D\bibinitperiod},
         given={D.},
         giveni={D\bibinitperiod},
      }}%
      {{hash=KH}{%
         family={Kakavand},
         familyi={K\bibinitperiod},
         given={H.},
         giveni={H\bibinitperiod},
      }}%
      {{hash=MJ}{%
         family={Mammen},
         familyi={M\bibinitperiod},
         given={J.},
         giveni={J\bibinitperiod},
      }}%
    }
    \strng{namehash}{DDKHMJ1}
    \strng{fullhash}{DDKHMJ1}
    \field{labelnamesource}{author}
    \field{labeltitlesource}{title}
    \field{booktitle}{Proc. IEEE Int. Symp. Inf. Theory (ISIT)}
    \field{pages}{1924\bibrangedash 1928}
    \field{title}{The simplest solution to an underdetermined system of linear
  equations}
    \list{location}{1}{%
      {Seattle, WA}%
    }
    \field{month}{07}
    \field{year}{2006}
  \endentry

  \entry{HN11}{incollection}{}
    \name{author}{2}{}{%
      {{hash=HJD}{%
         family={Haupt},
         familyi={H\bibinitperiod},
         given={J.\bibnamedelima D.},
         giveni={J\bibinitperiod\bibinitdelim D\bibinitperiod},
      }}%
      {{hash=NR}{%
         family={Nowak},
         familyi={N\bibinitperiod},
         given={R.},
         giveni={R\bibinitperiod},
      }}%
    }
    \list{publisher}{1}{%
      {Cambridge University Press}%
    }
    \strng{namehash}{HJDNR1}
    \strng{fullhash}{HJDNR1}
    \field{labelnamesource}{author}
    \field{labeltitlesource}{title}
    \field{booktitle}{Compressed Sensing: Theory and Applications}
    \field{title}{Adaptive sensing for sparse recovery}
    \field{year}{2012}
  \endentry

  \entry{Ramirez2011}{article}{}
    \name{author}{2}{}{%
      {{hash=RI}{%
         family={Ram{\'\i}rez},
         familyi={R\bibinitperiod},
         given={I.},
         giveni={I\bibinitperiod},
      }}%
      {{hash=SG}{%
         family={Sapiro},
         familyi={S\bibinitperiod},
         given={G.},
         giveni={G\bibinitperiod},
      }}%
    }
    \strng{namehash}{RISG1}
    \strng{fullhash}{RISG1}
    \field{labelnamesource}{author}
    \field{labeltitlesource}{title}
    \field{number}{6}
    \field{pages}{2913\bibrangedash 2927}
    \field{title}{An {MDL} framework for sparse coding and dictionary learning}
    \field{volume}{60}
    \field{journaltitle}{IEEE Trans. Signal Process.}
    \field{month}{06}
    \field{year}{2012}
  \endentry

  \entry{LiVitanyi2008}{book}{}
    \name{author}{2}{}{%
      {{hash=LM}{%
         family={Li},
         familyi={L\bibinitperiod},
         given={M.},
         giveni={M\bibinitperiod},
      }}%
      {{hash=VPMB}{%
         family={Vitanyi},
         familyi={V\bibinitperiod},
         given={P.\bibnamedelima M.\bibnamedelima B.},
         giveni={P\bibinitperiod\bibinitdelim M\bibinitperiod\bibinitdelim
  B\bibinitperiod},
      }}%
    }
    \list{publisher}{1}{%
      {Springer-Verlag, New York}%
    }
    \strng{namehash}{LMVPMB1}
    \strng{fullhash}{LMVPMB1}
    \field{labelnamesource}{author}
    \field{labeltitlesource}{title}
    \field{title}{An {I}ntroduction to {K}olmogorov {C}omplexity and {I}ts
  {A}pplications}
    \field{year}{2008}
  \endentry

  \entry{AharoEB_KSVD}{article}{}
    \name{author}{3}{}{%
      {{hash=AM}{%
         family={Aharon},
         familyi={A\bibinitperiod},
         given={M.},
         giveni={M\bibinitperiod},
      }}%
      {{hash=EM}{%
         family={Elad},
         familyi={E\bibinitperiod},
         given={M.},
         giveni={M\bibinitperiod},
      }}%
      {{hash=BA}{%
         family={Bruckstein},
         familyi={B\bibinitperiod},
         given={A.},
         giveni={A\bibinitperiod},
      }}%
    }
    \strng{namehash}{AMEMBA1}
    \strng{fullhash}{AMEMBA1}
    \field{labelnamesource}{author}
    \field{labeltitlesource}{title}
    \field{number}{11}
    \field{pages}{4311\bibrangedash 4322}
    \field{title}{{K-SVD}: An algorithm for designing overcomplete dictionaries
  for sparse representation}
    \field{volume}{54}
    \field{journaltitle}{{IEEE} Trans. Signal Process.}
    \field{month}{11}
    \field{year}{2006}
  \endentry

  \entry{Mairal2008}{inproceedings}{}
    \name{author}{5}{}{%
      {{hash=MJ}{%
         family={Mairal},
         familyi={M\bibinitperiod},
         given={J.},
         giveni={J\bibinitperiod},
      }}%
      {{hash=BF}{%
         family={Bach},
         familyi={B\bibinitperiod},
         given={F.},
         giveni={F\bibinitperiod},
      }}%
      {{hash=PJ}{%
         family={Ponce},
         familyi={P\bibinitperiod},
         given={J.},
         giveni={J\bibinitperiod},
      }}%
      {{hash=SG}{%
         family={Sapiro},
         familyi={S\bibinitperiod},
         given={G.},
         giveni={G\bibinitperiod},
      }}%
      {{hash=ZA}{%
         family={Zisserman},
         familyi={Z\bibinitperiod},
         given={A.},
         giveni={A\bibinitperiod},
      }}%
    }
    \strng{namehash}{MJBFPJSGZA1}
    \strng{fullhash}{MJBFPJSGZA1}
    \field{labelnamesource}{author}
    \field{labeltitlesource}{title}
    \field{booktitle}{Workshop Neural Inf. Process. Syst. (NIPS)}
    \field{title}{Supervised dictionary learning}
    \list{location}{1}{%
      {Vancouver, B.C., Canada}%
    }
    \field{month}{12}
    \field{year}{2008}
  \endentry

  \entry{Zhoul2011}{article}{}
    \name{author}{9}{}{%
      {{hash=ZM}{%
         family={Zhou},
         familyi={Z\bibinitperiod},
         given={M.},
         giveni={M\bibinitperiod},
      }}%
      {{hash=CH}{%
         family={Chen},
         familyi={C\bibinitperiod},
         given={H.},
         giveni={H\bibinitperiod},
      }}%
      {{hash=PJ}{%
         family={Paisley},
         familyi={P\bibinitperiod},
         given={J.},
         giveni={J\bibinitperiod},
      }}%
      {{hash=RL}{%
         family={Ren},
         familyi={R\bibinitperiod},
         given={L.},
         giveni={L\bibinitperiod},
      }}%
      {{hash=LL}{%
         family={Li},
         familyi={L\bibinitperiod},
         given={L.},
         giveni={L\bibinitperiod},
      }}%
      {{hash=XZ}{%
         family={Xing},
         familyi={X\bibinitperiod},
         given={Z.},
         giveni={Z\bibinitperiod},
      }}%
      {{hash=DD}{%
         family={Dunson},
         familyi={D\bibinitperiod},
         given={D.},
         giveni={D\bibinitperiod},
      }}%
      {{hash=SG}{%
         family={Sapiro},
         familyi={S\bibinitperiod},
         given={G.},
         giveni={G\bibinitperiod},
      }}%
      {{hash=CL}{%
         family={Carin},
         familyi={C\bibinitperiod},
         given={L.},
         giveni={L\bibinitperiod},
      }}%
    }
    \strng{namehash}{ZMCHPJRLLLXZDDSGCL1}
    \strng{fullhash}{ZMCHPJRLLLXZDDSGCL1}
    \field{labelnamesource}{author}
    \field{labeltitlesource}{title}
    \field{number}{1}
    \field{pages}{130\bibrangedash 144}
    \field{title}{Nonparametric {B}ayesian dictionary learning for analysis of
  noisy and incomplete images}
    \field{volume}{21}
    \field{journaltitle}{IEEE Trans. Image Process.}
    \field{month}{01}
    \field{year}{2012}
  \endentry

  \entry{Garrigues07learninghorizontal}{inproceedings}{}
    \name{author}{2}{}{%
      {{hash=GPJ}{%
         family={Garrigues},
         familyi={G\bibinitperiod},
         given={P.\bibnamedelima J.},
         giveni={P\bibinitperiod\bibinitdelim J\bibinitperiod},
      }}%
      {{hash=OBA}{%
         family={Olshausen},
         familyi={O\bibinitperiod},
         given={B.\bibnamedelima A.},
         giveni={B\bibinitperiod\bibinitdelim A\bibinitperiod},
      }}%
    }
    \strng{namehash}{GPJOBA1}
    \strng{fullhash}{GPJOBA1}
    \field{labelnamesource}{author}
    \field{labeltitlesource}{title}
    \field{booktitle}{Workshop Neural Inf. Process. Syst. (NIPS)}
    \field{pages}{1\bibrangedash 8}
    \field{title}{Learning horizontal connections in a sparse coding model of
  natural images}
    \list{location}{1}{%
      {Vancouver, B.C., Canada}%
    }
    \field{month}{12}
    \field{year}{2007}
  \endentry

  \entry{CodedMapReduce2015Allerton}{inproceedings}{}
    \name{author}{3}{}{%
      {{hash=LS}{%
         family={Li},
         familyi={L\bibinitperiod},
         given={S.},
         giveni={S\bibinitperiod},
      }}%
      {{hash=MAMA}{%
         family={Maddah-Ali},
         familyi={M\bibinitperiod-A\bibinitperiod},
         given={M.\bibnamedelima A.},
         giveni={M\bibinitperiod\bibinitdelim A\bibinitperiod},
      }}%
      {{hash=AAS}{%
         family={Avestimehr},
         familyi={A\bibinitperiod},
         given={A.\bibnamedelima S.},
         giveni={A\bibinitperiod\bibinitdelim S\bibinitperiod},
      }}%
    }
    \strng{namehash}{LSMAMAAAS1}
    \strng{fullhash}{LSMAMAAAS1}
    \field{labelnamesource}{author}
    \field{labeltitlesource}{title}
    \field{booktitle}{Proc. Allerton Conference Commun., Control, and Comput.}
    \field{pages}{964\bibrangedash 971}
    \field{title}{Coded {M}ap{R}educe}
    \field{month}{09}
    \field{year}{2015}
  \endentry

  \entry{LMYA2016ISIT}{inproceedings}{}
    \name{author}{4}{}{%
      {{hash=LS}{%
         family={Li},
         familyi={L\bibinitperiod},
         given={S.},
         giveni={S\bibinitperiod},
      }}%
      {{hash=MAMA}{%
         family={Maddah-Ali},
         familyi={M\bibinitperiod-A\bibinitperiod},
         given={M.\bibnamedelima A.},
         giveni={M\bibinitperiod\bibinitdelim A\bibinitperiod},
      }}%
      {{hash=YQ}{%
         family={Yu},
         familyi={Y\bibinitperiod},
         given={Qian},
         giveni={Q\bibinitperiod},
      }}%
      {{hash=AAS}{%
         family={Avestimehr},
         familyi={A\bibinitperiod},
         given={A.\bibnamedelima S.},
         giveni={A\bibinitperiod\bibinitdelim S\bibinitperiod},
      }}%
    }
    \strng{namehash}{LSMAMAYQAAS1}
    \strng{fullhash}{LSMAMAYQAAS1}
    \field{labelnamesource}{author}
    \field{labeltitlesource}{title}
    \field{booktitle}{Proc. IEEE Int. Symp. Inf. Theory (ISIT)}
    \field{pages}{1814\bibrangedash 1818}
    \field{title}{Fundamental tradeoff between computation and communication in
  distributed computing}
    \list{location}{1}{%
      {Barcelona, Spain}%
    }
    \field{month}{07}
    \field{year}{2016}
  \endentry

  \entry{Thanou2013}{article}{}
    \name{author}{4}{}{%
      {{hash=TD}{%
         family={Thanou},
         familyi={T\bibinitperiod},
         given={D.},
         giveni={D\bibinitperiod},
      }}%
      {{hash=KE}{%
         family={Kokiopoulou},
         familyi={K\bibinitperiod},
         given={E.},
         giveni={E\bibinitperiod},
      }}%
      {{hash=PY}{%
         family={Pu},
         familyi={P\bibinitperiod},
         given={Y.},
         giveni={Y\bibinitperiod},
      }}%
      {{hash=FP}{%
         family={Frossard},
         familyi={F\bibinitperiod},
         given={P.},
         giveni={P\bibinitperiod},
      }}%
    }
    \strng{namehash}{TDKEPYFP1}
    \strng{fullhash}{TDKEPYFP1}
    \field{labelnamesource}{author}
    \field{labeltitlesource}{title}
    \field{number}{1}
    \field{pages}{194\bibrangedash 205}
    \field{title}{Distributed average consensus with quantization refinement}
    \field{volume}{61}
    \field{journaltitle}{IEEE Trans. Signal Process.}
    \field{month}{01}
    \field{year}{2013}
  \endentry

  \entry{Tan2014}{article}{}
    \name{author}{3}{}{%
      {{hash=TJ}{%
         family={Tan},
         familyi={T\bibinitperiod},
         given={J.},
         giveni={J\bibinitperiod},
      }}%
      {{hash=CD}{%
         family={Carmon},
         familyi={C\bibinitperiod},
         given={D.},
         giveni={D\bibinitperiod},
      }}%
      {{hash=BD}{%
         family={Baron},
         familyi={B\bibinitperiod},
         given={D.},
         giveni={D\bibinitperiod},
      }}%
    }
    \strng{namehash}{TJCDBD1}
    \strng{fullhash}{TJCDBD1}
    \field{labelnamesource}{author}
    \field{labeltitlesource}{title}
    \field{number}{1}
    \field{pages}{150\bibrangedash 158}
    \field{title}{Signal estimation with additive error metrics in compressed
  sensing}
    \field{volume}{60}
    \field{journaltitle}{IEEE Trans. Inf. Theory}
    \field{month}{01}
    \field{year}{2014}
  \endentry

  \entry{Tan2014Infty}{article}{}
    \name{author}{3}{}{%
      {{hash=TJ}{%
         family={Tan},
         familyi={T\bibinitperiod},
         given={J.},
         giveni={J\bibinitperiod},
      }}%
      {{hash=BD}{%
         family={Baron},
         familyi={B\bibinitperiod},
         given={D.},
         giveni={D\bibinitperiod},
      }}%
      {{hash=DL}{%
         family={Dai},
         familyi={D\bibinitperiod},
         given={L.},
         giveni={L\bibinitperiod},
      }}%
    }
    \strng{namehash}{TJBDDL1}
    \strng{fullhash}{TJBDDL1}
    \field{labelnamesource}{author}
    \field{labeltitlesource}{title}
    \field{number}{10}
    \field{pages}{6626\bibrangedash 6635}
    \field{title}{Wiener filters in {G}aussian mixture signal estimation
  with~{$\ell_\infty$}-norm error}
    \field{volume}{60}
    \field{journaltitle}{IEEE Trans. Inf. Theory}
    \field{month}{10}
    \field{year}{2014}
  \endentry

  \entry{Berger71}{book}{}
    \name{author}{1}{}{%
      {{hash=BT}{%
         family={Berger},
         familyi={B\bibinitperiod},
         given={T.},
         giveni={T\bibinitperiod},
      }}%
    }
    \list{publisher}{1}{%
      {Prentice-Hall Englewood Cliffs, NJ}%
    }
    \strng{namehash}{BT1}
    \strng{fullhash}{BT1}
    \field{labelnamesource}{author}
    \field{labeltitlesource}{title}
    \field{pages}{xiii, 311 p.}
    \field{title}{Rate {D}istortion {T}heory: {M}athematical {B}asis for {D}ata
  {C}ompression}
    \field{type}{Book}
    \field{year}{1971}
  \endentry

  \entry{GershoGray1993}{book}{}
    \name{author}{2}{}{%
      {{hash=GA}{%
         family={Gersho},
         familyi={G\bibinitperiod},
         given={A.},
         giveni={A\bibinitperiod},
      }}%
      {{hash=GRM}{%
         family={Gray},
         familyi={G\bibinitperiod},
         given={R.\bibnamedelima M.},
         giveni={R\bibinitperiod\bibinitdelim M\bibinitperiod},
      }}%
    }
    \list{publisher}{1}{%
      {Kluwer}%
    }
    \strng{namehash}{GAGRM1}
    \strng{fullhash}{GAGRM1}
    \field{labelnamesource}{author}
    \field{labeltitlesource}{title}
    \field{title}{{Vector {Q}uantization and {S}ignal {C}ompression}}
    \field{year}{1993}
  \endentry

  \entry{WeidmannVetterli2012}{article}{}
    \name{author}{2}{}{%
      {{hash=WC}{%
         family={Weidmann},
         familyi={W\bibinitperiod},
         given={C.},
         giveni={C\bibinitperiod},
      }}%
      {{hash=VM}{%
         family={Vetterli},
         familyi={V\bibinitperiod},
         given={M.},
         giveni={M\bibinitperiod},
      }}%
    }
    \strng{namehash}{WCVM1}
    \strng{fullhash}{WCVM1}
    \field{labelnamesource}{author}
    \field{labeltitlesource}{title}
    \field{number}{8}
    \field{pages}{4969\bibrangedash 4992}
    \field{title}{Rate distortion behavior of sparse sources}
    \field{volume}{58}
    \field{journaltitle}{IEEE Trans. Inf. Theory}
    \field{month}{08}
    \field{year}{2012}
  \endentry

  \entry{DMM2009}{article}{}
    \name{author}{3}{}{%
      {{hash=DDL}{%
         family={Donoho},
         familyi={D\bibinitperiod},
         given={D.\bibnamedelima L.},
         giveni={D\bibinitperiod\bibinitdelim L\bibinitperiod},
      }}%
      {{hash=MA}{%
         family={Maleki},
         familyi={M\bibinitperiod},
         given={A.},
         giveni={A\bibinitperiod},
      }}%
      {{hash=MA}{%
         family={Montanari},
         familyi={M\bibinitperiod},
         given={A.},
         giveni={A\bibinitperiod},
      }}%
    }
    \strng{namehash}{DDLMAMA1}
    \strng{fullhash}{DDLMAMA1}
    \field{labelnamesource}{author}
    \field{labeltitlesource}{title}
    \field{number}{45}
    \field{pages}{18914\bibrangedash 18919}
    \field{title}{{Message passing algorithms for compressed sensing}}
    \field{volume}{106}
    \field{journaltitle}{Proc. Nat. Academy Sci.}
    \field{month}{11}
    \field{year}{2009}
  \endentry

  \entry{Bayati2011}{article}{}
    \name{author}{2}{}{%
      {{hash=BM}{%
         family={Bayati},
         familyi={B\bibinitperiod},
         given={M.},
         giveni={M\bibinitperiod},
      }}%
      {{hash=MA}{%
         family={Montanari},
         familyi={M\bibinitperiod},
         given={A.},
         giveni={A\bibinitperiod},
      }}%
    }
    \list{publisher}{1}{%
      {IEEE}%
    }
    \strng{namehash}{BMMA1}
    \strng{fullhash}{BMMA1}
    \field{labelnamesource}{author}
    \field{labeltitlesource}{title}
    \field{number}{2}
    \field{pages}{764\bibrangedash 785}
    \field{title}{The dynamics of message passing on dense graphs, with
  applications to compressed sensing}
    \field{volume}{57}
    \field{journaltitle}{IEEE Trans. Inf. Theory}
    \field{month}{02}
    \field{year}{2011}
  \endentry

  \entry{Montanari2012}{article}{}
    \name{author}{1}{}{%
      {{hash=MA}{%
         family={Montanari},
         familyi={M\bibinitperiod},
         given={A.},
         giveni={A\bibinitperiod},
      }}%
    }
    \list{publisher}{1}{%
      {Cambridge University Press}%
    }
    \strng{namehash}{MA1}
    \strng{fullhash}{MA1}
    \field{labelnamesource}{author}
    \field{labeltitlesource}{title}
    \field{pages}{394\bibrangedash 438}
    \field{title}{Graphical models concepts in compressed sensing}
    \field{journaltitle}{Compressed Sensing: Theory and Applications}
    \field{year}{2012}
  \endentry

  \entry{MaZhuBaronAllerton2014}{inproceedings}{}
    \name{author}{3}{}{%
      {{hash=MY}{%
         family={Ma},
         familyi={M\bibinitperiod},
         given={Y.},
         giveni={Y\bibinitperiod},
      }}%
      {{hash=ZJ}{%
         family={Zhu},
         familyi={Z\bibinitperiod},
         given={J.},
         giveni={J\bibinitperiod},
      }}%
      {{hash=BD}{%
         family={Baron},
         familyi={B\bibinitperiod},
         given={D.},
         giveni={D\bibinitperiod},
      }}%
    }
    \strng{namehash}{MYZJBD1}
    \strng{fullhash}{MYZJBD1}
    \field{labelnamesource}{author}
    \field{labeltitlesource}{title}
    \field{booktitle}{Proc. Allerton Conference Commun., Control, and Comput.}
    \field{title}{Compressed sensing via universal denoising and approximate
  message passing}
    \field{month}{10}
    \field{year}{2014}
  \endentry

  \entry{MaZhuBaron2016TSP}{article}{}
    \name{author}{3}{}{%
      {{hash=MY}{%
         family={Ma},
         familyi={M\bibinitperiod},
         given={Y.},
         giveni={Y\bibinitperiod},
      }}%
      {{hash=ZJ}{%
         family={Zhu},
         familyi={Z\bibinitperiod},
         given={J.},
         giveni={J\bibinitperiod},
      }}%
      {{hash=BD}{%
         family={Baron},
         familyi={B\bibinitperiod},
         given={D.},
         giveni={D\bibinitperiod},
      }}%
    }
    \strng{namehash}{MYZJBD1}
    \strng{fullhash}{MYZJBD1}
    \field{labelnamesource}{author}
    \field{labeltitlesource}{title}
    \field{number}{21}
    \field{pages}{5611\bibrangedash 5622}
    \field{title}{Approximate message passing algorithm with universal
  denoising and {G}aussian mixture learning}
    \field{volume}{65}
    \field{journaltitle}{IEEE Trans. Signal Process.}
    \field{month}{11}
    \field{year}{2016}
  \endentry

  \entry{ZhuBaronCISS2013}{inproceedings}{}
    \name{author}{2}{}{%
      {{hash=ZJ}{%
         family={Zhu},
         familyi={Z\bibinitperiod},
         given={J.},
         giveni={J\bibinitperiod},
      }}%
      {{hash=BD}{%
         family={Baron},
         familyi={B\bibinitperiod},
         given={D.},
         giveni={D\bibinitperiod},
      }}%
    }
    \strng{namehash}{ZJBD1}
    \strng{fullhash}{ZJBD1}
    \field{labelnamesource}{author}
    \field{labeltitlesource}{title}
    \field{booktitle}{Proc. IEEE Conf. Inf. Sci. Syst. (CISS)}
    \field{title}{Performance regions in compressed sensing from noisy
  measurements}
    \list{location}{1}{%
      {Baltimore, MD}%
    }
    \field{month}{03}
    \field{year}{2013}
  \endentry

  \entry{ZhuBaronKrzakala2016}{article}{}
    \name{author}{3}{}{%
      {{hash=ZJ}{%
         family={Zhu},
         familyi={Z\bibinitperiod},
         given={J.},
         giveni={J\bibinitperiod},
      }}%
      {{hash=BD}{%
         family={Baron},
         familyi={B\bibinitperiod},
         given={D.},
         giveni={D\bibinitperiod},
      }}%
      {{hash=KF}{%
         family={Krzakala},
         familyi={K\bibinitperiod},
         given={F.},
         giveni={F\bibinitperiod},
      }}%
    }
    \strng{namehash}{ZJBDKF1}
    \strng{fullhash}{ZJBDKF1}
    \field{labelnamesource}{author}
    \field{labeltitlesource}{title}
    \field{number}{9}
    \field{pages}{2444\bibrangedash 2454}
    \field{title}{Performance limits for noisy multimeasurement vector
  problems}
    \field{volume}{65}
    \field{journaltitle}{IEEE Trans. Signal Process.}
    \field{month}{05}
    \field{year}{2017}
  \endentry

  \entry{ZhuBeiramiBaron2016ISIT}{inproceedings}{}
    \name{author}{3}{}{%
      {{hash=ZJ}{%
         family={Zhu},
         familyi={Z\bibinitperiod},
         given={J.},
         giveni={J\bibinitperiod},
      }}%
      {{hash=BA}{%
         family={Beirami},
         familyi={B\bibinitperiod},
         given={A.},
         giveni={A\bibinitperiod},
      }}%
      {{hash=BD}{%
         family={Baron},
         familyi={B\bibinitperiod},
         given={D.},
         giveni={D\bibinitperiod},
      }}%
    }
    \strng{namehash}{ZJBABD1}
    \strng{fullhash}{ZJBABD1}
    \field{labelnamesource}{author}
    \field{labeltitlesource}{title}
    \field{booktitle}{Proc. IEEE Int. Symp. Inf. Theory (ISIT)}
    \field{pages}{680\bibrangedash 684}
    \field{title}{Performance trade-offs in multi-processor approximate message
  passing}
    \list{location}{1}{%
      {Barcelona, Spain}%
    }
    \field{month}{07}
    \field{year}{2016}
  \endentry

  \entry{ZhuBaronMPAMP2016ArXiv}{article}{}
    \name{author}{3}{}{%
      {{hash=ZJ}{%
         family={Zhu},
         familyi={Z\bibinitperiod},
         given={J.},
         giveni={J\bibinitperiod},
      }}%
      {{hash=BD}{%
         family={Baron},
         familyi={B\bibinitperiod},
         given={D.},
         giveni={D\bibinitperiod},
      }}%
      {{hash=BA}{%
         family={Beirami},
         familyi={B\bibinitperiod},
         given={A.},
         giveni={A\bibinitperiod},
      }}%
    }
    \strng{namehash}{ZJBDBA1}
    \strng{fullhash}{ZJBDBA1}
    \field{labelnamesource}{author}
    \field{labeltitlesource}{title}
    \field{title}{Optimal trade-offs in multi-processor approximate message
  passing}
    \field{journaltitle}{Arxiv preprint arXiv:1601.03790}
    \field{month}{11}
    \field{year}{2016}
  \endentry

  \entry{JZ2014SSP}{inproceedings}{}
    \name{author}{3}{}{%
      {{hash=ZJ}{%
         family={Zhu},
         familyi={Z\bibinitperiod},
         given={J.},
         giveni={J\bibinitperiod},
      }}%
      {{hash=BD}{%
         family={Baron},
         familyi={B\bibinitperiod},
         given={D.},
         giveni={D\bibinitperiod},
      }}%
      {{hash=DMF}{%
         family={Duarte},
         familyi={D\bibinitperiod},
         given={M.\bibnamedelima F.},
         giveni={M\bibinitperiod\bibinitdelim F\bibinitperiod},
      }}%
    }
    \strng{namehash}{ZJBDDMF1}
    \strng{fullhash}{ZJBDDMF1}
    \field{labelnamesource}{author}
    \field{labeltitlesource}{title}
    \field{booktitle}{Proc. IEEE Stat. Signal Process. Workshop (SSP)}
    \field{pages}{416\bibrangedash 419}
    \field{title}{Complexity--adaptive universal signal estimation for
  compressed sensing}
    \list{location}{1}{%
      {Gold Coast, Australia}%
    }
    \field{month}{06}
    \field{year}{2014}
  \endentry

  \entry{ZhuBaronDuarte2014_SLAM}{article}{}
    \name{author}{3}{}{%
      {{hash=ZJ}{%
         family={Zhu},
         familyi={Z\bibinitperiod},
         given={J.},
         giveni={J\bibinitperiod},
      }}%
      {{hash=BD}{%
         family={Baron},
         familyi={B\bibinitperiod},
         given={D.},
         giveni={D\bibinitperiod},
      }}%
      {{hash=DMF}{%
         family={Duarte},
         familyi={D\bibinitperiod},
         given={M.\bibnamedelima F.},
         giveni={M\bibinitperiod\bibinitdelim F\bibinitperiod},
      }}%
    }
    \strng{namehash}{ZJBDDMF1}
    \strng{fullhash}{ZJBDDMF1}
    \field{labelnamesource}{author}
    \field{labeltitlesource}{title}
    \field{number}{6}
    \field{pages}{1512\bibrangedash 1527}
    \field{title}{Recovery from linear measurements with complexity--matching
  universal signal estimation}
    \field{volume}{63}
    \field{journaltitle}{IEEE Trans. Signal Process.}
    \field{month}{03}
    \field{year}{2015}
  \endentry

  \entry{NeriBook}{article}{}
    \name{author}{1}{}{%
      {{hash=MN}{%
         family={Merhav},
         familyi={M\bibinitperiod},
         given={Neri},
         giveni={N\bibinitperiod},
      }}%
    }
    \strng{namehash}{MN1}
    \strng{fullhash}{MN1}
    \field{labelnamesource}{author}
    \field{labeltitlesource}{title}
    \field{number}{1--2}
    \field{pages}{1\bibrangedash 212}
    \field{title}{Statistical physics and information theory}
    \field{volume}{6}
    \field{journaltitle}{Foundations and Trends in Communications and
  Information Theory}
    \field{year}{2010}
  \endentry

  \entry{LBG1980}{article}{}
    \name{author}{3}{}{%
      {{hash=LY}{%
         family={Linde},
         familyi={L\bibinitperiod},
         given={Y.},
         giveni={Y\bibinitperiod},
      }}%
      {{hash=BA}{%
         family={Buzo},
         familyi={B\bibinitperiod},
         given={A.},
         giveni={A\bibinitperiod},
      }}%
      {{hash=GRM}{%
         family={Gray},
         familyi={G\bibinitperiod},
         given={R.\bibnamedelima M.},
         giveni={R\bibinitperiod\bibinitdelim M\bibinitperiod},
      }}%
    }
    \strng{namehash}{LYBAGRM1}
    \strng{fullhash}{LYBAGRM1}
    \field{labelnamesource}{author}
    \field{labeltitlesource}{title}
    \field{number}{1}
    \field{pages}{84\bibrangedash 95}
    \field{title}{An algorithm for vector quantizer design}
    \field{volume}{28}
    \field{journaltitle}{IEEE Trans. Commun.}
    \field{month}{01}
    \field{year}{1980}
  \endentry

  \entry{Gray1984}{article}{}
    \name{author}{1}{}{%
      {{hash=GRM}{%
         family={Gray},
         familyi={G\bibinitperiod},
         given={R.\bibnamedelima M.},
         giveni={R\bibinitperiod\bibinitdelim M\bibinitperiod},
      }}%
    }
    \strng{namehash}{GRM1}
    \strng{fullhash}{GRM1}
    \field{labelnamesource}{author}
    \field{labeltitlesource}{title}
    \field{number}{2}
    \field{pages}{4\bibrangedash 29}
    \field{title}{Vector quantization}
    \field{volume}{1}
    \field{journaltitle}{IEEE ASSP Mag.}
    \field{month}{04}
    \field{year}{1984}
  \endentry

  \entry{FunctionalAnalysisBook}{book}{}
    \name{author}{1}{}{%
      {{hash=KE}{%
         family={Kreyszig},
         familyi={K\bibinitperiod},
         given={Erwin},
         giveni={E\bibinitperiod},
      }}%
    }
    \list{publisher}{1}{%
      {Wiley}%
    }
    \strng{namehash}{KE1}
    \strng{fullhash}{KE1}
    \field{labelnamesource}{author}
    \field{labeltitlesource}{title}
    \field{title}{Introductory Functional Analysis with Applications}
    \field{year}{1989}
  \endentry

  \entry{Arimoto72}{article}{}
    \name{author}{1}{}{%
      {{hash=AS}{%
         family={Arimoto},
         familyi={A\bibinitperiod},
         given={S.},
         giveni={S\bibinitperiod},
      }}%
    }
    \strng{namehash}{AS1}
    \strng{fullhash}{AS1}
    \field{labelnamesource}{author}
    \field{labeltitlesource}{title}
    \field{number}{1}
    \field{pages}{14\bibrangedash 20}
    \field{title}{An algorithm for computing the capacity of an arbitrary
  discrete memoryless channel}
    \field{volume}{18}
    \field{journaltitle}{IEEE Trans. Inf. Theory}
    \field{month}{01}
    \field{year}{1972}
  \endentry

  \entry{Blahut72}{article}{}
    \name{author}{1}{}{%
      {{hash=BRE}{%
         family={Blahut},
         familyi={B\bibinitperiod},
         given={R.\bibnamedelima E.},
         giveni={R\bibinitperiod\bibinitdelim E\bibinitperiod},
      }}%
    }
    \strng{namehash}{BRE1}
    \strng{fullhash}{BRE1}
    \field{labelnamesource}{author}
    \field{labeltitlesource}{title}
    \field{number}{4}
    \field{pages}{460\bibrangedash 473}
    \field{title}{Computation of channel capacity and rate-distortion
  functions}
    \field{volume}{18}
    \field{journaltitle}{IEEE Trans. Inf. Theory}
    \field{month}{07}
    \field{year}{1972}
  \endentry

  \entry{Rose94}{article}{}
    \name{author}{1}{}{%
      {{hash=RK}{%
         family={Rose},
         familyi={R\bibinitperiod},
         given={K.},
         giveni={K\bibinitperiod},
      }}%
    }
    \strng{namehash}{RK1}
    \strng{fullhash}{RK1}
    \field{labelnamesource}{author}
    \field{labeltitlesource}{title}
    \field{number}{6}
    \field{pages}{1939\bibrangedash 1952}
    \field{title}{A mapping approach to rate-distortion computation and
  analysis}
    \field{volume}{40}
    \field{journaltitle}{IEEE Trans. Inf. Theory}
    \field{month}{11}
    \field{year}{1994}
  \endentry

  \entry{DuarteWakinBaronSarvothamBaraniuk2013}{article}{}
    \name{author}{5}{}{%
      {{hash=DMF}{%
         family={Duarte},
         familyi={D\bibinitperiod},
         given={M.\bibnamedelima F.},
         giveni={M\bibinitperiod\bibinitdelim F\bibinitperiod},
      }}%
      {{hash=WMB}{%
         family={Wakin},
         familyi={W\bibinitperiod},
         given={M.\bibnamedelima B.},
         giveni={M\bibinitperiod\bibinitdelim B\bibinitperiod},
      }}%
      {{hash=BD}{%
         family={Baron},
         familyi={B\bibinitperiod},
         given={D.},
         giveni={D\bibinitperiod},
      }}%
      {{hash=SS}{%
         family={Sarvotham},
         familyi={S\bibinitperiod},
         given={S.},
         giveni={S\bibinitperiod},
      }}%
      {{hash=BRG}{%
         family={Baraniuk},
         familyi={B\bibinitperiod},
         given={R.\bibnamedelima G.},
         giveni={R\bibinitperiod\bibinitdelim G\bibinitperiod},
      }}%
    }
    \strng{namehash}{DMFWMBBDSSBRG1}
    \strng{fullhash}{DMFWMBBDSSBRG1}
    \field{labelnamesource}{author}
    \field{labeltitlesource}{title}
    \field{number}{7}
    \field{pages}{4280\bibrangedash 4289}
    \field{title}{Measurement bounds for sparse signal ensembles via graphical
  models}
    \field{volume}{59}
    \field{journaltitle}{IEEE Trans. Inf. Theory}
    \field{month}{07}
    \field{year}{2013}
  \endentry

  \entry{tropp2006ass}{article}{}
    \name{author}{3}{}{%
      {{hash=TJA}{%
         family={Tropp},
         familyi={T\bibinitperiod},
         given={J.\bibnamedelima A.},
         giveni={J\bibinitperiod\bibinitdelim A\bibinitperiod},
      }}%
      {{hash=GAC}{%
         family={Gilbert},
         familyi={G\bibinitperiod},
         given={A.\bibnamedelima C.},
         giveni={A\bibinitperiod\bibinitdelim C\bibinitperiod},
      }}%
      {{hash=SMJ}{%
         family={Strauss},
         familyi={S\bibinitperiod},
         given={M.\bibnamedelima J.},
         giveni={M\bibinitperiod\bibinitdelim J\bibinitperiod},
      }}%
    }
    \strng{namehash}{TJAGACSMJ1}
    \strng{fullhash}{TJAGACSMJ1}
    \field{labelnamesource}{author}
    \field{labeltitlesource}{title}
    \field{number}{3}
    \field{pages}{572\bibrangedash 588}
    \field{title}{Algorithms for simultaneous sparse approximation. Part {I}:
  Greedy pursuit.}
    \field{volume}{86}
    \field{journaltitle}{Signal Process.}
    \field{month}{03}
    \field{year}{2006}
  \endentry

  \entry{malioutov2005ssr}{article}{}
    \name{author}{3}{}{%
      {{hash=MD}{%
         family={Malioutov},
         familyi={M\bibinitperiod},
         given={D.},
         giveni={D\bibinitperiod},
      }}%
      {{hash=CM}{%
         family={Cetin},
         familyi={C\bibinitperiod},
         given={M.},
         giveni={M\bibinitperiod},
      }}%
      {{hash=WAS}{%
         family={Willsky},
         familyi={W\bibinitperiod},
         given={A.\bibnamedelima S.},
         giveni={A\bibinitperiod\bibinitdelim S\bibinitperiod},
      }}%
    }
    \strng{namehash}{MDCMWAS1}
    \strng{fullhash}{MDCMWAS1}
    \field{labelnamesource}{author}
    \field{labeltitlesource}{title}
    \field{number}{8}
    \field{pages}{3010\bibrangedash 3022}
    \field{title}{A sparse signal reconstruction perspective for source
  localization with sensor arrays}
    \field{volume}{53}
    \field{journaltitle}{IEEE Trans. Signal Process.}
    \field{month}{08}
    \field{year}{2005}
  \endentry

  \entry{tropp2006ass2}{article}{}
    \name{author}{1}{}{%
      {{hash=TJA}{%
         family={Tropp},
         familyi={T\bibinitperiod},
         given={J.\bibnamedelima A.},
         giveni={J\bibinitperiod\bibinitdelim A\bibinitperiod},
      }}%
    }
    \strng{namehash}{TJA1}
    \strng{fullhash}{TJA1}
    \field{labelnamesource}{author}
    \field{labeltitlesource}{title}
    \field{number}{3}
    \field{pages}{589\bibrangedash 602}
    \field{title}{Algorithms for simultaneous sparse approximation. Part {II}:
  Convex relaxation}
    \field{volume}{86}
    \field{journaltitle}{Signal Process.}
    \field{month}{03}
    \field{year}{2006}
  \endentry

  \entry{LeeBreslerJunge2012}{article}{}
    \name{author}{3}{}{%
      {{hash=LK}{%
         family={Lee},
         familyi={L\bibinitperiod},
         given={K.},
         giveni={K\bibinitperiod},
      }}%
      {{hash=BY}{%
         family={Bresler},
         familyi={B\bibinitperiod},
         given={Y.},
         giveni={Y\bibinitperiod},
      }}%
      {{hash=JM}{%
         family={Junge},
         familyi={J\bibinitperiod},
         given={M.},
         giveni={M\bibinitperiod},
      }}%
    }
    \strng{namehash}{LKBYJM1}
    \strng{fullhash}{LKBYJM1}
    \field{labelnamesource}{author}
    \field{labeltitlesource}{title}
    \field{number}{6}
    \field{pages}{3613\bibrangedash 3641}
    \field{title}{Subspace methods for joint sparse recovery}
    \field{volume}{58}
    \field{journaltitle}{IEEE Trans. Inf. Theory}
    \field{month}{06}
    \field{year}{2012}
  \endentry

  \entry{YeKimBresler2015}{article}{}
    \name{author}{3}{}{%
      {{hash=YJC}{%
         family={Ye},
         familyi={Y\bibinitperiod},
         given={J.\bibnamedelima C.},
         giveni={J\bibinitperiod\bibinitdelim C\bibinitperiod},
      }}%
      {{hash=KJM}{%
         family={Kim},
         familyi={K\bibinitperiod},
         given={J.\bibnamedelima M.},
         giveni={J\bibinitperiod\bibinitdelim M\bibinitperiod},
      }}%
      {{hash=BY}{%
         family={Bresler},
         familyi={B\bibinitperiod},
         given={Y.},
         giveni={Y\bibinitperiod},
      }}%
    }
    \strng{namehash}{YJCKJMBY1}
    \strng{fullhash}{YJCKJMBY1}
    \field{labelnamesource}{author}
    \field{labeltitlesource}{title}
    \field{number}{24}
    \field{pages}{6595\bibrangedash 6605}
    \field{title}{Improving {M-SBL} for joint sparse recovery using a subspace
  penalty}
    \field{volume}{63}
    \field{journaltitle}{IEEE Trans. Signal Process.}
    \field{month}{12}
    \field{year}{2015}
  \endentry

  \entry{ZinielSchniter2011}{inproceedings}{}
    \name{author}{2}{}{%
      {{hash=ZJ}{%
         family={Ziniel},
         familyi={Z\bibinitperiod},
         given={J.},
         giveni={J\bibinitperiod},
      }}%
      {{hash=SP}{%
         family={Schniter},
         familyi={S\bibinitperiod},
         given={P.},
         giveni={P\bibinitperiod},
      }}%
    }
    \strng{namehash}{ZJSP1}
    \strng{fullhash}{ZJSP1}
    \field{labelnamesource}{author}
    \field{labeltitlesource}{title}
    \field{booktitle}{Proc. IEEE Asilomar Conf. Signals, Syst., and Comput.}
    \field{pages}{1447\bibrangedash 1451}
    \field{title}{Efficient message passing-based inference in the multiple
  measurement vector problem}
    \field{month}{11}
    \field{year}{2011}
  \endentry

  \entry{pottie2000}{article}{}
    \name{author}{2}{}{%
      {{hash=PGJ}{%
         family={Pottie},
         familyi={P\bibinitperiod},
         given={G.\bibnamedelima J.},
         giveni={G\bibinitperiod\bibinitdelim J\bibinitperiod},
      }}%
      {{hash=KWJ}{%
         family={Kaiser},
         familyi={K\bibinitperiod},
         given={W.\bibnamedelima J.},
         giveni={W\bibinitperiod\bibinitdelim J\bibinitperiod},
      }}%
    }
    \strng{namehash}{PGJKWJ1}
    \strng{fullhash}{PGJKWJ1}
    \field{labelnamesource}{author}
    \field{labeltitlesource}{title}
    \field{number}{5}
    \field{pages}{51\bibrangedash 58}
    \field{title}{{Wireless integrated network sensors}}
    \field{volume}{43}
    \field{journaltitle}{Commun. ACM}
    \field{month}{05}
    \field{year}{2000}
  \endentry

  \entry{BDMK2016}{article}{}
    \name{author}{4}{}{%
      {{hash=BJ}{%
         family={Barbier},
         familyi={B\bibinitperiod},
         given={J.},
         giveni={J\bibinitperiod},
      }}%
      {{hash=DM}{%
         family={Dia},
         familyi={D\bibinitperiod},
         given={M.},
         giveni={M\bibinitperiod},
      }}%
      {{hash=MN}{%
         family={Macris},
         familyi={M\bibinitperiod},
         given={N.},
         giveni={N\bibinitperiod},
      }}%
      {{hash=KF}{%
         family={Krzakala},
         familyi={K\bibinitperiod},
         given={F.},
         giveni={F\bibinitperiod},
      }}%
    }
    \strng{namehash}{BJDMMNKF1}
    \strng{fullhash}{BJDMMNKF1}
    \field{labelnamesource}{author}
    \field{labeltitlesource}{title}
    \field{title}{The mutual information in random linear estimation}
    \field{journaltitle}{Arxiv preprint arXiv:1607.02335}
    \field{month}{07}
    \field{year}{2016}
  \endentry

  \entry{ReevesPfister2016}{article}{}
    \name{author}{2}{}{%
      {{hash=RG}{%
         family={Reeves},
         familyi={R\bibinitperiod},
         given={G.},
         giveni={G\bibinitperiod},
      }}%
      {{hash=PHD}{%
         family={Pfister},
         familyi={P\bibinitperiod},
         given={H.\bibnamedelima D.},
         giveni={H\bibinitperiod\bibinitdelim D\bibinitperiod},
      }}%
    }
    \strng{namehash}{RGPHD1}
    \strng{fullhash}{RGPHD1}
    \field{labelnamesource}{author}
    \field{labeltitlesource}{title}
    \field{title}{The replica-symmetric prediction for compressed sensing with
  {G}aussian matrices is exact}
    \field{journaltitle}{Arxiv preprint arXiv:1607.02524}
    \field{month}{07}
    \field{year}{2016}
  \endentry

  \entry{ZinielSchniter2013MMV}{article}{}
    \name{author}{2}{}{%
      {{hash=ZJ}{%
         family={Ziniel},
         familyi={Z\bibinitperiod},
         given={J.},
         giveni={J\bibinitperiod},
      }}%
      {{hash=SP}{%
         family={Schniter},
         familyi={S\bibinitperiod},
         given={P.},
         giveni={P\bibinitperiod},
      }}%
    }
    \strng{namehash}{ZJSP1}
    \strng{fullhash}{ZJSP1}
    \field{labelnamesource}{author}
    \field{labeltitlesource}{title}
    \field{number}{2}
    \field{pages}{340\bibrangedash 354}
    \field{title}{Efficient high-dimensional inference in the multiple
  measurement vector problem}
    \field{volume}{61}
    \field{journaltitle}{IEEE Trans. Signal Process.}
    \field{month}{01}
    \field{year}{2013}
  \endentry

  \entry{WuVerdu2012}{article}{}
    \name{author}{2}{}{%
      {{hash=WY}{%
         family={Wu},
         familyi={W\bibinitperiod},
         given={Y.},
         giveni={Y\bibinitperiod},
      }}%
      {{hash=VS}{%
         family={Verd{\'u}},
         familyi={V\bibinitperiod},
         given={S.},
         giveni={S\bibinitperiod},
      }}%
    }
    \strng{namehash}{WYVS1}
    \strng{fullhash}{WYVS1}
    \field{labelnamesource}{author}
    \field{labeltitlesource}{title}
    \field{number}{10}
    \field{pages}{6241\bibrangedash 6263}
    \field{title}{Optimal phase transitions in compressed sensing}
    \field{volume}{58}
    \field{journaltitle}{{IEEE} Trans. Inf. Theory}
    \field{month}{10}
    \field{year}{2012}
  \endentry

  \entry{KimChangJungBaronYe2011}{article}{}
    \name{author}{5}{}{%
      {{hash=KJ}{%
         family={Kim},
         familyi={K\bibinitperiod},
         given={J.},
         giveni={J\bibinitperiod},
      }}%
      {{hash=CW}{%
         family={Chang},
         familyi={C\bibinitperiod},
         given={W.},
         giveni={W\bibinitperiod},
      }}%
      {{hash=JB}{%
         family={Jung},
         familyi={J\bibinitperiod},
         given={B.},
         giveni={B\bibinitperiod},
      }}%
      {{hash=BD}{%
         family={Baron},
         familyi={B\bibinitperiod},
         given={D.},
         giveni={D\bibinitperiod},
      }}%
      {{hash=YJC}{%
         family={Ye},
         familyi={Y\bibinitperiod},
         given={J.\bibnamedelima C.},
         giveni={J\bibinitperiod\bibinitdelim C\bibinitperiod},
      }}%
    }
    \strng{namehash}{KJCWJBBDYJC1}
    \strng{fullhash}{KJCWJBBDYJC1}
    \field{labelnamesource}{author}
    \field{labeltitlesource}{title}
    \field{title}{Belief propagation for jointly sparse recovery}
    \field{journaltitle}{Arxiv preprint arXiv:1102.3289}
    \field{month}{02}
    \field{year}{2011}
  \endentry

  \entry{DMM2011}{article}{}
    \name{author}{3}{}{%
      {{hash=DDL}{%
         family={Donoho},
         familyi={D\bibinitperiod},
         given={D.\bibnamedelima L.},
         giveni={D\bibinitperiod\bibinitdelim L\bibinitperiod},
      }}%
      {{hash=MA}{%
         family={Maleki},
         familyi={M\bibinitperiod},
         given={A.},
         giveni={A\bibinitperiod},
      }}%
      {{hash=MA}{%
         family={Montanari},
         familyi={M\bibinitperiod},
         given={A.},
         giveni={A\bibinitperiod},
      }}%
    }
    \strng{namehash}{DDLMAMA1}
    \strng{fullhash}{DDLMAMA1}
    \field{labelnamesource}{author}
    \field{labeltitlesource}{title}
    \field{number}{10}
    \field{pages}{6920\bibrangedash 6941}
    \field{title}{The noise-sensitivity phase transition in compressed sensing}
    \field{volume}{57}
    \field{journaltitle}{IEEE Trans. Inf. Theory}
    \field{month}{10}
    \field{year}{2011}
  \endentry

  \entry{JavanmardMontanari2012}{article}{}
    \name{author}{2}{}{%
      {{hash=JA}{%
         family={Javanmard},
         familyi={J\bibinitperiod},
         given={A.},
         giveni={A\bibinitperiod},
      }}%
      {{hash=MA}{%
         family={Montanari},
         familyi={M\bibinitperiod},
         given={A.},
         giveni={A\bibinitperiod},
      }}%
    }
    \strng{namehash}{JAMA1}
    \strng{fullhash}{JAMA1}
    \field{labelnamesource}{author}
    \field{labeltitlesource}{title}
    \field{title}{State evolution for general approximate message passing
  algorithms, with applications to spatial coupling}
    \field{journaltitle}{Arxiv preprint arXiv:1211.5164}
    \field{month}{12}
    \field{year}{2012}
  \endentry

  \entry{Donoho2013}{article}{}
    \name{author}{3}{}{%
      {{hash=DD}{%
         family={Donoho},
         familyi={D\bibinitperiod},
         given={D.},
         giveni={D\bibinitperiod},
      }}%
      {{hash=JI}{%
         family={Johnstone},
         familyi={J\bibinitperiod},
         given={I.},
         giveni={I\bibinitperiod},
      }}%
      {{hash=MA}{%
         family={Montanari},
         familyi={M\bibinitperiod},
         given={A.},
         giveni={A\bibinitperiod},
      }}%
    }
    \strng{namehash}{DDJIMA1}
    \strng{fullhash}{DDJIMA1}
    \field{labelnamesource}{author}
    \field{labeltitlesource}{title}
    \field{number}{6}
    \field{pages}{3396\bibrangedash 3433}
    \field{title}{Accurate prediction of phase transitions in compressed
  sensing via a connection to minimax denoising}
    \field{volume}{59}
    \field{journaltitle}{IEEE Trans. Inf. Theory}
    \field{month}{06}
    \field{year}{2013}
  \endentry

  \entry{Bayati2015}{article}{}
    \name{author}{3}{}{%
      {{hash=BM}{%
         family={Bayati},
         familyi={B\bibinitperiod},
         given={M.},
         giveni={M\bibinitperiod},
      }}%
      {{hash=LM}{%
         family={Lelarge},
         familyi={L\bibinitperiod},
         given={M.},
         giveni={M\bibinitperiod},
      }}%
      {{hash=MA}{%
         family={Montanari},
         familyi={M\bibinitperiod},
         given={A.},
         giveni={A\bibinitperiod},
      }}%
    }
    \strng{namehash}{BMLMMA1}
    \strng{fullhash}{BMLMMA1}
    \field{labelnamesource}{author}
    \field{labeltitlesource}{title}
    \field{number}{2}
    \field{pages}{753\bibrangedash 822}
    \field{title}{Universality in polytope phase transitions and message
  passing algorithms}
    \field{volume}{25}
    \field{journaltitle}{Ann. Appl. Probability}
    \field{month}{02}
    \field{year}{2015}
  \endentry

  \entry{Rush_ISIT2016_arxiv}{article}{}
    \name{author}{2}{}{%
      {{hash=RC}{%
         family={Rush},
         familyi={R\bibinitperiod},
         given={C.},
         giveni={C\bibinitperiod},
      }}%
      {{hash=VR}{%
         family={Venkataramanan},
         familyi={V\bibinitperiod},
         given={R.},
         giveni={R\bibinitperiod},
      }}%
    }
    \strng{namehash}{RCVR1}
    \strng{fullhash}{RCVR1}
    \field{labelnamesource}{author}
    \field{labeltitlesource}{title}
    \field{title}{Finite-Sample Analysis of Approximate Message Passing}
    \field{journaltitle}{Arxiv preprint arXiv:1606.01800}
    \field{month}{06}
    \field{year}{2016}
  \endentry

  \entry{GuoBaronShamai2009}{inproceedings}{}
    \name{author}{3}{}{%
      {{hash=GD}{%
         family={Guo},
         familyi={G\bibinitperiod},
         given={D.},
         giveni={D\bibinitperiod},
      }}%
      {{hash=BD}{%
         family={Baron},
         familyi={B\bibinitperiod},
         given={D.},
         giveni={D\bibinitperiod},
      }}%
      {{hash=SS}{%
         family={Shamai},
         familyi={S\bibinitperiod},
         given={S.},
         giveni={S\bibinitperiod},
      }}%
    }
    \strng{namehash}{GDBDSS1}
    \strng{fullhash}{GDBDSS1}
    \field{labelnamesource}{author}
    \field{labeltitlesource}{title}
    \field{booktitle}{Proc. Allerton Conference Commun., Control, and Comput.}
    \field{pages}{52\bibrangedash 59}
    \field{title}{A Single-letter Characterization of Optimal Noisy Compressed
  Sensing}
    \field{month}{09}
    \field{year}{2009}
  \endentry

  \entry{estrin2002}{article}{}
    \name{author}{4}{}{%
      {{hash=ED}{%
         family={Estrin},
         familyi={E\bibinitperiod},
         given={D.},
         giveni={D\bibinitperiod},
      }}%
      {{hash=CD}{%
         family={Culler},
         familyi={C\bibinitperiod},
         given={D.},
         giveni={D\bibinitperiod},
      }}%
      {{hash=PK}{%
         family={Pister},
         familyi={P\bibinitperiod},
         given={K.},
         giveni={K\bibinitperiod},
      }}%
      {{hash=SG}{%
         family={Sukhatme},
         familyi={S\bibinitperiod},
         given={G.},
         giveni={G\bibinitperiod},
      }}%
    }
    \strng{namehash}{EDCDPKSG1}
    \strng{fullhash}{EDCDPKSG1}
    \field{labelnamesource}{author}
    \field{labeltitlesource}{title}
    \field{number}{1}
    \field{pages}{59\bibrangedash 69}
    \field{title}{{Connecting the physical world with pervasive networks}}
    \field{volume}{1}
    \field{journaltitle}{IEEE Pervasive Comput.}
    \field{month}{01}
    \field{year}{2002}
  \endentry

  \entry{EC2}{inproceedings}{}
    \field{labeltitlesource}{title}
    \field{note}{https://aws.amazon.com/ec2/}
    \field{title}{Amazon {E}{C}2}
  \endentry

  \entry{bertsekas1995}{book}{}
    \name{author}{1}{}{%
      {{hash=BDP}{%
         family={Bertsekas},
         familyi={B\bibinitperiod},
         given={D.\bibnamedelima P.},
         giveni={D\bibinitperiod\bibinitdelim P\bibinitperiod},
      }}%
    }
    \list{publisher}{1}{%
      {Athena Scientific Belmont, MA}%
    }
    \strng{namehash}{BDP1}
    \strng{fullhash}{BDP1}
    \field{labelnamesource}{author}
    \field{labeltitlesource}{title}
    \field{title}{Dynamic {P}rogramming and {O}ptimal {C}ontrol}
    \field{volume}{1}
    \field{year}{1995}
  \endentry

  \entry{DasDennisPareto1998}{article}{}
    \name{author}{2}{}{%
      {{hash=DI}{%
         family={Das},
         familyi={D\bibinitperiod},
         given={I.},
         giveni={I\bibinitperiod},
      }}%
      {{hash=DJE}{%
         family={Dennis},
         familyi={D\bibinitperiod},
         given={J.\bibnamedelima E.},
         giveni={J\bibinitperiod\bibinitdelim E\bibinitperiod},
      }}%
    }
    \strng{namehash}{DIDJE1}
    \strng{fullhash}{DIDJE1}
    \field{labelnamesource}{author}
    \field{labeltitlesource}{title}
    \field{number}{3}
    \field{pages}{631\bibrangedash 657}
    \field{title}{Normal-boundary intersection: A new method for generating the
  {P}areto surface in nonlinear multicriteria optimization problems}
    \field{volume}{8}
    \field{journaltitle}{SIAM J. Optimization}
    \field{month}{08}
    \field{year}{1998}
  \endentry

  \entry{Frasca2008}{article}{}
    \name{author}{4}{}{%
      {{hash=FP}{%
         family={Frasca},
         familyi={F\bibinitperiod},
         given={P.},
         giveni={P\bibinitperiod},
      }}%
      {{hash=CR}{%
         family={Carli},
         familyi={C\bibinitperiod},
         given={R.},
         giveni={R\bibinitperiod},
      }}%
      {{hash=FF}{%
         family={Fagnani},
         familyi={F\bibinitperiod},
         given={F.},
         giveni={F\bibinitperiod},
      }}%
      {{hash=ZS}{%
         family={Zampieri},
         familyi={Z\bibinitperiod},
         given={S.},
         giveni={S\bibinitperiod},
      }}%
    }
    \strng{namehash}{FPCRFFZS1}
    \strng{fullhash}{FPCRFFZS1}
    \field{labelnamesource}{author}
    \field{labeltitlesource}{title}
    \field{number}{16}
    \field{pages}{1787\bibrangedash 1816}
    \field{title}{Average consensus on networks with quantized communication}
    \field{volume}{19}
    \field{journaltitle}{Int. J. Robust Nonlinear Control}
    \field{month}{11}
    \field{year}{2008}
  \endentry

  \entry{MaBaronBeirami2015ISIT}{inproceedings}{}
    \name{author}{3}{}{%
      {{hash=MY}{%
         family={Ma},
         familyi={M\bibinitperiod},
         given={Y.},
         giveni={Y\bibinitperiod},
      }}%
      {{hash=BD}{%
         family={Baron},
         familyi={B\bibinitperiod},
         given={D.},
         giveni={D\bibinitperiod},
      }}%
      {{hash=BA}{%
         family={Beirami},
         familyi={B\bibinitperiod},
         given={A.},
         giveni={A\bibinitperiod},
      }}%
    }
    \strng{namehash}{MYBDBA1}
    \strng{fullhash}{MYBDBA1}
    \field{labelnamesource}{author}
    \field{labeltitlesource}{title}
    \field{booktitle}{Proc. IEEE Int. Symp. Inf. Theory (ISIT)}
    \field{pages}{760\bibrangedash 764}
    \field{title}{Mismatched Estimation in Large Linear Systems}
    \list{location}{1}{%
      {Hong Kong, China}%
    }
    \field{month}{07}
    \field{year}{2015}
  \endentry

  \entry{widrow2008quantization}{book}{}
    \name{author}{2}{}{%
      {{hash=WB}{%
         family={Widrow},
         familyi={W\bibinitperiod},
         given={B.},
         giveni={B\bibinitperiod},
      }}%
      {{hash=KI}{%
         family={Koll{\'a}r},
         familyi={K\bibinitperiod},
         given={I.},
         giveni={I\bibinitperiod},
      }}%
    }
    \list{publisher}{1}{%
      {Cambridge University press}%
    }
    \strng{namehash}{WBKI1}
    \strng{fullhash}{WBKI1}
    \field{labelnamesource}{author}
    \field{labeltitlesource}{title}
    \field{title}{Quantization Noise: Roundoff Error in Digital Computation,
  Signal Processing, Control, and Communications}
    \field{year}{2008}
  \endentry

  \entry{GrayNeuhoff1998}{article}{}
    \name{author}{2}{}{%
      {{hash=GRM}{%
         family={Gray},
         familyi={G\bibinitperiod},
         given={R.\bibnamedelima M.},
         giveni={R\bibinitperiod\bibinitdelim M\bibinitperiod},
      }}%
      {{hash=NDL}{%
         family={Neuhoff},
         familyi={N\bibinitperiod},
         given={D.\bibnamedelima L.},
         giveni={D\bibinitperiod\bibinitdelim L\bibinitperiod},
      }}%
    }
    \strng{namehash}{GRMNDL1}
    \strng{fullhash}{GRMNDL1}
    \field{labelnamesource}{author}
    \field{labeltitlesource}{title}
    \field{pages}{2325\bibrangedash 2383}
    \field{title}{Quantization}
    \field{volume}{IT-44}
    \field{journaltitle}{IEEE Trans. Inf. Theory}
    \field{month}{10}
    \field{year}{1998}
  \endentry

  \entry{WuVerdu2011}{article}{}
    \name{author}{2}{}{%
      {{hash=WY}{%
         family={Wu},
         familyi={W\bibinitperiod},
         given={Y.},
         giveni={Y\bibinitperiod},
      }}%
      {{hash=VS}{%
         family={Verd{\'u}},
         familyi={V\bibinitperiod},
         given={S.},
         giveni={S\bibinitperiod},
      }}%
    }
    \strng{namehash}{WYVS1}
    \strng{fullhash}{WYVS1}
    \field{labelnamesource}{author}
    \field{labeltitlesource}{title}
    \field{number}{8}
    \field{pages}{4857\bibrangedash 4879}
    \field{title}{{MMSE} Dimension}
    \field{volume}{57}
    \field{journaltitle}{{IEEE} Trans. Inf. Theory}
    \field{month}{08}
    \field{year}{2011}
  \endentry

  \entry{CC2530}{manual}{}
    \list{organization}{1}{%
      {Texas Instruments}%
    }
    \field{labeltitlesource}{title}
    \field{note}{Rev. B}
    \field{number}{SWRS081B}
    \field{title}{A true system-on-chip solution for 2.4-{GH}z {IEEE} 802.15.4
  and {Z}ig{B}ee applications}
    \field{month}{04}
    \field{year}{2009}
  \endentry

  \entry{Tan_CompressiveImage2014}{article}{}
    \name{author}{3}{}{%
      {{hash=TJ}{%
         family={Tan},
         familyi={T\bibinitperiod},
         given={J.},
         giveni={J\bibinitperiod},
      }}%
      {{hash=MY}{%
         family={Ma},
         familyi={M\bibinitperiod},
         given={Y.},
         giveni={Y\bibinitperiod},
      }}%
      {{hash=BD}{%
         family={Baron},
         familyi={B\bibinitperiod},
         given={D.},
         giveni={D\bibinitperiod},
      }}%
    }
    \strng{namehash}{TJMYBD1}
    \strng{fullhash}{TJMYBD1}
    \field{labelnamesource}{author}
    \field{labeltitlesource}{title}
    \field{number}{8}
    \field{pages}{2085\bibrangedash 2092}
    \field{title}{Compressive imaging via approximate message passing with
  image denoising}
    \field{volume}{63}
    \field{journaltitle}{IEEE Trans. Signal Process.}
    \field{month}{04}
    \field{year}{2015}
  \endentry

  \entry{LZ77}{article}{}
    \name{author}{2}{}{%
      {{hash=ZJ}{%
         family={Ziv},
         familyi={Z\bibinitperiod},
         given={J.},
         giveni={J\bibinitperiod},
      }}%
      {{hash=LA}{%
         family={Lempel},
         familyi={L\bibinitperiod},
         given={A.},
         giveni={A\bibinitperiod},
      }}%
    }
    \strng{namehash}{ZJLA1}
    \strng{fullhash}{ZJLA1}
    \field{labelnamesource}{author}
    \field{labeltitlesource}{title}
    \field{number}{3}
    \field{pages}{337\bibrangedash 343}
    \field{title}{{A universal algorithm for sequential data compression}}
    \field{volume}{23}
    \field{journaltitle}{IEEE Trans. Inf. Theory}
    \field{month}{05}
    \field{year}{1977}
  \endentry

  \entry{Rissanen1983}{article}{}
    \name{author}{1}{}{%
      {{hash=RJ}{%
         family={Rissanen},
         familyi={R\bibinitperiod},
         given={J.},
         giveni={J\bibinitperiod},
      }}%
    }
    \strng{namehash}{RJ1}
    \strng{fullhash}{RJ1}
    \field{labelnamesource}{author}
    \field{labeltitlesource}{title}
    \field{number}{5}
    \field{pages}{656\bibrangedash 664}
    \field{title}{{A universal data compression system}}
    \field{volume}{29}
    \field{journaltitle}{IEEE Trans. Inf. Theory}
    \field{month}{09}
    \field{year}{1983}
  \endentry

  \entry{Ramirez2010}{article}{}
    \name{author}{2}{}{%
      {{hash=RI}{%
         family={Ramirez},
         familyi={R\bibinitperiod},
         given={I.},
         giveni={I\bibinitperiod},
      }}%
      {{hash=SG}{%
         family={Sapiro},
         familyi={S\bibinitperiod},
         given={G.},
         giveni={G\bibinitperiod},
      }}%
    }
    \strng{namehash}{RISG2}
    \strng{fullhash}{RISG2}
    \field{labelnamesource}{author}
    \field{labeltitlesource}{title}
    \field{number}{9}
    \field{pages}{3850\bibrangedash 3864}
    \field{title}{Universal regularizers for robust sparse coding and modeling}
    \field{volume}{21}
    \field{journaltitle}{IEEE Trans. Image Process.}
    \field{month}{09}
    \field{year}{2012}
  \endentry

  \entry{DonohoKolmogorov}{report}{}
    \name{author}{1}{}{%
      {{hash=DDL}{%
         family={Donoho},
         familyi={D\bibinitperiod},
         given={D.\bibnamedelima L.},
         giveni={D\bibinitperiod\bibinitdelim L\bibinitperiod},
      }}%
    }
    \strng{namehash}{DDL1}
    \strng{fullhash}{DDL1}
    \field{labelnamesource}{author}
    \field{labeltitlesource}{title}
    \field{number}{2002-4}
    \field{title}{The {K}olmogorov sampler}
    \list{location}{1}{%
      {Stanford, CA}%
    }
    \list{institution}{1}{%
      {Stanford University}%
    }
    \field{type}{Department of Statistics Technical Report}
    \field{month}{01}
    \field{year}{2002}
  \endentry

  \entry{Chaitin1966}{article}{}
    \name{author}{1}{}{%
      {{hash=CGJ}{%
         family={Chaitin},
         familyi={C\bibinitperiod},
         given={G.\bibnamedelima J.},
         giveni={G\bibinitperiod\bibinitdelim J\bibinitperiod},
      }}%
    }
    \strng{namehash}{CGJ1}
    \strng{fullhash}{CGJ1}
    \field{labelnamesource}{author}
    \field{labeltitlesource}{title}
    \field{number}{4}
    \field{pages}{547\bibrangedash 569}
    \field{title}{On the length of programs for computing finite binary
  sequences}
    \field{volume}{13}
    \field{journaltitle}{J. ACM}
    \field{year}{1966}
  \endentry

  \entry{Solomonoff1964}{article}{}
    \name{author}{1}{}{%
      {{hash=SRJ}{%
         family={Solomonoff},
         familyi={S\bibinitperiod},
         given={R.\bibnamedelima J.},
         giveni={R\bibinitperiod\bibinitdelim J\bibinitperiod},
      }}%
    }
    \strng{namehash}{SRJ1}
    \strng{fullhash}{SRJ1}
    \field{labelnamesource}{author}
    \field{labeltitlesource}{title}
    \field{number}{1}
    \field{pages}{1\bibrangedash 22}
    \field{title}{A formal theory of inductive inference. {Part I}}
    \field{volume}{7}
    \field{journaltitle}{Inf. and Control}
    \field{month}{03}
    \field{year}{1964}
  \endentry

  \entry{Kolmogorov1965}{article}{}
    \name{author}{1}{}{%
      {{hash=KAN}{%
         family={Kolmogorov},
         familyi={K\bibinitperiod},
         given={A.\bibnamedelima N.},
         giveni={A\bibinitperiod\bibinitdelim N\bibinitperiod},
      }}%
    }
    \strng{namehash}{KAN1}
    \strng{fullhash}{KAN1}
    \field{labelnamesource}{author}
    \field{labeltitlesource}{title}
    \field{number}{1}
    \field{pages}{1\bibrangedash 7}
    \field{title}{Three approaches to the quantitative definition of
  information}
    \field{volume}{1}
    \field{journaltitle}{Problems Inf. Transmission}
    \field{year}{1965}
  \endentry

  \entry{JalaliMaleki2011}{inproceedings}{}
    \name{author}{2}{}{%
      {{hash=JS}{%
         family={Jalali},
         familyi={J\bibinitperiod},
         given={S.},
         giveni={S\bibinitperiod},
      }}%
      {{hash=MA}{%
         family={Maleki},
         familyi={M\bibinitperiod},
         given={A.},
         giveni={A\bibinitperiod},
      }}%
    }
    \strng{namehash}{JSMA1}
    \strng{fullhash}{JSMA1}
    \field{labelnamesource}{author}
    \field{labeltitlesource}{title}
    \field{booktitle}{Proc. Allerton Conference Commun., Control, Comput.}
    \field{pages}{1764\bibrangedash 1770}
    \field{title}{Minimum complexity pursuit}
    \field{month}{09}
    \field{year}{2011}
  \endentry

  \entry{JalaliMalekiRichB2014}{article}{}
    \name{author}{3}{}{%
      {{hash=JS}{%
         family={Jalali},
         familyi={J\bibinitperiod},
         given={S.},
         giveni={S\bibinitperiod},
      }}%
      {{hash=MA}{%
         family={Maleki},
         familyi={M\bibinitperiod},
         given={A.},
         giveni={A\bibinitperiod},
      }}%
      {{hash=BRG}{%
         family={Baraniuk},
         familyi={B\bibinitperiod},
         given={R.\bibnamedelima G.},
         giveni={R\bibinitperiod\bibinitdelim G\bibinitperiod},
      }}%
    }
    \strng{namehash}{JSMABRG1}
    \strng{fullhash}{JSMABRG1}
    \field{labelnamesource}{author}
    \field{labeltitlesource}{title}
    \field{number}{4}
    \field{pages}{2253\bibrangedash 2268}
    \field{title}{Minimum complexity pursuit for universal compressed sensing}
    \field{volume}{60}
    \field{journaltitle}{IEEE Trans. Inf. Theory}
    \field{month}{04}
    \field{year}{2014}
  \endentry

  \entry{BaronFinland2011}{inproceedings}{}
    \name{author}{1}{}{%
      {{hash=BD}{%
         family={Baron},
         familyi={B\bibinitperiod},
         given={D.},
         giveni={D\bibinitperiod},
      }}%
    }
    \strng{namehash}{BD1}
    \strng{fullhash}{BD1}
    \field{labelnamesource}{author}
    \field{labeltitlesource}{title}
    \field{booktitle}{Workshop Inf. Theoretic Methods Sci. Eng. (WITMSE)}
    \field{title}{Information complexity and estimation}
    \list{location}{1}{%
      {Helsinki, Finland}%
    }
    \field{month}{08}
    \field{year}{2011}
  \endentry

  \entry{BaronDuarteAllerton2011}{inproceedings}{}
    \name{author}{2}{}{%
      {{hash=BD}{%
         family={Baron},
         familyi={B\bibinitperiod},
         given={D.},
         giveni={D\bibinitperiod},
      }}%
      {{hash=DMF}{%
         family={Duarte},
         familyi={D\bibinitperiod},
         given={M.\bibnamedelima F.},
         giveni={M\bibinitperiod\bibinitdelim F\bibinitperiod},
      }}%
    }
    \strng{namehash}{BDDMF1}
    \strng{fullhash}{BDDMF1}
    \field{labelnamesource}{author}
    \field{labeltitlesource}{title}
    \field{booktitle}{Proc. Allerton Conference Commun., Control, and Comput.}
    \field{pages}{768\bibrangedash 775}
    \field{title}{Universal {MAP} estimation in compressed sensing}
    \field{month}{09}
    \field{year}{2011}
  \endentry

  \entry{Rissanen1978}{article}{}
    \name{author}{1}{}{%
      {{hash=RJ}{%
         family={Rissanen},
         familyi={R\bibinitperiod},
         given={J.},
         giveni={J\bibinitperiod},
      }}%
    }
    \strng{namehash}{RJ1}
    \strng{fullhash}{RJ1}
    \field{labelnamesource}{author}
    \field{labeltitlesource}{title}
    \field{number}{5}
    \field{pages}{465\bibrangedash 471}
    \field{title}{Modeling by shortest data description}
    \field{volume}{14}
    \field{journaltitle}{Automatica}
    \field{month}{09}
    \field{year}{1978}
  \endentry

  \entry{schwarz1978estimating}{article}{}
    \name{author}{1}{}{%
      {{hash=SG}{%
         family={Schwarz},
         familyi={S\bibinitperiod},
         given={G.},
         giveni={G\bibinitperiod},
      }}%
    }
    \list{publisher}{1}{%
      {Institute of Mathematical Statistics}%
    }
    \strng{namehash}{SG1}
    \strng{fullhash}{SG1}
    \field{labelnamesource}{author}
    \field{labeltitlesource}{title}
    \field{number}{2}
    \field{pages}{461\bibrangedash 464}
    \field{title}{Estimating the dimension of a model}
    \field{volume}{6}
    \field{journaltitle}{Ann. Stat.}
    \field{month}{03}
    \field{year}{1978}
  \endentry

  \entry{Wallace1968}{article}{}
    \name{author}{2}{}{%
      {{hash=WCS}{%
         family={Wallace},
         familyi={W\bibinitperiod},
         given={C.\bibnamedelima S.},
         giveni={C\bibinitperiod\bibinitdelim S\bibinitperiod},
      }}%
      {{hash=BDM}{%
         family={Boulton},
         familyi={B\bibinitperiod},
         given={D.\bibnamedelima M.},
         giveni={D\bibinitperiod\bibinitdelim M\bibinitperiod},
      }}%
    }
    \strng{namehash}{WCSBDM1}
    \strng{fullhash}{WCSBDM1}
    \field{labelnamesource}{author}
    \field{labeltitlesource}{title}
    \field{number}{2}
    \field{pages}{185\bibrangedash 194}
    \field{title}{An information measure for classification}
    \field{volume}{11}
    \field{journaltitle}{Comput. J.}
    \field{year}{1968}
  \endentry

  \entry{BRY98}{article}{}
    \name{author}{3}{}{%
      {{hash=BA}{%
         family={Barron},
         familyi={B\bibinitperiod},
         given={A.},
         giveni={A\bibinitperiod},
      }}%
      {{hash=RJ}{%
         family={Rissanen},
         familyi={R\bibinitperiod},
         given={J.},
         giveni={J\bibinitperiod},
      }}%
      {{hash=YB}{%
         family={Yu},
         familyi={Y\bibinitperiod},
         given={B.},
         giveni={B\bibinitperiod},
      }}%
    }
    \strng{namehash}{BARJYB1}
    \strng{fullhash}{BARJYB1}
    \field{labelnamesource}{author}
    \field{labeltitlesource}{title}
    \field{number}{6}
    \field{pages}{2743\bibrangedash 2760}
    \field{title}{The minimum description length principle in coding and
  modeling}
    \field{volume}{44}
    \field{journaltitle}{IEEE Trans. Inf. Theory}
    \field{month}{10}
    \field{year}{1998}
  \endentry

  \entry{Geman1984}{article}{}
    \name{author}{2}{}{%
      {{hash=GS}{%
         family={Geman},
         familyi={G\bibinitperiod},
         given={S.},
         giveni={S\bibinitperiod},
      }}%
      {{hash=GD}{%
         family={Geman},
         familyi={G\bibinitperiod},
         given={D.},
         giveni={D\bibinitperiod},
      }}%
    }
    \strng{namehash}{GSGD1}
    \strng{fullhash}{GSGD1}
    \field{labelnamesource}{author}
    \field{labeltitlesource}{title}
    \field{pages}{721\bibrangedash 741}
    \field{title}{Stochastic relaxation, {G}ibbs distributions, and the
  {B}ayesian restoration of images}
    \field{volume}{6}
    \field{journaltitle}{IEEE Trans. Pattern Anal. Mach. Intell.}
    \field{month}{11}
    \field{year}{1984}
  \endentry

  \entry{Rangan2010CISS}{inproceedings}{}
    \name{author}{1}{}{%
      {{hash=RS}{%
         family={Rangan},
         familyi={R\bibinitperiod},
         given={S.},
         giveni={S\bibinitperiod},
      }}%
    }
    \strng{namehash}{RS1}
    \strng{fullhash}{RS1}
    \field{labelnamesource}{author}
    \field{labeltitlesource}{title}
    \field{booktitle}{Proc. IEEE Conf. Inf. Sci. Syst. (CISS)}
    \field{title}{Estimation with random linear mixing, belief propagation and
  compressed sensing}
    \list{location}{1}{%
      {Princeton, NJ}%
    }
    \field{month}{03}
    \field{year}{2010}
  \endentry

  \entry{Turing1950}{article}{}
    \name{author}{1}{}{%
      {{hash=TAM}{%
         family={Turing},
         familyi={T\bibinitperiod},
         given={A.\bibnamedelima M.},
         giveni={A\bibinitperiod\bibinitdelim M\bibinitperiod},
      }}%
    }
    \list{publisher}{1}{%
      {JSTOR}%
    }
    \strng{namehash}{TAM1}
    \strng{fullhash}{TAM1}
    \field{labelnamesource}{author}
    \field{labeltitlesource}{title}
    \field{number}{236}
    \field{pages}{433\bibrangedash 460}
    \field{title}{Computing machinery and intelligence}
    \field{volume}{59}
    \field{journaltitle}{Mind}
    \field{month}{10}
    \field{year}{1950}
  \endentry

  \entry{BaronWeissman2012}{article}{}
    \name{author}{2}{}{%
      {{hash=BD}{%
         family={Baron},
         familyi={B\bibinitperiod},
         given={D.},
         giveni={D\bibinitperiod},
      }}%
      {{hash=WT}{%
         family={Weissman},
         familyi={W\bibinitperiod},
         given={T.},
         giveni={T\bibinitperiod},
      }}%
    }
    \list{publisher}{1}{%
      {IEEE}%
    }
    \strng{namehash}{BDWT1}
    \strng{fullhash}{BDWT1}
    \field{labelnamesource}{author}
    \field{labeltitlesource}{title}
    \field{number}{10}
    \field{pages}{5230\bibrangedash 5240}
    \field{title}{An {MCMC} approach to universal lossy compression of analog
  sources}
    \field{volume}{60}
    \field{journaltitle}{IEEE Trans. Signal Process.}
    \field{month}{10}
    \field{year}{2012}
  \endentry

  \entry{Jalali2008}{inproceedings}{}
    \name{author}{2}{}{%
      {{hash=JS}{%
         family={Jalali},
         familyi={J\bibinitperiod},
         given={S.},
         giveni={S\bibinitperiod},
      }}%
      {{hash=WT}{%
         family={Weissman},
         familyi={W\bibinitperiod},
         given={T.},
         giveni={T\bibinitperiod},
      }}%
    }
    \strng{namehash}{JSWT1}
    \strng{fullhash}{JSWT1}
    \field{labelnamesource}{author}
    \field{labeltitlesource}{title}
    \field{booktitle}{Proc. IEEE Int. Symp. Inf. Theory (ISIT)}
    \field{pages}{852\bibrangedash 856}
    \field{title}{Rate-distortion via {M}arkov chain {M}onte {C}arlo}
    \list{location}{1}{%
      {Toronto, Ontario, Canada}%
    }
    \field{month}{07}
    \field{year}{2008}
  \endentry

  \entry{Jalali2012}{article}{}
    \name{author}{2}{}{%
      {{hash=JS}{%
         family={Jalali},
         familyi={J\bibinitperiod},
         given={S.},
         giveni={S\bibinitperiod},
      }}%
      {{hash=WT}{%
         family={Weissman},
         familyi={W\bibinitperiod},
         given={T.},
         giveni={T\bibinitperiod},
      }}%
    }
    \strng{namehash}{JSWT1}
    \strng{fullhash}{JSWT1}
    \field{labelnamesource}{author}
    \field{labeltitlesource}{title}
    \field{number}{8}
    \field{pages}{2187\bibrangedash 2198}
    \field{title}{{Block and sliding-block lossy compression via MCMC}}
    \field{volume}{60}
    \field{journaltitle}{IEEE Trans. Commun.}
    \field{month}{08}
    \field{year}{2012}
  \endentry

  \entry{Yang1997}{article}{}
    \name{author}{3}{}{%
      {{hash=YE}{%
         family={Yang},
         familyi={Y\bibinitperiod},
         given={E.},
         giveni={E\bibinitperiod},
      }}%
      {{hash=ZZ}{%
         family={Zhang},
         familyi={Z\bibinitperiod},
         given={Z.},
         giveni={Z\bibinitperiod},
      }}%
      {{hash=BT}{%
         family={Berger},
         familyi={B\bibinitperiod},
         given={T.},
         giveni={T\bibinitperiod},
      }}%
    }
    \strng{namehash}{YEZZBT1}
    \strng{fullhash}{YEZZBT1}
    \field{labelnamesource}{author}
    \field{labeltitlesource}{title}
    \field{number}{5}
    \field{pages}{1465\bibrangedash 1476}
    \field{title}{{Fixed-slope universal lossy data compression}}
    \field{volume}{43}
    \field{journaltitle}{IEEE Trans. Inf. Theory}
    \field{month}{09}
    \field{year}{1997}
  \endentry

  \entry{Willems1995CTW}{article}{}
    \name{author}{3}{}{%
      {{hash=WFMJ}{%
         family={Willems},
         familyi={W\bibinitperiod},
         given={F.\bibnamedelima M.\bibnamedelima J.},
         giveni={F\bibinitperiod\bibinitdelim M\bibinitperiod\bibinitdelim
  J\bibinitperiod},
      }}%
      {{hash=SYM}{%
         family={Shtarkov},
         familyi={S\bibinitperiod},
         given={Y.\bibnamedelima M.},
         giveni={Y\bibinitperiod\bibinitdelim M\bibinitperiod},
      }}%
      {{hash=TTJ}{%
         family={Tjalkens},
         familyi={T\bibinitperiod},
         given={T.\bibnamedelima J.},
         giveni={T\bibinitperiod\bibinitdelim J\bibinitperiod},
      }}%
    }
    \strng{namehash}{WFMJSYMTTJ1}
    \strng{fullhash}{WFMJSYMTTJ1}
    \field{labelnamesource}{author}
    \field{labeltitlesource}{title}
    \field{number}{3}
    \field{pages}{653\bibrangedash 664}
    \field{title}{The context tree weighting method: {B}asic properties}
    \field{volume}{41}
    \field{journaltitle}{IEEE Trans. Inf. Theory}
    \field{month}{05}
    \field{year}{1995}
  \endentry

  \entry{Cosamp08}{article}{}
    \name{author}{2}{}{%
      {{hash=ND}{%
         family={Needell},
         familyi={N\bibinitperiod},
         given={D.},
         giveni={D\bibinitperiod},
      }}%
      {{hash=TJA}{%
         family={Tropp},
         familyi={T\bibinitperiod},
         given={J.\bibnamedelima A.},
         giveni={J\bibinitperiod\bibinitdelim A\bibinitperiod},
      }}%
    }
    \strng{namehash}{NDTJA1}
    \strng{fullhash}{NDTJA1}
    \field{labelnamesource}{author}
    \field{labeltitlesource}{title}
    \field{number}{3}
    \field{pages}{301\bibrangedash 321}
    \field{title}{Co{S}a{MP}: Iterative signal recovery from incomplete and
  inaccurate samples}
    \field{volume}{26}
    \field{journaltitle}{Appl. Computational Harmonic Anal.}
    \field{month}{05}
    \field{year}{2009}
  \endentry

  \entry{CandesCSdictonary2011}{article}{}
    \name{author}{4}{}{%
      {{hash=CEJ}{%
         family={Cand\`es},
         familyi={C\bibinitperiod},
         given={E.\bibnamedelima J.},
         giveni={E\bibinitperiod\bibinitdelim J\bibinitperiod},
      }}%
      {{hash=EYC}{%
         family={Eldar},
         familyi={E\bibinitperiod},
         given={Y.\bibnamedelima C.},
         giveni={Y\bibinitperiod\bibinitdelim C\bibinitperiod},
      }}%
      {{hash=ND}{%
         family={Needell},
         familyi={N\bibinitperiod},
         given={D.},
         giveni={D\bibinitperiod},
      }}%
      {{hash=RP}{%
         family={Randall},
         familyi={R\bibinitperiod},
         given={P.},
         giveni={P\bibinitperiod},
      }}%
    }
    \strng{namehash}{CEJEYCNDRP1}
    \strng{fullhash}{CEJEYCNDRP1}
    \field{labelnamesource}{author}
    \field{labeltitlesource}{title}
    \field{number}{1}
    \field{pages}{59\bibrangedash 73}
    \field{title}{Compressed sensing with coherent and redundant dictionaries}
    \field{volume}{31}
    \field{journaltitle}{Appl. Computational Harmonic Anal.}
    \field{month}{07}
    \field{year}{2011}
  \endentry

  \entry{Sivaramakrishnan2008}{article}{}
    \name{author}{2}{}{%
      {{hash=SK}{%
         family={Sivaramakrishnan},
         familyi={S\bibinitperiod},
         given={K.},
         giveni={K\bibinitperiod},
      }}%
      {{hash=WT}{%
         family={Weissman},
         familyi={W\bibinitperiod},
         given={T.},
         giveni={T\bibinitperiod},
      }}%
    }
    \strng{namehash}{SKWT1}
    \strng{fullhash}{SKWT1}
    \field{labelnamesource}{author}
    \field{labeltitlesource}{title}
    \field{number}{12}
    \field{pages}{5632\bibrangedash 5660}
    \field{title}{Universal denoising of discrete-time continuous-amplitude
  signals}
    \field{volume}{54}
    \field{journaltitle}{IEEE Trans. Inf. Theory}
    \field{month}{12}
    \field{year}{2008}
  \endentry

  \entry{SW_Context2009}{article}{}
    \name{author}{2}{}{%
      {{hash=SK}{%
         family={Sivaramakrishnan},
         familyi={S\bibinitperiod},
         given={K.},
         giveni={K\bibinitperiod},
      }}%
      {{hash=WT}{%
         family={Weissman},
         familyi={W\bibinitperiod},
         given={T.},
         giveni={T\bibinitperiod},
      }}%
    }
    \strng{namehash}{SKWT1}
    \strng{fullhash}{SKWT1}
    \field{labelnamesource}{author}
    \field{labeltitlesource}{title}
    \field{number}{6}
    \field{pages}{2110\bibrangedash 2129}
    \field{title}{A context quantization approach to universal denoising}
    \field{volume}{57}
    \field{journaltitle}{IEEE Trans. Signal Process.}
    \field{month}{06}
    \field{year}{2009}
  \endentry

  \entry{FigueiredoJain2002}{article}{}
    \name{author}{2}{}{%
      {{hash=FM}{%
         family={Figueiredo},
         familyi={F\bibinitperiod},
         given={M.},
         giveni={M\bibinitperiod},
      }}%
      {{hash=JA}{%
         family={Jain},
         familyi={J\bibinitperiod},
         given={A.},
         giveni={A\bibinitperiod},
      }}%
    }
    \strng{namehash}{FMJA1}
    \strng{fullhash}{FMJA1}
    \field{labelnamesource}{author}
    \field{labeltitlesource}{title}
    \field{number}{3}
    \field{pages}{381\bibrangedash 396}
    \field{title}{Unsupervised learning of finite mixture models}
    \field{volume}{24}
    \field{journaltitle}{IEEE Trans. Pattern Anal. Mach. Intell.}
    \field{month}{03}
    \field{year}{2002}
  \endentry

  \entry{JalaliPoor2014}{article}{}
    \name{author}{2}{}{%
      {{hash=JS}{%
         family={Jalali},
         familyi={J\bibinitperiod},
         given={S.},
         giveni={S\bibinitperiod},
      }}%
      {{hash=PHV}{%
         family={Poor},
         familyi={P\bibinitperiod},
         given={H.\bibnamedelima V.},
         giveni={H\bibinitperiod\bibinitdelim V\bibinitperiod},
      }}%
    }
    \strng{namehash}{JSPHV1}
    \strng{fullhash}{JSPHV1}
    \field{labelnamesource}{author}
    \field{labeltitlesource}{title}
    \field{title}{Universal compressed sensing of {M}arkov sources}
    \field{journaltitle}{Arxiv preprint arXiv:1406.7807}
    \field{month}{06}
    \field{year}{2014}
  \endentry

  \entry{Stratanovitch-Wiki}{inproceedings}{}
    \field{labeltitlesource}{title}
  \field{note}{https://en.wikipedia.org/wiki/Hubbard-\\Stratonovich\_transformation}
    \field{title}{Hubbard--{S}tratonovich transformation}
  \endentry

  \entry{PCC}{inproceedings}{}
    \field{labeltitlesource}{title}
  \field{note}{https://en.wikipedia.org/wiki/\\Pearson\_product-moment\_correlation\_coefficient}
    \field{title}{Pearson product-moment correlation coefficient}
  \endentry

  \entry{Bremaud1999}{book}{}
    \name{author}{1}{}{%
      {{hash=BP}{%
         family={Br{\'e}maud},
         familyi={B\bibinitperiod},
         given={P.},
         giveni={P\bibinitperiod},
      }}%
    }
    \list{publisher}{1}{%
      {Springer Verlag}%
    }
    \strng{namehash}{BP1}
    \strng{fullhash}{BP1}
    \field{labelnamesource}{author}
    \field{labeltitlesource}{title}
    \field{title}{{Markov {C}hains: Gibbs {F}ields, {M}onte {C}arlo
  {S}imulation, and {Q}ueues}}
    \field{volume}{31}
    \field{year}{1999}
  \endentry
\endsortlist

 \blx@bblend
  \endgroup
  \csnumgdef{blx@labelnumber@\the\c@refsection}{0}%
  \iftoggle{blx@reencode}{\blx@reencode}{}}

 \defbibheading{myheading}[BIBLIOGRAPHY]{
 \chapter*{#1}
 \markboth{#1}{#1}}

\usepackage{dcolumn}
\usepackage{bm}
\usepackage{cancel}
\usepackage{verbatim}
\usepackage{ifthen}
\usepackage{url}
\usepackage{sectsty}
\usepackage{balance} 
\usepackage{graphicx} 
\usepackage{lastpage}
\usepackage[format=plain,justification=justified,singlelinecheck=false,font=small,labelfont=bf,labelsep=space]{caption} 
\usepackage{fancyhdr}
\usepackage{abbrevs}
\usepackage{etoolbox}
\robustify{\DateMark} 

\usepackage{units} 

\usepackage[sharp]{easylist} 

\usepackage[table]{xcolor}
\usepackage{array,booktabs}
\usepackage{colortbl}
\newcolumntype{L}{@{}>{\kern\tabcolsep}l<{\kern\tabcolsep}}

\usepackage{titlesec}

\titleformat{\paragraph}
{\normalfont\normalsize\bfseries}{\theparagraph}{1em}{}
\titlespacing*{\paragraph}
{0pt}{3.25ex plus 1ex minus .2ex}{1.5ex plus .2ex}

\usepackage{placeins}

\dispositionformat{\bfseries}            

\headingformat{\large\MakeUppercase}   

\usepackage[Rejne]{fncychap}

\usepackage[
  pdfauthor={Carlos Pompeyo Ortiz},
  pdftitle={Rigidity of Microsphere Heaps},
  pdfcreator={pdftex},
  pdfsubject={NC State ETD Thesis},
  pdfkeywords={microfluidics, hard sphere, jamming, suspension, rigidity, friction, microscopy},
  colorlinks=true,
  linkcolor=black,
  citecolor=black,
  filecolor=black,
  urlcolor=black,
]{hyperref}

\usepackage{microtype}

\usepackage[T1]{fontenc}
\usepackage[adobe-utopia]{mathdesign}

\usepackage{algorithm,algpseudocode,bigints}

\makeatletter
\@addtoreset{algorithm}{chapter}
\makeatother

\usepackage[T1]{fontenc}
\usepackage{mwe}
\usepackage{bbm}

\usepackage{amsthm}

\def \x {\mathbf{x}}
\def \y {\mathbf{y}}

\def \v {\mathbf{v}}
\def \z {\mathbf{z}}
\def \g {\mathbf{g}}
\def \h {\mathbf{h}}
\def \l {\left}
\def \r {\right}
\def \T {\mathbf{T}}

\def \G {\mathbf{G}}

\def \w {\mathbf{w}}
\def \W {\mathbf{W}}
\def \u {\mathbf{u}}
\def \a {\mathbf{a}}
\def \q {\mathbf{q}}
\def \R {\mathbf{R}}
\def \ss {\mathcal{S}}

\def \P {\mathbf{P}}

\def \Z {{\cal{Z}}}
\def \X {{\cal{X}}}
\def \alphabet {\mathcal{Z}}
\def \map {{{\mathcal{A}}}}
\def \N {\mathcal{L}}
\def \replevels{{\mathcal{R}}}
\def \breplevels{\mathcal{R}_F}

\newcommand{\sij}[2] {\sum_{#1}^{#2}}
\newcommand{\pij}[2] {\prod_{#1}^{#2}}

\newcommand{\n}{{\bf n}}
\newcommand{\A}{{\bf A}}
\newcommand{\f}{{\bf f}}

\usepackage[normalem]{ulem}

\algnewcommand{\LineComment}[1]{\State \(\triangleright\) #1}
\algnewcommand\algorithmicswitch{\textbf{switch}}
\algnewcommand\algorithmiccase{\textbf{case}}
\algnewcommand\algorithmicdowhile{\textbf{dowhile}}
\algdef{SE}[SWITCH]{Switch}{EndSwitch}[1]{\algorithmicswitch\ #1\ \algorithmicdo}{\algorithmicend\ \algorithmicswitch}%
\algdef{SE}[CASE]{Case}{EndCase}[1]{\algorithmiccase\ #1}{\algorithmicend\ \algorithmiccase}%
\algdef{SE}[DOWHILE]{Do}{doWhile}{\algorithmicdo}[1]{\algorithmicwhile\ #1}%

\newtheorem{myDef}{Definition}[chapter]
\newtheorem{myTheorem}{Theorem}[chapter]
\newtheorem{myLemma}{Lemma}[chapter]
\newtheorem{myCoro}[myTheorem]{Corollary}

\newtheorem{COND}{Condition}[chapter]

\newtheorem{myRemark}{Remark}[chapter]
\newtheorem{myConj}{Conjecture}[chapter] 
\committeesize{5}

\chair{Dror Baron}
\memberI{Huaiyu Dai}
\memberII{Karen Daniels}    
\memberIII{Brian Hughes}
\memberIV{David Ricketts}   

\student{Junan}{Zhu} 

\program{Electrical Engineering\break}

\thesistitle{Statistical Physics and Information Theory Perspectives on Linear Inverse Problems}

\degreeyear{2017}

\usepackage{color}

\graphicspath{{./}}

\usepackage{citesort}
\usepackage{calc}
\newlength{\chaptercapitalheight}
\newlength{\chapterfootskip}
\renewcommand{\bibname}{BIBLIOGRAPHY}

\newlength\graphht

\usepackage{pdflscape}
\usepackage{tikz}
\fancypagestyle{lscapedplain}{%
  \fancyhf{}
  \fancyfoot{%
    \tikz[remember picture,overlay]
      \node[outer sep=1cm,above,rotate=90] at (current page.east) {\thepage};}

}

\begin{document}
\pagestyle{plain}
\frontmatter

\begin{abstract}
Many real-world problems in machine learning, signal processing, and communications assume that an unknown vector $\x$ is measured by a matrix $\A$, resulting in a vector $\y=\A\x+\z$, where $\z$ denotes the noise; we call this a single measurement vector (SMV) problem. Sometimes, multiple dependent vectors $\x^{(j)},\ j\in\{1,\cdots,J\}$, are measured at the same time, forming the so-called multi-measurement vector (MMV) problem. Both SMV and MMV are linear models (LM's), and the process of estimating the underlying vector(s) $\x$ from an LM 
given the matrices, noisy measurements, and knowledge of the noise statistics, is called a linear inverse problem. In some scenarios, the matrix $\A$ is stored in a single processor and this processor also records its measurements $\y$; this is called centralized LM. In other scenarios, multiple sites are measuring the same underlying unknown vector $\x$, where each site only possesses part of the matrix $\A$; we call this multi-processor LM. Recently, due to an ever-increasing amount of data and ever-growing dimensions in LM's, 
it has become more important to study large-scale linear inverse problems. In this dissertation, we take advantage of  tools in statistical physics and information theory to advance the understanding of large-scale linear inverse problems. The intuition of the application of statistical physics to our problem is that statistical physics deals with large-scale problems, and we can make an analogy between an LM and a thermodynamic system~\cite{Tanaka2002,GuoVerdu2005,Krzakala2012probabilistic,krzakala2012statistical,MezardMontanariBook,Barbier2015}. Therefore, we can apply statistical physics analysis tools as well as algorithmic tools into understanding large-scale LM's and their corresponding linear inverse problems. In terms of information theory~\cite{Cover06}, although it was originally developed to characterize the theoretic limits of digital communication systems, information theory was later found to be rather useful in analyzing and understanding other inference problems. We use some of the concepts and ideas of information theory to understand the theoretic performance limits in various aspects of linear inverse problems.

There exist numerous algorithms for solving linear inverse problems. However, 
only a partial understanding of the theoretic characterization of the minimum mean squared error (MMSE) when solving linear inverse problems appears in the literature~\cite{RFG2012,Tanaka2002,GuoVerdu2005}. Such a theoretic analysis helps practitioners appreciate the gap between their estimation quality and the theoretically optimal quality. Therefore, in this dissertation we use the replica analysis~\cite{Tanaka2002,GuoVerdu2005,Montanari2006,Krzakala2012probabilistic,
krzakala2012statistical,MezardMontanariBook,Barbier2015,Lesieur2015} from statistical physics to study the MMSE in MMV problems. We obtain different performance regions in which the MMSE behaves differently. Besides the quality of the estimation, there are also other ``costs'' that practitioners might care about, especially in the big data era. 
Some prior art has focused on reducing certain costs such as the communication cost~\cite{Han2014} and the computation cost~\cite{MaBaronNeedell2014}, but there has been less progress relating different costs and achieving optimal trade-offs among them. Despite the lack of such works, these trade-offs are important to system designers in order to produce efficient systems.
To address these issues, in this dissertation we use a distributed algorithm  as an example and study the behavior of the optimal communication scheme in the limit of low excess mean squared error beyond the MMSE for that distributed algorithm. Furthermore, we study the optimal trade-offs among the computation cost, the communication cost, and the quality of the estimate.

Finally, we discuss estimation algorithm design for an SMV setting.
There are numerous estimation algorithms for SMV in the prior art, but they all require some statistical knowledge about the underlying vector $\x$; in a practical setting, such knowledge might be inaccurate or unavailable. Therefore, it is important to design a {\em universal} estimation algorithm that is more agnostic to the prior knowledge of the unknown vector $\x$. In this dissertation, we design an algorithmic framework based on Markov chain Monte Carlo (MCMC) borrowed from statistical physics, and in extensive numerical experiments the algorithm achieves a mean squared error that is close to the MMSE. 
\end{abstract}

\makecopyrightpage

\maketitlepage

\begin{dedication}
 \centering To people who care about me and people who I care about. To  world peace.
\end{dedication}

\begin{biography}
Junan Zhu is currently pursuing the Ph.D. degree in the Department of Electrical and Computer Engineering (ECE) at North Carolina State University (NCSU), Raleigh, North Carolina, U.S. His research interests include compressed sensing, statistical signal processing, information theory, statistical physics, machine learning, optimization, distributed algorithms, and computational imaging.
Before joining NCSU, Mr. Zhu received the B.E. degree in Electrical Engineering with a focus on optoelectronics from the University of Shanghai for Science and Technology (USST), Shanghai, China in 2011. His research in USST focused on Terahertz waveguides and black silicon. 

Junan Zhu received the Graduate Student Fellowship at NCSU in 2011, which was awarded to the top 3 incoming ECE graduate students. He also received the National Scholarship in 2009 and the Baosteel Scholarship in 2010, both at USST.

\end{biography}

\begin{acknowledgements}
First, I would like to express my sincere gratitude to my advisor Dr. Dror Baron. It is his patient guidance and advice in research that has enlightened me and made my research life easier. It is his helpful mentoring about life in the U.S. that has provided me with enough information to merge into this new society. It is his abundant financial support that has allowed me to focus on research. (In particular, I would like to thank the generous support of the National Science Foundation and Army Research Office.\footnote{More specifically, the author was supported in part  by  the National Science Foundation under the Grants CCF-1217749 and ECCS-1611112, and by the U.S. Army Research Office under the Grants W911NF-04-D-0003 and W911NF-14-1-0314.}) Dr. Baron is more than an academic advisor. He is a mentor and a friend. I am very grateful for Dr. Baron's help and advice, and I hope to work on research projects with him in the future as well. 

Next, I would like to thank my committee, in alphabetical order: Dr. Huaiyu Dai, Dr. Karen Daniels, Dr. Brian Hughes, and Dr. David Ricketts, as well as former committee members Dr. W. Rhett Davis and Dr. Edgar Lobaton. Their helpful comments about my work and enlightening feedback greatly improved the quality of my work and dissertation. Besides my committee, I would like to thank Dr. Ahmad Beirami, Dr. Marco F. Duarte, Dr. Florent Krzakala, and Dr. Lenka Zdeborova for their advice and collaboration.
I also want to thank the lecturers of all the courses I attended. It is their clear explanations that granted me a solid understanding of various subjects in my field.

I also want to thank my dear roommates, Dr. Shikai Luo and Shuiqing Wang, whom I started my endeavor in the U.S. with and whom I shared joy and sadness with.
Also, I would like to thank them for their help on technical subjects. In addition, I would like to thank my colleagues and friends, in alphabetical order, Nicholas Casale, Miao Feng, Qian Ge, Dr. Fengyuan Gong, Dr. Xiaofan He, Yufan Huang, Richeng Jin, Nikhil Krishnan, Dr. Chengzhi Li, Wuyuan Li, Feier Lian, Dr. Juan Liu, Dr. Yuan Lu, Yanting Ma, Ryan Pilgrim, Macey Ruble, Rafael Silva, Dr. Jin Tan, Joseph Young, and Dr. Huazi Zhang. Without their help and friendship, I could not have lived a happy life while I am working toward my Ph.D. 

Furthermore, I would like to thank my college buddies, Shijie Li and Jiaming Xu, who are now pursuing their Ph.D.s as well. Without their encouragement and help, I could not have even dreamed of coming to the U.S. to pursue my Ph.D. I hope their research progress goes well and that they graduate soon. I am also very grateful to Dr. Yiming Zhu, my advisor in China, who changed my life.

At last, I would like to thank my dear parents. Whenever I need them, they are ready to help. It is their unconditional love and support that enable the endeavor of my life. They give me the courage to conquer every difficulty in the pursuit of my dream and teach me to love this world so that I am not alone. My special thanks goes to my beloved Meizhu, who accompanied me when I felt lonely, encouraged me when I was lost, and shared happiness with me whenever there were good news; life is like a box of chocolate, and you are the sweetest one.
\end{acknowledgements}

\thesistableofcontents

\thesislistoftables

\thesislistoffigures

\mainmatter

\newgeometry{margin=1in,lmargin=1.25in,footskip=\chapterfootskip, includefoot}
\chapter{Introduction}
\label{chap-intro}

Many problems in science and engineering can be approximated as linear, where an unknown vector $\x \in \mathbb{R}^N$ is measured via a matrix multiplication, $\w=\A\x$, with $\A$ being an $M \times N$ matrix. The measurements $\y$ are collected after $\w$ is corrupted by measurement noise $\z\in \mathbb{R}^M$,
\begin{equation}\label{eq:SMV}
\y = \A \x + \z.
\end{equation}
In some machine learning problems, the training set consists of $\A$ and $\y$, where $\A$ contains the features and $\y$ contains the outcomes~\cite{clickPrediction_MS2007,clickPrediction_Google2013}; $\x$ is usually called the coefficient vector that describes the relation between the features and the outcomes. In signal processing, $\A$ describes the signal acquisition system, $\y$ contains the measurements, and $\x$ is the underlying signal~\cite{DonohoCS}. For communication systems such as CDMA, the matrix $\A$ contains the spreading sequences that spread the input (channel) symbol from each user, and then the receiver mixes the spread symbols from different users and obtains $\y$~\cite{GuoVerdu2005}. The input symbols from different users at a certain time interval form the vector $\x$. For  ease of presentation, we call the underlying input vector $\x$ the {\em signal}, $\A$ the {\em measurement matrix}, and  $\y$ the {\em measurements} vector.
In the following, we introduce several variants of our setting~\eqref{eq:SMV} and then discuss the prior art in solving the linear models.

\section{Linear Models and Linear Inverse Problems}
\subsection{Problem setting}\label{sec:Chap1-setting}
There are some variants of linear models (LM's). Based on how the measurements $\y$ and the matrix $\A$ are stored, we form centralized LM's or multi-processor LM's. We can also define linear models based on the number of underlying unknown vectors $\x$: if there is only one unknown vector $\x$, then it is a single measurement vector (SVM) problem; if there are more than one unknown vector $\x$, then we form a multi-measurement vector (MMV) problem.

{\bf Centralized vs. multi-processor LM's:}
If the matrix $\A$ and the measurements $\y$ in~\eqref{eq:SMV} are stored in a single processor, then we call the LM a  {\em centralized LM}.
Recently, there is an increasing amount of data being generated in various applications. For example, the trend of relying on Internet services and social networks is more prevalent than ever before; users of web services are generating numerous log files daily. As another example, financial analysts need to predict the changes in prices based on historical price information. Given the amount of financial derivatives and the high frequency of changes in prices, financial institutions are also overwhelmed by a vast amount of data. Another example involves recent advances in wearable devices. Health care providers can provide patients with wearable sensors that record and report the health status of patients frequently, so that the health care providers can react quickly once there is an emergency. With these ever-growing amounts of data, it is no longer practical to fit these data into a single machine, and distributed and scalable file systems such as Hadoop Distributed File Systems (HDFS)~\cite{DeanGhemawat2008} have been developed. For the case of LM, if the matrix $\A$ and the measurements $\y$ are so big that they have to be stored in a distributed file system such as HDFS, then we form a {\em multi-processor (MP) LM}~\cite{Mota2012,Patterson2013,Patterson2014,Han2014,Han2015ICASSP,Ravazzi2015,Han2015SPARS,HanZhuNiuBaron2016ICASSP}. Consider an MP-LM with $P$ distributed {\em processor nodes} and a {\em fusion center}. Each distributed processor node stores $\frac{M}{P}$ rows of the matrix $\A$, and acquires the corresponding measurements of the underlying signal $\x$. Without loss of generality, the LM in distributed processor node $p\in \{1,\cdots,P\}$ can be written as
 \begin{equation}\label{eq:one-node-meas_intro}
    y_i=\A_i \x+z_i,\ i\in \left\{\frac{M(p-1)}{P}+1,\cdots,\frac{Mp}{P}\right\},
 \end{equation}
 where $\A_i$ is the $i$-th row of $\A$, and $y_i$ and $z_i$ are the $i$-th entries of $\y$ and $\z$, respectively.

{\bf Single measurement vector vs. multiple measurement vectors:}
Apart from the MP-LM, another type of distributed linear model involves multiple sensors. Using multiple sensors can accelerate the sensing speed by pointing different sensors at different regions of interest, which we call {\em distributed sensing}~\cite{Duarte2006IPSN,HN05,BaronDCStech}.
In distributed sensing, suppose that $J$ sensors are measuring $J$ signal vectors, $\x^{(1)},\cdots,\x^{(J)}$. Each signal vector $\x^{(j)}$ is measured by
a matrix $\A^{(j)}$, which models the sensing mechanism of each sensor, and the measurements $\y^{(j)}$ are corrupted by independent and identically distributed (i.i.d.) noise $\z^{(j)}$,
\begin{equation}\label{eq:MMVmodel_intro}
\y^{(j)}=\A^{(j)}\x^{(j)}+\z^{(j)},\quad j\in\{1,\cdots,J\},
\end{equation}
where the $(j)$ in the super-script denotes the index of the corresponding sensor.
Of particular interest in reducing the number of measurements while achieving similar signal estimation quality, distributed sensing leads to a proliferation of research on the MMV problem~\cite{chen2006trs,cotter2005ssl,Mishali08rembo,Berg09jrmm}, in which
the $J$ sparse signal vectors $\x^{(j)},\ j\in\{1,\cdots,J\}$, share common non-zero supports, as explained below. Let us construct a {\em super-symbol} $\x_l=\l[x_l^{(1)},\cdots,x_l^{(J)}\r]^{\top}$, where $\{\cdot\}^{\top}$ denotes the transpose, and $x^{(j)}_l$ is the $l$-th entry of the signal vector $\x^{(j)}$. The super-symbols $\x_l,\ l\in\{1,\cdots,N\}$, follow an i.i.d.  $J$-dimensional joint distribution,
\begin{equation}\label{eq:chap1-jsm}
f(\x_l)=\rho \phi(\x_l)+(1-\rho)\delta(\x_l),
\end{equation}
where $\rho$ is the {\em sparsity rate}, $\phi(\x_l)$ is a $J$-dimensional joint distribution, and $\delta(\x_l)$ is the Dirac delta function for $J$-dimensional vectors.
When the number of signal vectors becomes 1, i.e., $J=1$, this MMV problem~\eqref{eq:MMVmodel_intro} becomes an SMV problem.
The MMV problem has many
applications such as radar array signal processing,
acoustic sensing with multiple speakers, magnetic resonance imaging
with multiple coils~\cite{JuYeKi07,JuSuNaKiYe09}, and diffuse optical tomography using multiple
illumination patterns~\cite{LeeKimBreslerYe2011}.

{\bf Linear inverse problem:} Usually, estimation algorithms need to be designed to estimate the
signal $\x$ given the matrix $\A$, noisy measurements $\y$, and possible statistical knowledge about the noise $\z$. We call this a linear inverse problem.

In this work, we focus on the {\em large system limit} defined below.
\begin{myDef}[Large system limit~\cite{GuoWang2008}]\label{def:chap1-largeSystemLimit}
The signal length $N$ scales to infinity, and the
number of measurements $M=M(N)$ depends on $N$ and also scales to infinity, where
the ratio approaches a positive constant $\kappa$,
\begin{equation*}
\lim_{N\rightarrow\infty} \frac{M(N)}{N} = \kappa>0.
\end{equation*}
\end{myDef}
We call $\kappa$ the measurement rate.

\subsection{Prior art and open questions}\label{sec:priorArt}
Linear models are widely studied and find extensive real-world applications. Over the years, people have developed various algorithms to solve the underlying signal vectors for linear models.
Many estimation algorithms pose a sparsity prior on the signal $\x$ or the coefficient vector $\theta$~\cite{CandesRUP,DonohoCS,GPSR2007}, where
$\theta=\W^{-1} \x$, and $\W$ is called the {\em sparsifying transform} that renders a sparse coefficient vector $\theta$.
A second, separate class of Bayesian algorithms to solve the linear inverse problem poses a
probabilistic prior for the coefficients of $\x$ in a known transform
domain~\cite{DMM2010ITW1,RanganGAMP2011ISIT,BCS2008,BCSEx2008,CSBP2010}. Given a probabilistic model, some related message passing approaches
learn the parameters of the signal model and achieve the minimum mean squared error (MMSE) in some settings; examples include EM-GM-AMP-MOS~\cite{EMGMTSP}, turboGAMP~\cite{turboGAMP}, and AMP-MixD~\cite{MTKB2014ITA}. As a third alternative,
complexity-penalized least square methods~\cite{Figueiredo2003,DonohoKolmogorovCS2006,HN05,HN11,Ramirez2011} can use arbitrary prior information on the
signal model and provide analytical guarantees, but are only computationally
efficient for specific signal models, such as the independent-entry Laplacian
model~\cite{HN05}. For example, Donoho et al.~\cite{DonohoKolmogorovCS2006} relies on Kolmogorov complexity, which cannot be computed~\cite{Cover06,LiVitanyi2008}.
As a fourth alternative, there exist algorithms that can formulate dictionaries that yield
sparse representations for the signals of interest when a large amount of training data is
available~\cite{Ramirez2011,AharoEB_KSVD,Mairal2008,Zhoul2011}.
When the signal is non-i.i.d., existing algorithms require either prior knowledge of the probabilistic model~\cite{turboGAMP} or the use of training data~\cite{Garrigues07learninghorizontal}. In spite of the numerous algorithms to solve the linear inverse problem, there are many important gaps in the prior art, such as those listed below.
\begin{enumerate}
\item {\bf What is the best we can do?} Along with existing algorithms for solving linear inverse problems, researchers often provide theoretic estimation accuracy guarantees for these algorithms. However, what is often missing is the optimal estimation quality associated with the linear inverse problem itself, instead of the optimal estimation quality for a specific algorithm. Such a theoretic analysis will help us evaluate the quality of each algorithm and identify the gap between a specific algorithm and the theoretically optimal estimation quality.
\item {\bf What are the costs of running an algorithm?} Nowadays, due to the large amounts of data mentioned in Section~\ref{sec:Chap1-setting}, many systems are designed in a distributed fashion. Hence, estimation algorithms need to run in a distributed network and thus incur communication costs. There exists some work trying to save communication by designing cache systems so that each node in the network does not need to send every piece of data every time~\cite{CodedMapReduce2015Allerton,LMYA2016ISIT}. There are also some works using heuristics in reducing the precision of the floating-point numbers sent across the network~\cite{clickPrediction_Google2013,Thanou2013}. However, there is little prior art discussing the ``optimal'' communication scheme.
\item {\bf Better algorithms?} At the beginning of this section, we briefly discussed some classes of algorithms. In certain cases, one might not be certain about the structure or statistics of the signal prior to estimation. Uncertainty about such structure may result in a sub-optimal choice of
the sparsifying transform $\W$, yielding a coefficient vector $\theta$  that requires more measurements to achieve reasonable
estimation quality; uncertainty about the statistics of the signal will make it difficult to
select a prior or model for Bayesian algorithms. Thus, we think that a ``better'' algorithm should be more agnostic to the particular statistics of the signal while still achieving reasonable estimation results.
\end{enumerate}

\subsection{Contributions}
In the following, we briefly discuss our contributions corresponding to each of the unsolved problems raised in Section~\ref{sec:priorArt}. Most of our contributions are made possible by taking advantage of  statistical physics tools and information theory.
\begin{enumerate}
\item {\bf Characterizing the optimal estimation quality:}  In Chapter~\ref{chap-MMV}, we make an analogy between the MMV problem~\eqref{eq:MMVmodel_intro} and a thermodynamic system and use the
replica analysis~\cite{Tanaka2002,GuoVerdu2005,Montanari2006,Krzakala2012probabilistic,
krzakala2012statistical,MezardMontanariBook,Barbier2015,Lesieur2015} from statistical physics to analyze the information theoretic MMSE for MMV problems with i.i.d. Gaussian measurement matrices and i.i.d. Gaussian noise. Our analysis is readily extended to other i.i.d. measurement matrices and i.i.d. measurement noise.
Note that the MMSE is associated with the MMV problem~\eqref{eq:MMVmodel_intro} itself and is not associated with any specific estimation algorithms. Realizing that mean squared error (MSE) might not be the only metric that is of interest, we propose a future direction to extend the work of Tan and coauthors~\cite{Tan2014,Tan2014Infty} to analyze the average error based on arbitrary user-defined error metrics for MMV problems.
\item {\bf Optimal trade-offs among different costs:} In Chapter~\ref{chap-MP-AMP}, we
apply rate-distortion theory~\cite{Cover06,Berger71,GershoGray1993,WeidmannVetterli2012} to optimize the communication cost in a specific distributed algorithm, and propose a method to find the optimal combined cost of computation and communication. In addition, we study the asymptotic behavior of the optimal communication scheme in the limit of
low excess MSE beyond the MMSE. Also, recognizing that we cannot minimize the computation cost, communication cost, and the quality of the estimate simultaneously, we study the optimal trade-offs among these different costs.
\item {\bf Designing better algorithms:} In Chapter~\ref{chap-SLAM}, we propose a {\em universal} algorithm that is based on the mild assumption of the signal being ``simple,'' i.e., there is some structure in the signal that is simple. Our algorithm is based on ``simulated annealing,'' a mathematical analogy to a statistical physics concept, and achieves favorable estimation accuracy while using limited prior information about the signal models. In Chapter~\ref{chap-SLAM}, we also briefly discuss another universal algorithm that is based on belief propagation~\cite{DMM2009,CSBP2010,Bayati2011,Montanari2012,Krzakala2012probabilistic,krzakala2012statistical,Barbier2015}, which originates from statistical physics and information theory. We refer interested readers to Ma et al.~\cite{MaZhuBaronAllerton2014,MaZhuBaron2016TSP}.
\end{enumerate}

The underlying intuition of why statistical physics and information theory can be useful in tackling our problems is that they both deal with large systems, and fortunately, the problems that we are targeting in this dissertation are indeed large systems. Moreover, the general formulations of our problems create analogies between our problems and thermodynamic systems and communication systems, so that we can take advantage of the existing analytical and algorithmic tools in the rich fields of statistical physics and information theory.

\section{Organization, Notations, and Acronyms}
\subsection{Organization}
This dissertation is organized as follows. Chapter~\ref{chap-basics} introduces some background on statistical physics and information theory. Chapter~\ref{chap-MMV} studies the MMSE and its behavior for MMV problems; we also propose a future direction to study arbitrary user-defined error metrics for MMV problems. The  limiting behavior of the optimal communication scheme and the optimal trade-offs among different costs in MP-LM's are discussed in Chapter~\ref{chap-MP-AMP}. In Chapter~\ref{chap-SLAM}, we propose a universal algorithmic framework that achieves favorable estimation quality. Chapter~\ref{chap-discuss} concludes the dissertation and proposes some future directions. Details about some proofs appear in the appendices. 

Note that Chapter~\ref{chap-MMV} is based on our work with Baron~\cite{ZhuBaronCISS2013} and with Baron and Krzakala~\cite{ZhuBaronKrzakala2016}. Chapter~\ref{chap-MP-AMP} is based on our work with Han et al.~\cite{HanZhuNiuBaron2016ICASSP} and with Baron and Beirami~\cite{ZhuBeiramiBaron2016ISIT,ZhuBaronMPAMP2016ArXiv}. Chapter~\ref{chap-SLAM} is based on our work with Baron and Duarte~\cite{JZ2014SSP,ZhuBaronDuarte2014_SLAM}.

\subsection{Notations}
In this dissertation, bold capital letters represent matrices, bold lower case letters represent vectors, and normal font letters represent scalars. The entry (scalar) in the $i$-th row, $j$-th column of a matrix $\A$ is denoted by $A_{i,j}$, where the comma is often omitted. The $i$-th entry (scalar) in a vector $\z$ is denoted by $z_i$. Following are some frequently used notations.
\begin{itemize}
\item $\A$: Measurement matrix
\item $\mathbb{C}$: The set of complex numbers
\item $D$: Distortion
\item $\delta(\cdot)$: Dirac delta function
\item $f(\cdot)$: Probability density function (continuous variable)
\item $\mathbb{E}[\cdot]$: Expectation
\item $\kappa$: Measurement rate
\item $M$: Number of measurements
\item $N$: Signal length
\item $\mathbb{N}$: The set of natural numbers, i.e., $\{0,1,\cdots \}$
\item $\mathcal{N}(\mu,\sigma^2)$: Gaussian distribution with mean $\mu$ and variance $\sigma^2$
\item $R$: Coding rate
\item $\mathbb{R}$: The set of real numbers
\item $\mathbb{P}$: Probability
\item $\mathbb{P}(\cdot)$: Probability mass function (discrete variable)
\item $\rho$: Sparsity rate (percentage of non-zeros in a vector)
\item $\sigma_Z^2$: Variance of the noise $\z$
\item $t$: Iteration index
\item $\A^{\top}$: Transpose of matrix $\A$
\item $\x$: Signal
\item $\|\x\|_p$: $\ell_p$ norm of a vector $\x$; if $p$ is not specified, then we refer to $\ell_2$ norm
\item $\y$: Measurements
\item $\z$: Noise
\item $[x_1,x_2,\cdots,x_N]$: The vector consists of $x_1,x_2,\cdots,x_N$
\item $\{1,2,\cdots,N\}$: The set consists of $1,2,\cdots,N$
\end{itemize}

\subsection{Acronyms}
\begin{itemize}
\item AMP: Approximate message passing
\item BP: Belief propagation
\item CS: Compressed sensing
\item i.i.d.: Independent and identically distributed
\item LM: Linear model
\item MMSE: Minimum mean squared error
\item MMV: Multi-measurement vector
\item MP: Multi-processor
\item MSE: Mean squared error
\item PMF: Probability mass function
\item RD: Rate-distortion
\item SDR: Signal-to-distortion ratio
\item SMV: Single measurement vector
\item SNR: Signal-to-noise ratio
\end{itemize}

\chapter{Statistical Physics and Information Theory Background}
\label{chap-basics}
\chaptermark{Stat. Phys. \& Inf. Theory}

In Chapter~\ref{chap-intro}, we discussed the prior art and mentioned that our contributions are made possible by tools in statistical physics and information theory. Due to the interdisciplinary nature of this dissertation, this chapter briefly reviews some concepts and methodologies that are used in our work. We refer readers who are interested in delving into these subjects to the books by M{\'e}zard and Montanari~\cite{MezardMontanariBook} and by Cover and Thomas~\cite{Cover06}.

\section{Relevant Statistical Physics Concepts}\label{sec:background_statPhys}
\sectionmark{Relevant Stat. Phys. Concepts}
Statistical physics studies a disordered thermodynamic system containing a large number of particles that are interacting with each other by the internal force between (among) the particles as well as the external force applied to the entire disordered system.

\subsection{Basics}\label{sec:statPhysBasics}
In this section, we briefly introduce some concepts that are frequently used in statistical physics.

{\bf Entropy (thermodynamics):} Entropy quantifies the amount of disorder of a thermodynamic system,
\begin{equation}\label{eq:def_entropy}
\mathcal{S}(\x)=-\sum_{\x} \mathbb{P}(\x)\log \mathbb{P}(\x),
\end{equation}
where the vector $\x$ describes the {\em configuration} of a certain thermodynamic system and $\mathbb{P}(\x)$ is the probability of a certain configuration existing in the disordered system. By summing over all possible configurations and accounting for their corresponding probability, we are able to obtain the level of disorder, or the {\em entropy} of this particular thermodynamic system.

{\bf Boltzmann distribution:} In a thermodynamic system, the higher the temperature is, the more disordered the system is. The Boltzmann distribution is a probability distribution used to describe various possible configurations in a thermodynamic system,
\begin{equation}\label{eq:def_Boltzmann_background}
\mathbb{P}(\x)=\frac{1}{Z}\text{exp}\l(-\frac{H(\x)}{T}\r),
\end{equation}
where the vector $\x$ describes the configuration of a thermodynamic system, $T$ is the temperature of this system, $H(\x)$ is the energy for a certain configuration, and $Z$ is a normalizer called the {\em partition function}. If the thermodynamic system is in a high temperature, i.e., $T$ is large, then the probabilities for configurations with different energy are approximately the same and the system reaches the maximum entropy~\eqref{eq:def_entropy}, which corresponds to the greatest amount of disorder.

{\bf Annealing and quench:} The configuration associated with the lowest energy can be obtained through a process called annealing, where a disordered system gradually cools down. Intuitively, when the temperature $T$ decreases, the configurations with lower energy becomes more and more likely in the disordered system, according to~\eqref{eq:def_Boltzmann_background}. Given enough time that allows a slow enough decrease in the temperature, we can guarantee to obtain the globally minimum energy configuration. A related concept is {\em quench}, in which the temperature is quickly decreased, so that the disordered system is likely to achieve a local minimum energy configuration. Since the temperature is quickly decreased, once a local minimum energy configuration appears, it will be difficult to generate other lower energy configurations according to~\eqref{eq:def_Boltzmann_background}.

\subsection{Spin glass theory basics}
A basic understanding of spin glass theory provides new perspectives when solving linear inverse problems. In the following, we introduce some basics of spin glass theory. The goal is to provide intuition, and we refer interested readers to M{\'e}zard and Montanari~\cite{MezardMontanariBook} for rigorous and detailed explanations.

{\bf Mean-field spin glasses:} As discussed in Section~\ref{sec:statPhysBasics}, the thermodynamic system we are interested in contains many particles. A {\em simple model} in the mean-field spin glass theory models each of the particles as a spinning glass, where each glass has two spinning states. {\em In this simple model}, there exist internal forces between {\em each pair} of the spinning glasses. Moreover, we assume that there is an external force that can affect the states of the glasses. Hence, the overall energy
of a specific thermodynamic system for a specific {\em configuration} $\x$ is
\begin{equation}\label{eq:def_Hamiltonian}
H(\x)=-\sum_{i}\sum_{j<i} r_{ij}x_ix_j-\sum_{i} h_i x_i,
\end{equation}
where $x_i$ is the $i$-th element of the configuration (vector) $\x$ and it represents the state of the $i$-th glass, $r_{ij}$ models the force between glass $i$ and glass $j$, and $h_i$ models the external force applied to glass $i$. This model is illustrated in Figure~\ref{fig:spinGlass}, where each dot represents a glass, and the vertical arrows denote the state of each glass. The remaining arrows illustrate the internal forces between pairs of spin glasses and the curve in the bottom panel illustrates the external force. The energy function~\eqref{eq:def_Hamiltonian} is often called the {\em Hamiltonian}. Note that the Hamiltonian~\eqref{eq:def_Hamiltonian} is {\em quenched}, because we assume that $r_{ij}$ and $h_i$ are constant.

\begin{figure}
  \centering
  \includegraphics[width=8cm]{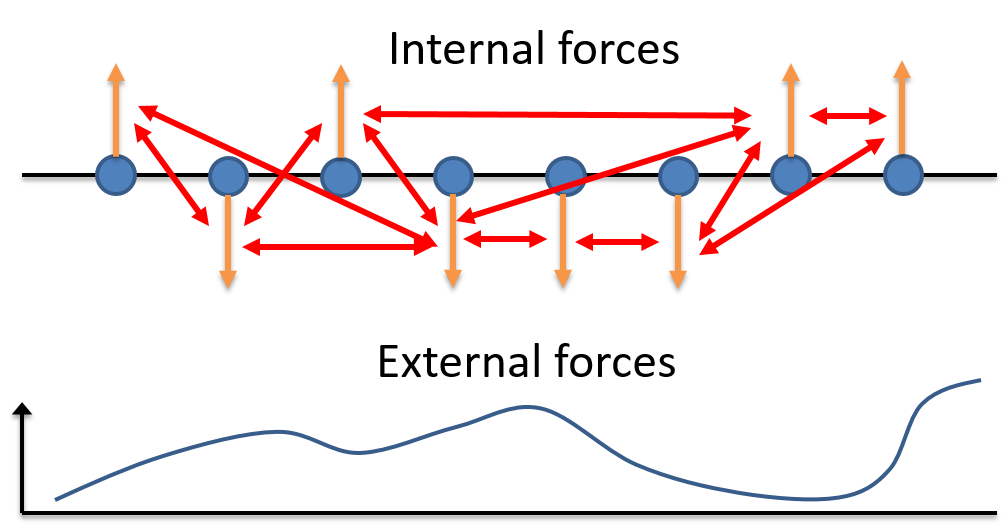}
  \caption{Illustration of spin glasses with internal and external forces. Each dot represents a spin glass. Vertical arrows denote the state of each glass. The remaining arrows illustrate the internal forces between pairs of spin glasses and the  curve in the bottom panel illustrates the external force. Figure inspired by Ralf R. M\"{u}ller.}\label{fig:spinGlass}
\end{figure}

One of the things that nature does is maximizing the entropy~\eqref{eq:def_entropy} of a thermodynamic system for a given energy (because energy is assumed to be conserved),
\begin{equation}\label{eq:def_energy}
\mathcal{E}=\sum_{\x} \mathbb{P}(\x)H(\x).
\end{equation}
It can be proved that the Boltzmann distribution~\eqref{eq:def_Boltzmann_background} maximizes the entropy~\eqref{eq:def_entropy} for a given energy~\eqref{eq:def_energy}. Moreover, the energy $H(\x)$ in the Boltzmann distribution~\eqref{eq:def_Boltzmann_background} is the Hamiltonian for configuration $\x$~\eqref{eq:def_Hamiltonian}.

{\bf Free energy and self-averaging:} Sometimes, instead of (mathematically) evaluating the maximum entropy~\eqref{eq:def_entropy}, it is more convenient to evaluate the minimum {\em free energy} given by
\begin{equation}\label{eq:def_freeEnergy_0}
\mathcal{F}=\mathcal{E}-T\mathcal{S}.
\end{equation}
Using~\eqref{eq:def_entropy},~\eqref{eq:def_Boltzmann_background}, and~\eqref{eq:def_energy} with normalization by the number of spin glasses $N$, we simplify~\eqref{eq:def_freeEnergy_0} as
\begin{equation}\label{eq:def_freeEnergy}
\mathcal{F}=-\frac{T}{N}\log Z,
\end{equation}
where the partition function $Z$ is the normalizer in~\eqref{eq:def_Boltzmann_background}.
Note that because the Hamiltonian~\eqref{eq:def_Hamiltonian} is quenched, the free energy~\eqref{eq:def_freeEnergy} is quenched.

The expression in~\eqref{eq:def_freeEnergy} is undesirable, because we have to calculate the free energy for each of the quenched Hamiltonians. Physically, it means that we need to carry out this calculation for every specific piece of material. It turns out that when the size of the system is sufficiently large, the properties of the system do not depend on the specific settings of $r_{ij}$ and $h_i$ any more~\eqref{eq:def_Hamiltonian}, which is the so-called {\em self-averaging} property of a thermodynamic system, given sufficiently many particles. Hence, we define the free energy as
\begin{equation}\label{eq:def_freeEnergy_selfAveraging}
\mathcal{F}=-\lim_{N\rightarrow \infty}\frac{T}{N} \mathbb{E}\l[\log Z\r].
\end{equation}

\sectionmark{Relevant Stat. Phys. Concepts}
\section{Information Theory and Coding Theory}
\sectionmark{Inf. \& Coding Theory}

This section discusses some important results from information theory and coding theory that are relevant to this dissertation. The author refers interested readers to the book by Cover and Thomas~\cite{Cover06} for further details and more comprehensive explanations. Coding theory and information theory are quite related and are both widely used in digital communication systems, and we simply call them ``information theory'' for brevity. Seeing that information theory is widely used in digital communication systems, we start by introducing the components of a typical digital communication system. But before that, we must understand the 
most basic of concepts: the bit.

{\bf Bit:} A bit is a unit that can represent two states. We could call these two states  0 and 1, or -1 and +1, and so on. Why are bits so important? Before entering the digital world, people used analog electronics. One of the key challenges was the noise in the signal. For example, in order to represent a number $1.2$, a waveform of magnitude 1.2 needs to be formed and transmitted. However, due to various noise and distortions, what the receiver receives is not exactly 1.2, which is undesirable. In the digital world, devices use sequences of bits to represent a number such as 1.2. The advantage of digital electronics is that they use ``bits'' that only have two states: the circuit is either on or off. The recognition and identification of a bit are much easier than recognizing and identifying analog waveforms. Information theory provides theoretical bounds for various errors when using bits, and proves that by using bits a digital communication system can exploit the communication channel as well as an analog communication system does. Moreover, information theory develops many techniques to achieve these theoretical bounds.

{\bf Components of a digital communication system:}
As illustrated in Figure~\ref{fig:commSys}, there are 7 key components of a typical digital communication system. First, the signal is encoded (compressed), so that the communication system does not need to send as many bits as required by the original signal; this step is called {\em source encoding}. Then, the encoded (compressed) signal is passed through a channel encoder, in which redundancy is introduced to the bit sequence. This redundancy is crucial to better utilize the energy of the transmitter and the channel. Next, the redundant sequence is modulated to an analog waveform by one of the available modulation schemes. After modulation, the transmitter sends the modulated signals (analog) through a noisy channel and the receiver receives a noisy sequence that contains the information of the original signal. Then, the receiver demodulates the noisy analog waveform into a sequence of bits. After that, the receiver decodes (channel decoder) the sequence to remove redundancy.\footnote{There will be errors in the demodulated sequence. By introducing redundancy in the channel encoding step, the channel decoder can identify and correct errors due to the noisy channel.} Finally, with an error-free (hopefully) sequence of bits, the last step is to decompress the data. 

\begin{figure}
  \centering
  \includegraphics[width=8cm]{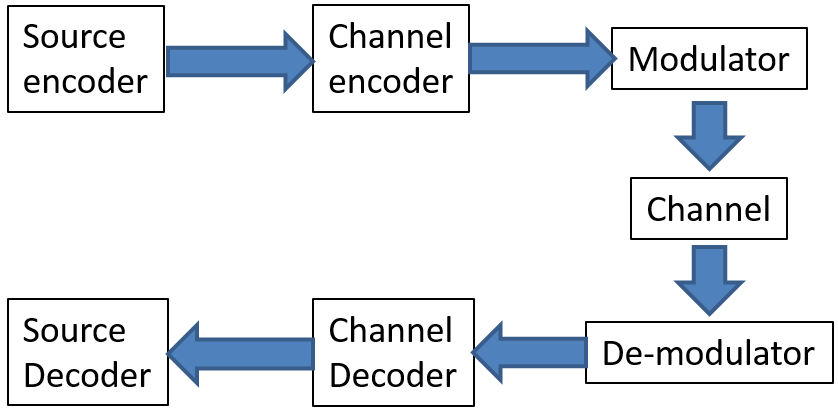}
  \caption{Illustration of a typical digital communication system. Figure inspired by Brian Hughes' slides.}\label{fig:commSys}
\end{figure}

{\bf Link between statistical physics and information theory:}\footnote{Interested readers may want to refer to Merhav~\cite{NeriBook}.} In Section~\ref{sec:background_statPhys}, we denote the configuration of a thermodynamic system by a vector $\x=[x_1,\cdots,x_N]$, where $x_i,\ i\in\{1,\cdots,N\}$, represents the state of the $i$-th spin glass. In information theory, we typically use $\x=[x_1,\cdots,x_N]$ to represent a length-$N$ signal. This signal $\x$ is passed through a channel. The counterparts of the channel in digital communication systems for statistical physics are the internal and external forces that interact with the particles of the thermodynamic system. With this brief analogy, we start introducing some important concepts and results in information theory.

{\bf Entropy (information theory):}
We have introduced entropy~\eqref{eq:def_entropy} in statistical physics. In information theory, entropy quantifies the amount of information carried by a certain signal $\x$.
If the entries of $\x$ take discrete values, then the expression for entropy in information theory is the same as~\eqref{eq:def_entropy}, and the only difference is that $\mathbb{P}(\x)$ represents the joint probability mass function of a signal $\x$. If the entries of $\x$ are continuous, then the entropy in information theory for a signal $\x$ is
\begin{equation}\label{eq:entropy_continuous}
\mathcal{S}(\x)=-\int_{\x} f(\x)\log [f(\x)] d \x,
\end{equation}
where $f(\x)$ is the joint probability density function of $\x$. 

{\bf Coding rate (source encoder):}
Before transmitting the signal $\x\in\mathbb{R}^N$ to the receiver, a communication system typically first compresses the signal, so that it can save in communication load. The coding rate is defined as
\begin{equation}\label{eq:codingRate}
R=\frac{\text{Number of bits after compression}}{N}.
\end{equation}

{\bf Distortion:} After receiving the encoded signal,\footnote{According to Figure~\ref{fig:commSys}, after data compression and before transmitting the sequence, there is typically a channel encoding step, which helps to exploit the channel to a greater extent. Here, we assume perfect channel decoding. Interested readers can refer to Cover and Thomas~\cite{Cover06}.} the receiver needs to decode it. There are two types of data compression that can be used in the source encoder. One is {\em lossless compression} and the other is {\em lossy compression}. In lossless compression, after the source decoder decodes the data sequence, it obtains a signal that is identical to the original signal. In lossy compression, the signal obtained after decoding is somewhat distorted from the original signal. The cause of this {\em distortion} is the {\em quantization} process when encoding the signal in a lossy way. A typical quantizer builds a ``grid'' in the space of  value(s) to be quantized. Next, the quantizer rounds the value(s) to the nearest point on the grid. As an example, the scalar quantizer~\cite{GershoGray1993,Cover06} rounds each (scalar) entry in the signal to the nearest grid point. The vector quantizer~\cite{LBG1980,Gray1984,GershoGray1993} rounds sequence of scalars to the nearest hyper-grid point.

Denote the distance between a certain entry in the original signal $x_i$ and the corresponding entry in the decoded signal $\widehat{x}_i$ by $d(x_i,\widehat{x}_i)$, where we can use various distance functions~\cite{FunctionalAnalysisBook} for $d(\cdot,\cdot)$. The average distortion of the entire signal is given by
\begin{equation}\label{eq:avgDistortion}
D=\frac{1}{N}\sum_{i=1}^N d(x_i,\widehat{x}_i).
\end{equation}

\begin{figure}[t]
  \centering
  \includegraphics[width=6cm]{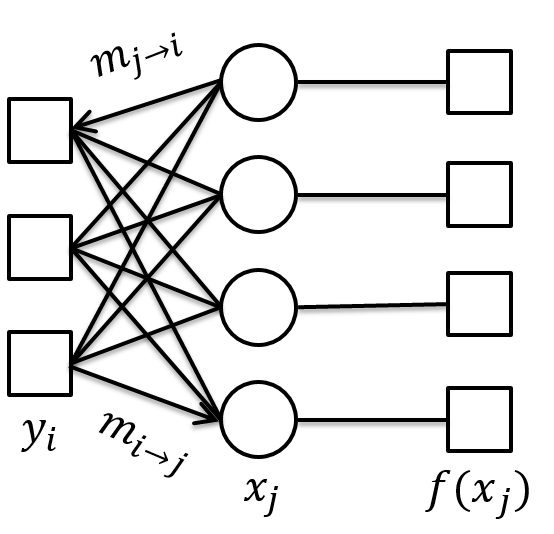}
  \caption{Illustration of belief propagation. The boxes are called the factor nodes and the circles are called the variable nodes.}\label{fig:TannerGraph}
\end{figure}

{\bf Rate-distortion theory:} There is a fundamental information theoretic relation between the rate~\eqref{eq:codingRate} and distortion~\eqref{eq:avgDistortion}. With a certain quantization scheme and  knowledge about the distribution of the signal, we can calculate the coding rate $R$~\eqref{eq:codingRate} and the expected distortion $D$~\eqref{eq:avgDistortion}. Although this calculation is not always an easy task~\cite{Arimoto72,Blahut72,Rose94}, a pivotal message from rate-distortion theory is that we can save a lot in the coding rate $R$~\eqref{eq:codingRate} by allowing a small distortion $D$~\eqref{eq:avgDistortion}.

{\bf Cavity method and belief propagation:} We can regard the linear model in~\eqref{eq:SMV} as a communication channel, where $\x$ is the signal to be transmitted, $\A$ models the transmission scheme, $\z$ is the noise in the receiver, and $\y$ is the received sequence.
Belief propagation (BP)~\cite{DMM2009,CSBP2010,Bayati2011,Montanari2012,Krzakala2012probabilistic,krzakala2012statistical,Barbier2015}  is an algorithm that can be used to infer the underlying signal $\x$ in the channel~\eqref{eq:SMV}. BP was invented independently by researchers in coding theory, statistical physics, and artificial intelligence. First of all, we represent the channel~\eqref{eq:SMV} as a Tanner graph in Figure~\ref{fig:TannerGraph}, where we express each entry $x_j$ of the signal $\x$ by a variable node (circles in Figure~\ref{fig:TannerGraph}), driven by its distribution $f(x_j)$ from a factor node (boxes in Figure~\ref{fig:TannerGraph}). Then, variable nodes are interacting with the factor nodes $y_i$'s. 

The messages $m_{i\rightarrow j}(x_j)$ and $m_{j\rightarrow i}(x_j)$ given by the canonical BP updating rules for the posterior distribution $f(\x|\y)$ are as follows,
\begin{equation}\label{eq:canonicalBP}
\begin{split}
m_{i\rightarrow j}(x_j)&=\frac{1}{Z_{i\rightarrow j}} \bigintss \l[\prod_{k\neq j} m_{k\rightarrow i}(x_k)\r]\operatorname{e}^{-\frac{1}{2\sigma_Z^2}\l(\sum_{k\neq j} A_{ik}x_k+A_{ik}x_k-y_i\r)^2}\l[\prod_{k\neq j}d x_k\r],\\
m_{j\rightarrow i}(x_j)&=\frac{1}{Z_{j\rightarrow i}}f(x_j)\prod_{q\neq j} m_{q\rightarrow j}(x_j).
\end{split}
\end{equation}
Note that in statistical physics, the factor nodes model the forces between (or among) spin glasses (variable nodes). When $\A$ is sparse or locally tree-like, BP yields an estimate that converges to the true posterior distribution $f(\x|\y)$. With this posterior distribution, we obtain the estimate $\widehat{\x}=\mathbb{E}[\x|\y]$ of the original signal that achieves the smallest mean squared error~\cite{RanganGAMP2011ISIT}.

\chapter{Minimum Mean Squared Error for Multi-measurement Vector Problem}
\label{chap-MMV}
\chaptermark{MMSE for MMV}

The multi-measurement vector (MMV) problem~\eqref{eq:MMVmodel_intro} considers the estimation of a set of sparse signal vectors
that share common supports, and has applications such as radar array signal processing,
acoustic sensing with multiple speakers, magnetic resonance imaging
with multiple coils~\cite{JuYeKi07,JuSuNaKiYe09}, and diffuse optical tomography using multiple
illumination patterns~\cite{LeeKimBreslerYe2011}.  
In this chapter, which is based on our work with Baron~\cite{ZhuBaronCISS2013} and with Baron and Krzakala~\cite{ZhuBaronKrzakala2016}, two related MMV settings are studied. In the first setting, each signal vector is measured by a different independent and identically distributed (i.i.d.) measurement matrix,
while in the second setting, all signal vectors are measured by the same i.i.d. matrix. Although there are many algorithms~\cite{DuarteWakinBaronSarvothamBaraniuk2013,tropp2006ass,chen2006trs,malioutov2005ssr,tropp2006ass2,cotter2005ssl, Mishali08rembo,LeeBreslerJunge2012,YeKimBresler2015,ZinielSchniter2011} for solving the unknown vectors in the MMV problem~\eqref{eq:MMVmodel_intro}, the performance limits of MMV signal estimation
in the presence of measurement noise have not been studied. In this chapter, replica analysis~\cite{Tanaka2002,GuoVerdu2005,Montanari2006,Krzakala2012probabilistic,
krzakala2012statistical,MezardMontanariBook,Barbier2015,Lesieur2015}, borrowed from statistical physics, is performed for these two MMV settings, and the minimum mean squared error (MMSE), which turns out to be identical for both settings, is obtained as a function of the noise variance and number of measurements. To showcase the application of MMV models, the MMSE's of complex single measurement vector (SMV) problems with both real and complex measurement matrices are also analyzed. Multiple performance regions for MMV are identified where the MMSE behaves differently as a function of the noise variance and the number of measurements.

Belief propagation (BP) is a signal estimation framework for linear inverse problems that often achieves the MMSE asymptotically. A phase transition for BP is identified. This phase transition, verified by numerical results, separates the regions where BP achieves the MMSE and where it is sub-optimal. Numerical results also illustrate that more signal vectors in the jointly sparse signal ensemble lead to a better phase transition.

Realizing that the mean squared error might not be the only error metric that is of interest, we propose some future directions involving the study of optimal performance for arbitrary user-defined additive error metrics for MMV problems by extending the work of Tan and coauthors~\cite{Tan2014,Tan2014Infty}.

\section{Related Work and Contributions}\label{sec:MMVintro}

In multi-measurement vector (MMV) problems, thanks to the common
support, the number of sparse coefficients that can be successfully estimated
increases with the number of
measurements. This property was evaluated rigorously for noiseless
measurements using
$\ell_0$ minimization~\cite{DuarteWakinBaronSarvothamBaraniuk2013}.
To address measurement noise, estimation approaches for MMV problems
have included greedy algorithms such as SOMP~\cite{tropp2006ass,chen2006trs},
$\ell_1$ convex relaxation~\cite{malioutov2005ssr,tropp2006ass2}, and M-FOCUSS~\cite{cotter2005ssl}. REduce MMV and BOost (ReMBo) has
been shown to outperform conventional methods~\cite{Mishali08rembo}, and subspace methods have also
been used to solve MMV problems~\cite{LeeBreslerJunge2012,YeKimBresler2015}.
Statistical approaches~\cite{ZinielSchniter2011} often achieve the oracle minimum mean squared error (MMSE).
However, the performance limits of MMV signal estimation in the presence
of measurement noise have not been studied.

Replica analysis is a statistical physics method that can be used to analyze the MMSE and phase transition for inverse problems~\cite{Tanaka2002,GuoVerdu2005,Montanari2006,Krzakala2012probabilistic,krzakala2012statistical,MezardMontanariBook,Barbier2015,Lesieur2015}.
Barbier and Krzakala~\cite{Barbier2015} studied the MMSE for estimating superposition codes using replica analysis. In this chapter, we extend the derivation in Barbier and Krzakala~\cite{Barbier2015} to two related yet different MMV settings: ({\em i}) $J$ jointly sparse signals are measured by $J$ different dense matrices that are independent and identically distributed (i.i.d.), and ({\em ii}) $J$ jointly sparse signals are measured by $J$ identical i.i.d. matrices. We only consider dense i.i.d. Gaussian matrices in this work, while our analysis can be extended to other i.i.d. matrices easily.

We make several contributions in this chapter. First, we obtain the information theoretic MMSE for the two MMV settings above under the Bayesian setting. Second, we show that in the large system limit (defined in Definition~\ref{def:chap1-largeSystemLimit}) the MMSE's for these two settings are identical to the single measurement vector (SMV) problem with a dense measurement matrix and a block sparse signal with fixed length blocks. Third, we derive the MMSE for complex SMV problems by noticing that complex SMV is essentially an MMV problem. Fourth,
we identify several performance regions for MMV, where the MMSE
has different characteristics based on the channel noise variance and measurement rate. Finally, we find a phase transition for belief propagation algorithms (BP)~\cite{DMM2009,CSBP2010,Bayati2011,Montanari2012,Krzakala2012probabilistic,krzakala2012statistical,Barbier2015} applied to MMV problems, which separates regions where BP achieves the MMSE asymptotically and where it is sub-optimal. BP simulation results
confirm the phase transition results. Seeing that the mean squared error (MSE) might not be the only error metric that is of interest, we propose a future direction to extend the work of Tan and coauthors~\cite{Tan2014,Tan2014Infty} to MMV settings, so that we can analyze the performance limits for arbitrary user-defined additive error metrics, as well as design an algorithmic framework that can achieve such performance limits.

The remainder of this chapter is organized as follows.
We introduce our signal and measurement models in Section~\ref{sec:model},
followed by replica analysis for two MMV settings as well as two complex SMV problems in Section~\ref{sec:main}. Section~\ref{sec:proof} proves the results of Section~\ref{sec:main}.
Numerical results are discussed in Section~\ref{sec:numeric}. Section~\ref{sec:errorMetric} proposes some future directions to study the performance of arbitrary user-defined additive error metrics for MMV problems and we conclude in Section~\ref{sec:conclusion}. Some detailed derivations appear in Appendix~\ref{chap:append-A}.

\section{Signal and Measurement Models}
\label{sec:model}

\setcounter{equation}{0} \indent

{\bf Signal model}:
We consider an ensemble of $J$ signal vectors, $\underline{\x}^{(j)}\in\mathbb{R}^N,\ j\in\{1,\cdots,J\}$, where $j$ is the index of the signal. As in Section~\ref{sec:Chap1-setting}, we consider
a {\em super-symbol} $\x_l=\l[\underline{x}_l^{(1)},\cdots,\underline{x}_l^{(J)}\r]^{\top},\ l\in\{1,\cdots,N\}$, where $\{\cdot\}^{\top}$ denotes the transpose. The super-symbol $\x_l$ follows a $J$-dimensional Bernoulli-Gaussian distribution (defined in~\eqref{eq:chap1-jsm}),
\begin{equation}\label{eq:jsm}
f(\x_l)=\rho \phi(\x_l)+(1-\rho)\delta(\x_l),
\end{equation}
where $\rho$ is the sparsity rate, $\phi(\x_l)$ is a $J$-dimensional Gaussian distribution with zero mean and identity covariance matrix, and $\delta(\x_l)$ is the delta function for $J$-dimensional vectors.

\begin{myDef}[Jointly sparse]\label{def:jointly_sparse}
{\em Ensembles of signals that obey~\eqref{eq:jsm} are called jointly sparse.}
\end{myDef}
{\bf Measurement models}:
Each signal $\underline{\x}^{(j)}$ is measured by
an i.i.d. Gaussian measurement matrix $\underline{\A}^{(j)}\in\mathbb{R}^{M\times N}$, $\underline{A}_{\mu l}^{(j)} \sim \mathcal{N}(0,\frac{1}{N})$, where $\mu$ refers to the row index and $l$ is the column index. The measurements $\underline{\y}^{(j)}$ are corrupted by i.i.d. Gaussian noise $\underline{\z}^{(j)}$ consisting of entries $\underline{z}_{\mu}^{(j)}\sim \mathcal{N}(0,\sigma_Z^2)$,
\begin{equation}\label{eq:MMVmodel}
\underline{\y}^{(j)}=\underline{\A}^{(j)}\underline{\x}^{(j)}+\underline{\z}^{(j)},\quad j\in\{1,\cdots,J\}.
\end{equation}
When the number of signal vectors becomes $J=1$, this MMV model~\eqref{eq:MMVmodel} becomes an SMV problem. Note that SMV and MMV problems were motivated in~\eqref{eq:SMV} and~\eqref{eq:MMVmodel_intro}, respectively.
Our analysis in this chapter is readily extended to other i.i.d. matrices, jointly sparse signals~\eqref{eq:jsm}, and other i.i.d. noise distributions.

\begin{myDef}[MMV-1]\label{def:MMV_set1}
{\em The setting MMV-1 refers to the measurement model in~\eqref{eq:MMVmodel} with all matrices $\underline{\A}^{(j)}$ being different.}
\end{myDef}
\begin{myDef}[MMV-2]\label{def:MMV_set2}
{\em The setting MMV-2 refers to the measurement model in~\eqref{eq:MMVmodel} with all matrices $\underline{\A}^{(j)}$ being equal.}
\end{myDef}

In the signal model~\eqref{eq:jsm} and measurement model~\eqref{eq:MMVmodel}, the sparsity rate $\rho$, channel noise variance $\sigma_Z^2$, and number of channels $J$ are constant. We are interested in the large system limit, which has been defined in Definition~\ref{def:chap1-largeSystemLimit} in Section~\ref{sec:Chap1-setting}. For readers' convenience, we restate the definition of the large system limit as follows.

\begin{myDef}[Large system limit~\cite{GuoWang2008}]\label{def:largeSystemLimit}
The signal length $N$ scales to infinity, and the
number of measurements $M=M(N)$ depends on $N$ and also scales to infinity, where
the ratio approaches a positive constant $\kappa$,
\begin{equation}\label{eq:measurementRate}
\lim_{N\rightarrow\infty} \frac{M(N)}{N} = \kappa>0.
\end{equation}
\end{myDef}
We call $\kappa$ the measurement rate.

\section{Replica Analysis for MMV Settings}
\label{sec:main}

Section~\ref{sec:model} discussed two MMV settings. Both settings have applications in real-world problems such as magnetic resonance imaging~\cite{JuYeKi07,JuSuNaKiYe09} and sensor networks~\cite{pottie2000}. Although numerous algorithms for MMV signal estimation have been proposed~\cite{tropp2006ass,chen2006trs,malioutov2005ssr,tropp2006ass2,cotter2005ssl,Mishali08rembo,ZinielSchniter2011}, what is often missing is an information theoretic analysis of the best possible MSE performance. In this chapter, we only consider the MSE as our performance metric, except for Section~\ref{sec:errorMetric}.

\subsection{Statistical physics background and replica method}\label{sec:set1}
\begin{figure}[t]
\centering
\includegraphics[width=8.5cm]{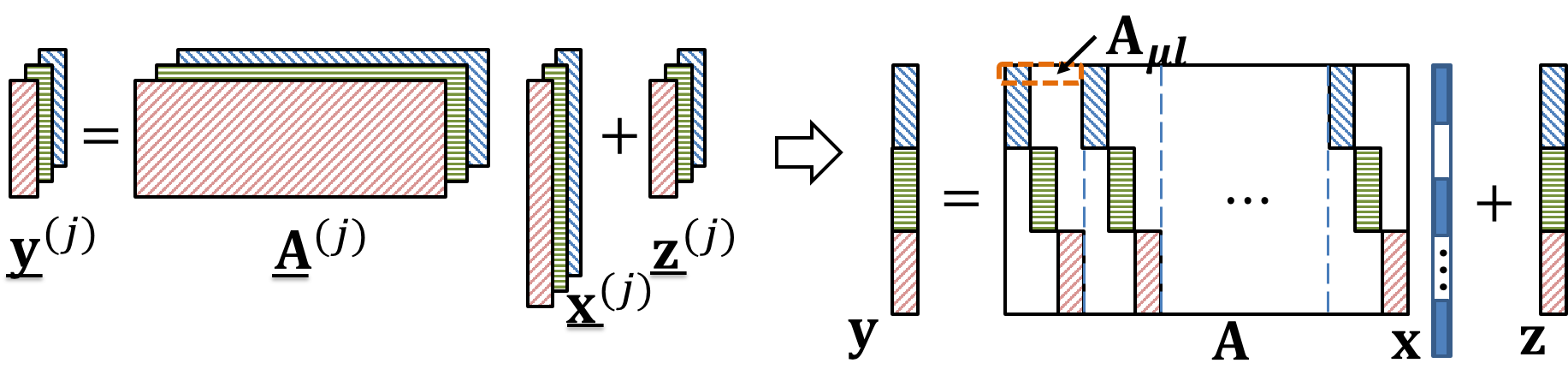}
\caption{Illustration of MMV channel~\eqref{eq:MMVmodel} with $J=3$ signal vectors (left), and one of its possible SMV forms (right). Different background patterns differentiate entries from different channels, and blank space denotes zeros.}\label{fig:channel}
\end{figure}

In order to express~\eqref{eq:MMVmodel} using a single channel, we transform it to an SMV form. One possible way to do so is illustrated in Figure~\ref{fig:channel}.
The equivalent SMV problem is
\begin{equation}\label{eq:MMVchannel}
  \y=\A\x+\z,
\end{equation}
where $\A\in\mathbb{R}^{MJ\times NJ}$ is the matrix, $\y\in\mathbb{R}^{MJ}$ are the measurements, and the noise is $\z\in\mathbb{R}^{MJ}$.
Entries of the signal vectors $\underline{\x}^{(j)}$, measurement vectors $\underline{\y}^{(j)}$, and noise vectors $\underline{\z}^{(j)}$ in~\eqref{eq:MMVmodel} form the SMV signal $\x$, measurements $\y$, and noise $\z$~\eqref{eq:MMVchannel} with
\begin{equation*}
x_{(l-1)J+j}=\underline{x}^{(j)}_l,\ y_{(j-1)M+\mu}=\underline{y}^{(j)}_{\mu},\ \text{and}\ z_{(j-1)M+\mu}=\underline{z}^{(j)}_{\mu},
\end{equation*}
respectively.
Entries of the matrix $\underline{\A}^{(j)}$~\eqref{eq:MMVmodel} form the SMV matrix $\A$~\eqref{eq:MMVchannel} with
$A_{(j-1)M+\mu,(l-1)J+j}=\underline{A}^{(j)}_{\mu l}$; other entries of $\A$ are zeros.
The posterior for the estimate $\widehat{\x} \in\mathbb{R}^{NJ}$, comprised of super-symbols $\widehat{\x}_l=\l[\widehat{x}_{(l-1)J+1},\cdots,\widehat{x}_{lJ}\r]^{\top},\ l\in\{1,\cdots,N\}$, is
\begin{equation}\label{eq:pxy2}
  f(\widehat{\x}|\y)=\frac{1}{Z}\pij{l=1}{N}f(\widehat{\x}_l)\pij{\mu=1}{MJ}\l[\frac{\operatorname{e}^{-\frac{1}{2\sigma_Z^2}\l(y_{\mu}-\sij{l=1}{N}\A_{\mu l}\widehat{\x}_l\r)^2}}{\sqrt{2\pi\sigma_Z^2}}\r],
\end{equation}
where $\A_{\mu l}=[A_{\mu,(l-1)J+1}, \cdots, A_{\mu,lJ}]$ is a super-symbol highlighted by the dashed area in Figure~\ref{fig:channel}, and the denominator $Z$ is the partition function~\cite{Tanaka2002,GuoVerdu2005,Krzakala2012probabilistic,krzakala2012statistical,MezardMontanariBook,Barbier2015},
\begin{equation}\label{eq:partition}
  Z=\bigintss \pij{l=1}{N}f(\widehat{\x}_l)\pij{\mu=1}{MJ}\l[\frac{\operatorname{e}^{-\frac{1}{2\sigma_Z^2}\l(y_{\mu}-\sij{l=1}{N}\A_{\mu l}\widehat{\x}_l\r)^2}}{\sqrt{2\pi\sigma_Z^2}}\r]\pij{l=1}{N}d\widehat{\x}_l.
\end{equation}
Note that multi-dimensional integrations such as~\eqref{eq:partition} are denoted by a single $\int$ operator for brevity.
Confining our attention to the Bayesian setting~\cite{Krzakala2012probabilistic,krzakala2012statistical,Barbier2015}, $f(\widehat{\x}_l)$ follows the true distribution~\eqref{eq:jsm},
$f(\widehat{\x}_l)=\rho \phi(\widehat{\x}_l)+(1-\rho)\delta(\widehat{\x}_l)$.

By creating an analogy between the channel~\eqref{eq:MMVchannel} and a many-body thermodynamic system~\cite{Tanaka2002,GuoVerdu2005,Krzakala2012probabilistic,krzakala2012statistical,MezardMontanariBook,Barbier2015}, the posterior~\eqref{eq:pxy2} can be interpreted as the Boltzmann measure on a disordered system with the following Hamiltonian,
\begin{equation}\label{eq:Hamiltonian}
H(\widehat{\x})=\sum_{l=1}^N \log [f(\widehat{\x}_l)]+\sum_{\mu=1}^{MJ} \frac{1}{2\sigma_Z^2}\l(y_{\mu}-\sum_{l=1}^N\A_{\mu l} \widehat{\x}_l\r)^2.
\end{equation}

The averaged free energy of the disordered system given by~\eqref{eq:Hamiltonian} characterizes the thermodynamic properties of the system. Evaluating the fixed points (local maxima) in the free energy expression provides the MMSE for the channel~\eqref{eq:MMVchannel}~\cite{Tanaka2002,GuoVerdu2005,Krzakala2012probabilistic,krzakala2012statistical,MezardMontanariBook,Barbier2015}. {\em Under the assumption of self-averaging}~\cite{Tanaka2002,GuoVerdu2005,Krzakala2012probabilistic,krzakala2012statistical,MezardMontanariBook,Barbier2015}, the free energy is defined as\footnote{Part of the literature~\cite{Tanaka2002,GuoVerdu2005}, including~\eqref{eq:def_freeEnergy_selfAveraging} in this dissertation, defines the free energy as the negative of~\eqref{eq:free_energy}, so that fixed points of the free energy correspond to local minima.}
\begin{equation}\label{eq:free_energy}
  \mathcal{F}=\lim_{N\rightarrow\infty}\frac{1}{N}\mathbb{E}_{\A,\x,\z}[\log (Z)],
\end{equation}
which is difficult to evaluate. 
Note that $\mathbb{E}_{\A,\x,\z}[\cdot]$ denotes expectation with respect to (w.r.t.) $\A,\x$, and $\z$.
The replica method~\cite{Tanaka2002,GuoVerdu2005,Krzakala2012probabilistic,krzakala2012statistical,MezardMontanariBook,Barbier2015} introduces $n$ replicas of the estimate $\widehat{\x}$ as $\widehat{\x}^a,\ a\in\{1,\cdots,n\}$, and the free energy~\eqref{eq:free_energy} can be approximated by the replica trick~\cite{Krzakala2012probabilistic,krzakala2012statistical,MezardMontanariBook,Barbier2015},
\begin{equation}\label{eq:replicaTrick}
  \mathcal{F}=\lim_{N\rightarrow\infty}\lim_{n\rightarrow 0} \frac{\mathbb{E}_{\A,\x,\z}[ Z^n]-1}{Nn}.
\end{equation}
Note that the self-averaging property that leads to~\eqref{eq:free_energy} and the replica trick~\eqref{eq:replicaTrick}, as well as the replica symmetry assumptions that appear in latter parts of this chapter, are assumed to be valid in this work, and their rigorous justification is still an open problem in mathematical physics~\cite{Tanaka2002,GuoVerdu2005,Krzakala2012probabilistic,krzakala2012statistical,MezardMontanariBook,Barbier2015}.\footnote{Recently, the replica Gibbs free
energy has been proven rigorously for the SMV case
by Barbier et al.~\cite{BDMK2016} and
Reeves and Pfister~\cite{ReevesPfister2016}. We conjecture that by generalizing these two works~\cite{BDMK2016,ReevesPfister2016},
our MMV analysis can be made rigorous; we leave it for future work.}

\textbf{Evaluating the free energy}:
To evaluate the free energy~\eqref{eq:replicaTrick}, we calculate $\mathbb{E}_{\A,\x,\z}\l[Z^n\r]$ as follows,
\begin{equation}\label{eq:EZn1}
  \mathbb{E}_{\A,\x,\z}\l[Z^n\r]=(2\pi\sigma_Z^2)^{-\frac{nMJ}{2}}\times \mathbb{E}_{\x}\l[\displaystyle{\int \pij{l=1}{N}\pij{a=1}{n} f(\widehat{\x}_l^a)\pij{\mu=1}{M}\mathbb{X}_{\mu}\pij{l=1}{N}\pij{a=1}{n}d\widehat{\x}_l^a}\r],
\end{equation}
where $Z$ is given in~\eqref{eq:partition},
\begin{equation}\label{eq:Xmu}
  \mathbb{X}_{\mu}=\mathbb{E}_{\A,\z}\l[\operatorname{e}^{-\frac{1}{2\sigma_Z^2}\sij{j=1}{J}\sij{a=1}{n}(v_{\mu j}^a)^2}\r],
\end{equation}
$a$ is the replica index, $\widehat{\x}^a_l$ is the $l$-th  super-symbol of $\widehat{\x}^a$, and
\begin{equation}\label{eq:v_mu_a}
v_{\mu j}^a=\sij{l=1}{N}\A_{\mu+M(j-1),l}(\x_l-\widehat{\x}_l^a)+z_{\mu+M(j-1)}.
\end{equation}

\begin{myLemma}\label{lemma:covIsSame}
In the large system limit, the quantity $\mathbb{X}_{\mu}$~\eqref{eq:Xmu} is the same for both MMV-1 and MMV-2.
\end{myLemma}

Lemma~\ref{lemma:covIsSame} is proved in Section~\ref{sec:proof}. Because of Lemma~\ref{lemma:covIsSame}, the free energy expressions for MMV-1 and MMV-2 should be identical in the large system limit. We state the result as a theorem and the detailed derivations appear in Appendix~\ref{chap:append-A}.

\begin{myTheorem}[Free energy for MMV]\label{th:free_energy}
For settings MMV-1 and MMV-2, the free energy expressions as functions of $E$ are identical in the large system limit and are given below,
\begin{eqnarray}
  \mathcal{F}(E)&=&-\frac{J}{2}\kappa\l\{\log[2\pi(\sigma_Z^2+E)]+\frac{\rho+\sigma_Z^2}{E+\sigma_Z^2}\r\}\!+\nonumber\\
  & &\quad \int f(\x_1) \int \log \l[ \int f(\widehat{\x}_1)\operatorname{e}^{-\frac{\widehat{Q}+\widehat{q}}{2}\widehat{\x}_1^{\top}\widehat{\x}_1+\widehat{m}\widehat{\x}_1^{\top}\x_1+
  \sqrt{\widehat{q}}\h^{\top}\widehat{\x}_1} d\widehat{\x}_1\r]\mathcal{D}\h \ d\x_1\label{eq:free_energy3}\\
  &=&-\frac{J}{2}\kappa\l\{\log[2\pi(\sigma_Z^2+E)]+\frac{\sigma_Z^2}{E+\sigma_Z^2}\r\}+\frac{JR(1-\rho)}{2(\kappa+E+\sigma_Z^2)}+\nonumber\\
  & &\quad \rho\int \log \Bigg[ \rho \l(\frac{E+\sigma_Z^2}{\kappa+E+\sigma_Z^2}\r)^{J/2}+(1-\rho)\operatorname{e}^{-\frac{\kappa}{2(E+\sigma_Z^2)}\g^{\top}\g}\Bigg]\mathcal{D}\g+\nonumber\\
  & &\quad (1-\rho)\int \log \l[ \rho \l(\frac{E+\sigma_Z^2}{\kappa+E+\sigma_Z^2}\r)^{J/2}+(1-\rho)\operatorname{e}^{-\frac{\kappa}{2(\kappa+E+\sigma_Z^2)}\h^{\top}\h}\r]\mathcal{D}\h,\label{eq:free_energy4}
\end{eqnarray}
where $\h,\x_1$, and $\g$ are $J$-dimensional super-symbols, and the differential $\mathcal{D}\h=\prod_{j=1}^J\frac{1}{\sqrt{2\pi}}\operatorname{e}^{-h_j^2/2}d h_j$; the same rule applies to $\mathcal{D}\g$.\footnote{The $J$-dimensional integrals in~\eqref{eq:free_energy4} can be simplified to one-dimensional integrals using a change of coordinates to $J$-sphere coordinates. Note also that $E$ approaches the MSE in the large system limit; details appear in Appendix~\ref{chap:append-A}.}
\end{myTheorem}

\textbf{MMSE}: The $E$ that maximizes the free energy~\eqref{eq:free_energy4} {\em corresponds to} the MMSE~\cite{Krzakala2012probabilistic,krzakala2012statistical,Barbier2015}.
After finding the $E_0$ that maximizes the free energy~\eqref{eq:free_energy4}, we obtain the MMSE, $D_0=E_0$, in the large system limit.

\begin{myCoro}
The MMSE for MMV-1 and MMV-2 is the same for the same measurement rate $\kappa$, noise variance $\sigma_Z^2$, and number of signal vectors $J$.
\end{myCoro}

\begin{myRemark} As the reader can see from the proof of Lemma~\ref{th:free_energy} in Section~\ref{sec:proof}, the key reason that both MMV-1 and MMV-2 have an identical MMSE is that the entries in the super-symbols $\x_l$ and $\widehat{\x}_l^{\{\cdot\}}$ are i.i.d. That said, we suspect that the MMSE for MMV-1 and MMV-2 could differ by some higher order terms.  If the entries of these super-symbols are not i.i.d., which is true in some practical MMV applications~\cite{ZinielSchniter2013MMV}, then it becomes  more difficult to analyze the covariance matrix $\G_{\mu}$ as in Section~\ref{sec:proof}. Therefore, we do not have an analysis for non-i.i.d. entries within $\x_l$ and $\widehat{\x}_l^{\{\cdot\}}$. However, we speculate that MMV-1 might have lower MMSE than MMV-2 in that case.
\end{myRemark}

{\bf Link to SMV with block sparse signal:}
The signal $\x$ in~\eqref{eq:MMVchannel} is a block sparse signal comprised of $N$ blocks of length $J$.
We study an SMV problem by replacing the measurement matrix $\A$ in~\eqref{eq:MMVchannel} with an i.i.d. Gaussian matrix $\widehat{\A}\in\mathbb{R}^{MJ\times NJ}$, i.e., $\y=\widehat{\A}\x+\z$. The entries of $\widehat{\A}$ follow the distribution, $\widehat{A}_{\mu l}\sim \mathcal{N}(0,\frac{1}{NJ})$. This SMV is similar to the setting in Barbier and Krzakala~\cite{Barbier2015}, except for the different priors and different $\ell_2$ norms in each row of $\widehat{\A}$. We consider these differences while following their derivation~\cite{Barbier2015}, and obtain the same free energy expression as~\eqref{eq:free_energy4}. We have also shown that MMV-1 and MMV-2 have the same MMSE in the large system limit.
Hence, the three settings have the same free energy expression and their MMSE's are the same under the same noise variance $\sigma_Z^2$ and measurement rate $\kappa$ in the large system limit.

\subsection{Extension to complex SMV}\label{sec:complex}
The MMV model with jointly sparse signals is a versatile model that can be adapted to other problems. As an example, we show how the MMV model can be used to analyze the MMSE of a complex SMV.
Consider the complex SMV, $\y^{\mathcal{C}}=\A^{\mathcal{C}}\x^{\mathcal{C}}+\z^{\mathcal{C}}$,
where $\x^{\mathcal{C}}=\x^{\mathcal{R}}+i\x^{\mathcal{I}}\in\mathbb{C}^N$, $\A^{\mathcal{C}}=\A^{\mathcal{R}}+i\A^{\mathcal{I}}\in\mathbb{C}^{M\times N}$, $\z^{\mathcal{C}}=\z^{\mathcal{R}}+i\z^{\mathcal{I}}\in\mathbb{C}^M$, $\y^{\mathcal{C}}=\y^{\mathcal{R}}+i\y^{\mathcal{I}}\in\mathbb{C}^M$, $i=\sqrt{-1}$, and $\mathcal{R}$ and $\mathcal{I}$ refer to the real and imaginary parts, respectively. The real and imaginary parts of the entries of $\z^{\mathcal{C}}$ both follow a Gaussian distribution, $z_l^{\mathcal{R}}, z_l^{\mathcal{I}}\sim\mathcal{N}(0,\sigma_Z^2), l\in\{1,\cdots,M\}$.
Assume that the complex signal $\x^{\mathcal{C}}$ is comprised of two jointly sparse signals, $\x^{\mathcal{R}}$ and $\x^{\mathcal{I}}$, that satisfy the $J=2$ dimensional Bernoulli-Gaussian distribution~\eqref{eq:jsm}.
We can extend the analysis of Section~\ref{sec:set1} to two settings of complex SMV:
({\em i}) the measurement matrix $\A^{\mathcal{C}}$ is real and ({\em ii}) $\A^{\mathcal{C}}$ is complex.\footnote{A replica analysis for complex SMV with a real measurement matrix appears in Guo and Verd{\'u}~\cite{GuoVerdu2005}. Their derivation does not cover complex matrices.}

{\bf Real measurement matrix:}
Suppose that $\A^{\mathcal{C}}$ is real, $\A^{\mathcal{C}}=\A^{\mathcal{R}}\in \mathbb{R}^{M\times N}$, and the entries of $\A^{\mathcal{R}}$ follow a Gaussian distribution, $A^{\mathcal{R}}_{\mu l}\sim \mathcal{N}(0,\frac{1}{N})$. Complex SMV with a real measurement matrix can be written as real-valued MMV,
\begin{equation}\label{eq:complexRealMat}
    \y^{\mathcal{R}}=\A^{\mathcal{R}}\x^{\mathcal{R}}+\z^{\mathcal{R}}\ \text{and}\     \y^{\mathcal{I}}=\A^{\mathcal{R}}\x^{\mathcal{I}}+\z^{\mathcal{I}},
\end{equation}
where $\x^{\mathcal{R}}$ and $\x^{\mathcal{I}}$ are jointly sparse and follow~\eqref{eq:jsm}.
This formulation~\eqref{eq:complexRealMat} fits into MMV-2 for $J=2$. Hence, we can obtain the MMSE according to~\eqref{eq:free_energy4}.\footnote{As a reminder, the free energy of MMV-2 is identical to that of MMV-1 in the large system limit.}

{\bf Complex measurement matrix:}
Consider a complex $\A^{\mathcal{C}}=\A^{\mathcal{R}}+i\A^{\mathcal{I}}\in\mathbb{C}^{M\times N}$ with entries $A_{\mu l}^{\mathcal{R}}, A_{\mu l}^{\mathcal{I}}\sim \mathcal{N}(0,\frac{1}{2N})$.
Expanding out the complex channel, $\y^{\mathcal{C}}=\A^{\mathcal{C}}\x^{\mathcal{C}}+\z^{\mathcal{C}}$, we obtain the equivalent real-valued SMV channel,
\begin{equation}\label{eq.realComplexChannel}
  \begin{bmatrix}
    \y^{\mathcal{R}} \\
    \y^{\mathcal{I}}
  \end{bmatrix}
=
  \begin{bmatrix}
    \A^{\mathcal{R}} & -\A^{\mathcal{I}}  \\
    \A^{\mathcal{I}}  & \A^{\mathcal{R}}
  \end{bmatrix}
    \begin{bmatrix}
    \x^{\mathcal{R}} \\
    \x^{\mathcal{I}}
  \end{bmatrix}
  +
    \begin{bmatrix}
    \z^{\mathcal{R}} \\
    \z^{\mathcal{I}}
  \end{bmatrix}.
\end{equation}

We rearrange~\eqref{eq.realComplexChannel} as follows,
\begin{equation}\label{eq:complexMatRearrange}
\underbrace{\begin{bmatrix}
    \y^{\mathcal{R}} \\
    \y^{\mathcal{I}}
  \end{bmatrix}}_{\overline{\y}}
=
\underbrace{\begin{bmatrix}
    \A_{:,1}^{\mathcal{R}},-\A_{:,1}^{\mathcal{I}},\cdots,\A_{:,N}^{\mathcal{R}}, -\A_{:,N}^{\mathcal{I}} \\
    \A_{:,1}^{\mathcal{I}},\ \ \A_{:,1}^{\mathcal{R}},\cdots,\A_{:,N}^{\mathcal{I}},\ \  \A_{:,N}^{\mathcal{R}}
  \end{bmatrix}}_{\overline{\A}}
    \underbrace{\begin{bmatrix}
    x_1^{\mathcal{R}}\\
    x_1^{\mathcal{I}}\\
    \vdots\\
    x_N^{\mathcal{R}}\\
    x_N^{\mathcal{I}}
  \end{bmatrix}}_{\overline{\x}}
  \!+\!
    \underbrace{\begin{bmatrix}
    \z^{\mathcal{R}} \\
    \z^{\mathcal{I}}
  \end{bmatrix}}_{\overline{\z}},
\end{equation}
where $\{:\}$ refers to all the rows.
In the rearranged channel~\eqref{eq:complexMatRearrange}, the measurement matrix $\overline{\A}$ consists of super-symbols,
\begin{equation}\label{eq:SMV_F}
    \overline{\A}_{\mu l}=\left\{
                \begin{array}{ll}
                 &[A_{\mu l}^{\mathcal{R}},-A_{\mu l}^{\mathcal{I}}],\ \mu\in\{1,\cdots,M\}\\
                 &[A_{\mu l}^{\mathcal{I}}, A_{\mu l}^{\mathcal{R}}],\ \mu\in\{M+1,\cdots,2M\}
                \end{array}
      		\right.,\\
\end{equation}      		
and the signal $\overline{\x}$ consists of
$\overline{\x}_{l}=\begin{bmatrix}
    x_l^{\mathcal{R}} \\
    x_l^{\mathcal{I}}
  \end{bmatrix},\ l\in\{1,\cdots,N\}$.
The measurements and noise are $\overline{\y}=\begin{bmatrix}
    \y^{\mathcal{R}} \\
    \y^{\mathcal{I}}
  \end{bmatrix}$ and
$\overline{\z}=\begin{bmatrix}
    \z^{\mathcal{R}} \\
    \z^{\mathcal{I}}
  \end{bmatrix}$, respectively.
Hence, $\overline{y}_{\mu}=\sum_{l=1}^N \overline{\A}_{\mu l}\overline{\x}_l+\overline{z}_{\mu},\ \mu\in\{1,\cdots,2M\}$.

Section~\ref{sec:proof} shows that the free energy and MMSE for complex SMV with complex measurement matrices are the same as MMV-1 with $J=2$.
Note that in the free energy expression~\eqref{eq:free_energy4}, the MSE, $D=E$~\eqref{eq:DandE}, is the average MSE of the $J$ entries of $\x_l$. Therefore, in this complex SMV setting, $D$ is the average MSE of the real and imaginary parts of the signal entries.

\section[Proof of Lemma~3.1]{Proof of Lemma~\ref{lemma:covIsSame}}\label{sec:proof}
In this section, we show that the quantity $\mathbb{X_{\mu}}$~\eqref{eq:Xmu} is the same for MMV-1 and MMV-2. Moreover, we show that  complex SMV with a complex measurement matrix also yields the same $\mathbb{X_{\mu}}$ with $J=2$.

First, we rewrite~\eqref{eq:Xmu} in the vector form
\begin{equation}\label{eq:XMuVector}
\mathbb{X}_{\mu}\!=\!\mathbb{E}_{\v_{\mu}}\!\left[\operatorname{e}^{-\frac{1}{2\sigma_Z^2}\sij{j=1}{J}\sij{a=1}{n}(v_{\mu j}^a)^2}\right]\!
=\!\mathbb{E}_{\v_{\mu}}\!\left[\operatorname{e}^{-\frac{1}{2\sigma_Z^2}\v_{\mu}^{\top}\v_{\mu}}\right],
\end{equation}
where $\v_{\mu}=[v_{\mu 1}^1,\cdots,v_{\mu 1}^a,\cdots,v_{\mu J}^1$, $\cdots,v_{\mu J}^n]^{\top}$ and $v_{\mu j}^a$ is given in~\eqref{eq:v_mu_a}.
In order to calculate the expectation w.r.t. $\v_{\mu}$ in~\eqref{eq:XMuVector}, we calculate the distribution of $\v_{\mu}$, which is approximated by a Gaussian distribution, due to the central limit theorem. The mean is $\mathbb{E}_{\A,\z}[v_{\mu j}^a]=0$.

We now calculate the covariance matrix, $\G_{\mu}=\mathbb{E}[\v_{\mu}\v_{\mu}^{\top}]$. The matrix $\G_{\mu}$ is separated into $J\times J$ blocks of size $n\times n$, as shown in Figure~\ref{fig.cov}. The main diagonal of $\G_{\mu}$ consists of entries $w_1=\mathbb{E}_{\A,\z}[(v_{\mu j}^a)^2]$. The entries in the blocks along the main diagonal (other than entries along the main diagonal itself) are $w_3=\mathbb{E}_{\A,\z}[v_{\mu j}^a v_{\mu j}^b]$. The main diagonals of other blocks have entries $w_2=\mathbb{E}_{\A,\z}[v_{\mu j}^a v_{\mu\eta}^a]$, and other entries in these blocks are $w_4=\mathbb{E}_{\A,\z}[v_{\mu j}^a v_{\mu \eta}^b]$. We now calculate each of these values as follows for MMV-1, MMV-2, and complex SMV with a complex measurement matrix.

\begin{figure}[t]
\centering
\includegraphics[width=3.5cm]{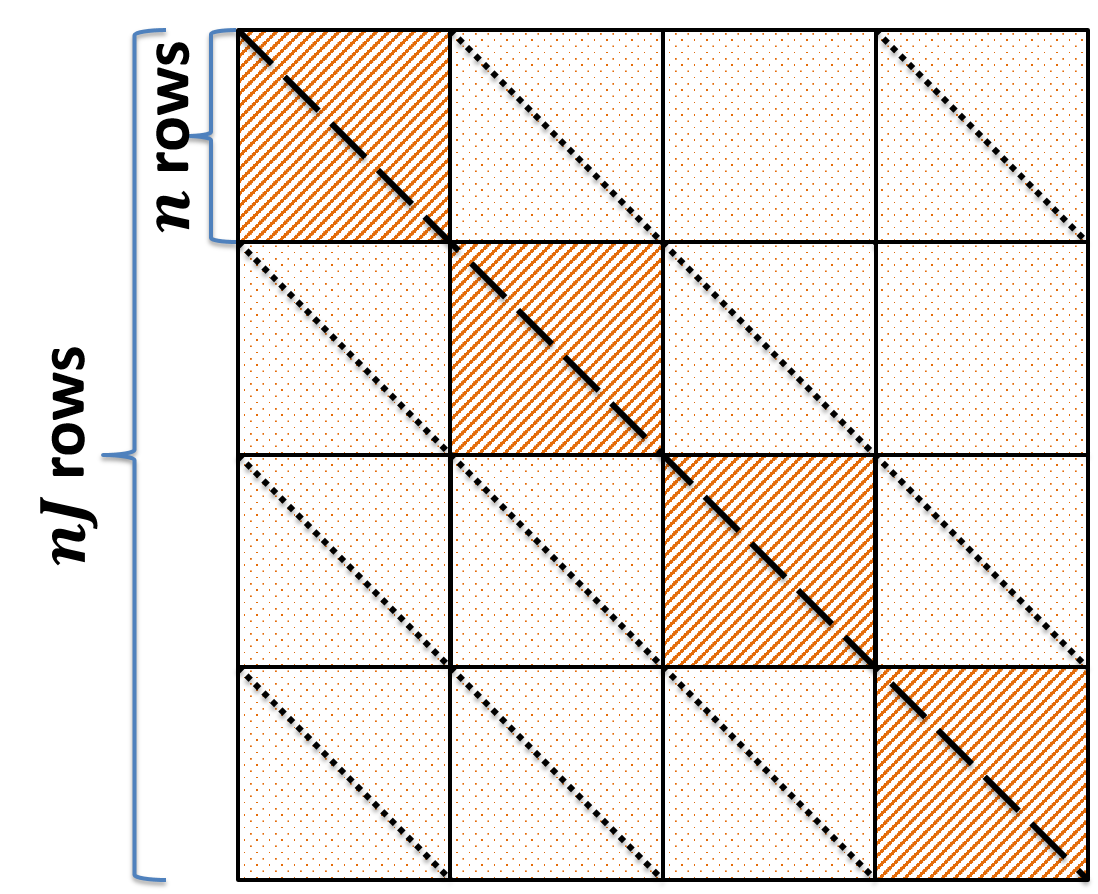}
\caption{Covariance matrix $\G_{\mu}\in\mathbb{R}^{nJ\times nJ}$. Each block in $\G_{\mu}$ has a size of $n\times n$.
The entries in the heavily marked blocks take the value $w_3$, except that entries along the dashed diagonal are $w_1$. The entries in the lightly marked blocks take the value $w_4$, except that entries along the dotted diagonals are $w_2$.}\label{fig.cov}
\end{figure}

{\bf MMV-1:}
We begin by calculating the diagonal entries of the covariance matrix $\G_{\mu}=\mathbb{E}[\v_{\mu}\v_{\mu}^{\top}]$,
\begin{equation}\label{eq:vVar1}
    w_1=\mathbb{E}_{\A,\z}\l[(v_{\mu j}^a)^2\r]=\sij{l,k=1}{N,N}\Bigg\{(\x_l-\widehat{\x}_l^a)^{\top}
     \mathbb{E}_{\A}\l[\A_{\mu+M(j-1), l}^{\top}\A_{\mu+M(j-1), k}\r](\x_k-\widehat{\x}_k^a)\Bigg\}+\sigma_Z^2.
\end{equation}
In~\eqref{eq:vVar1}, $\mathbb{E}_{\A}\l[\A_{\mu+M(j-1), l}^{\top}\A_{\mu+M(j-1), k}\r]=\frac{\delta_{k,l}}{N}\widetilde{\mathbf{I}}_J$ (cf. Figure~\ref{fig:channel}), where $\widetilde{\mathbf{I}}_J$ is a $J\times J$ matrix with only one 1 located at the $j$-th row, $j$-th column, and $\delta_{k,l}=1$ when $k=l$, else zero. Hence,~\eqref{eq:vVar1} becomes
\begin{eqnarray}
 w_1&=& \mathbb{E}_{\A,\z}\l[(v_{\mu j}^a)^2\r]=\frac{1}{N}\sij{l=1}{N} (x_{l,j}-\widehat{x}_{l,j}^a)^2+\sigma_Z^2\label{eq:vVar2}\\
  &=&\frac{1}{NJ}\sij{l=1}{N} (\x_l-\widehat{\x}_l^a)^{\top} (\x_l-\widehat{\x}_l^a)+\sigma_Z^2,\label{eq:vVar2_1}
\end{eqnarray}
where $x_{l,j}$ and $\widehat{x}_{l,j}^a$~\eqref{eq:vVar2} denote the $j$-th entries in super-symbols $\x_l$ and $\widehat{\x}_l^a$, respectively, and~\eqref{eq:vVar2_1} holds because all $J$ entries within the same super-symbol ($\x_l$ or $\widehat{\x}_l^a$) are i.i.d.

Similarly, we obtain
\begin{equation}
\begin{split}
w_2&=\mathbb{E}_{\A,\z}[v_{\mu j}^a v_{\mu\eta}^{a}] = \frac{1}{N}\sij{l=1}{N}(x_{l,j}-\widehat{x}_{l,j}^a)(x_{l,\eta}-\widehat{x}_{l,\eta}^{a})\\
&= \frac{1}{NJ}\sij{l=1}{N}(\x_l-\widehat{\x}_l^a)^{\top} (\x_l^a-\widehat{\x}_l^b),\label{eq:sx_iid1}
\end{split}
\end{equation}
where entries of $\x_l^{\{\cdot\}}$ and $\widehat{\x}_l^{\{\cdot\}}$ follow the same distribution as entries of $\x_l$ given $l$, and~\eqref{eq:sx_iid1} is due to
({\em i}) entries of $\x_l$ being i.i.d., ({\em ii}) entries of $\widehat{\x}_l^{\{\cdot\}}$ being i.i.d. for fixed $l$, and ({\em iii}) the replica symmetry assumption~\cite{Krzakala2012probabilistic,krzakala2012statistical}.
We also obtain
\begin{equation}\label{eq:v_j_eta_a}
\begin{split}
w_3=\mathbb{E}_{\A,\z}[v_{\mu j}^a v_{\mu j}^b]&=\frac{1}{NJ}\sij{l=1}{N}(\x_l-\widehat{\x}_l^a)^{\top} (\x_l-\widehat{\x}_l^b)+\sigma_Z^2,\\
w_4=\mathbb{E}_{\A,\z}[v_{\mu j}^a v_{\mu\eta}^{b}]&=\frac{1}{NJ}\sij{l=1}{N}(\x_l-\widehat{\x}_l^a)^{\top} (\x_l^a-\widehat{\x}_l^b).
\end{split}
\end{equation}

We now define the following auxiliary parameters
\begin{equation}\label{eq:auxParamsSet1}
m_a=\frac{\displaystyle\sij{l=1}{N} (\widehat{\x}_l^a)^{\top}\x_l}{NJ},\quad Q_a=\frac{\displaystyle\sij{l=1}{N} (\widehat{\x}_l^a)^{\top}\widehat{\x}_l^a}{NJ},\quad q_{ab}=\frac{\displaystyle\sij{l=1}{N} (\widehat{\x}_l^a)^{\top}\widehat{\x}_l^b}{NJ},\quad
q_0=\frac{1}{NJ}\sij{l=1}{N}(\x_l^a)^{\top} \x_l,
\end{equation}
which allow us to express \eqref{eq:vVar2_1}--\eqref{eq:v_j_eta_a} as
\begin{equation*}
w_1=\rho-2m_a+Q_a+\sigma_Z^2,
\end{equation*}
\begin{equation}\label{eq:ws2}
w_2=q_0-(m_a+m_b)+q_{ab},
\end{equation}
\begin{equation*}
w_3 = \rho-(m_a+m_b)+q_{ab}+\sigma_Z^2,
\end{equation*}
\begin{equation}\label{eq:ws4}
w_4=q_0-(m_a+m_b)+q_{ab}.
\end{equation}

Up to this point, we have obtained the entries of $\G_{\mu}$.
Plugging the distribution of $\v_{\mu}$, approximated by $f(\v_{\mu})=[(2\pi)^n\det (\G_{\mu})]^{-\frac{1}{2}}\text{exp}(-\frac{1}{2}\v_{\mu}^{\top}\G_{\mu}^{-1}\v_{\mu})$, into~\eqref{eq:XMuVector}, we obtain
\begin{equation}\label{eq:Xmu2}
\mathbb{X}_{\mu}=\l[\det\l(\mathbb{I}_n+\frac{1}{\sigma_Z^2}\G_{\mu}\r)\r]^{-1/2},
\end{equation}
where $\mathbb{I}_n$ denotes an identity matrix of size $n\times n$ and $\det(\cdot)$ is the determinant of a matrix.

{\bf MMV-2:}
For the matrix $\A$~\eqref{eq:MMVchannel} in MMV-2, rows $jM+1,\cdots,(j+1)M,\ 2\leq j \leq J$, will be the right-shift of rows $(j-1)M+1,\cdots,jM$.
We express $v_{\mu j}^a$~\eqref{eq:v_mu_a} as
\begin{equation}\label{eq:v_mu_j_a_mmv2}
  v_{\mu j}^a=\sij{l=1}{N}\A_{\mu l}{\bf T}_{j}(\x_l-\widehat{\x}_l^a)+z_{\mu+M(j-1)},\ \mu\in \{1,\cdots,M\},
\end{equation}
where $\T_{j}$ is a $J\times J$ transform matrix with the $j$-th entry of the first row being one and all other entries in $\T_j$ being zeros. Using the same derivations as in MMV-1, it can be proved that the covariance matrix $\G_{\mu}=\mathbb{E}[\v_{\mu}\v_{\mu}^{\top}]$  in MMV-2 is identical to that of MMV-1.
Therefore, $\mathbb{X}_{\mu}$ in MMV-1 and MMV-2 are identical in the large system limit.

{\bf Complex SMV with complex measurement matrix:}
The derivations are the same as in MMV-2 above, except that we need to change $\A_{\mu l}$
in~\eqref{eq:v_mu_j_a_mmv2}  to $\overline{\A}_{\mu l}$~\eqref{eq:SMV_F} and
replace $\T_j$ by
\begin{equation*}
{\bf T}=\begin{bmatrix}
    0 & 1 \\
    -1 & 0
  \end{bmatrix},
\end{equation*}
because $\overline{\A}_{(\mu+M)l}=\overline{\A}_{\mu l}{\bf T},\ \mu\in\{1,\cdots,M\}$.
Using similar steps as above, we obtain that the covariance matrix $\G_{\mu}$ in this case is also the same as that of MMV-1 with $J=2$.

\textbf{Solving $\mathbb{X}_{\mu}$}: For such a structured matrix $\G_{\mu}$ (Figure~\ref{fig.cov}), elementary transforms show that the eigenvalues (EV's) are comprised of one EV equal to $\alpha_1=[w_1+(J-1)w_2]+(n-1)[w_3+(J-1)w_4],\ (J-1)$ EV's equal to $\alpha_2=(w_1-w_2)+(n-1)(w_3-w_4),\ (n-1)$ EV's equal to $\alpha_3=[w_1+(J-1)w_2]-[w_3+(J-1)w_4]$, and $(J-1)(n-1)$ EV's equal to $\alpha_4=(w_1-w_2)-(w_3-w_4)$.

Owing to replica symmetry~\cite{Krzakala2012probabilistic,krzakala2012statistical}, we have $m_a=m_b=m$, $Q_a=Q$, and $q_{ab}=q$, cf.~\eqref{eq:auxParamsSet1}. Also, in the Bayesian setting, we have $m=q_0=q$ and $Q=\rho$.
Thus, $w_2=w_4=0$ (\eqref{eq:ws2} and~\eqref{eq:ws4}),
and
\begin{equation}\label{eq:detSet2}
\begin{split}
\det \l(\mathbb{I}_{nJ}+\frac{1}{\sigma_Z^2}\G_{\mu}\r)&= \l(1+\frac{\alpha_1}{\sigma_Z^2}\r)\l(1+\frac{\alpha_2}{\sigma_Z^2}\r)^{J-1} \l(1+\frac{\alpha_1}{\sigma_Z^2}\r)^{n-1}\l(1+\frac{\alpha_1}{\sigma_Z^2}\r)^{(n-1)(J-1)}\\
&=\l(1+n\frac{w_3}{\sigma_Z^2+\alpha_4}\r)^J\!\l(1+\frac{1}{\sigma_Z^2}\alpha_4\r)^{Jn}\!.
\end{split}
\end{equation}
Considering~\eqref{eq:detSet2}, we simplify~\eqref{eq:Xmu2},
\begin{equation}\label{eq:XmuNew}
\lim_{n\rightarrow 0}\mathbb{X}_{\mu}=\operatorname{e}^{-\frac{nJ}{2}\l[\frac{\rho-2m+\sigma_Z^2+q}{Q-q+\sigma_Z^2}+\log(Q-q+\sigma_Z^2)-\log(\sigma_Z^2)\r]},
\end{equation}
where we rely on the following Taylor series,
\begin{equation*}
\operatorname{e}^{nk}\approx 1+nk\Rightarrow \operatorname{e}^{-\frac{n}{2}k}\approx (1+nk)^{-1/2},\ n\rightarrow 0.
\end{equation*}

\begin{figure}[t] 
\centering
\includegraphics[width=8cm]{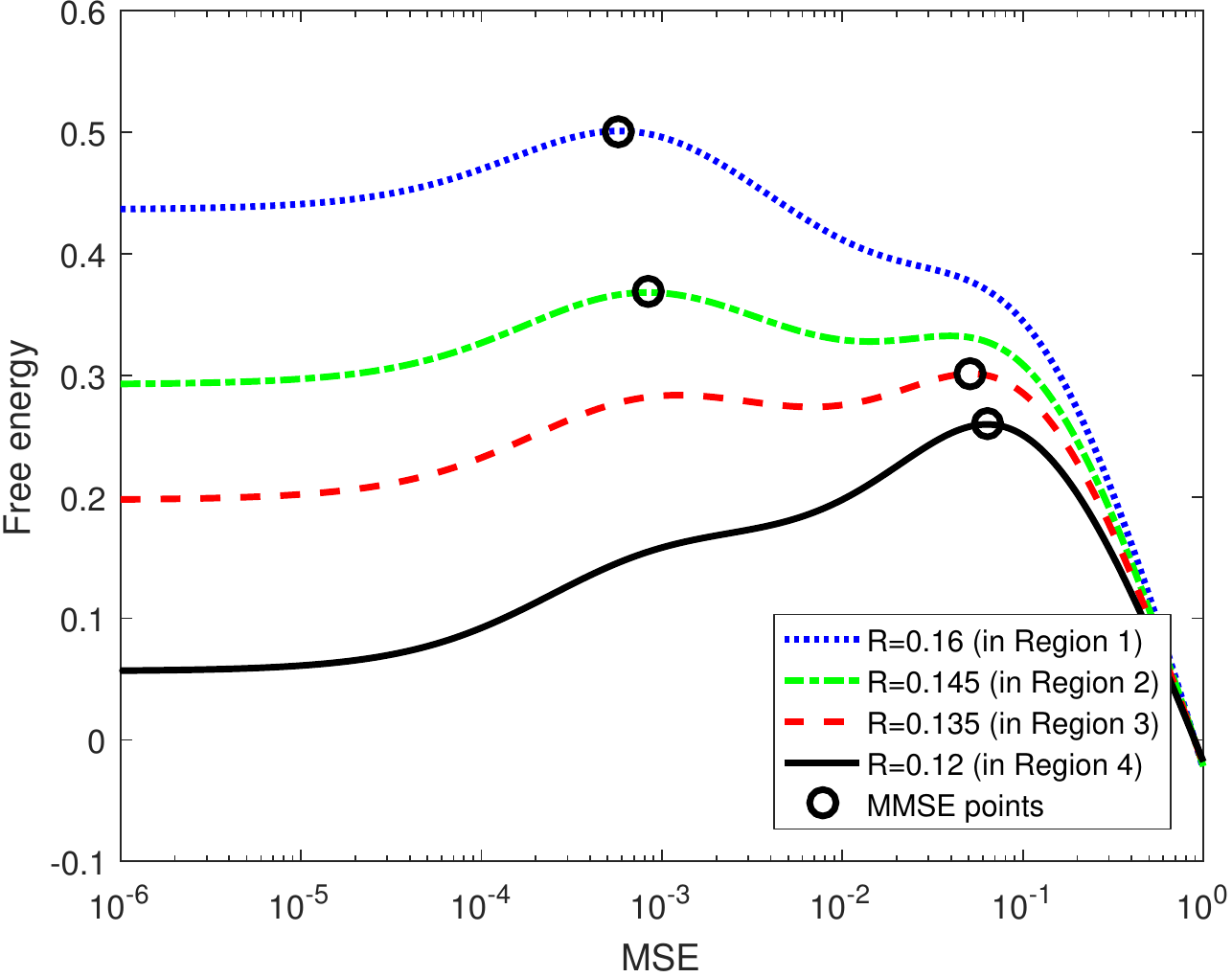}
\caption{Free energy as a function of the MSE for different measurement rates $\kappa$  (number of jointly sparse signal vectors $J=3$ and noise variance $\sigma_Z^2=-35$ dB). The black circles mark the largest free energy, and so they correspond to the MMSE.}\label{fig:freeEnergyProf}
\end{figure}

\section{Numerical Results}\label{sec:numeric}
Given a free energy expression for an MMV problem, the MMSE can be obtained by evaluating the largest free energy~\cite{Tanaka2002,GuoVerdu2005,Krzakala2012probabilistic,krzakala2012statistical,MezardMontanariBook,Barbier2015}.  Having derived the free energy for the two  MMV settings in Section~\ref{sec:main}, this section calculates the MMSE under various cases. Different performance regions of MMV are identified, where the MMSE behaves differently as a function of the noise variance $\sigma_Z^2$ and measurement rate $\kappa$. We identify a phase transition of belief propagation (BP) that separates regions where BP is optimal asymptotically or not. Simulation results match the performance predicted for BP.

\subsection{Performance regions: Definitions and numerical results}\label{sec:PerfRegion}
When calculating the MMSE~\eqref{eq:DandE} for different settings from the free energy expression~\eqref{eq:free_energy4}, four different {\em performance regions} will appear, as  discussed below; the free energy as a function of the MSE is shown in Figure~\ref{fig:freeEnergyProf} for different performance regions.

\begin{figure*}[t]
  \subfloat[$J=1$\label{fig:MMV_J1}]{
    \includegraphics[width=0.33\textwidth]{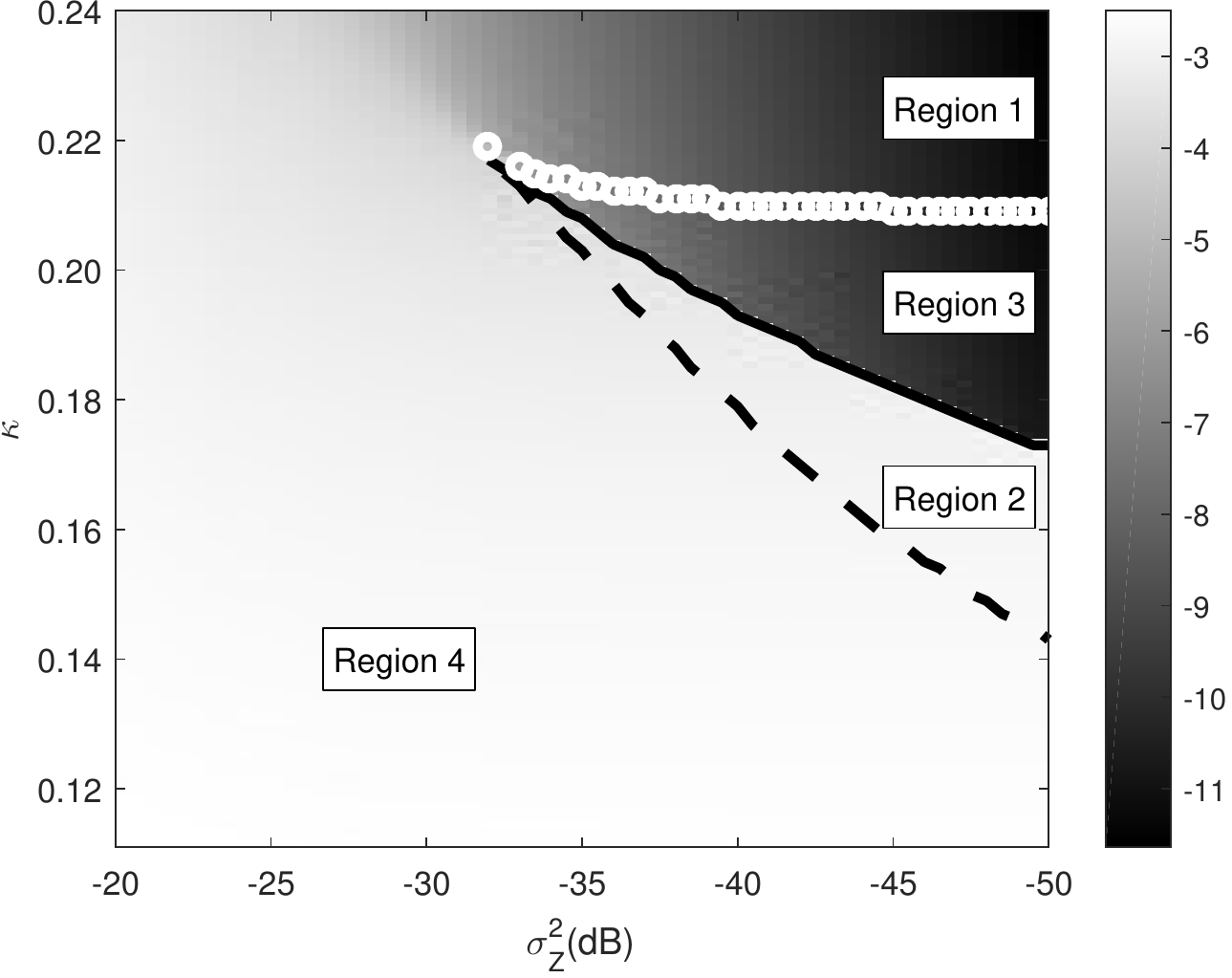}}
  \subfloat[$J=3$\label{fig:MMV_J3}]{
    \includegraphics[width=0.33\textwidth]{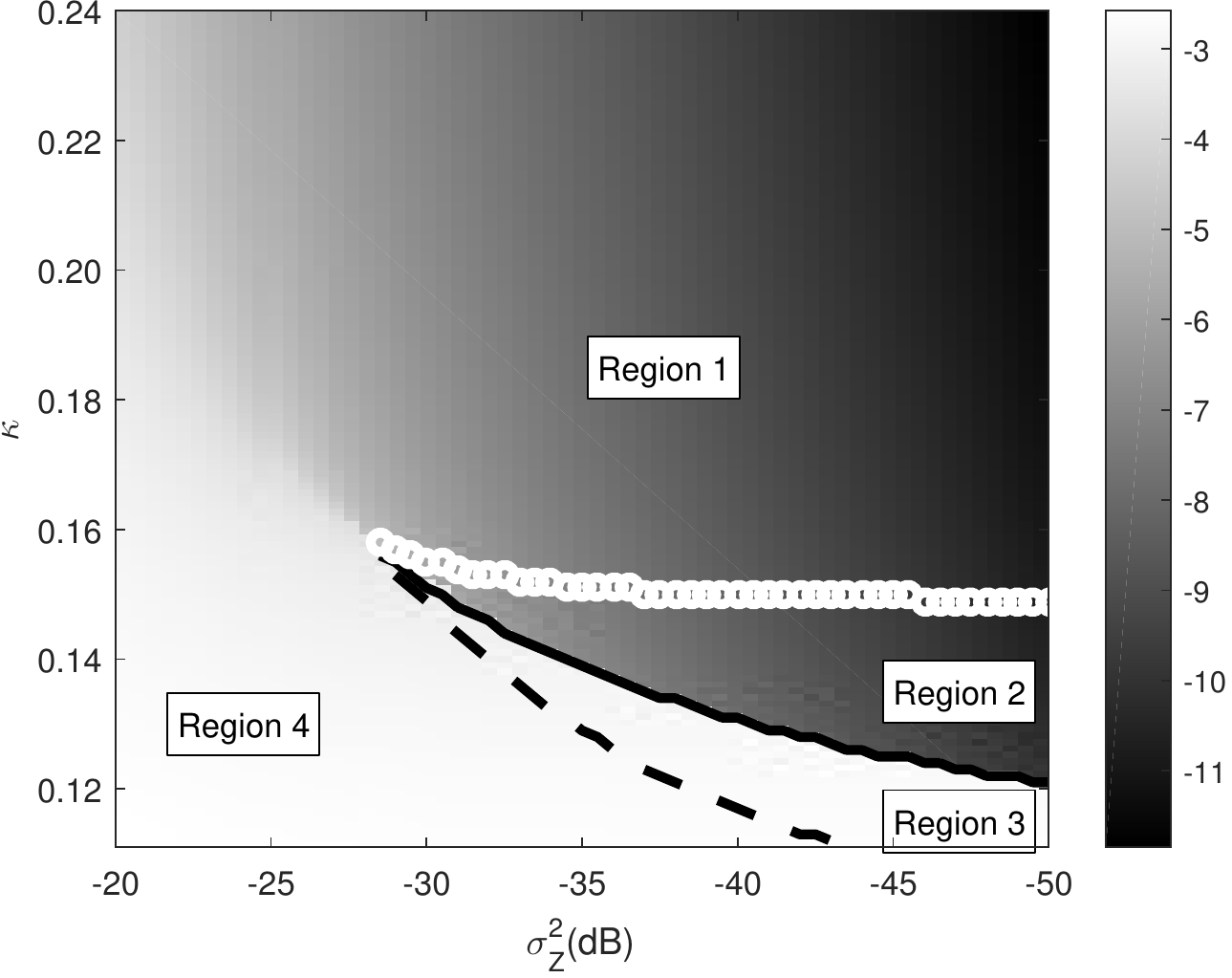}}
  \subfloat[$J=5$\label{fig:MMV_J5}]{
    \includegraphics[width=0.33\textwidth]{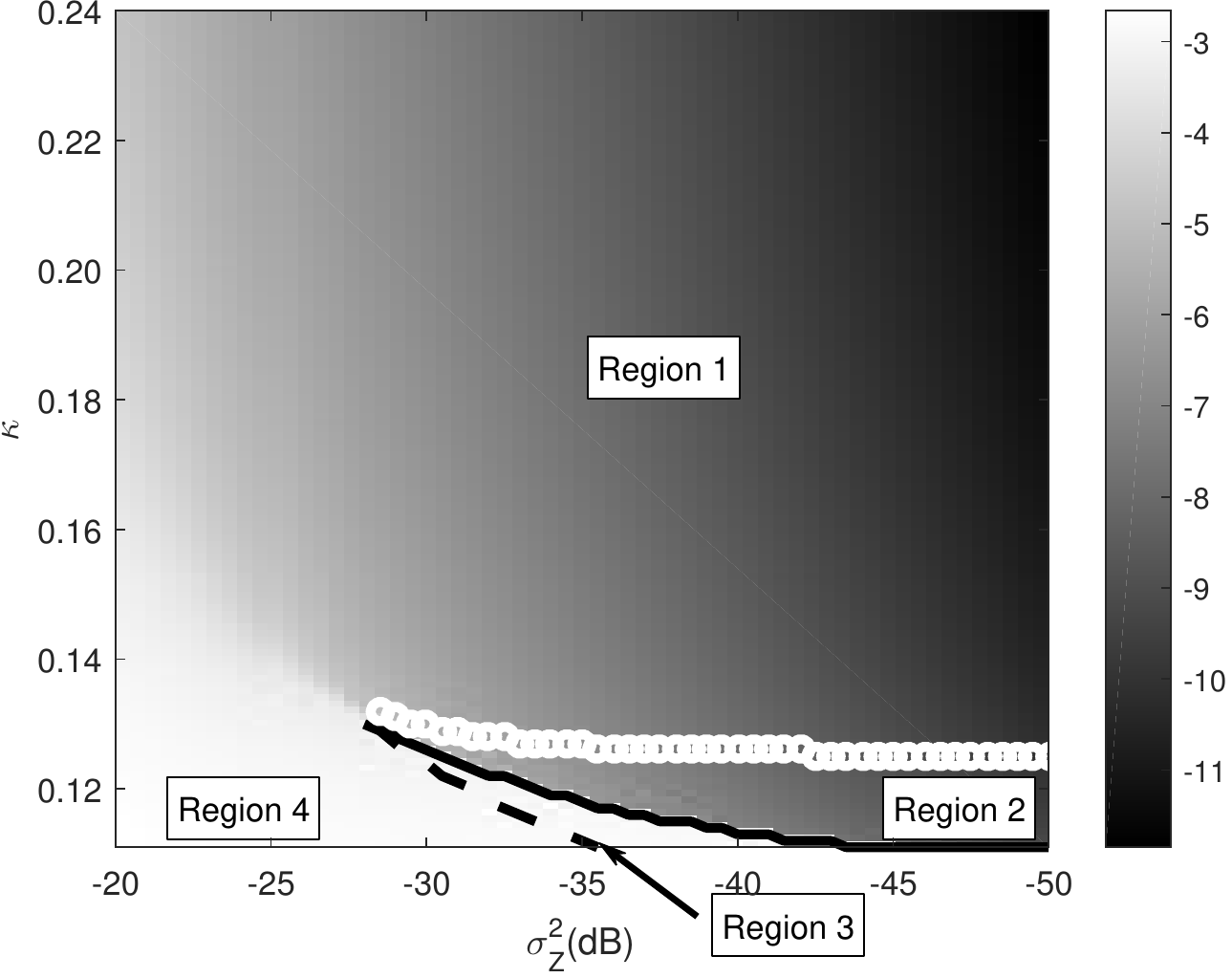}}
\caption{Performance regions for MMV with different $J$. The darkness of the shades corresponds to $\ln$(MMSE) for a certain noise variance $\sigma_Z^2$ and measurement rate $\kappa$. There are 4 regions, Regions~1 to~4, where the MMSE as a function of the noise variance $\sigma_Z^2$ and measurement rate $\kappa$ behaves differently. Regions~1 to~4 are separated by 3 thresholds, $\kappa_c(\sigma_Z^2)$ (the dashed curves), $\kappa_l(\sigma_Z^2)$ (the solid curves), and $\kappa_{BP}(\sigma_Z^2)$ (the curves comprised of little white circles); note that Section~\ref{sec:PerfRegion} discusses how to obtain these thresholds. (a) MMV with $J=1$, (b) MMV with $J=3$, and (c) MMV with $J=5$.}\label{fig:PerformanceRegions}
\end{figure*}

{\bf Regions 1 and 4:} The free energy~\eqref{eq:free_energy4} has one local maximum point w.r.t. the MSE $D$~\eqref{eq:DandE}. This $D$ leads to the globally maximum free energy and is the MMSE.

{\bf Regions~2 and~3:} There are 2 local maxima in the free energy, $D_1$ and $D_2$, where $D_1<D_2$. In Region~2, the smaller MSE, $D_1$, leads to the larger local maximum free energy~\eqref{eq:free_energy4} (hence, $\mathcal{F}(D_1)$ is the global maximum), and is the MMSE. In Region~3, the larger MSE, $D_2$, is the MMSE.

{\bf Boundaries between regions:} We denote the boundary separating regions~1 and~2 by the {\em BP threshold} $\kappa_{BP}(\sigma_Z^2)$, the boundary separating regions~2 and~3 by the {\em low noise threshold} $\kappa_l(\sigma_Z^2)$, and  the boundary separating regions~3 and~4 by the {\em critical threshold} $\kappa_c(\sigma_Z^2)$.

{\bf Numerical results:}
Consider $J$-dimensional Bernoulli-Gaussian signals~\eqref{eq:jsm} with sparsity rate $\rho=0.1$.
Evaluating the free energy~\eqref{eq:free_energy4} with the noise variance $\sigma_Z^2$ from -20 dB to -50 dB and measurement rate $\kappa$ from 0.11 to 0.24, we obtain the MMSE as a function of $\sigma_Z^2$ and $\kappa$ for $J=1,3$, and $5$, as shown in Figure~\ref{fig:PerformanceRegions}.\footnote{The MMV with $J=1$ becomes an SMV. The MMSE results in Figure~\ref{fig:MMV_J1} match with the SMV MMSE in Krzakala et. al.~\cite{Krzakala2012probabilistic,krzakala2012statistical} and Zhu and Baron~\cite{ZhuBaronCISS2013}.} The darkness of the shades represents the natural logarithm of the MMSE, $\ln$(MMSE). In all panels,
the critical threshold $\kappa_c(\sigma_Z^2)$, low noise threshold $\kappa_l(\sigma_Z^2)$, and BP threshold $\kappa_{BP}(\sigma_Z^2)$, as well as Regions~1-4, are marked.

\begin{figure}[t] 
\centering
\includegraphics[width=8cm]{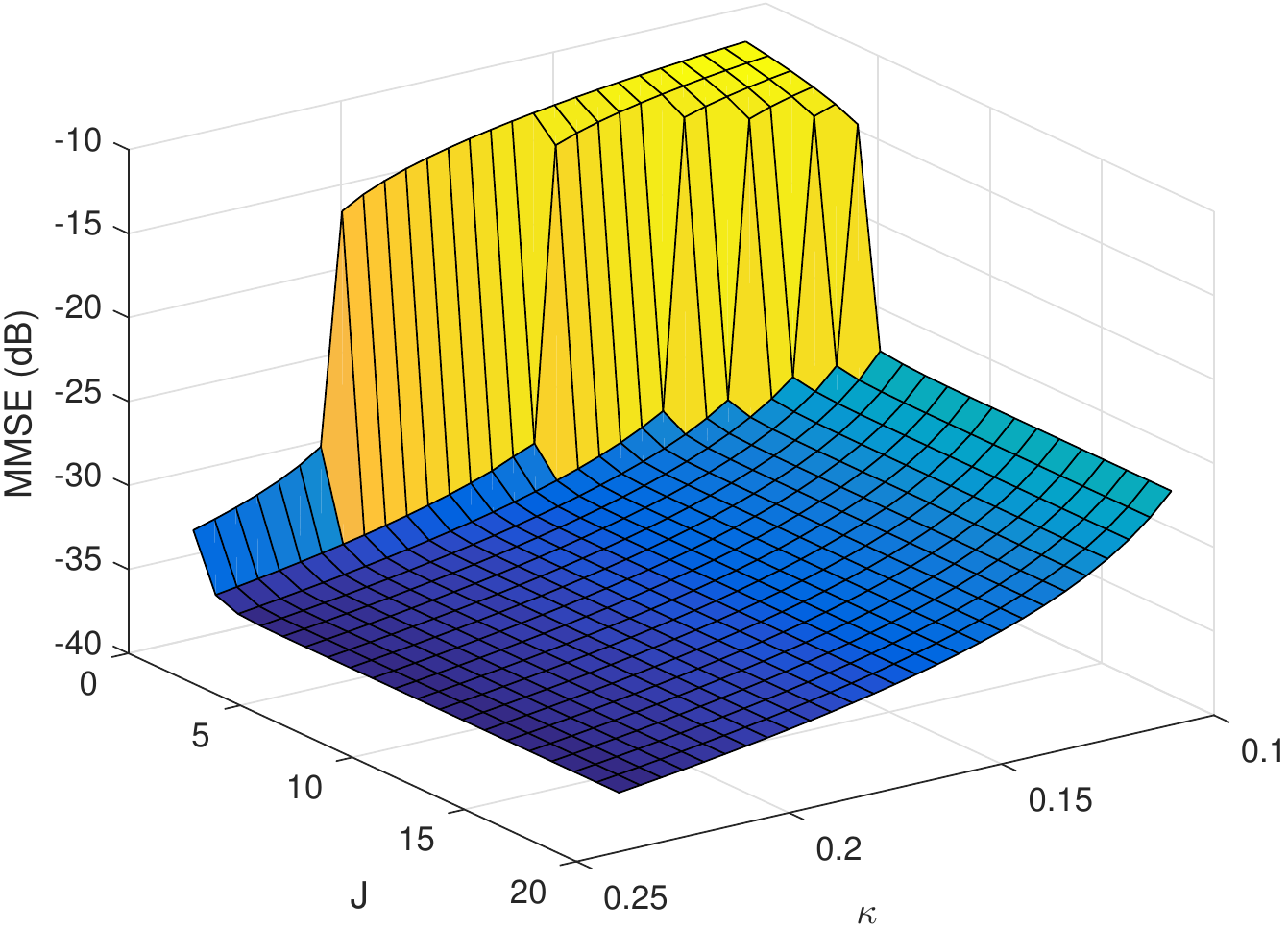}
\caption{MMSE in dB as a function of the number of jointly sparse signal vectors $J$ and the measurement rate $\kappa$ (noise variance $\sigma_Z^2=-35$ dB).}\label{fig:MMSE_R_J}
\end{figure}

In Regions~3 and~4, the best-possible algorithm yields a large MMSE for all noise variances. In contrast, in Regions~1 and~2, the optimal algorithm yields an MMSE that decreases with the noise variance $\sigma_Z^2$. To summarize, the optimal algorithm yields poor estimation performance below the low noise threshold $\kappa_{l}(\sigma_Z^2)$, and good performance above $\kappa_{l}(\sigma_Z^2)$.

We further examine the MMSE as a function of the number of jointly sparse signal vectors $J$ and the measurement rate $\kappa$. We plot the MMSE in dB scale in Figure~\ref{fig:MMSE_R_J}. The noise variance is -35 dB. We can see that the MMSE decreases with more signal vectors $J$ and greater measurement rate $\kappa$. However, the MMSE depends less on $J$ as $J$ is increased. Note that the discontinuity in the MMSE surface in Figure~\ref{fig:MMSE_R_J} is a result of the different performance regions that the various settings (different $J$ and $\kappa$) lie in.

\subsection{BP phase transition}\label{sec:phaseTrans}
Belief propagation (BP)~\cite{DMM2009,CSBP2010,Montanari2012,Bayati2011,Krzakala2012probabilistic,krzakala2012statistical,Barbier2015} is an algorithmic framework invented independently by researchers in coding theory, statistical physics, and artificial intelligence, which can often achieve the optimal estimation performance (MMSE) for linear inverse problems. The canonical BP updating rules appeared in~\eqref{eq:canonicalBP}.
When there are multiple local maxima $D_1<D_2$ in the free energy~\eqref{eq:free_energy4}, BP converges to the local maximum with the larger MSE, $D_2$~\cite{DMM2009,Montanari2012,Bayati2011,Krzakala2012probabilistic,krzakala2012statistical}. Hence, $D_2$ characterizes the MSE {\em predicted} for BP. Moving from Region~1 to Region~2 by decreasing the measurement rate $\kappa$ with fixed noise variance $\sigma_Z^2$, the number of local maxima increases from 1 to 2. Therefore, BP estimation performance experiences a sudden deterioration (increase in MSE) when the measurement rate $\kappa$ drops such that the combination of the noise variance $\sigma_Z^2$ and measurement rate $\kappa$ moves from Region~1 to Region~2. The BP threshold, $\kappa_{BP}(\sigma_Z^2)$, is the boundary between Regions~1 and~2, and is where the BP phase transition happens. That is, BP achieves poor estimation performance below $\kappa_{BP}(\sigma_Z^2)$, and good performance above $\kappa_{BP}(\sigma_Z^2)$.

\begin{myRemark} In Figure~\ref{fig:PerformanceRegions}, we see that increasing $J$ reduces the BP threshold $\kappa_{BP}(\sigma_Z^2)$. Since BP achieves the MMSE when $\kappa>\kappa_{BP}(\sigma_Z^2)$, increasing $J$ is beneficial to applications that use BP as the estimation algorithm.
\end{myRemark}

\begin{myRemark} We further numerically analyzed the low noise ($\sigma_Z^2\rightarrow 0$) and zero noise ($\sigma_Z^2=0$) cases. The low noise threshold $\kappa_l(\sigma_Z^2)$ converges to $\rho$ as the noise variance $\sigma_Z^2$ is decreased for $J=1,3$, and $5$. We believe that this numerical result holds for every $J$. Moreover, this result matches the theoretical robust threshold of Wu and Verd{\'u}~\cite{WuVerdu2012} for $J=1$ in the low noise limit. Our numerical results also show that the BP threshold $\kappa_{BP}(\sigma_Z^2)$ converges to some value for different $J$ as $\sigma_Z^2\rightarrow 0$. Analyzing these observations rigorously is left for future work.
\end{myRemark}

\begin{algorithm}[t]
\caption{AMP for MMV}
\label{algo:AMP_MMV}
\begin{algorithmic}[1]
\\{\bf Inputs:} Maximum number of iterations $T$, threshold $\epsilon$, sparsity rate $\rho$, noise variance $\sigma_Z^2$, measurements $\y^{(j)}$, and measurement matrices $\A^{(j)}, \forall j$
\\{\bf Initialize:} $t=1,\delta=\infty,\w^{(j)}=\y^{(j)},\Theta_j=0,v^{(j)}_l=\rho\sigma_Z^2,a^{(j)}_l=0,\forall l,j$
\While{$t<T$ and $\delta>\epsilon$}
\For{$j\leftarrow 1$ to $J$}
\\\quad\quad\quad$\q^{(j)}=\frac{\y^{(j)}-\w^{(j)}}{\sigma_Z^2+\Theta_j}$
\\\quad\quad\quad$\Theta_j=\frac{1}{N}\sum_{l=1}^N v^{(j)}_l$
\\\quad\quad\quad$\w^j=\A^{(j)} \a^{(j)}-\Theta_j \q^{(j)}$
\\\quad\quad\quad$\Sigma_j=\frac{N(\sigma_Z^2+\Theta_j)}{M}$ \Comment{Scalar channel noise variance}
\\\quad\quad\quad$\R^{(j)}=\a^{(j)}+\Sigma_j \l(\A^{(j)}\r)^{\top} \frac{\y^{(j)}-\w^{(j)}}{\sigma_Z^2+\Theta_j}$ \Comment{Pseudodata}
\\\quad\quad\quad$\widehat{\a}^{(j)}=\a^{(j)}$ \Comment{Save current estimate}
\EndFor
\For{$l\leftarrow 1$ to $N$}
\\\quad\quad\quad$\l\{v^{(j)}_l\r\}_{j=1}^J=f_{v_l}\l(\{\Sigma_j\}_{j=1}^J,\l\{R^{(j)}_l\r\}_{j=1}^J\r)$ \Comment{Variance}
\\\quad\quad\quad$\l\{a^{(j)}_l\r\}_{j=1}^J=f_{a_l}\l(\{\Sigma_j\}_{j=1}^J,\l\{R^{(j)}_l\r\}_{j=1}^J\r)$ \Comment{Estimate}
\EndFor
\\\quad\ \ $t=t+1$ \Comment{Increment iteration index.}
\\\quad\ \ $\delta=\frac{1}{NJ}\sum_{l=1}^N\sum_{j=1}^J\l(\widehat{a}^{(j)}_l-a^{(j)}_l\r)^2$ \Comment{Change in estimate}
\EndWhile
\\{\bf Outputs:} Estimate $\a^{(j)},\forall j$
\end{algorithmic}
\end{algorithm}

\subsection{BP simulation}\label{sec:AMPsim}

After obtaining the theoretic MMSE for MMV, as well as the MSE predicted for BP, we run some simulations to estimate the $\underline{\x}^{(j)}$ of channel~\eqref{eq:MMVmodel} in a Bayesian setting.
The algorithm we use is approximate message passing (AMP)~\cite{DMM2009,Montanari2012,Bayati2011,Krzakala2012probabilistic,krzakala2012statistical,Barbier2015}, which is an approximation to the BP algorithm; related algorithms have been proposed by Ziniel and Schniter~\cite{ZinielSchniter2013MMV} and Kim et al.~\cite{KimChangJungBaronYe2011}.
In the SMV case, when the measurement matrix and the signal have i.i.d. entries, AMP has the state evolution (SE) formalism~\cite{DMM2011,Bayati2011,JavanmardMontanari2012,Donoho2013,Bayati2015} that tracks the evolution of the MSE at each iteration. Recently, Javanmard and Montanari proved that SE tracks AMP rigorously in an SMV setting with a spatially coupled measurement matrix~\cite{JavanmardMontanari2012}. According to our transform in Figure~\ref{fig:channel}, we can see that the proof~\cite{JavanmardMontanari2012} could be extended to the MMV setting. Note that SE allows to compute the
highest equilibrium of Gibbs free energy~\cite{DMM2011,Bayati2011,JavanmardMontanari2012,Donoho2013,Bayati2015}, which corresponds to the local optimum $D_2$ in Section~\ref{sec:phaseTrans}. Hence, AMP often achieves the same MSE as BP
and we use AMP simulation results to demonstrate that the MMSE can often be achieved.\footnote{When the assumptions about the measurement matrix and signal~\cite{DMM2009,Montanari2012,Bayati2011,Krzakala2012probabilistic,krzakala2012statistical,Barbier2015} are violated, AMP might suffer from divergence issues.}
Considering~\eqref{eq:MMVmodel}, we simplify the AMP algorithm in Barbier and Krzakala~\cite{Barbier2015} to obtain Algorithm~\ref{algo:AMP_MMV},\footnote{Note that Algorithm~\ref{algo:AMP_MMV} is a straightforward simplification of the AMP algorithm by Barbier and Krzakala~\cite{Barbier2015}.} where $\{\Sigma_j\}_{j=1}^J$, $\l\{R^{(j)}_l\r\}_{j=1}^J$, $\l\{a_l^{(j)}\r\}_{j=1}^J$ and $\l\{v_l^{(j)}\r\}_{j=1}^J$ refer to sets of all intermediate variables $\Sigma_j$, pseudodata $R^{(j)}_l$, estimates $a_l^{(j)}$, and variances $v^{(j)}_l,\ j\in\{1,\cdots,J\},\ l\in\{1,\cdots,N\}$, respectively.
The current iteration $t$, change in the estimate $\delta$, and intermediate variables $\Theta_j,\ j\in\{1,\cdots,J\}$, are scalars. The intermediate variables $\q^{(j)}$ and $\w^{(j)}$ are vectors of length $M$. The functions $f_{a_l}\l(\{\Sigma_j\}_{j=1}^J,\l\{R^{(j)}_l\r\}_{j=1}^J\r)$ and $f_{v_l}\l(\{\Sigma_j\}_{j=1}^J,\l\{R^{(j)}_l\r\}_{j=1}^J\r)$ are given by
\begin{equation*}
f_{a_l}\l(\{\Sigma_j\}_{j=1}^J,\l\{R^{(j)}_l\r\}_{j=1}^J\r)=\frac{\rho\frac{1}{\Sigma_j+1}\l\{R^{(j)}_l\r\}_{j=1}^J}{\rho+(1-\rho)\pij{j=1}{J}\l\{\sqrt{1+\frac{1}{\Sigma_j}}\text{exp}\l[-\frac{\l(R^{(j)}_l\r)^2}{2\Sigma_j(\Sigma_j+1)}\r]\r\}},
\end{equation*}
\begin{equation*}
f_{v_l}\l(\{\Sigma_j\}_{j=1}^J,\l\{R^{(j)}_l\r\}_{j=1}^J\r)\!=\!-\l[f_{a_l}\l(\{\Sigma_j\}_{j=1}^J,\l\{R^{(j)}_l\r\}_{j=1}^J\r)\r]^2
+\frac{\rho\frac{1}{\Sigma_j+1}\l[\l(\l\{R^{(j)}_l\r\}_{j=1}^J\r)^2\frac{1}{\Sigma_j+1}\!+\!\Sigma_j\r]}{\rho+(1-\rho)\pij{j=1}{J}\!\l\{\!\sqrt{1+\frac{1}{\Sigma_j}}\text{exp}\l[-\frac{\l(R^{(j)}_l\r)^2}{2\Sigma_j(\Sigma_j+1)}\r]\r\}},
\end{equation*}
for $J$-dimensional Bernoulli-Gaussian signals~\eqref{eq:jsm}.

We simulated the signals in~\eqref{eq:jsm} with $J=3$ signal vectors and sparsity rate $\rho=0.1$ measured by a channel~\eqref{eq:MMVmodel} with measurement rate $\kappa\in[0.11,0.24]$ and noise variance $\sigma_Z^2\in[-20,-50]$ dB. For each setting, we generated 50 signals of length $N=5000$, and the resulting MSE compared to the MSE predicted for BP is shown in Figure~\ref{fig:AMPoverMSE}.\footnote{We simulated both $J$ different measurement matrices $\underline{\A}^{(j)}$ and $J$ identical $\underline{\A}^{(j)}$. Both results match the MSE predicted for BP, which support our conclusion that the MMSE's of both settings are the same. Figure~\ref{fig:AMPoverMSE} is with $J$ different $\underline{\A}^{(j)}$.}
\begin{figure}[t] 
\centering
\includegraphics[width=8cm]{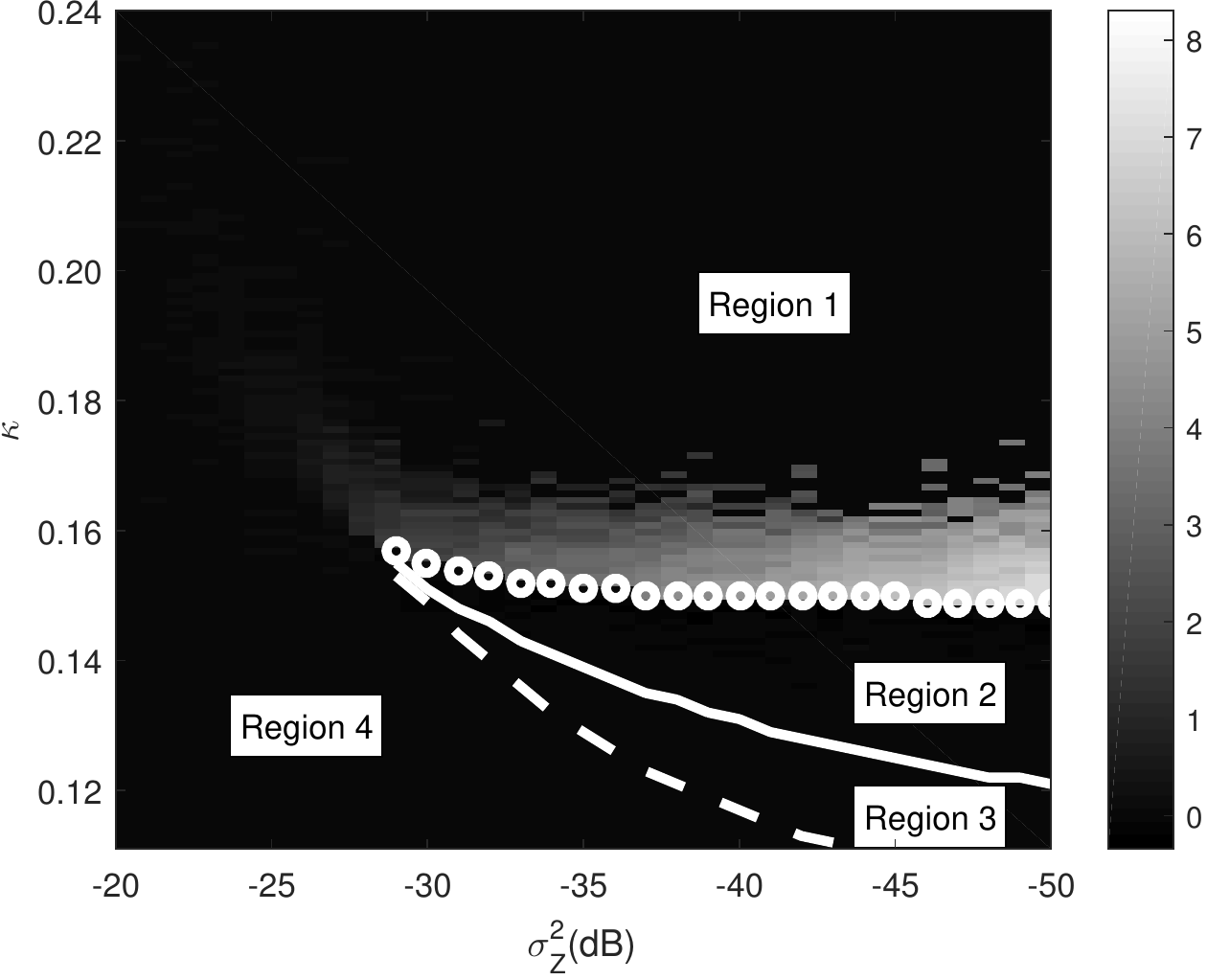}
\caption{AMP simulation results ($\text{MSE}_{\text{AMP}}$) compared to the MSE predicted for BP ($\text{MSE}_{\text{BP}}$) with $J=3$ jointly sparse signal vectors. The dashed curve, solid curve, and the curve comprised of little circles correspond to thresholds $\kappa_c(\sigma_Z^2),\ \kappa_l(\sigma_Z^2)$, and $\kappa_{BP}(\sigma_Z^2)$, respectively. Regions 1-4 are also marked. The darkness of the shades denotes $\ln \l(\frac{\text{MSE}_{\text{AMP}}}{\text{MSE}_{\text{BP}}}\r)$, which we expect to be zero (completely dark shades) in the entire $\kappa$ versus $\sigma_Z^2$ plane. The narrow bright band above the BP threshold indicates the mismatch between the MSE from the simulation and the MSE predicted for BP.}\label{fig:AMPoverMSE}
\end{figure}
The labels of the thresholds are omitted for brevity. We can see that AMP simulation results match with the MSE predicted for BP and BP phase transition from the replica analysis of Section~\ref{sec:phaseTrans}. Note that there is a narrow band of light shades above the BP threshold, $\kappa_{BP}(\sigma_Z^2)$ (the top threshold), meaning that the MSE from the simulation is greater than the MSE predicted for BP; this is due to randomness in our generated signals and channels. Note that we also compared the AMP simulation results to that of the M-SBL algorithm~\cite{YeKimBresler2015}, a widely used algorithm to solve the MMV problem. The M-SBL results were not as good. Indeed, because AMP is often an
approach that achieves the MMSE, other algorithms are expected to provide greater MSE.

\section{Extension to Arbitrary Error Metrics}\label{sec:errorMetric}
In this chapter, we have obtained the MMSE for MMV problems.
As mentioned in Section~\ref{sec:MMVintro}, there are many estimation approaches for MMV problems~\cite{tropp2006ass,chen2006trs,malioutov2005ssr,tropp2006ass2,cotter2005ssl,
Mishali08rembo,LeeBreslerJunge2012,YeKimBresler2015,ZinielSchniter2011}. However, when running estimation algorithms for MMV problems, people might be interested in obtaining an estimate whose ``user-defined'' error is as small as possible. For example, if estimating the underlying signal is important, people may use the MSE metric; when there might be outliers in the estimate, using the mean absolute error metric might be more appropriate. For applications such as compressive diffuse optical tomography~\cite{LeeKimBreslerYe2011}, estimating the support set of the jointly sparse underlying signals is of more interest. Seeing that there are different algorithms minimizing different error metrics, but there is no prior work discussing the optimal performance with user-defined (arbitrary) error metrics in MMV, it is of interest to study the optimal performance with user-defined error metrics in MMV problems and also design algorithms to achieve such optimal performance.

Tan and coauthors~\cite{Tan2014,Tan2014Infty} studied the optimal performance for arbitrary additive error metrics for an SMV problem~\eqref{eq:SMV} by taking advantage of the properties of BP~\cite{DMM2009,CSBP2010,Bayati2011,Montanari2012,Krzakala2012probabilistic,krzakala2012statistical,Barbier2015}: BP yields an equivalent scalar channel
\begin{equation}\label{eq:equivScalarChannel}
\widetilde{\y}=\x+\widetilde{\z},
\end{equation}
whose posterior $f(\x|\widetilde{\y})$ approaches the true posterior distribution $f(\x|\y)$ under certain conditions~\cite{RanganGAMP2011ISIT}. Using $f(\x|\widetilde{\y})$, Tan and coauthors designed the denoiser that minimizes the (additive) user-defined error metrics for~\eqref{eq:equivScalarChannel}.

According to Section~\ref{sec:model} and Figure~\ref{fig:channel}, we can transform the MMV problem~\eqref{eq:MMVmodel} into an SMV problem~\eqref{eq:MMVchannel}. Hence, we can extend the work of Tan and coauthors~\cite{Tan2014,Tan2014Infty} to study the optimal performance for arbitrary additive error metrics, as well as to build algorithms that achieve the optimal performance for MMV~\eqref{eq:MMVmodel}. The
details are left for future work.

\section{Conclusion}\label{sec:conclusion}
We analyzed the minimum mean squared error (MMSE) for two settings of multi-measurement vector (MMV) problems, where the entries in the signal vectors are independent and identically distributed (i.i.d.), and share the same support.
One MMV setting has i.i.d. Gaussian measurement matrices, while the other MMV setting has identical i.i.d. Gaussian measurement matrices. Replica analysis yields identical free energy expressions for these two settings in the large system limit when the signal length goes to infinity and the number of measurements scales with the signal length. Because of the identical free energy expressions, the MMSE's for both MMV settings are identical. By numerically evaluating the free energy expression, we identified different performance regions for MMV where the MMSE as a function of the channel noise variance and the measurement rate behaves differently. We also identified a phase transition for belief propagation algorithms (BP) that separates regions where BP achieves the MMSE asymptotically and where it is sub-optimal. Simulation results of an approximated version of BP matched with the mean squared error (MSE) predicted by replica analysis. As a special case of MMV, we extended our replica analysis to complex single measurement vector (SMV) problems, so that we can calculate the MMSE for complex SMV with real or complex measurement matrices. Seeing that the MSE might not be the only error metric that is of interest, we proposed to extend the work of Tan and coauthors~\cite{Tan2014,Tan2014Infty} to MMV problems, so that we can optimize over different user-defined additive error metrics in MMV applications.

\chapter{Performance Trade-offs in Multi-Processor Approximate Message Passing}
\label{chap-MP-AMP}
\chaptermark{Performance Trade-offs}

In Chapter~\ref{chap-MMV}, we focused on analyzing the information theoretic performance limits for multi-measurement vector problems~\eqref{eq:MMVmodel_intro}. Our analysis is readily extended to single measurement vector problems~\eqref{eq:SMV}.
In practice, many algorithms run in distributed networks, especially as we are entering the ``big data'' era.
Running estimation algorithms across distributed networks can incur different costs besides the quality of the estimation.
Some prior art has focused on reducing certain costs such as the communication cost~\cite{Han2014} and the computation cost~\cite{MaBaronNeedell2014}, but there has been less
progress relating different costs and achieving optimal trade-offs among them. Despite the lack of
such works, these trade-offs are important to system designers in order to produce efficient systems.
Studying the relation between different costs is a broad problem with a rich design space. Therefore, in this chapter, we focus our discussion on  one specific distributed algorithm as an example: the ``multi-processor approximate message passing'' algorithm (MP-AMP)~\cite{Han2014,HanZhuNiuBaron2016ICASSP}, and study the optimal trade-offs among different costs.
In each MP-AMP iteration, nodes of the multi-processor system and its fusion center
exchange lossily compressed messages pertaining to their estimates of the input.
In this setup, we derive the optimal per-iteration coding rates using dynamic programming.
We analyze the excess mean squared error (EMSE) beyond the minimum mean squared error, and
prove that, in the limit of low EMSE,
the optimal coding rates increase approximately linearly per iteration. Additionally, we obtain that the combined cost of computation
and communication scales with the desired estimation quality according to $O(\log^2(1/\text{EMSE}))$.
Finally, we study trade-offs between the physical
costs of the estimation process including computation time,
communication loads, and the estimation quality as a multi-objective optimization problem,
and characterize the properties of the Pareto optimal surfaces. This chapter is based on our work with Han et al.~\cite{HanZhuNiuBaron2016ICASSP} and with Baron and Beirami~\cite{ZhuBeiramiBaron2016ISIT,ZhuBaronMPAMP2016ArXiv}.

\section{Related Work and Contributions}

\subsection{Related work}
Many scientific and engineering problems~\cite{DonohoCS,CandesRUP} can be approximated
using a linear model,
\begin{equation}
\y = \A\x + \z,
\label{eq:matrix_channel}
\end{equation}
where $\x\in\mathbb{R}^N$ is the unknown input signal, $\A\in\mathbb{R}^{M\times N}$ is the matrix that characterizes the linear model, and $\z\in\mathbb{R}^M$ is measurement noise.
The goal is to estimate $\x$ from the noisy measurements $\y$ given $\A$ and statistical information about $\z$; this is a {\em linear inverse problem}. Alternately, one could view the estimation of $\x$ as fitting or learning a linear model for the data comprised of $\y$ and $\A$.

When $M\ll N$, the setup~\eqref{eq:matrix_channel} is known as compressed sensing (CS)~\cite{DonohoCS,CandesRUP}; by posing a sparsity or compressibility
requirement on the signal,
it is indeed possible to accurately recover $\x$ from the ill-posed linear model~\cite{DonohoCS,CandesRUP} when the number of measurements $M$ is large enough, and the noise level is modest. However, we might need $M>N$ when the signal is dense or the noise is substantial. Hence, we do not constrain ourselves to the case of $M\ll N$.

Approximate message passing (AMP)~\cite{DMM2009,Montanari2012,Bayati2011,Rush_ISIT2016_arxiv} is an iterative framework that solves  linear inverse problems by successively decoupling~\cite{Tanaka2002,GuoVerdu2005,GuoWang2008} the problem in~\eqref{eq:matrix_channel} into scalar denoising
problems with additive white Gaussian noise (AWGN). AMP has received considerable attention, because of its fast convergence and the state evolution (SE) formalism~\cite{DMM2009,Bayati2011,Rush_ISIT2016_arxiv}, which offers a precise
characterization of the AWGN denoising problem in each iteration.
In the Bayesian setting, AMP often achieves the minimum mean squared error
(MMSE)~\cite{GuoBaronShamai2009,RFG2012,ZhuBaronCISS2013,Krzakala2012probabilistic} in the limit of large linear systems ($N\rightarrow\infty, \frac{M}{N}\rightarrow \kappa$, cf. Definition~\ref{def:chap1-largeSystemLimit}).

In real-world applications, a multi-processor (MP) version of the linear model could be of interest, due to either storage limitations in each individual processor node, or the need for fast computation. This chapter considers  multi-processor linear model (MP-LM)~\cite{Mota2012,Patterson2014,Han2014,Ravazzi2015,
Han2015SPARS,HanZhuNiuBaron2016ICASSP}, in which there are $P$ {\em processor nodes} and a {\em fusion center}.
Recall from~\eqref{eq:one-node-meas_intro} that in an MP-LM, each
processor node stores $\frac{M}{P}$ rows of the matrix $\A$, and acquires the corresponding linear measurements of the underlying signal $\x$. Without loss of generality, we model the measurement system in processor node $p\in \{1,\cdots,P\}$ as
 \begin{equation}\label{eq:one-node-meas}
    y_i=\A_i \x+z_i,\ i\in \left\{\frac{M(p-1)}{P}+1,\cdots,\frac{Mp}{P}\right\},
 \end{equation}
 where $\A_i$ is the $i$-th row of $\A$, and $y_i$ and $z_i$ are the $i$-th entries of $\y$ and $\z$, respectively.
Once every $y_i$ is collected, we run distributed algorithms among the fusion center and $P$ processor nodes to estimate the signal $\x$.
MP versions of AMP (MP-AMP) for MP-LM have been studied in the literature~\cite{Han2014,HanZhuNiuBaron2016ICASSP}.
Usually, MP platforms are designed for distributed settings such as sensor networks~\cite{pottie2000,estrin2002} or large-scale ``big data" computing systems~\cite{EC2}, where the computational and communication burdens can differ among different settings. We reduce the communication costs of MP platforms by applying lossy compression~\cite{Berger71,Cover06,GershoGray1993} to the communication portion of MP-AMP.
Our key idea in this work is to minimize the total communication and computation costs by varying the lossy compression schemes in different iterations of MP-AMP.

\subsection{Contributions}
Rate-distortion (RD) theory suggests that we can transmit data with greatly reduced coding rates, if we allow some distortion at the output.
However, the MP-AMP problem does
not directly fall into the RD framework, because the quantization error in the current iteration feeds into estimation errors in future iterations. We quantify the interaction between these two forms of error by studying the excess mean squared error (EMSE)
of MP-AMP above the MMSE (EMSE=MSE-MMSE, where MSE denotes the mean squared error).
Our first contribution (Section~\ref{sec:DP}) is to use dynamic programming (DP, cf. Bertsekas~\cite{bertsekas1995}) to find a sequence of coding rates that yields a desired EMSE while achieving the smallest combined cost of
communication and computation; our DP-based scheme is proved to yield optimal coding rates.

Our second contribution (Section~\ref{sec:linRateTh}) is to pose the task of finding the optimal coding rate at each iteration in the low EMSE limit as a convex optimization problem. We prove that the optimal coding rate grows approximately linearly in the low EMSE limit. At the same time, we also provide
the theoretic asymptotic growth rate of the optimal coding rates in the limit of low EMSE.
This provides practitioners with a heuristic to find a near-optimal coding rate sequence without solving the optimization problem.
The linearity of the  optimal coding rate sequence (defined in Section~\ref{sec:DP}) is also illustrated numerically.
With the rate being approximately linear, we obtain that the combined cost of computation and communication scales as $O(\log^2(1/\text{EMSE}))$.

In Section~\ref{sec:Pareto}, we further consider a rich design space that includes various costs, such as the number of iterations $T$, aggregate coding rate $R_{agg}$, which is the sum of the coding rates in all iterations and is formally defined in~\eqref{eq:R_agg}, and the MSE achieved by the estimation algorithm. In such a rich design space, reducing any cost is likely to incur an increase in other costs, and it is impossible to simultaneously minimize all the costs.
Han et al.~\cite{Han2014} reduce the communication costs, and Ma et al.~\cite{MaBaronNeedell2014} develop an algorithm with reduced computation; both works~\cite{Han2014,MaBaronNeedell2014} achieve a reasonable MSE. However, the optimal trade-offs in this rich design space have not been studied.
Our third contribution is to pose the problem of finding the best trade-offs among the individual costs $T,\ R_{agg}$, and $\text{MSE}$ as a multi-objective optimization problem (MOP), and study the properties of Pareto optimal tuples~\cite{DasDennisPareto1998} of this MOP. These properties are verified numerically using
the DP-based scheme developed in this chapter.

Finally, we emphasize that although this chapter is presented for the specific framework of MP-AMP,
similar methods could be applied to other iterative distributed  algorithms, such as
consensus averaging~\cite{Frasca2008,Thanou2013}, to obtain the optimal coding rate as well as
optimal trade-offs between communication and computation costs.

{\bf  Organization:}
The rest of the chapter is organized as follows. Section~\ref{sec:setting_MP_Chap} provides
background content. Section~\ref{sec:DP} formulates a DP scheme that finds an optimal coding rate. Section~\ref{sec:linRateTh} proves that any optimal coding rate in the low EMSE limit grows approximately linearly as iterations proceed. Section~\ref{sec:Pareto} studies the optimal trade-offs among the computation cost, communication cost, and the MSE of the estimate. Section~\ref{sec:realworld} uses some real-world examples to showcase the different trade-offs between communication and computation costs, and Section~\ref{sec:conclude} concludes the chapter.

\section{Background}\label{sec:setting_MP_Chap}
\subsection{Centralized linear model using AMP}\label{sec:centralAMP}
In our linear model~\eqref{eq:matrix_channel}, we consider an independent and identically distributed (i.i.d.) Gaussian measurement matrix $\A$, i.e.,
$A_{i,j}\sim\mathcal{N}(0,\frac{1}{M})$, where $\mathcal{N}(\mu,\sigma^2)$
denotes a Gaussian distribution with mean $\mu$ and variance $\sigma^2$.
The signal entries follow an i.i.d. distribution, $f_X(x)$.
The noise entries obey $z_i\sim\mathcal{N}(0,\sigma_Z^2)$, where $\sigma_Z^2$ is the noise variance.

Starting from ${\bf x}_0={\bf 0}$, the AMP framework~\cite{DMM2009} proceeds iteratively according to\footnote{AMP is an approximation to the belief propagation algorithm~\eqref{eq:canonicalBP}.}
\begin{align}
{\bf x}_{t+1}&=\eta_t({\bf A}^{\top}{\bf r}_t+{\bf x}_t)\label{eq:AMPiter1},\\
{\bf r}_t&={\bf y}-{\bf Ax}_t+\frac{1}{\kappa}{\bf r}_{t-1}
\langle \eta_{t-1}'({\bf A}^{\top}{\bf r}_{t-1}+{\bf x}_{t-1})\rangle\label{eq:AMPiter2},
\end{align}
where $\eta_t(\cdot)$ is a denoising function, $\eta_{t}'(\cdot)=\frac{d \eta_t({\cdot})}{d\{\cdot\}}$ is the derivative of $\eta_t(\cdot)$, and~$\langle{\bf u}\rangle=\frac{1}{N}\sum_{i=1}^N u_i$
for any vector~${\bf u}\in\mathbb{R}^N$. The subscript $t$ represents the iteration index, ${\{\cdot\}}^\top$ denotes the matrix transpose operation, and $\kappa=\frac{M}{N}$ is the measurement rate.
Owing to the decoupling effect~\cite{Tanaka2002,GuoVerdu2005,GuoWang2008}, in each AMP iteration~\cite{Bayati2011,Montanari2012,Rush_ISIT2016_arxiv},
the vector~$\f_t={\bf A}^{\top}{\bf r}_t+{\bf x}_t$
in (\ref{eq:AMPiter1}) is statistically equivalent to
the input signal ${\bf x}$ corrupted by AWGN $\w_t$ generated by a source $W\sim \mathcal{N}(0,\sigma_t^2)$,
\begin{equation}\label{eq:equivalent_scalar_channel}
\f_t=\x+\w_t.
\end{equation}
We call~\eqref{eq:equivalent_scalar_channel} the {\em equivalent scalar channel}.
In large systems ($N\rightarrow\infty, \frac{M}{N}\rightarrow \kappa$),\footnote{Note that the results of this chapter only hold for large systems.} a useful property of AMP~\cite{Bayati2011,Montanari2012,Rush_ISIT2016_arxiv} is that
the noise variance $\sigma_t^2$ evolves following state evolution (SE):
\begin{equation}
\sigma_{t+1}^2=\sigma^2_Z+\frac{1}{\kappa}\text{MSE}(\eta_t,\sigma_t^2),\label{eq:ori_SE}
\end{equation}
where
$\text{MSE}(\eta_t,\sigma_t^2)=\mathbb{E}_{X,W}\left[\left( \eta_t\left( X+W \right)-X \right)^2\right]$, $\mathbb{E}_{X,W}(\cdot)$ is expectation with respect to (w.r.t.) $X$ and $W$,
and $X$ is the source that generates $\x$. Note that $\sigma_1^2=\sigma_Z^2+\frac{\mathbb{E}[X^2]}{\kappa}$, because of the all-zero initial estimate for $\x$.
Formal statements for SE appear
in prior work~\cite{Bayati2011,Montanari2012,Rush_ISIT2016_arxiv}.

In this chapter, we confine ourselves to the Bayesian setting, in which we assume knowledge
of the true prior, $f_X(x)$, for the signal $\x$.
Therefore, throughout this chapter we use conditional expectation, $\eta_t(\cdot)=\mathbb{E}[\x|\f_t]$, as the MMSE-achieving
denoiser.\footnote{Tan et al.~\cite{Tan2014} showed that AMP with MMSE-achieving denoisers can be used as a building block for algorithms that minimize arbitrary user-defined error metrics.} The derivative of $\eta_t(\cdot)$, which is continuous, can be easily obtained, and is omitted for brevity.
Other denoisers such as soft thresholding~\cite{DMM2009,Montanari2012,Bayati2011} yield MSE's that are larger than that of the MMSE denoiser, $\eta_t(\cdot)=\mathbb{E}[\x|\f_t]$.
When the true prior for $\x$ is unavailable, parameter estimation techniques
can be used~\cite{MaZhuBaron2016TSP}; Ma et al.~\cite{MaBaronBeirami2015ISIT} study the behavior of AMP when the denoiser uses a mismatched prior.

\subsection{MP-LM using lossy MP-AMP}\label{sec:MP-CS_for_MP-AMP}
In the sensing problem formulated in~\eqref{eq:one-node-meas}, the measurement matrix is stored in a distributed manner in each processor node. Lossy MP-AMP~\cite{HanZhuNiuBaron2016ICASSP} iteratively solves MP-LM using lossily compressed messages:
\begin{equation}
\mbox{Processor nodes:}\ {\bf r}_t^p={\bf y}^p-\A^p\x_t+\frac{1}{\kappa}{\bf r}_{t-1}^p
\omega_{t-1},\label{eq:slave1}
\end{equation}
\begin{equation}
\quad \quad \quad {\bf f}_t^p=\frac{1}{P}\x_t+(\A^p)^{\top}{\bf r}_t^p,\label{eq:slave2}
\end{equation}
\begin{equation}
\mbox{Fusion center:}\ {\bf f}_{Q,t}=\sum_{p=1}^P Q({\bf f}_{t}^p),\ \omega_{t}=\langle d\eta_{t}({\bf f}_{Q,t})\rangle,\label{eq:master0}
\end{equation}
\begin{equation}
 \x_{t+1}=\eta_{t}( {\bf f}_{Q,t}),\label{eq:master}
\end{equation}
where $Q(\cdot)$ denotes quantization, and
an MP-AMP iteration refers to the process from~\eqref{eq:slave1} to~\eqref{eq:master}.
The processor nodes send quantized (lossily compressed) messages, $Q(\f_t^p)$, to the fusion center. The reader might notice that the fusion center also needs to transmit the denoised signal vector $\x_t$ and a scalar $\omega_{t-1}$ to the processor nodes. The transmission of
$\omega_{t-1}$ is negligible, and the fusion center may broadcast $\x_t$ so that naive compression of $\x_t$, such as compression with a fixed quantizer, is sufficient. Hence, we will not discuss possible  compression of messages transmitted by the fusion center.

Assume that we quantize $\f_t^p, \forall p$, and use $C$ bits to encode the quantized vector $Q(\f_t^p)\in\mathbb{R}^N$. According to~\eqref{eq:codingRate}, the {\em coding rate} is $R=\frac{C}{N}$. We incur an {\em expected distortion}
\begin{equation*}
D_t^p=\mathbb{E}\left[\frac{1}{N}\sum_{i=1}^N(Q(f_{t,i}^p)-f_{t,i}^p)^2\right]
\end{equation*}
at iteration $t$ in each processor node,\footnote{Because we assume that
$\A$ and $\z$ are both i.i.d., the expected distortions are the same over
all $P$ nodes, and can be denoted by $D_t$ for simplicity.
Note also that $D_t=\mathbb{E}[(Q(f_{t,i}^p)-f_{t,i}^p)^2]$
due to $\x$ being i.i.d.}
where $Q(f_{t,i}^p)$ and $f_{t,i}^p$ are the $i$-th entries of the vectors $Q(\f_t^p)$ and $\f_t^p$, respectively,
and the expectation is over $\f_t^p$.
When the size of the problem grows, i.e., $N\rightarrow\infty$, the rate-distortion (RD) function, denoted by $R(D)$, offers the fundamental information theoretic limit on the coding rate $R$ for communicating a long sequence up to distortion $D$~\cite{Cover06,Berger71,GershoGray1993,WeidmannVetterli2012}.
A pivotal conclusion from RD theory is that coding rates can be greatly reduced even if $D$ is small.
The function $R(D)$ can be computed in various ways~\cite{Arimoto72,Blahut72,Rose94}, and can be achieved by an RD-optimal quantization scheme in the limit of large $N$.
Other quantization schemes may require larger coding rates to achieve the same expected distortion $D$.

The goal of this chapter is to understand the fundamental trade-offs for MP-LM using MP-AMP. Hence,
unless otherwise stated, we assume that
appropriate vector quantization (VQ)
schemes~\cite{LBG1980,Gray1984,GershoGray1993}, which achieve $R(D)$,
are applied within each MP-AMP iteration, although our analysis is readily extended to practical quantizers such as entropy coded scalar quantization (ECSQ)~\cite{GershoGray1993,Cover06}. (Note that the cost of running quantizers
in each processor node is not considered, because
the cost of processing a bit is usually much smaller than the cost of transmitting it.)
Therefore, the signal {\em at the fusion center} before denoising can be modeled as
\begin{align}
\f_{Q,t}=\sum_{p=1}^P Q(\f_t^p)=\x+\w_t+\n_t,\label{eq:indpt_noises}
\end{align}
where $\w_t$ is the equivalent scalar channel noise~\eqref{eq:equivalent_scalar_channel} and $\n_t$ is the overall quantization error whose entries follow $\mathcal{N}(0,PD_t)$.
Because the quantization error, $\n_t$, is a sum of quantization errors in the $P$ processor nodes,
$\n_t$ resembles Gaussian noise due to the central limit theorem.
Han et al. suggest that SE for lossy MP-AMP~\cite{HanZhuNiuBaron2016ICASSP} (called lossy SE) follows
\begin{equation}\label{eq:SE_Q}
\sigma_{t+1}^2=\sigma^2_Z+\frac{1}{\kappa}\text{MSE}(\eta_t,\sigma_t^2+PD_t),
\end{equation}
where $\sigma_t^2$ can be estimated by
$\widehat{\sigma}_t^2 = \frac{1}{M}\|{\bf r}_t\|_2^2$ with $\|\cdot\|_p$ denoting the $\ell_p$ norm~\cite{Bayati2011,Montanari2012}, and $\sigma_{t+1}^2$ is the variance of $\w_{t+1}$.

The rigorous justification of~\eqref{eq:SE_Q} by extending the framework put forth by Bayati and Montanari~\cite{Bayati2011} and Rush and Venkataramanan~\cite{Rush_ISIT2016_arxiv} is left for future work. Instead, we argue that lossy SE~\eqref{eq:SE_Q} asymptotically tracks the evolution of $\sigma_t^2$ in lossy MP-AMP in the limit of $\frac{PD_t}{\sigma_t^2}\rightarrow 0$.
Our argument is comprised of three parts: ({\em i}) $\w_t$ and $\n_t$~\eqref{eq:indpt_noises} are approximately independent in the limit of $\frac{PD_t}{\sigma_t^2}\rightarrow 0$,   ({\em ii})  $\w_t+\n_t$ is approximately independent of $\x$ in the limit of $\frac{PD_t}{\sigma_t^2}\rightarrow 0$, and ({\em iii}) lossy SE~\eqref{eq:SE_Q} holds if ({\em i}) and ({\em ii}) hold.
The first part ($\w_t$ and $\n_t$ are independent) ensures that we can track the variance of $\w_t+\n_t$ with $\sigma_t^2+PD_t$. The second part ($\w_t+\n_t$ is independent of $\x$) ensures that lossy MP-AMP
follows lossy SE~\eqref{eq:SE_Q} as it falls under the general framework discussed in Bayati and Montanari~\cite{Bayati2011} and Rush and Venkataramanan~\cite{Rush_ISIT2016_arxiv}. Hence, the third part of our argument holds.
The first two parts are backed up by extensive numerical evidence in Appendix~\ref{app:verifyIndpt}, where ECSQ~\cite{GershoGray1993,Cover06} is used; ECSQ
approaches $R(D)$ within 0.255 bits in the high rate limit (corresponds to small distortion)~\cite{GershoGray1993}. Furthermore, Appendix~\ref{app:verifyLossySE} provides extensive numerical evidence to show that lossy SE~\eqref{eq:SE_Q} indeed tracks the evolution of the MSE when $\w_t$ and $\n_t$ are independent and $\w_t+\n_t$ and $\x$ are independent.

Although lossy SE~\eqref{eq:SE_Q} requires $\frac{PD_t}{\sigma_t^2}\rightarrow 0$, if scalar quantization is used in
a practical implementation,
then lossy SE approximately holds when $\gamma<\frac{2\sigma_t}{\sqrt{P}}$, where $\gamma$ is the quantization bin size of the scalar quantizer (details in
Appendices~\ref{app:verifyIndpt} and~\ref{app:verifyLossySE}).
Note that the condition $\gamma<\frac{2\sigma_t}{\sqrt{P}}$ is motivated by Widrow and Koll{\'a}r~\cite{widrow2008quantization}. If appropriate VQ schemes~\cite{LBG1980,Gray1984,GershoGray1993} are used, then we might need milder requirements than $\frac{PD_t}{\sigma_t^2}\rightarrow 0$ in the scalar quantizer case, in order for $\w_t$ and $\n_t$ to be independent and for $\w_t+\n_t$ and $\x$ to be independent.

Denote the coding rate used to transmit $Q(\f^p_t)$ at iteration $t$ by $R_t$. The sequence $\mathbf{R}=(R_1,\cdots,R_T)$ is called
the {\em coding rate sequence}, where $T$ is the total number of MP-AMP iterations. Given
$\mathbf{R}$, the distortion $D_t$ can be evaluated with $R(D)$, and the scalar channel noise variance $\sigma_t^2$ can be evaluated with~\eqref{eq:SE_Q}.
Hence, the MSE for $\mathbf{R}$ can be predicted. The MSE at the last iteration is called the {\em final MSE}.

\sectionmark{Background}
\section{Optimal Rates Using Dynamic Programming}\label{sec:DP}
\sectionmark{Optimal Rates Using DP}

In this section, we first define the cost of running MP-AMP. We then use DP to find an optimal coding rate sequence with minimum cost, while achieving a desired EMSE.

\begin{myDef}[Combined cost]\label{def:costFunc}
Define the cost of estimating a signal in an MP system as
\begin{equation}\label{eq:cost}
C^b(\mathbf{R})=b \|\mathbf{R}\|_0+ \|\mathbf{R}\|_1,
\end{equation}
where $\|\mathbf{R}\|_0=T$ is the number of iterations to run, and
$\|\mathbf{R}\|_1$ is
 the aggregate coding rate, denoted also by $R_{agg}$,
\begin{equation}\label{eq:R_agg}
R_{agg}=\| \mathbf{R} \|_1=\sum_{t=1}^T R_t.
\end{equation}
The parameter $b$ is the cost of computation
in one MP-AMP iteration normalized by the cost of transmitting $Q(\f_t^p)$~\eqref{eq:master0} at a coding rate of 1 bit/entry.
Also, the cost at iteration $t$ is
\begin{equation}\label{eq:cost_one_iter}
C^b_t(R_t)  = b\times {\mathbbm{1}}_{R_t \neq 0} + R_t,
\end{equation}
where the indicator function ${\mathbbm{1}}_{\mathcal{A}}$
is 1 if the condition $\mathcal{A}$ is met,
else 0.
Hence, $C^b(\mathbf{R}) = \sum_{t =1}^T C_t^b (R_t)$.
\end{myDef}

In some applications, we may want to obtain a sufficiently small EMSE at minimum cost~\eqref{eq:cost}, where the physical meaning of the cost varies in different problems (cf. Section~\ref{sec:realworld}).
Denote the EMSE at iteration $t$ by $\epsilon_t$. Hence, the {\em final EMSE} at the output of MP-AMP is $\epsilon_T$.

Let us formally state the problem.
Our goal is to obtain a coding rate sequence $\mathbf{R}$ for MP-AMP iterations, which is the solution of the following optimization problem:
\begin{equation}\label{eq:optimizationSetup}
\text{minimize } C^b(\mathbf{R}) \quad \quad
\text{subject to }  \epsilon_T\leq \Delta.
\end{equation}
We now have a definition for the optimal coding rate sequence.

\begin{myDef}[Optimal coding rate sequence]\label{def:optRate}
An optimal coding rate sequence $\mathbf{R}^*$ is a solution of~\eqref{eq:optimizationSetup}.
\end{myDef}

To compute $\mathbf{R}^*$, we derive a dynamic programming (DP)~\cite{bertsekas1995}
scheme, and then prove that it is optimal.

{\bf Dynamic programming scheme:}
Suppose that MP-AMP is at iteration $t$. Define the smallest cost for the $(T-t)$ remaining iterations to achieve the EMSE constraint, $\epsilon_T\leq \Delta$, as $\Phi_{T-t}(\sigma_t^2)$, which is a function of the  scalar channel noise variance at iteration $t$, $\sigma^2_t$~\eqref{eq:indpt_noises}.
Hence, $\Phi_{T-1}(\sigma_1^2)$ is the cost for solving~\eqref{eq:optimizationSetup}, where $\sigma_1^2=\sigma_Z^2+\frac{1}{\kappa}\mathbb{E}[X^2]$ is due to the all-zero initialization of the signal estimate.

DP uses a base case and recursion steps to find $\Phi_{T-1}(\sigma_1^2)$.
In the base case of DP, $T-t=0$, the cost of running MP-AMP is $C^b_T(R_T)=b\times \mathbbm{1}_{R_{T}\neq 0}+R_{T}$~\eqref{eq:cost_one_iter}.
If $\sigma^2_{T}$ is not too large, then there exist some values for $R_T$ that satisfy $\epsilon_T\leq \Delta$; for these $\sigma^2_{T}$ and $R_T$, we have $\Phi_{0}(\sigma^2_{T})=\min_{R_T} C^b_T(R_T)$.
If $\sigma^2_{T}$ is too large, even lossless transmission of $\f_{T}^p$ during the single remaining MP-AMP iteration~\eqref{eq:SE_Q} does not yield an EMSE that satisfies the constraint, $\epsilon_T\leq \Delta$, and
we assign $\Phi_{0}(\sigma^2_{T})=\infty$ for such $\sigma_{T}^2$.

Next, in the recursion steps of DP,
we iterate back in time by decreasing $t$ (equivalently, increasing $T-t$),
\begin{equation}\label{eq:DPrecursion}
 \Phi_{T-t}(\sigma^2_t)\! =\! \min_{\widehat{R}}\left\{ C_t^b(\widehat{R})+ \Phi_{T-(t+1)}(\sigma^2_{t+1}(\widehat{R}))\right\},
\end{equation}
where $\widehat{R}$ is the coding rate used in the current
MP-AMP iteration~$t$,
the equivalent scalar channel noise variance at the fusion center is
$\sigma_t^2$~\eqref{eq:indpt_noises},
and $\sigma^2_{t+1}(\widehat{R})$, which is obtained from (\ref{eq:SE_Q}),
is the variance of the scalar channel noise~\eqref{eq:indpt_noises} in the
next iteration after transmitting $\f_t^p$ at rate $\widehat{R}$.
The terms on the right hand side are the current cost of MP-AMP~\eqref{eq:cost_one_iter}
(including computational and communication costs) and the minimum combined cost
in all later iterations, $t+1,\cdots,T$.

The coding rates $\widehat{R}$ that yield the smallest cost $\Phi_{T-t}(\sigma_t^2)$ for different $t$ and $\sigma^2_t$ are stored in a table $\mathcal{R}(t,\sigma^2_t)$.
After DP finishes, we obtain the coding rate for the first MP-AMP iteration as
$R_1=\mathcal{R}(1,\sigma_Z^2+\frac{1}{\kappa}\mathbb{E}[X^2])$.
Using $R_1$, we calculate $\sigma_t^2$ from~\eqref{eq:SE_Q} for $t=2$ {and
find $R_2=\mathcal{R}(2,\sigma^2_2)$. Iterating from $t=1$ to $T$, we obtain $\mathbf{R}=(R_1,\cdots,R_T)$.

To be computationally tractable, the proposed DP scheme should operate in
discretized search spaces for $\sigma^2_{\{\cdot\}}$ and $R_{\{\cdot\}}$.
Details about the resolutions of $\sigma^2_{\{\cdot\}}$ and $R_{\{\cdot\}}$ appear in
Appendix~\ref{app:Integrity}.

In the following, we state that our DP scheme yields the optimal solution. The proof appears in Appendix~\ref{app:proofDPoptimal}.

\begin{myLemma}\label{lemma:DPoptimal}
The dynamic programming formulation in~\eqref{eq:DPrecursion} yields an optimal coding rate sequence $\mathbf{R}^*$, which is a solution of~\eqref{eq:optimizationSetup} for the discretized search spaces
of $R_t$ and $\sigma_t^2,\ \forall t$.
\end{myLemma}

Lemma~\ref{lemma:DPoptimal} focuses on the optimality of our DP scheme in
discretized search spaces for $R_t$ and $\sigma_t^2$. It can be shown that we can achieve a desired accuracy level in $\mathbf{R}^*$ by adjusting the resolutions of the discretized search spaces for $R_t$ and $\sigma_t^2$.
Suppose that the discretized search spaces for $\sigma^2_{\{\cdot\}}$ and $R_{\{\cdot\}}$ have
$K_1$ and $K_2$ different values, respectively. Then, the computational complexity of our DP scheme is $O(TK_1K_2)$.

\begin{figure}[t]
\begin{center}
\includegraphics[width=8cm]{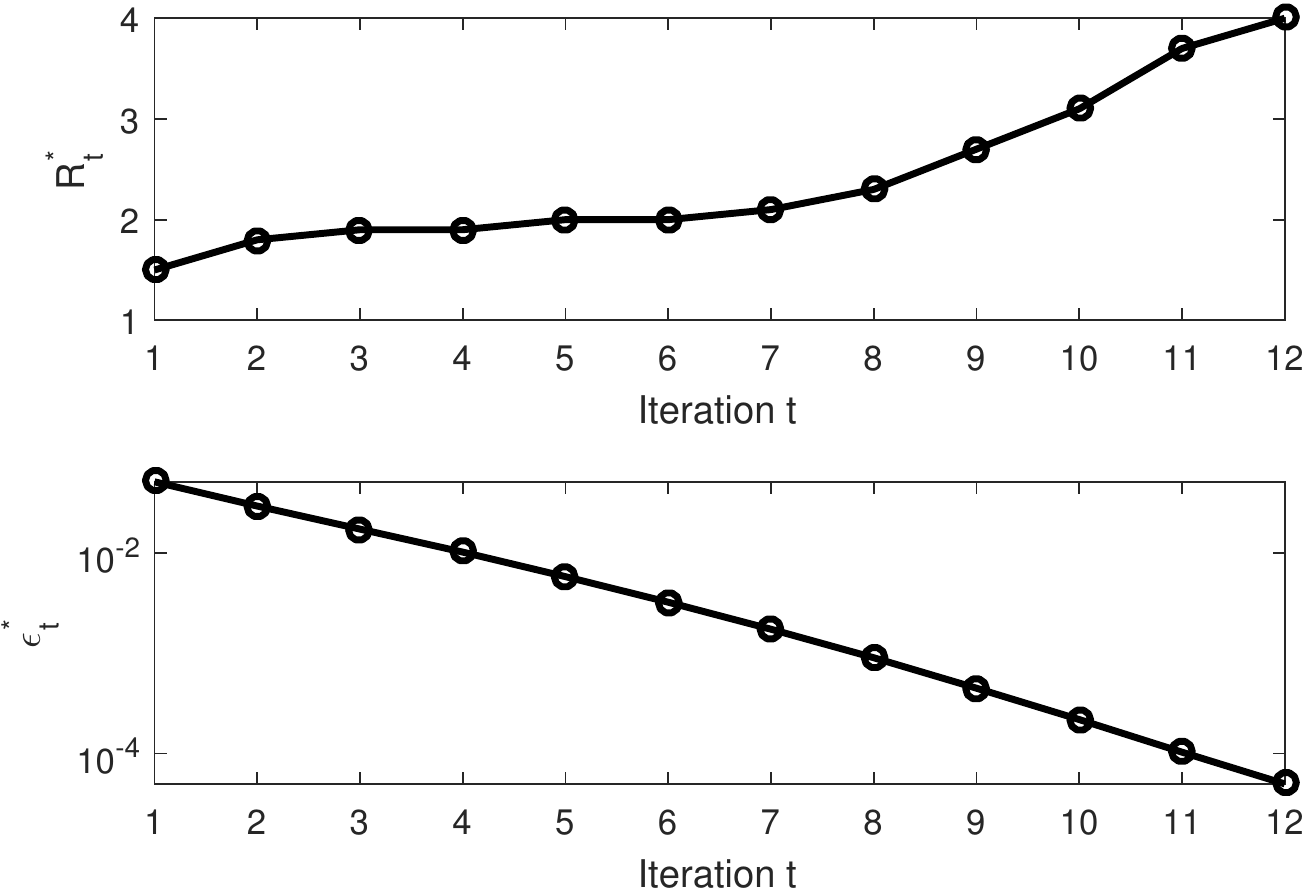}
\end{center}
\caption{The optimal coding rate sequence $\mathbf{R}^*$ (top panel) and optimal EMSE $\epsilon_t^*$ (bottom) given by DP are  shown as functions of $t$. (Bernoulli-Gaussian signal~\eqref{eq:BG} with $\rho=0.1$, $\kappa=0.4$, $P=100$, $\sigma_Z^2=\frac{1}{400}$, and $b=2$.)}
\label{fig:RtAndEMSE}
\end{figure}

{\bf Optimal coding rate sequence given by DP:}
Consider estimating a {\em Bernoulli-Gaussian} signal,
\begin{equation}
X=X_BX_G,\label{eq:BG}
\end{equation}
where $X_B\sim \text{Ber}(\rho)$ is a Bernoulli random variable,
$\rho$ is called the {\em sparsity rate} of the signal,
and $X_G\sim {\cal N}(0,1)$;
here we use $\rho=0.1$. Note that the results in this chapter apply to priors, $f_X(x)$,
other than~\eqref{eq:BG}.

We run our DP scheme on a problem with relatively small
desired EMSE,  $\Delta=5\times10^{-5}$, in the last iteration $T$.
The signal is measured in an MP platform with $P=100$ processor nodes according to~\eqref{eq:one-node-meas}.
The measurement rate is $\kappa=\frac{M}{N}=0.4$, and the noise variance is $\sigma_Z^2=\frac{1}{400}$. The parameter $b=2$~\eqref{eq:cost}.
We use ECSQ~\cite{GershoGray1993,Cover06} as the quantizer in each processor node, and use  the corresponding relation between the rate $R_t$ and distortion $D_t$ of ECSQ in our DP scheme. Note that we require the quantization bin size to be smaller than $\frac{2\sigma_t}{\sqrt{P}}$, according to Section~\ref{sec:MP-CS_for_MP-AMP}.
Figure~\ref{fig:RtAndEMSE} illustrates the optimal coding rate sequence $\mathbf{R}^*$ and optimal EMSE $\epsilon_t^*$ given by DP as functions of the iteration number $t$.

It is readily seen that after the first 5--6 iterations the coding rate seems near-linear.
The next section proves that any optimal coding rate sequence $\mathbf{R}^*$
is approximately linear in the limit of EMSE$\rightarrow 0$.
However, our proof involves the large $t$ limit, and does not provide insights for small $t$.
We ran DP for various configurations.
Examining all $\mathbf{R}^*$ from our DP results, we notice that the coding rate is
monotone non-decreasing, i.e., $R^*_1\leq R^*_2\leq \cdots\leq R^*_T$. This seems intuitive, because in
early iterations of (MP-)AMP, the scalar channel noise $\w_t$ is large, which does not require
transmitting $\f_t^p$ (cf.~\eqref{eq:slave2}) at high fidelity. Hence, a
low rate $R^*_t$ suffices. As the iterations proceed, the scalar channel noise
$\w_t$ in~\eqref{eq:indpt_noises} decreases, and the large quantization
error $\n_t$ would be unfavorable for the final MSE.
Hence, higher rates are needed in later iterations.

\sectionmark{Optimal Rates Using DP}
\section{Properties of Optimal Coding Rate Sequences}\label{sec:linRateTh}
\sectionmark{Properties of Optimal Rates}

\subsection{Intuition}\label{sec:intuition}

We start this section by providing some brief intuitions about why optimal coding rate sequences are approximately linear when the EMSE is small.

Consider a case where we aim to reach a low $\text{EMSE}$. Montanari~\cite{Montanari2012}
provided a geometric interpretation of the relation between the MSE performance of AMP at
iteration $t$ and
the denoiser $\eta_t(\cdot)$ being used.\footnote{We will also provide such an interpretation in Section~\ref{sec:geoInterp}.}
In the limit of small $\text{EMSE}$, the $\text{EMSE}$ decreases by a nearly-constant multiplicative factor
per AMP iteration, yielding a geometric decay of the EMSE.
In MP-AMP, in addition to the equivalent scalar channel noise $\w_t$, we have additive quantization error $\n_t$~\eqref{eq:indpt_noises}.
In order for the $\text{EMSE}$ in an MP-AMP system to decay geometrically, the distortion $D_t$ must decay at least as quickly. To obtain this geometric decay in $D_t$,
recall that in the high rate limit, the distortion-rate function typically takes the form $D(R)\approx C_1 2^{-2R}$~\cite{GrayNeuhoff1998} for some positive constant $C_1$.
We propose for $R_t$ to have the form,
$R_t\approx C_2+C_3 t$,
where $C_2$ and $C_3$ are constants. In the remainder of this section, we first discuss the geometric interpretation of AMP state evolution, followed by our results about the linearity of optimal coding rate sequences. The detailed proofs appear in the appendices.

\begin{figure*}[t]
\centering
  \subfloat[]{
    \label{fig:losslessSE} 
    \includegraphics[height=0.23\textwidth]{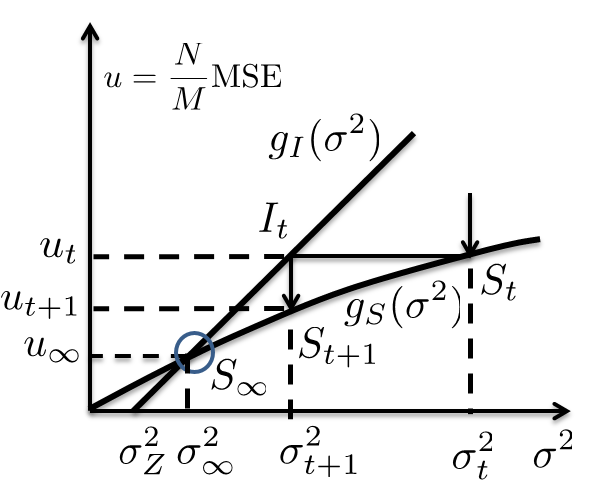}}
    \subfloat[]{
    \label{fig:zoomIn} 
    \includegraphics[height=0.23\textwidth]{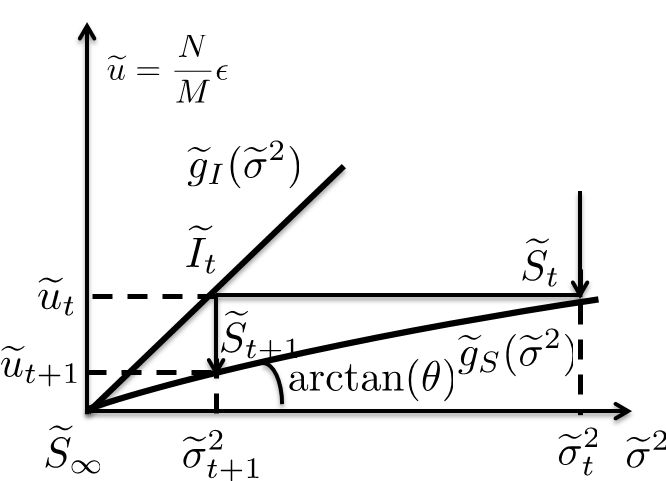}}
  \subfloat[]{
    \label{fig:lossySE} 
    \includegraphics[height=0.23\textwidth]{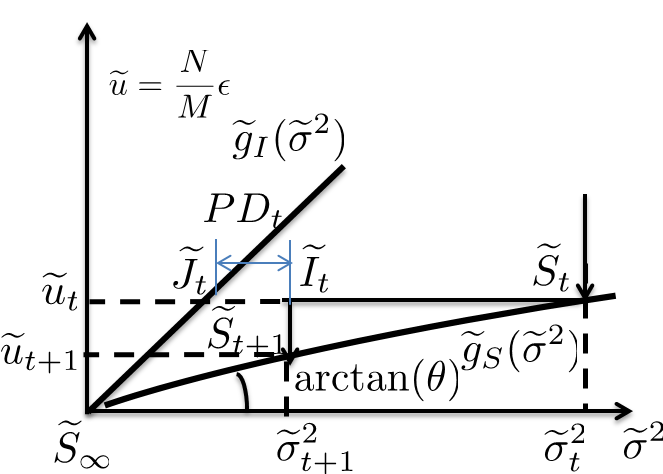}}
\caption{Geometric interpretation of SE. In all panels, the thick solid curves correspond to $g_I(\cdot)$ and $g_S(\cdot)$, and their offset versions $\widetilde{g}_I(\cdot)$ and $\widetilde{g}_S(\cdot)$. The solid lines with arrows correspond to the SE of AMP.   Dashed lines without arrows are auxiliary lines. Panel (a): Illustration of centralized SE. Panel (b): Zooming in to the small region just above point $S_\infty$. Panel (c): Illustration of lossy SE.}\label{fig:AMP_SE_geo}
\end{figure*}

\subsection{Geometric interpretation of AMP state evolution}\label{sec:geoInterp}
{\bf Centralized SE:}
The equivalent scalar channel of AMP is given by~\eqref{eq:equivalent_scalar_channel}.
We rewrite the centralized AMP SE~\eqref{eq:ori_SE} as follows~\cite{DMM2009,Bayati2011,Rush_ISIT2016_arxiv},
\begin{equation}\label{eq:SE}
\underbrace{\sigma_{t+1}^2-\sigma_Z^2}_{g_I(\sigma_{t+1}^2)}=\underbrace{\frac{N}{M}\text{MSE}_{\eta_t}(\sigma_t^2)}_{g_S(\sigma_t^2)},
\end{equation}
where $\text{MSE}_{\eta_t}(\sigma_t^2)$ denotes the MSE after denoising $\f_t$~\eqref{eq:equivalent_scalar_channel}
using $\eta_t(\cdot)$.
The functions $g_I(\cdot)$ and $g_S(\cdot)$ are illustrated in Figure~\ref{fig:losslessSE}
with solid curves; the meanings of $I$ and $S$ will become clear below.
We see that $g_I(\sigma_t^2)$ is an affine function with unit slope, whereas $g_S(\sigma_t^2)$ is generally a non-linear function of $\sigma_t^2$ (see Figure~\ref{fig:losslessSE}). The lines with arrows illustrate the state evolution (SE). Details appear below.

In Figure~\ref{fig:losslessSE}, we present a
geometric interpretation of SE. The horizontal axis is
the scalar channel noise variance $\sigma^2$ and the vertical axis represents the scaled MSE, $u=\frac{N}{M}\text{MSE}$.
Let $S_t = (\sigma^2_t, u_t)$ be the {\em state} point that is reached by SE in
iteration $t$.
We follow the SE trajectory $S_t \to I_t \to S_{t+1} \to \cdots$
in Figure~\ref{fig:losslessSE},  
where $I_t = (\sigma^2_{t+1}, u_t)$ represents the {\em intermediate} point in the transition between states $S_t$ and $S_{t+1}$ corresponding to iterations $t$ and $t+1$, respectively. 
Observe that the points $S_t$ and $I_t$ have the same ordinate ($u_t$), while $S_{t+1}$ and $I_t$ have the same abscissa ($\sigma_{t+1}^2$), which are related as
$\sigma^2_{t+1} = g^{-1}_I(u_t)$
and
$u_{t+1} = g_S(\sigma^2_{t+1})$.
As $t$ grows, $\sigma^2_t$  converges to $\sigma^2_\infty$, which is the abscissa of the point $S_\infty$. The ordinate of point $S_\infty$ is $u_\infty=\frac{N}{M}$MSE$_\infty$, where $\text{MSE}_\infty = \text{MMSE}$. If we stop the algorithm at iteration $T$, or equivalently at point $S_T = (\sigma^2_T, u_T)$,
the corresponding MSE, MSE$_T$, has an EMSE of $\epsilon_T=\text{MSE}_T-\text{MMSE}$.

In Figure~\ref{fig:zoomIn}, we zoom into the neighborhood of point $S_\infty$. To make the presentation more concise, we vertically offset $g_I(\cdot)$ and $g_S(\cdot)$ by $\frac{N}{M}$MMSE and horizontally offset them by $\sigma_\infty^2$; we call the resulting functions $\widetilde{g}_I(\cdot)$ and $\widetilde{g}_S(\cdot)$, respectively. Hence, the vertical axis in Figure~\ref{fig:zoomIn} represents the scaled EMSE, $\widetilde{u}=\frac{N}{M}\text{EMSE}=\frac{N}{M}\epsilon$, and we have $\widetilde{g}_I(\widetilde{\sigma}^2_t)=g_I(\widetilde{\sigma}^2_t+\sigma_\infty^2)-\frac{N}{M}$MMSE and $\widetilde{g}_S(\widetilde{\sigma}^2_t)=g_S(\widetilde{\sigma}^2_t+\sigma_\infty^2)-\frac{N}{M}$MMSE.
Observe that $\widetilde{g}_I(0) = \widetilde{g}_S(0) = 0$.
Additionally, the slope of $\widetilde{g}_I(\widetilde{\sigma}^2_t)$ is $\widetilde{g}_I'(\widetilde{\sigma}^2_t)=1$, where $\widetilde{g}_I'(\cdot)$ is the first-order derivative of $\widetilde{g}_I(\cdot)$ w.r.t. $\widetilde{\sigma}^2_t$ (Figure~\ref{fig:zoomIn}).
Because the MSE function for the MMSE-achieving denoiser is continuous and differentiable twice~\cite{WuVerdu2011}, we can invoke Taylor's theorem to express
\begin{equation}
\widetilde{g}_S(\widetilde{\sigma}^2_t)
=\widetilde{g}_S'(0)\widetilde{\sigma}^2_t+\frac{1}{2}\widetilde{g}_S''(\zeta_t)\widetilde{\sigma}^4_t,
\label{eq:Taylor}
\end{equation}
where $\zeta_t\in (0,\widetilde{\sigma}^2_t),$ and $\widetilde{g}_S'(\widetilde{\sigma}^2_t)$ and $\widetilde{g}_S''(\widetilde{\sigma}^2_t)$ are the first- and second-order derivatives of $\widetilde{g}_S(\cdot)$ w.r.t.  $\widetilde{\sigma}^2_t$, respectively.
Due to continuity and differentiability of the denoising function, $\widetilde{g}_S(\cdot)$ is invertible in a neighborhood around $0$, and its inverse is denoted by $\widetilde{g}^{-1}_S(\cdot).$
Invoking Taylor's theorem,
\begin{equation}
\widetilde{g}_S^{-1}(\widetilde{u}_t)=(\widetilde{g}_S^{-1})'(0) \widetilde{u}_t+\frac{1}{2}(\widetilde{g}_S^{-1})''(\zeta_t)\widetilde{u}_t^2,
\label{eq:Taylor_inverse}
\end{equation}
where $\zeta_t\in (0,\widetilde{u}_t)$, and
$(\widetilde{g}^{-1}_S)'(\widetilde{u}_t)$ and
$(\widetilde{g}^{-1}_S)''(\widetilde{u}_t)$ are the first- and second-order derivatives of $\widetilde{g}^{-1}_S(\cdot)$ w.r.t.  $\widetilde{u}_t$, respectively. When $t\rightarrow \infty$, $\widetilde{\sigma}^2_t \to 0$ and $\widetilde{u}_t \to 0$,
and the higher-order terms become
$\frac{1}{2}\widetilde{g}_S''(\xi_t)\widetilde{\sigma}_t^4=O(\widetilde{\sigma}_t^4)$ and
$\frac{1}{2}(\widetilde{g}_S^{-1})''(\zeta_t)\widetilde{u}_t^2=O(\widetilde{u}_t^2)$. 
In other words, both $\widetilde{g}_S(\widetilde{\sigma}^2_t)$ and $\widetilde{g}^{-1}_S(\widetilde{u}_t)$ become approximately linear functions, as shown in Figure~\ref{fig:zoomIn}. We further denote the slope of $\widetilde{g}_S(0)$ by $\theta$, i.e.,
\begin{equation}
\theta=\widetilde{g}_S'(0) = \frac{1}{(\widetilde{g}_S^{-1})'(0)}.
\label{eq:theta}
\end{equation}

To calculate the slope $\theta$, we first calculate the scalar channel noise variance for point $S_\infty$, $\sigma_\infty^2$, by using replica analysis~\cite{ZhuBaronCISS2013,Krzakala2012probabilistic},\footnote{The outcome of replica analysis~\cite{ZhuBaronCISS2013,Krzakala2012probabilistic} is close to simulating SE~\eqref{eq:SE} with a large number of iterations.} and obtain
$\theta=g_S'(\sigma_\infty^2)=\widetilde{g}_S'(0)$.
Moreover, the slope of $\widetilde{g}_S(0)$ satisfies $\theta=\widetilde{g}_S'(0)\in(0,1)$; otherwise, the curves $\widetilde{g}_I(\cdot)$ and $\widetilde{g}_S(\cdot)$ would not intersect at point $S_\infty$.

{\bf Lossy SE:}
Considering lossy SE~\eqref{eq:SE_Q}, we have
\begin{equation}\label{eq:lossySE}
\underbrace{\sigma_{t+1}^2-\sigma_Z^2}_{g_I(\sigma_{t+1}^2)}=\underbrace{\frac{N}{M}\text{MSE}_{\eta_t}(\sigma_t^2+PD_t)}_{g_S(\sigma_t^2+PD_t)},
\end{equation}
where $P$ is the number of processor  nodes in an MP network, and $D_t$ is the expected
distortion incurred by each node at iteration $t$.
Note that lossy SE has not been rigorously proved in the literature, although we argued in Section~\ref{sec:MP-CS_for_MP-AMP} that it tracks the evolution of the equivalent scalar channel noise variance $\sigma_t^2$ when $D_t\ll \frac{1}{P}\sigma_t^2$.

We notice the additional term $PD_t$, which corresponds to the distortion {\em at the fusion center}.
Because the $P$ nodes transmit their signals $\f_t^p$
with distortion $D_t$, and their messages are independent,
the fusion center's signal has distortion $PD_t$.
The lines with arrows in Figure~\ref{fig:lossySE} illustrate the lossy SE after vertically offsetting $g_I(\cdot)$ and $g_S(\cdot)$ by $\frac{N}{M}$MMSE and horizontally offsetting $g_I(\cdot)$ and $g_S(\cdot)$ by $\sigma_\infty^2$.
After arriving at point $\widetilde{S}_t$, we move horizontally to
$\widetilde{J}_t$, and obtain the ordinate of $\widetilde{I}_t$, $\widetilde{u}_t$, from $\widetilde{g}_S(\widetilde{\sigma}_t^2+PD_t)=\widetilde{u}_t$. Geometrically,
SE is dragged to the right by distance $PD_t$ from point $\widetilde{J}_t$ to
$\widetilde{I}_t$, and then SE descends from $\widetilde{I}_t$
to $\widetilde{S}_{t+1}$.

\subsection{Asymptotic linearity of the optimal coding rate sequence}\label{sec:linearRateTheorem_subsection}

Recall from~\eqref{eq:Taylor} that
 $\lim_{t\rightarrow\infty}\widetilde{\sigma}^2_t=0$. 
Hence, as $t$ grows, $f_{t,i}$~\eqref{eq:equivalent_scalar_channel} converges in distribution to
$x_i +\mathcal{N}(0, \sigma^2_{\infty})$.
Therefore, the RD function converges to some fixed function as
$t$ grows. For large coding rate $R$, this function} has the form
\begin{equation}\label{eq:DR}
R_t=\frac{1}{2}\log_2\l(\frac{C_1}{D_t}\r) (1+ o_t(1)),
\end{equation}
for some constant $C_1$ that does not depend on $t$~\cite{GrayNeuhoff1998}.
Note that the assumption of $\widetilde{\sigma}^2_t$ being small implicitly requires the
coding rate used in the corresponding iteration to be large.

For an optimal coding rate sequence $\mathbf{R}^*$, we call the distortion $D_t^*$, derived from~\eqref{eq:DR}, incurred by the optimal coding rate $R_t^*$ at a certain iteration $t$ the {\em optimal distortion}. Correspondingly,
we call the EMSE achieved by MP-AMP with $\mathbf{R}^*$,
denoted by $\epsilon_t^*$, the {\em optimal EMSE} at iteration $t$.
In the following, we state our main results on the optimal coding rate, the optimal distortion, and the optimal EMSE.

\begin{myTheorem}[Linearity of the optimal coding rate sequence]\label{th:optRateLinear}
Supposing that lossy SE~\eqref{eq:lossySE} holds, we have 
\begin{equation}
\lim_{t \to \infty} \frac{D_{t+1}^*}{D_t^*} = \theta,
\label{eq:theorem1-1}
\end{equation}
where $\theta$ is defined in~\eqref{eq:theta}. Furthermore,
\begin{equation}\label{eq:theorem1}
\lim_{t\rightarrow\infty} \l(R^*_{t+1} - R^*_{t }\r)= \frac{1}{2}\log_2 \l(\frac{1}{\theta}\r).
\end{equation}
\end{myTheorem}

Theorem~\ref{th:optRateLinear} is
proved in Appendix~\ref{app:optRateLinear}.

\begin{myRemark}
Define the additive growth rate of an optimal coding rate sequence $\mathbf{R}^*$ at iteration $t$ as $R_{t+1}^*-R_{t}^*$.
Theorem~\ref{th:optRateLinear} not only shows that any optimal coding rate sequence grows approximately linearly in the low EMSE limit, but also provides a way to calculate its additive growth rate in the low EMSE limit. Hence, if the goal is to achieve a low EMSE, practitioners could simply use a coding rate sequence that has a fixed coding rate in the first few iterations and then increases linearly with additive growth rate $\frac{1}{2}\log_2\l(\frac{1}{\theta}\r)$.
\end{myRemark}

The following theorem provides ({\em i}) the relation between the optimal distortion $D_{t+1}^*$ and the optimal EMSE $\epsilon_t^*$ in the large $t$ limit, and ({\em ii}) the convergence rate of the optimal EMSE $\epsilon_t^*$.
\begin{myTheorem}\label{th:convergence}
Assuming that lossy SE~\eqref{eq:lossySE} holds, we have
\begin{equation}\label{eq:theorem2_2}
\lim_{t\rightarrow\infty} \frac{D_t^*}{\epsilon_{t}^*} =0.
\end{equation}
Furthermore, the convergence rate of the optimal EMSE is
\begin{equation}\label{eq:theorem2_1}
\lim_{t\rightarrow\infty} \frac{\epsilon_{t+1}^*}{\epsilon_{t}^*}=\theta.
\end{equation}
\end{myTheorem}

Theorem~\ref{th:convergence} is proved in Appendix~\ref{app:convergence}. Note that $\lim_{t\rightarrow\infty} \frac{D_t^*}{\epsilon_{t}^*} =0$ meets the
requirement $\frac{PD_t}{\sigma_t^2}\rightarrow 0$ discussed
in Section~\ref{sec:MP-CS_for_MP-AMP}.
Extending Theorems~\ref{th:optRateLinear} and~\ref{th:convergence}, we have the following result.

\begin{myCoro}
Assuming that lossy SE~\eqref{eq:SE_Q} holds, the combined computation and communication cost~\eqref{eq:cost} scales as $O(\log^2(1/\Delta))$, $\forall b>0$, where $\Delta$ is the desired EMSE.
\end{myCoro}
\begin{proof}
Given Theorem~\ref{th:convergence},
we obtain that the optimal EMSE, $\epsilon_t^*$,
indeed decreases geometrically in the large $t$ limit (as a reminder, we provided such intuition in Section~\ref{sec:intuition}). Considering~\eqref{eq:R_agg} and Theorem~\ref{th:optRateLinear}, the total computation and communication cost~\eqref{eq:cost} for running $T$ iterations is
$C^b(\mathbf{R^*})=O(T^2)=O(\log^2(1/\epsilon_T^*)) = O(\log^2(1/\Delta))$.
\end{proof}

\begin{myRemark}
The key to the proofs of Theorems~\ref{th:optRateLinear} and~\ref{th:convergence} is lossy SE~\eqref{eq:lossySE}. We expect that the linearity of the optimal coding rate sequence could be extended to other iterative distributed algorithms provided that ({\em i}) they have formulations similar to lossy SE~\eqref{eq:lossySE} that track their estimation errors
and ({\em ii}) their estimation errors converge geometrically.
Moreover, formulations that track the estimation error in such
algorithms might require less restrictive constraints than AMP. For example,
consensus averaging~\cite{Frasca2008,Thanou2013}
only requires i.i.d. entries in the vector that each node in the network averages.
\end{myRemark}

\begin{figure}
  \centering
  \includegraphics[width=8cm]{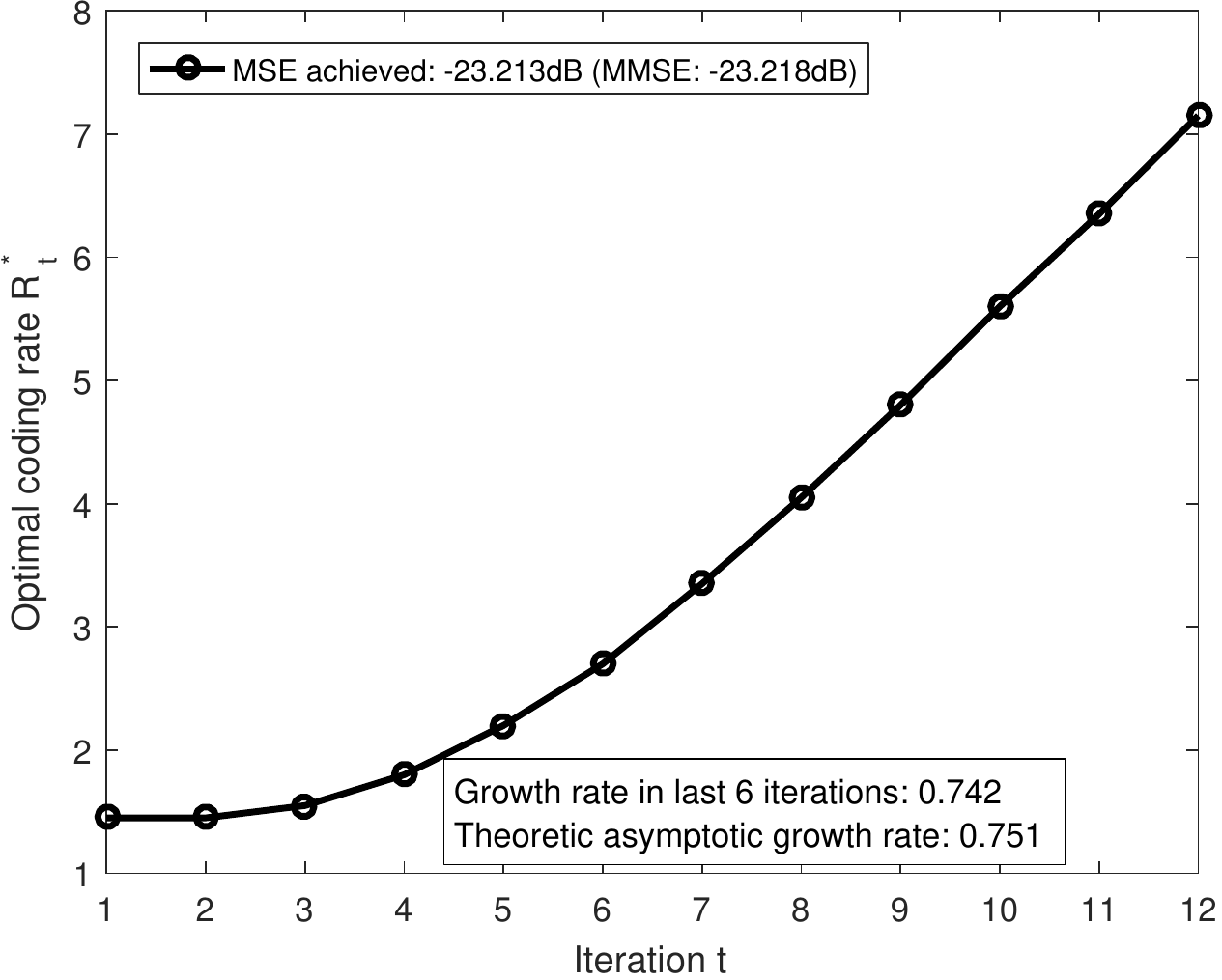}
  \caption{Comparison of the additive growth rate of the optimal coding rate sequence given by DP at low EMSE and the asymptotic  additive growth rate $\frac{1}{2}\log_2\l(\frac{1}{\theta}\r)$. (Bernoulli-Gaussian signal~\eqref{eq:BG} with $\rho=0.2,\ \kappa=1,\ P=100, \sigma_Z^2=0.01,\ b=0.782$.)}\label{fig:asympSlope}
\end{figure}

\subsection{Comparison of DP results to Theorem~\ref{th:optRateLinear}}

We run DP (cf. Section~\ref{sec:DP}) to find an optimal coding rate sequence $\mathbf{R}^*$ for the setting of $P=100$ nodes, a Bernoulli-Gaussian signal~\eqref{eq:BG} with sparsity rate $\rho=0.2$, measurement rate $\kappa=1$, noise variance $\sigma_Z^2=0.01$, and  parameter $b=0.782$. The goal is to achieve a desired EMSE of 0.005 dB, i.e.,
$10\log_{10}\l(1 + \frac{\Delta}{\text{MMSE}}\r)=0.005$.
We use ECSQ~\cite{GershoGray1993,Cover06} as the quantizer in each processor node and
use the corresponding relation between the
rate $R_t$ and distortion $D_t$ of ECSQ in the DP scheme. Note that we require the
quantization bin size $\gamma$ to be smaller than
$\frac{2\sigma_t}{\sqrt{P}}$, according to Section~\ref{sec:MP-CS_for_MP-AMP}. We know that ECSQ achieves a coding rate
within an additive constant of the RD function $R(D)$~\cite{GershoGray1993}. Therefore, the additive
growth rate of the optimal coding rate sequence
obtained for ECSQ will be the same as the additive growth rate if the RD relation is modeled by $R(D)$~\cite{Cover06,Berger71,GershoGray1993,WeidmannVetterli2012}.

The resulting optimal coding rate sequence is plotted in Figure~\ref{fig:asympSlope}. The additive growth rate of the last six iterations is $\frac{1}{6}(R_{12}^*-R_{6}^*)=0.742$, and the asymptotic additive growth rate according to Theorem~\ref{th:optRateLinear} is $\frac{1}{2}\log_2\l(\frac{1}{\theta}\r)\approx 0.751$.
Note that we use $\Delta R_t=0.05$ in the discretized search space for $R_t$. Hence, the discrepancy of 0.009 between the additive growth rate from the simulation and the asymptotic additive growth rate is
within our numerical precision.
In conclusion, our numerical result matches the theoretical prediction of Theorem~\ref{th:optRateLinear}.

\section{Achievable Performance Region}\label{sec:Pareto}

Following the discussion of Section~\ref{sec:setting_MP_Chap}, we can see that the lossy compression of
${\bf f}_t^p, \forall p \in \{1,\cdots,P\}$, can reduce communication costs. On the other hand,
the greater the savings in the coding rate sequence $\mathbf{R}$, the worse the final MSE is expected to be.
If a certain level of final MSE is desired
despite a small coding rate budget, then more iterations $T$ will be needed.
As mentioned above, there is a trade-off between $T$, $R_{agg}$, and the final
MSE, i.e., $\text{MMSE} + \Delta$, and there is no
solution that minimizes them simultaneously.
To deal with such trade-offs, which implicitly correspond to sweeping
$b$ in~\eqref{eq:cost} in a multi-objective optimization (MOP) problem, it is customary to think about
{\em Pareto optimality}~\cite{DasDennisPareto1998}.

\subsection{Properties of achievable region}\label{sec:property}

For notational convenience, denote
the set of all MSE values achieved by the pair $(T,R_{agg})$
for some parameter $b$~\eqref{eq:cost} by ${\cal E}(T, R_{agg})$.
Within $(T,R_{agg})$, let the smallest MSE be $\text{MSE}^*(T,R_{agg})$.
We now define the achievable set $\cal C$,
$$
{\cal C} := \{(T,R_{agg}, \text{MSE}) \in \mathbb{R}_{\geq 0}^3: \text{MSE} \in {\cal E}(T, R_{agg})\},
$$
where $\mathbb{R}_{\geq 0}$ is the set of non-negative real numbers.
That is, ${\cal C}$ contains all tuples $(T,R_{agg}, \text{MSE})$ for which
some instantiation of MP-AMP estimates the signal at the
desired MSE level using $T$ iterations and aggregate coding rate $R_{agg}$.

\begin{myDef}\label{def:Pareto}
{\em  The point $\mathcal{X}_1\in\mathcal{C}$ is said to dominate another point $\mathcal{X}_2\in\mathcal{C}$, denoted by $\mathcal{X}_1\prec \mathcal{X}_2$, if $T_1\leq T_2$, $R_{agg_1}\leq R_{agg_2}$, and $\text{MSE}_1\leq \text{MSE}_2$. A point $\mathcal{X}^*\in \mathcal{C}$ is
Pareto optimal if there does not exist $\mathcal{X}\in \mathcal{C}$ satisfying $\mathcal{X}\prec \mathcal{X}^*$.
Furthermore, let $\mathcal{P}$ denote the set of all Pareto optimal points,}
\begin{equation}\label{eq:setP}
\mathcal{P} := \{\mathcal{X}\in \mathcal{C}: \text{$\mathcal{X}$ is Pareto optimal}\}.
\end{equation}
\end{myDef}

In words, the tuple $(T,R_{agg},\text{MSE})$ is Pareto optimal if no other tuple
$(\widehat{T},\widehat{R}_{agg},\widehat{\text{MSE}})$ exists such that
$\widehat{T}\leq T$, $\widehat{R}_{agg}\leq R_{agg}$, and $\widehat{\text{MSE}}\leq \text{MSE}$.
Thus, the Pareto optimal tuples belong to the boundary of $\cal C$.

We extend the definition of the number of iterations $T$ to a probabilistic one. To do so, suppose
that the number of iterations is drawn from a probability distribution $\pi$ over $\mathbb{N}$, such that $\sum_{i=1}^{\infty} \pi_i = 1$. Of course, this definition contains a deterministic $T = j$ as a special case with $\pi_j = 1$ and $\pi_i =0$ for all $i \neq j$.
Armed with this definition of Pareto optimality and the probabilistic definition of the number of iterations, we have the following lemma.

\begin{myLemma}\label{th:convex1}
{\it For a fixed noise variance $\sigma^2_Z$, measurement rate $\kappa$, and $P$ processor nodes in MP-AMP, the achievable set $\mathcal{C}$ is a convex set.}
\end{myLemma}

\begin{proof}
We need to show that for any $(T^{(1)},$ $R^{(1)}_{agg},\text{MSE}^{(1)})$, $(T^{(2)},R^{(2)}_{agg},\text{MSE}^{(2)})$ $\in \mathcal{C}$ and any $0<\lambda<1$,
\begin{equation}\label{eq:suff1}
(\lambda T^{(1)} + (1-\lambda)T^{(2)},\lambda R^{(1)}_{agg} + (1-\lambda)R^{(2)}_{agg},
 \lambda \text{MSE}^{(1)}+(1-\lambda)\text{MSE}^{(2)}) \in \mathcal{C}.
\end{equation}
This result is shown using time-sharing arguments
(see Cover and Thomas~\cite{Cover06}).
Assume that $(T^{(1)},R^{(1)}_{agg},\text{MSE}^{(1)})$, $(T^{(2)},R^{(2)}_{agg},\text{MSE}^{(2)}) \in \mathcal{C}$ are achieved by probability distributions $\pi^{(1)}$ and $\pi^{(2)}$, respectively.
Let us select all parameters of the first tuple with probability
$\lambda$ and those of the second with probability $(1-\lambda)$.
Hence, we have
 $\pi = \lambda \pi^{(1)} + (1-\lambda) \pi^{(2)}$.
 Due to the linearity of expectation,
 $T = \lambda T^{(1)} + (1-\lambda) T^{(2)}$
 and $\text{MSE} = \lambda \text{MSE}^{(1)} + (1-\lambda) \text{MSE}^{(2)}$.
 Again, due to the linearity  of expectation,
 $R_{agg} = \lambda R^{(1)}_{agg} + (1-\lambda)R^{(2)}_{agg}$,
 implying that~\eqref{eq:suff1} is satisfied, and the proof is complete.
\end{proof}

\begin{myDef}\label{def:funcs}
{\em Let the function $R^*(T,\text{MSE}):\mathbb{R}_{\geq 0}^2\rightarrow \mathbb{R}_{\geq 0}$
be the Pareto optimal rate function, which is implicitly described as
$R^*(T,\text{MSE})=R_{agg}^* \Leftrightarrow (T,R_{agg}^*,\text{MSE}) \in \mathcal{P}$.
We further define implicit functions $T^*(R_{agg}, \text{MSE})$ and $\text{MSE}^*(T, R_{agg})$ in a similar way.}
\end{myDef}

\begin{myCoro}\label{coro:convexArguments}
The functions $R^*(T,\text{MSE})$, $T^*(R_{agg}, \text{MSE})$, and $\text{MSE}^*(T,R_{agg})$ are convex in their arguments.
\end{myCoro}

Note that our proof for the convexity of the set $\mathcal{C}$ might be extended to other iterative distributed learning algorithms
that transmit lossily compressed messages.

\begin{figure*}[t]
\centering
  \subfloat[]{
    \label{fig:rateVSiter} 
    \includegraphics[width=0.32\textwidth]{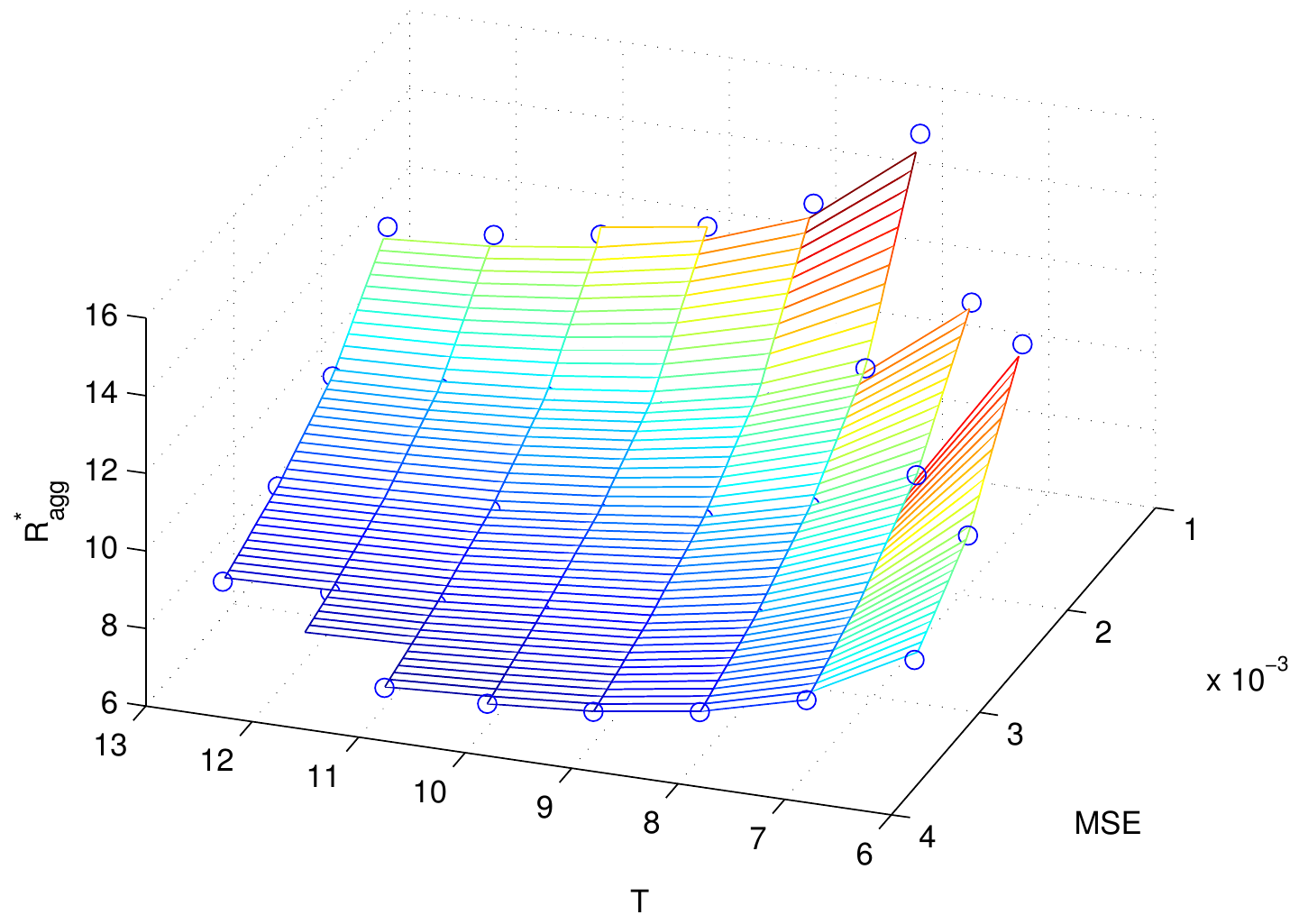}}
  \subfloat[]{
    \label{fig:Pareto2d_fixT} 
    \includegraphics[width=0.3\textwidth]{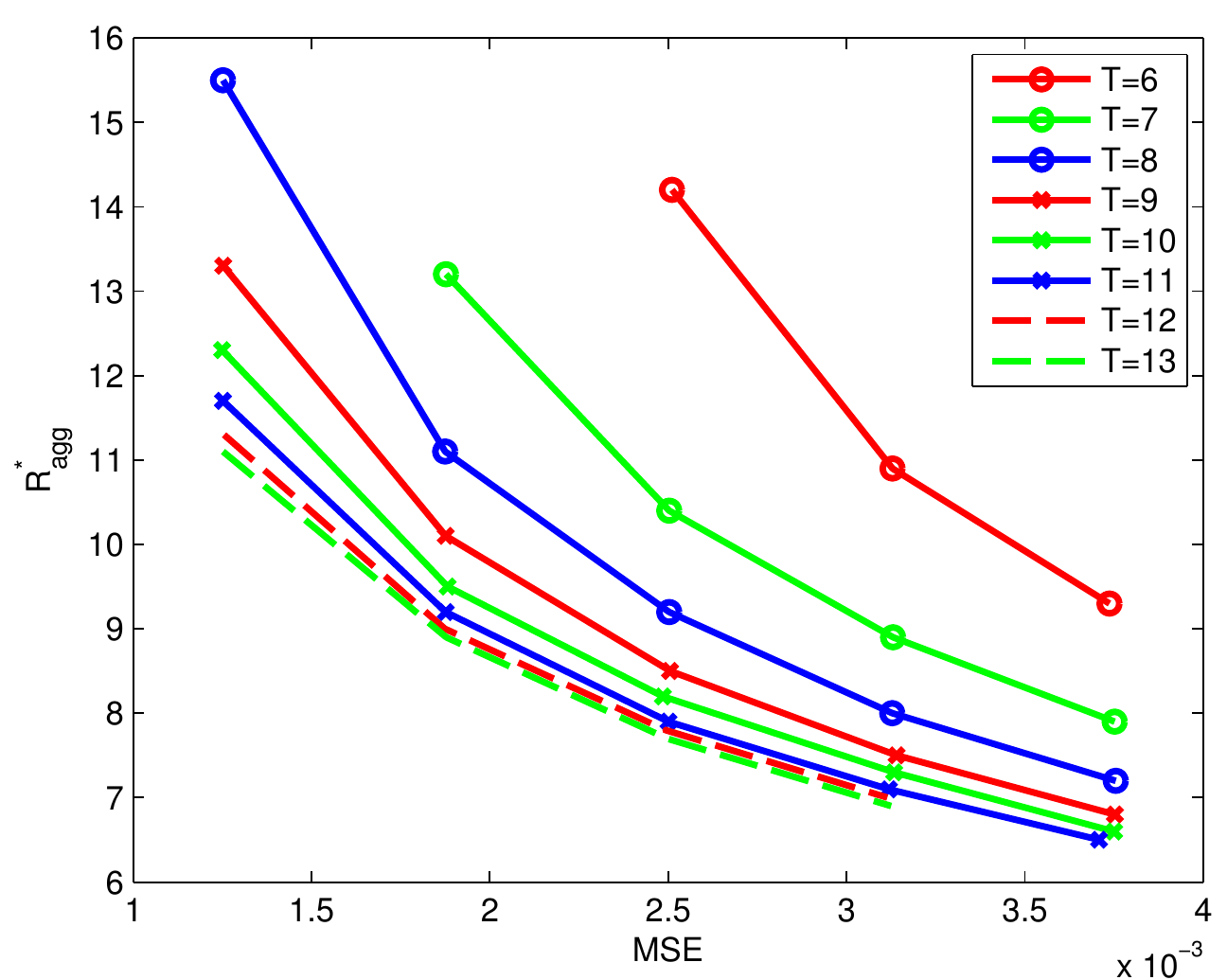}}
  \subfloat[]{
    \label{fig:Pareto2d_fixMSE} 
    \includegraphics[width=0.3\textwidth]{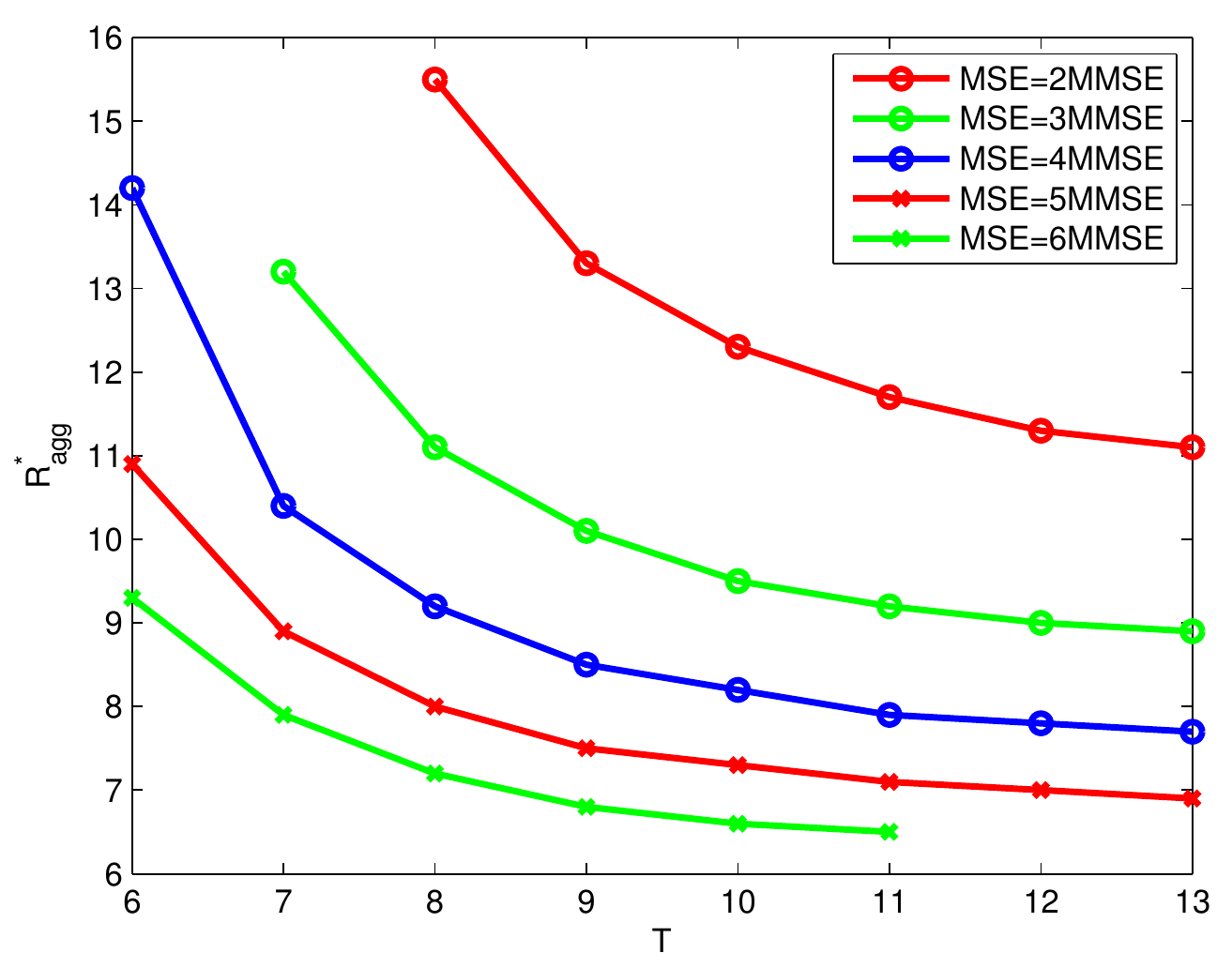}}
\caption{Pareto optimal results provided by DP under a variety of parameters
$b$~\eqref{eq:cost}: (a) Pareto optimal surface, (b) Pareto optimal aggregate coding rate $R_{agg}^*$~\eqref{eq:R_agg} versus the achieved MSE for different optimal MP-AMP iterations $T$, and (c) Pareto optimal $R_{agg}^*$~\eqref{eq:R_agg} versus the number of iterations $T$ for different optimal MSE's.
The signal is Bernoulli-Gaussian~\eqref{eq:BG} with $\rho=0.1$. ($\kappa=0.4$, $P=100$, and $\sigma_Z^2=\frac{1}{400}$.)}\label{fig:Pareto}
\end{figure*}

\subsection{Pareto optimal points via DP}

After proving that the achievable set $\mathcal{C}$ is convex, we apply  DP in Section~\ref{sec:DP}
to find the Pareto optimal points, and validate the convexity of the achievable set.

According to Definition~\ref{def:Pareto}, the resulting tuple $(T,$ $R_{agg},\text{MSE})$
computed using DP (Section~\ref{sec:DP}) is Pareto optimal on the discretized search spaces.
Hence, in this subsection, we run  DP to obtain the Pareto optimal points for a certain distributed linear model by sweeping the parameter $b$~\eqref{eq:cost}.

Consider the same setting as in Figure~\ref{fig:RtAndEMSE}, except that we analyze MP platforms~\cite{pottie2000,estrin2002,EC2} for different $b$~\eqref{eq:cost}. Running the  DP scheme
of Section~\ref{sec:DP},
we obtain the optimal coding rate sequence $\mathbf{R}^*$ that yields the lowest combined cost while
providing a desired EMSE that is at most $\Delta\in \{1,2,\cdots,5\}\times \text{MMSE}$ or equivalently
$\text{MSE}\in \{2,3,\cdots,6\}\times \text{MMSE}$.
In Figure~\ref{fig:rateVSiter}, we draw the Pareto optimal surface obtained by our DP scheme,
where the circles are Pareto optimal points.
Figure~\ref{fig:Pareto2d_fixT} plots the aggregate coding rate $R_{agg}$
as a function of MSE for different optimal numbers of MP-AMP iterations $T$.
Finally, Figure~\ref{fig:Pareto2d_fixMSE} plots the aggregate coding rate $R_{agg}$
as a function of $T$ for different optimal MSE's.
We can see that the surface comprised of the Pareto optimal points is indeed convex.
Note that when running DP to generate Figure~\ref{fig:Pareto}, we used the  RD function~\cite{Cover06,Berger71,GershoGray1993,WeidmannVetterli2012} to model the
relation between the rate $R_t$ and distortion $D_t$ at each iteration,
which could be approached by VQ at sufficiently high
rates. We also ignored the constraint on the
quantization bin size (Section~\ref{sec:MP-CS_for_MP-AMP}).
Therefore, we only present Figure~\ref{fig:Pareto} for illustration purposes.

When a smaller MSE  (or equivalently smaller EMSE)
is desired, more iterations $T$ and greater aggregate coding rates
$R_{agg}$~\eqref{eq:R_agg} are needed. Optimal coding rate sequences increase
$R_{agg}$ to reduce  $T$ when communication costs are
low (examples are commercial cloud computing systems~\cite{EC2}, multi-processor CPUs, and graphic processing units), whereas more iterations
allow to reduce the coding rate when communication is costly (for example, in sensor networks~\cite{pottie2000,estrin2002}).
These applications are discussed in Section~\ref{sec:realworld}.

{\bf Discussion of corner points:} We further discuss the
corners of the Pareto optimal surface (Figure~\ref{fig:Pareto}) below.

\begin{enumerate}
\item First, consider the corner points along the MSE coordinate.
\begin{itemize}
\item If MSE$^* \rightarrow$ MMSE (or equivalently $\Delta \to 0$),
then MP-AMP needs to run infinite iterations with infinite coding rates. Hence, $R_{agg}^*\rightarrow \infty$ and $T^*\rightarrow \infty$. The rate of growth of $R_{agg}^*$ can be deduced from Theorem~\ref{th:optRateLinear}.
\item If MSE$^*$ $\rightarrow \rho$ (the variance of the signal~\eqref{eq:BG}), then MP-AMP does not need to run any iterations at all. Instead, MP-AMP outputs an all-zero estimate. Therefore, $\lim_{\text{MSE}^*\rightarrow\rho} R_{agg}^*= 0$ and $\lim_{\text{MSE}^*\rightarrow\rho} T^*=0$.
\end{itemize}
\item Next, we discuss the corner points along the $T$ coordinate.
\begin{itemize}
\item If $T^*\rightarrow 0$, then the best MP-AMP can do is to output an all-zero estimate. Hence, $\lim_{T^*\rightarrow 0} \text{MSE}^* = \rho$ and $\lim_{T^*\rightarrow 0} R_{agg}^*=0$.
\item The other extreme, $T^*\rightarrow \infty$, occurs only when
we want to achieve an MSE$^*\rightarrow$ MMSE.
Hence, $R_{agg}^*\rightarrow\infty$.
\end{itemize}
\item We conclude with corner points along the $R_{agg}$ coordinate.
\begin{itemize}
\item If $R_{agg}^*\rightarrow 0$, then the best MP-AMP can do is to output an all-zero estimate without running any iterations at all. Hence,
$\lim_{R_{agg}^*\rightarrow 0}\text{MSE}^*=\rho$ and $\lim_{R_{agg}^*\rightarrow 0} T^* = 0$.
\item If $R_{agg}^*\rightarrow\infty$, then the optimal scheme will use high rates in all iterations,
and  MP-AMP resembles centralized AMP. Therefore, the MSE$^*$ as a function of $T^*$ converges to
that of centralized AMP SE~\eqref{eq:ori_SE}.
\end{itemize}
\end{enumerate}

\section{Real-world Case Study}\label{sec:realworld}

To showcase the difference between optimal coding rate sequences in different platforms,
this section discusses several MP platforms including
sensor networks~\cite{pottie2000,estrin2002}
and large-scale cloud servers~\cite{EC2}. The costs in these platforms are
quite different due to the different constraints in these platforms, and we will see how they affect the optimal coding rate sequence $\mathbf{R}^*$.
The changes in the optimal $\mathbf{R}^*$ highlight the importance of optimizing for the correct costs.

\subsection{Sensor networks}\label{sec:sn}
In sensor networks~\cite{pottie2000,estrin2002}, distributed sensors are typically dispatched to remote locations where they collect data and communicate with the fusion center. However, distributed sensors may have severe power consumption constraints. Therefore, low power chips such as the CC253X from Texas Instruments~\cite{CC2530} are commonly used in distributed sensors. Some typical parameters for such low power chips are: central processing unit (CPU) clock frequency 32MHz, data transmission rate 250Kbps, voltage between 2V-3.6V, and transceiver current 25mA~\cite{CC2530}, where the CPU current resembles the transceiver current. Because these chips are generally designed to be low power, when transmitting and receiving data, the CPU helps the transceiver and cannot carry out computing tasks. Therefore, the power consumption can be viewed as constant. Hence, in order to minimize the power consumption, we minimize the total runtime when estimating a signal from MP-LM measurements~\eqref{eq:one-node-meas} collected by the distributed sensors.

The runtime in each MP-AMP iteration~\eqref{eq:slave1}-\eqref{eq:master} consists of ({\em i}) time for computing~\eqref{eq:slave1} and~\eqref{eq:slave2}, ({\em ii}) time for encoding $\f_t^p$~\eqref{eq:slave2}, and ({\em iii}) data transmission time for $Q(\f_t^p)$~\eqref{eq:master0}.
As discussed in Section~\ref{sec:MP-CS_for_MP-AMP}, the fusion center may broadcast $\x_t$~\eqref{eq:master}, and simple compression schemes can reduce the coding rate. Therefore, we consider the data reception time in the $P$ processor nodes to be constant.
The overall computational complexity for~\eqref{eq:slave1} and~\eqref{eq:slave2} is $O(\frac{MN}{P})$.
Suppose further that ({\em i}) each processor node needs to carry out two matrix-vector
products in each iteration, ({\em ii}) the overhead of moving data in memory is assumed to
be 10 times greater than the actual computation, and ({\em iii}) the clock frequency is 32MHz.
Hence, we assume that the actual time needed for computing~\eqref{eq:slave1} and~\eqref{eq:slave2}
is $C_4=\frac{20MN}{32\times 10^6 P}$ sec. Transmitting $Q(\f_t^p)$ of length $N$ at coding rate
$R$ requires $\frac{RN}{250\times 10^3}$ sec, where the denominator is the data transmission rate
of the transceiver. Assuming that the overhead in communication is approximately the same as the
communication load caused by  the actual messages, we obtain that the time requested for
transmitting $Q(\f_t^p)$ at coding rate $R$ is $C_5 R$ sec, where $C_5=\frac{2N}{250\times 10^3}$.
Therefore, the total cost can be calculated from~\eqref{eq:cost} with $b=\frac{C_4}{C_5}$~\eqref{eq:cost}.

Because low power chips equipped in distributed sensors have limited memory (around 10KB, although sometimes external flash is allowed)~\cite{CC2530}, the signal length $N$ and number of measurements $M$ cannot be too large. We consider $N=1000$ and $M=400$ spread over $P=100$ sensors, sparsity rate $\rho=0.1$, and $\sigma_Z^2=\frac{1}{400}$.
We set the desired MSE to be $0.5$ dB above the MMSE, i.e., $10\log_{10}\l(1+ \frac{\Delta}{\text{MMSE}}\r)=0.5$,
and run DP as in
Section~\ref{sec:DP}.\footnote{Throughout Section~\ref{sec:realworld}, we use the RD function~\cite{Cover06,Berger71,GershoGray1993,WeidmannVetterli2012} to model the relation between
rate $R_t$ and distortion $D_t$ at each iteration. We also ignore the constraint on the quantizer
(Section~\ref{sec:MP-CS_for_MP-AMP}). Therefore, the optimal coding rate sequences in
Section~\ref{sec:realworld} are only for illustration purposes.
}
The coding rate sequence provided by DP is $\mathbf{R}^*=(0.1$, $0.1$, $0.6$, $0.8$, $1.0$, $1.0$, $1.1$, $1.1$, $1.2$, $1.4$, $1.6$, $1.9$, $2.3$, $2.7$, $3.1)$. In total we have $T=15$ MP-AMP iterations with $R_{agg}=20.0$ bits aggregate coding rate~\eqref{eq:R_agg}.
The final MSE ($\text{MMSE} + \Delta$) is $7.047\times 10^{-4}$, which is 0.5 dB from the MMSE ($6.281\times 10^{-4}$)~\cite{ZhuBaronCISS2013,Krzakala2012probabilistic,GuoBaronShamai2009,RFG2012}.

\subsection{Large-scale cloud server}\label{sec:largeScaleCase}
Having discussed sensor networks~\cite{pottie2000,estrin2002}, we now discuss an application of  DP (cf. Section~\ref{sec:DP}) to large-scale cloud servers. Consider the dollar cost for users of Amazon EC2~\cite{EC2}, a commercial cloud computing service. A typical cost for CPU time is $\$ 0.03$/hour, and the data transmission cost is $\$ 0.03$/GB. Assuming that the CPU clock frequency is 2.0GHz and considering various overheads, we need a runtime of $\frac{20MN}{2\times 10^9 P}$ sec and the computation cost is $C_4=\$ \frac{20MN}{2\times 10^9 P}\times \frac{0.03}{3600}$ per MP-AMP iteration.
Similar to Section~\ref{sec:sn}, the communication cost for coding rate $R$ is $C_5 R=\$ 2R N\frac{0.03}{8\times 10^9}$.
Note that the multiplicative factors of 20 in $C_4$ and 2 in $C_5$ are due
to the same considerations as in Section~\ref{sec:sn}, and the $8\times 10^9$ in $C_5$ is the number of bits per GB.
Therefore, the total cost with $T$ MP-AMP iterations can still be modeled as in~\eqref{eq:cost}, where $b=\frac{C_4}{C_5}$.

We consider a problem with the same signal and channel model as the setting of Section~\ref{sec:sn}, while the size of the problem grows to $N=50000$ and $M=20000$ spread over $P=100$ computing nodes. Running DP, we obtain the coding rate sequence $\mathbf{R}^*=(1.3$, $1.6$, $1.8$, $1.8$, $1.8$, $1.9$, $2.1$, $2.3$, $2.6$, $3.1$, $3.7)$ for a total of $T=11$ MP-AMP iterations with $R_{agg}=24.0$ bits aggregate coding rate. The final MSE is $7.031\times 10^{-4}$, which is 0.49 dB
above the MMSE. Note that this final MSE is 0.01 dB better than our goal of
$0.5$ dB above the MMSE due to the discretized search spaces used in DP.

{\bf Settings with even cheaper communication costs:} Compared to large-scale cloud servers, the relative
cost of communication is even cheaper in multi-processor CPU and graphics processing unit (GPU)
systems. We reduce $b$ by a factor of 100 compared to the large-scale cloud server case above. We rerun DP, and obtain the
coding rate sequence $\mathbf{R}^*=(2.3$, $2.5$, $2.6$, $2.7$, $2.7$, $2.8$, $3.0$, $3.4$, $3.7$, $4.5)$ for $T=10$ and $R_{agg}=30.2$ bits. Note that 10 iterations are needed for centralized AMP to converge in this setting. With the low-cost communication of this setting, DP yields a coding rate sequence $\mathbf{R}^*$ within 0.5 dB of the MMSE with the same number of iterations as centralized AMP, while using an average coding rate of only 3.02 bits  per iteration.

\begin{myRemark}
Let us review the cost tuples $(T,R_{agg},\text{MSE})$ for our three cases. For sensor networks, $(T,R_{agg},\text{MSE})_{\text{sensornet}}=(15,20,7.047\times 10^{-4})$; for cloud servers,
$(T,R_{agg},\text{MSE})_{\text{cloud}}$ $=$ $(11,24,7.031\times 10^{-4})$; and for GPUs, $(T,R_{agg},\text{MSE})_{\text{GPU}}=(10,30.2,7.047\times 10^{-4})$. These cost tuples are different points in the Pareto optimal set $\mathcal{P}$~\eqref{eq:setP}.
We can see for sensor networks that the optimal coding rate sequence
reduces $R_{agg}$ while adding
iterations, because sensor networks have relatively expensive communications.
The optimal coding rate sequences use higher rates in cloud servers and GPUs,
because their communication costs are relatively lower.
Indeed, different trade-offs between computation and communication
lead to different aggregate coding rates $R_{agg}$ and numbers of MP-AMP iterations $T$.
Moreover, the optimal coding rate sequences for sensor networks, cloud servers, and GPUs
use average coding rates of $1.33$, $2.18$, and $3.02$ bits/entry/iteration, respectively.
Compared to $32$ bits/entry/iteration single-precision floating point communication
schemes, optimal coding rate sequences reduce the communication costs significantly.
\end{myRemark}

\section{Conclusion}\label{sec:conclude}
This chapter used lossy compression in multi-processor (MP) approximate message passing (AMP) for solving MP linear inverse problems. Dynamic programming (DP) was used to obtain the optimal coding rate sequence for MP-AMP that incurs the lowest combined cost of communication and computation while achieving a desired
mean squared error (MSE). We posed the problem of finding the optimal coding rate sequence in the low excess MSE (EMSE=MSE-MMSE, where MMSE refers to the minimum MSE) limit as a convex optimization problem and proved that optimal coding rate sequences are approximately linear when the EMSE is small.
Additionally, we obtained that the combined cost of computation
and communication scales with $O(\log^2(1/\text{EMSE}))$.
Furthermore, realizing that there is a trade-off among the communication cost, computation cost, and MSE, we formulated a multi-objective optimization problem (MOP) for these costs and studied the Pareto optimal points that exploit this trade-off. We proved that the achievable region of the MOP is convex.

We further emphasize that there is little work in the prior art discussing the optimization of communication schemes in iterative distributed algorithms. Although we focused on the MP-AMP algorithm, our conclusions such as the linearity of the optimal coding rate sequence and the convexity of the
achievable set of communication/computation trade-offs could be extended to
other iterative distributed algorithms including consensus averaging~\cite{Frasca2008,Thanou2013}.


\chapter{Universal Algorithm}
\label{chap-SLAM}

Previous chapters discussed the information theoretic performance limits for multi-measurement vector problems~\eqref{eq:MMVmodel_intro} and also studied the optimal trade-offs among different costs in multi-processor linear models~\eqref{eq:one-node-meas_intro}. When the number of rows $M$ is smaller than the number of columns $N$ in the measurement matrix $\A$, we call the corresponding linear model a compressed sensing (CS) problem. In this chapter,
we study the CS signal estimation problem. While CS usually assumes sparsity or compressibility in the input signal during estimation, the signal structure that can be leveraged is often not known a priori. In this chapter, we consider universal CS signal estimation, where the statistics of a stationary ergodic signal source are estimated simultaneously with the signal itself. Inspired by Kolmogorov complexity and minimum description length, we focus on a maximum a posteriori (MAP) estimation framework that leverages universal priors to match the complexity of the source. Our
framework can also be applied to general linear inverse problems where more measurements than the signal length might be needed. We provide theoretical results that support the algorithmic feasibility of universal MAP estimation using a Markov chain Monte Carlo implementation (an algorithmic framework mimicking the annealing process in statistical physics, cf. Section~\ref{sec:background_statPhys}), which is computationally challenging. We incorporate some techniques to accelerate the algorithm while providing comparable and in many cases better estimation quality than existing algorithms. Experimental results show the promise of universality in CS, particularly for low-complexity sources that do not exhibit standard sparsity or compressibility. This chapter is based on our work with Baron and Duarte~\cite{JZ2014SSP,ZhuBaronDuarte2014_SLAM}.
\section{Motivation and Contributions}
\label{sec:intro}

Since many systems in science and engineering are approximately linear~\eqref{eq:SMV},
linear inverse problems have attracted great attention in the
signal processing community. Recall from~\eqref{eq:SMV} that an
input signal $\x \in \mathbb{R}^N$ is recorded via a linear operator under additive noise:
\begin{equation}
\label{eq:def_y}
\y = \A \x + \z,
\end{equation}
where $\A$ is an $M \times N$ matrix and $\z \in \mathbb{R}^M$ denotes the noise.
The goal is to estimate $\x$ from the measurements $\y$ given knowledge of $\A$
and a model for the noise $\z$. When $M \ll N$, the setup is known as compressed
sensing (CS) and the estimation problem is commonly referred to as recovery
or reconstruction; by posing a sparsity or compressibility\footnote{We use the term compressibility in this chapter as defined by Cand\`{e}s et al.~\cite{CandesRUP} to refer to signals whose sparse approximation error decays sufficiently quickly.} requirement on the signal
and using this requirement as a prior during estimation, it is indeed possible to
accurately estimate $\x$ from $\y$~\cite{CandesRUP,DonohoCS}. On the other hand,
we might need more measurements than the signal length when the signal is dense or the noise is substantial.

Wu and Verd{\'u}~\cite{WuVerdu2012} have shown that independent and identically distributed (i.i.d.) Gaussian sensing matrices achieve the same phase transition threshold as the optimal (potentially non-linear) measurement operator, for any i.i.d. signals following the discrete/continuous mixture distribution $f_X(x)=\rho\cdot f_c(x)+(1-\rho)\cdot \mathbb{P}_d(x)$, where $\rho$ is the probability for a scalar $x$ to take a continuous distribution $f_c(x)$ and $\mathbb{P}_d(x)$ is an arbitrary discrete distribution. For non-i.i.d. signals, Gaussian matrices also work well~\cite{Donoho2013,Tan_CompressiveImage2014,MaZhuBaronAllerton2014}. Hence,
in CS the acquisition can be designed independently of the particular signal
prior through the use of randomized Gaussian matrices $\A$. Nevertheless, the majority of (if not all)
existing estimation algorithms require knowledge of the sparsity structure of
$\x$, i.e., the choice of a {\em sparsifying transform} $\W$ that renders a sparse coefficient vector $\theta=\W^{-1} \x$ for
the signal.

The large majority of CS signal estimation algorithms pose a sparsity prior on the signal $\x$ or the coefficient vector $\theta$, e.g.,~\cite{CandesRUP,DonohoCS,GPSR2007}.
A second, separate class of Bayesian CS signal estimation algorithms poses a
probabilistic prior for the coefficients of $\x$ in a known transform
domain~\cite{DMM2010ITW1,RanganGAMP2011ISIT,BCS2008,BCSEx2008,CSBP2010}. Given a probabilistic model, some related message passing approaches
learn the parameters of the signal model and achieve the minimum mean squared error (MMSE) in some settings; examples include EM-GM-AMP-MOS~\cite{EMGMTSP}, turboGAMP~\cite{turboGAMP}, and AMP-MixD~\cite{MTKB2014ITA}. As a third alternative,
complexity-penalized least square methods~\cite{Figueiredo2003,DonohoKolmogorovCS2006,HN05,HN11,Ramirez2011} can use arbitrary prior information on the
signal model and provide analytical guarantees, but are only computationally
efficient for specific signal models, such as the independent-entry Laplacian
model~\cite{HN05}. For example, Donoho et al.~\cite{DonohoKolmogorovCS2006} relies on Kolmogorov complexity, which cannot be computed~\cite{Cover06,LiVitanyi2008}.
As a fourth alternative, there exist algorithms that can formulate dictionaries that yield
sparse representations for the signals of interest when a large amount of training data is
available~\cite{Ramirez2011,AharoEB_KSVD,Mairal2008,Zhoul2011}.
When the signal is non-i.i.d., existing algorithms require either prior knowledge of the probabilistic model~\cite{turboGAMP} or the use of training data~\cite{Garrigues07learninghorizontal}.

In certain cases, one might not be certain about the structure or statistics of the
source prior to estimation. Uncertainty about such structure may result in a sub-optimal choice of
the sparsifying transform $\W$, yielding a coefficient vector $\theta$ that requires more measurements to achieve reasonable
estimation quality; uncertainty about the statistics of the source will make it difficult to
select a prior or model for Bayesian algorithms. Thus, it would  be desirable to formulate algorithms
to estimate $\x$ that are more agnostic to the particular statistics of the signal.
Therefore, we shift our focus from the standard sparsity or compressibility priors to
{\em universal} priors~\cite{LZ77,Rissanen1983,Ramirez2010}. Such concepts have been
previously leveraged in the Kolmogorov sampler universal denoising
algorithm~\cite{DonohoKolmogorov}, which minimizes Kolmogorov
complexity~\cite{Chaitin1966,Solomonoff1964,Kolmogorov1965,LiVitanyi2008,JalaliMaleki2011,JalaliMalekiRichB2014,BaronFinland2011,BaronDuarteAllerton2011}. Related approaches
based on minimum description length
(MDL)~\cite{Rissanen1978,schwarz1978estimating,Wallace1968,BRY98} minimize
the complexity of the estimated signal with respect to (w.r.t.) some class of sources.

Approaches for non-parametric sources based
on Kolmogorov complexity are not computable in practice~\cite{Cover06,LiVitanyi2008}. To address this
computational problem, we confine our attention to the class of stationary ergodic sources and
develop an algorithmic framework for {\em universal} signal estimation in CS systems that will approach the MMSE as closely as possible for the class of stationary ergodic sources. Our framework can be applied to general
linear inverse problems where more measurements might be needed.
Our framework leverages the fact that for stationary ergodic sources,
both the per-symbol empirical entropy and Kolmogorov complexity converge
asymptotically almost surely to the entropy rate of the source~\cite{Cover06}. We aim
to minimize the empirical entropy; our minimization is regularized by introducing a log
likelihood for the noise model, which is equivalent to the standard least squares under
additive white Gaussian noise. Other noise distributions are readily supported.

We make the following contributions toward our universal CS framework.
\begin{itemize}
  \item We apply a specific quantization grid to a maximum {\em a posteriori} (MAP) estimator driven by a universal prior, providing a finite-computation
universal estimation scheme; our scheme can also be applied to general linear inverse problems where more measurements might be needed.
  \item We propose an estimation algorithm based on
Markov chain Monte Carlo (MCMC)~\cite{Geman1984} to approximate this estimation procedure.
  \item We prove that for a sufficiently large number of iterations the output of our
MCMC estimation algorithm converges to the correct MAP estimate.
  \item We identify computational bottlenecks in the implementation of our MCMC estimator and show
approaches to reduce their complexity.
  \item We develop an adaptive quantization
scheme that tailors a set of reproduction levels to minimize the quantization error
within the MCMC iterations and that provides an accelerated implementation.
  \item We propose a framework that adaptively adjusts the cardinality (size) of the adaptive quantizer to match
the complexity of the input signal, in order to further reduce the quantization error and computation.
  \item We note in passing that averaging over the outputs of different runs of the same signal with the same measurements will yield
lower mean squared error (MSE) for our proposed algorithm.
\end{itemize}

This chapter is organized as follows. Section~\ref{sec:setting} provides background
content. Section~\ref{sec:theory} overviews MAP estimation, quantization, and introduces universal MAP estimation.
Section~\ref{sec:MCMC} formulates an initial MCMC algorithm for universal MAP
estimation, Section~\ref{sec:adaptive} describes several improvements to this initial algorithm, and Section~\ref{sec:numerical} presents
experimental results. We conclude in Section~\ref{sec:conclusions}.
The proof of our main theoretical result appears in Appendix~\ref{chap:append-SLAM}.

\section{Background and Related Work}
\label{sec:setting}
\subsection{Compressed sensing}

Consider the noisy measurement setup via a linear operator (\ref{eq:def_y}).
The input signal $\x\in \mathbb{R}^N$ is generated by a stationary ergodic source $X$,
and must be estimated from $\y$ and $\A$. Note that the stationary ergodicity assumption enables us to model the potential memory in the source.
{\em The distribution $f_X(\cdot)$ that generates $\x$ is unknown.} The matrix
$\A \in \mathbb{R}^{M \times N}$ has i.i.d. Gaussian
entries, $A_{m,n} \sim \mathcal{N}(0,\frac{1}{M})$.\footnote{In contrast to our analytical
and numerical results, the algorithm presented in Section~\ref{sec:MCMC} is not
dependent on a particular choice for the matrix $\A$.} These moments ensure that the
columns of the matrix have unit norm on average. For concrete analysis, we assume that
the noise $\z\in\mathbb{R}^M$ is i.i.d.\ Gaussian, with mean zero and known\footnote{We assume that the noise variance is known or can be estimated~\cite{DMM2010ITW1,MTKB2014ITA}.} variance
$\sigma_Z^2$ for simplicity.

We focus on the large system limit (cf. Definition~\ref{def:chap1-largeSystemLimit} in Chapter~\ref{chap-intro}).
Similar settings have been discussed in the literature~\cite{Rangan2010CISS,GuoWang2008}.
When $M\ll N$, this setup is known as CS; otherwise, it is a general linear inverse problem setting.
Since $\x$ is generated by an unknown source, we must search for an estimation
mechanism that is agnostic to the specific distribution $f_X(\cdot)$.

\subsection{Related work}
\label{subsec:Kolmogorov}

For a scalar channel with a discrete-valued signal $\x$, e.g., $\A$ is an identity matrix and $\y=\x+\z$,
Donoho proposed the Kolmogorov sampler for
denoising~\cite{DonohoKolmogorov},
\begin{equation}
\label{eq:x_KS}
\x_{KS}  \triangleq \arg\min_{\w} K(\w)\mbox{, subject to}~\|\w-\y\|^2<\tau,
\end{equation}
where $K(\x)$ denotes the Kolmogorov complexity of $\x$, defined as the length
of the shortest input to a Turing machine~\cite{Turing1950} that generates the
output $\x$ and then halts,\footnote{For real-valued $\x$, Kolmogorov complexity  can be approximated using a fine quantizer. Note that the algorithm developed in this chapter uses a coarse quantizer and does not rely on Kolmogorov complexity
due to the absence of a feasible method for its computation~\cite{Cover06,LiVitanyi2008} (cf.\ Section~\ref{sec:adaptive}).} and $\tau = N\sigma_Z^2$ controls for the presence
of noise. It can be shown that $K(\x)$ asymptotically captures the statistics of
the stationary ergodic source $X$, and the per-symbol complexity achieves the
entropy rate $H \triangleq H(X)$, i.e., $\lim_{N\to\infty} \frac{1}{N}K(\x)=H$ almost
surely~[\cite{Cover06}, p.~154, Theorem~7.3.1]. Noting that universal lossless compression
algorithms~\cite{LZ77,Rissanen1983} achieve the entropy rate for any
discrete-valued finite state machine source $X$, we see that these algorithms
achieve the per-symbol Kolmogorov complexity almost surely.

Donoho et al. expanded Kolmogorov sampler to the linear CS measurement setting
$\y=\A \x$ but did not consider measurement noise~\cite{DonohoKolmogorovCS2006}.
Recent papers by Jalali and coauthors~\cite{JalaliMaleki2011,JalaliMalekiRichB2014}, which appeared
simultaneously with Baron~\cite{BaronFinland2011} and Baron and Duarte~\cite{BaronDuarteAllerton2011},
provide an analysis of a modified Kolmogorov sampler suitable for measurements corrupted by
noise of bounded magnitude. Inspired by Donoho et al.~\cite{DonohoKolmogorovCS2006}, we
estimate $\x$ from noisy measurements $\y$ using the empirical entropy as a proxy
for the Kolmogorov complexity (cf.\ Section~\ref{sec:compressor}).

Separate notions of complexity-penalized least square{s} have also been shown
to be well suited for denoising and CS
signal estimation~\cite{Figueiredo2003,DonohoKolmogorovCS2006,Rissanen1978,schwarz1978estimating,Wallace1968,HN05,HN11,Ramirez2011}. For example,
minimum description length (MDL)~\cite{Rissanen1978,schwarz1978estimating,Wallace1968,Ramirez2011} provides
a framework composed of classes of models for which
the signal complexity can be defined sharply. In general, complexity-penalized least square
approaches can yield MDL-flavored CS signal estimation algorithms that are
adaptive to parametric classes of sources~\cite{DonohoKolmogorovCS2006,Figueiredo2003,HN05,HN11}. An
alternative universal denoising approach computes the universal conditional expectation of the signal~\cite{BaronFinland2011,MTKB2014ITA}.

\sectionmark{Background and Related Work}
\section{Universal MAP Estimation and Discretization}
\label{sec:theory}
\sectionmark{Universal MAP}

This section briefly reviews MAP estimation and then applies it over a quantization grid, where a universal prior is used for the signal. Additionally, we provide a conjecture for the MSE achieved by our universal MAP scheme.

\subsection{Discrete MAP estimation}\label{sec:MAP}
In this subsection, we assume for exposition purposes that we know the signal distribution $f_X(\cdot)$.
Given the measurements $\y$, the MAP estimator for $\x$ has the form
\begin{equation}
\x_{MAP} \triangleq \arg\max_{\w} f_X(\w) f_{Y|X}(\y|\w).
\label{eq:map}
\end{equation}
Because $\z$ is i.i.d.\ Gaussian with mean zero and known variance $\sigma_Z^2$,
\begin{equation*}
 f_{Y|X}(\y|\w) = c_1 \operatorname{e}^{-c_2 \|\y-\A \w\|^2},
\end{equation*}where
$c_1= (2\pi\sigma_Z^2)^{-M/2}$ and
$c_2= \frac{1}{2\sigma_Z^2}$
are constants, and $\|\cdot\|$ denotes the Euclidean norm.\footnote{Other noise distributions are readily supported, e.g., for i.i.d. Laplacian noise, we need to change the $\ell_2$ norm to an $\ell_1$ norm and adjust $c_1$ and $c_2$ accordingly.}
Plugging into (\ref{eq:map}) and taking log likelihoods, we obtain
$\displaystyle \x_{MAP} = \arg\min_{\w} \Psi^X(\w)$,
where $\Psi^X(\cdot)$ denotes the objective function (risk)
\begin{equation*}
\Psi^X(\w) \triangleq -\ln(f_X(\w)) + c_2\|\y-\A \w\|^2;
\end{equation*}our ideal risk would be $\Psi^X(\x_{MAP})$.

Instead of performing continuous-valued MAP estimation, we optimize for the MAP
in the discretized domain $\replevels^N$, with $\replevels$ being defined as follows.
Adapting the approach of Baron and Weissman~\cite{BaronWeissman2012},
we define the set of data-independent reproduction levels for quantizing $\x$ as
\begin{equation}
\label{eq:def:replevels}
\replevels \triangleq \left\{
\cdots,-\frac{1}{\gamma},0,\frac{1}{\gamma},\cdots
\right\},
\end{equation}
where $\gamma=\lceil\ln(N)\rceil$. As $N$ increases, $\replevels$ will quantize
$\x$ to a greater resolution.
These reproduction levels simplify the estimation problem from continuous to discrete.

Having discussed our reproduction levels in the set $\replevels$, we provide a technical condition on boundedness of the signal.
\begin{COND}
\label{cond:tech1}
We require that the probability density $f_X(\cdot)$ has bounded support, i.e., there exists $\Lambda = [x_\textrm{min},x_\textrm{max}]$ such that (s.t.) $f_X(\x) = 0$ for $\x \notin \Lambda^N$.
\end{COND}

A limitation of the data-independent reproduction level set (\ref{eq:def:replevels})
is that $\replevels$ has infinite cardinality (or size for short).
Thanks to Condition~\ref{cond:tech1}, for each value of $\gamma$ there exists
a constant $c_3>0$ s.t. a finite set of reproduction levels
\begin{equation}
\replevels_F \triangleq \left\{
-\frac{c_3\gamma^2}{\gamma},-\frac{c_3\gamma^2-1}{\gamma},\cdots,
\frac{c_3\gamma^2}{\gamma}
\right\}
\label{eq:def:replevels2}
\end{equation}
will quantize the range of values $\Lambda$ to the same accuracy as that of (\ref{eq:def:replevels}). We call $\replevels_F$ the {\em reproduction alphabet}, and each element in it a ({\em reproduction}) {\em level}.
This finite quantizer reduces the complexity of the estimation problem
from infinite to combinatorial. In fact,
$x_i\in [x_\textrm{min},x_\textrm{max}]$ under Condition~\ref{cond:tech1}. Therefore, for all $c_3 >0$ and sufficiently large $N$, this set of levels will cover the range $[x_\textrm{min},x_\textrm{max}]$.
The resulting reduction in complexity is due to the
structure in $\breplevels$ and independent of the particular statistics of the
source $X$.

Now that we have set up a quantization grid $(\breplevels)^N$ for $\x$,
we convert the distribution $f_X(\cdot)$ to a probability mass function (PMF)
$\mathbb{P}_X(\cdot)$ over $(\breplevels)^N$. Let
$\displaystyle f_{\breplevels} \triangleq \sum_{\w \in (\breplevels)^N} f_X(\w)$,
and define a PMF $\mathbb{P}_X(\cdot)$ as
$\displaystyle \mathbb{P}_X(\w) \triangleq \frac{f_X(\w)}{f_{\breplevels}}$.
Then,
\begin{equation*}
\x_{MAP}(\breplevels) \triangleq \arg\min_{\w \in (\breplevels)^N} \left[-\ln(  \mathbb{P}_X(\w)  ) + c_2 \|\y-\A \w\|^2\right]
\end{equation*}gives the MAP estimate of $\x$ over $(\breplevels)^N$.
Note that we use the PMF formulation above, instead of the more common bin integration formulation, in order to simplify our presentation and analysis. Luckily, as $N$ increases, $\mathbb{P}_X(\cdot)$ will approximate $f_X(\cdot)$ more closely under (\ref{eq:def:replevels2}).

\subsection{Universal MAP estimation}
\label{sec:univtheory}
We now describe a
universal estimator for CS over a quantized grid. Consider a prior
$\mathbb{P}_U(\cdot)$ that might involve Kolmogorov
complexity~\cite{Chaitin1966,Solomonoff1964,Kolmogorov1965},
e.g., $\mathbb{P}_U(\w)=2^{-K(\w)}$, or MDL complexity w.r.t.\ some class of
parametric sources~\cite{Rissanen1978,schwarz1978estimating,Wallace1968}.
We call $\mathbb{P}_U(\cdot)$ a {\em universal prior} if it has the fortuitous property that for every stationary ergodic
source $X$ and fixed $\epsilon > 0$,  there exists some minimum $N_0(X,\epsilon)$ s.t.
\begin{equation*}
 -\frac{\ln(\mathbb{P}_U(\w))}{N} < -\frac{\ln(\mathbb{P}_X(\w))}{N} + \epsilon
\end{equation*}for all $\w\in(\breplevels)^N$ and $N > N_0(X,\epsilon)$~\cite{LZ77,Rissanen1983}.
We optimize over an objective function that incorporates $\mathbb{P}_U(\cdot)$ and the presence
of additive white Gaussian noise in the measurements:
\begin{equation}
\Psi^U(\w) \triangleq -\ln(\mathbb{P}_U(\w)) + c_2 \|\y-\A \w\|^2,
\label{eq:psidef}
\end{equation}
resulting in\footnote{This formulation of $\x_{MAP}^U$  corresponds to a Lagrangian relaxation of the approach studied
in~\cite{JalaliMaleki2011,JalaliMalekiRichB2014}.}
$\displaystyle
\label{eq:universal_x}
\x_{MAP}^U \triangleq \arg \min_{\w \in (\breplevels)^N} \Psi^U(\w)$.
Our universal MAP estimator does not require $M\ll N$, and $\x_{MAP}^U$ can be used in general linear inverse problems.

\subsection{Conjectured MSE performance}\label{sec:conjecture}
Donoho~\cite{DonohoKolmogorov} showed for the scalar
channel $\y=\x+\z$ that: ($i$) the Kolmogorov sampler $\x_{KS}$ (\ref{eq:x_KS})
is drawn from the posterior distribution $\mathbb{P}_{X|Y}(\x|\y)$; and ($ii$) the
MSE of this estimate $\mathbb{E}_{X,Z,\A}[\|\y-\x_{KS}\|^2]$
is no greater than twice the MMSE.
Based on this result, which requires a large reproduction alphabet, we now present a conjecture on the quality of the estimate $\x^U_{MAP}$. Our conjecture is based on observing that
({\em i}) in the setting~\eqref{eq:def_y}, Kolmogorov sampling achieves optimal rate-distortion performance;
({\em ii}) the Bayesian posterior distribution is the solution to the rate-distortion problem; and
({\em iii}) sampling from the Bayesian posterior yields a squared error that is no greater
than twice the MMSE. Hence, $\x^U_{MAP}$ behaves as if we sample
from the Bayesian posterior distribution and yields no greater
than twice the MMSE; some experimental evidence to assess this conjecture is presented in Figures~\ref{fig:Ber} and~\ref{fig:sparseL}.

\begin{myConj}
\label{conj:double_MMSE}
Assume that $\A\in\mathbb{R}^{M\times N}$ is an i.i.d.\ Gaussian measurement matrix
where each entry has mean zero and variance $\frac{1}{M}$. Suppose that
Condition~\ref{cond:tech1} holds, the aspect ratio
$\kappa=\frac{M}{N}$, and the noise $\z\in\mathbb{R}^M$ is i.i.d.\ zero-mean
Gaussian with finite variance.
Then for all $\epsilon>0$, the mean squared error of the
universal MAP estimator $\x^U_{MAP}$ satisfies
\begin{equation*}
\frac{\mathbb{E}_{X,Z,\A}\left[\|\x-\x^U_{MAP}\|^2\right]}{N} < \frac{2 \mathbb{E}_{X,Z,\A}\left[\|\x-\mathbb{E}_{X}[\x|\y,\A]\|^2\right]}{N} +\epsilon
\end{equation*}
for sufficiently large $N$.
\end{myConj}

\section{Fixed Reproduction Alphabet Algorithm}
\label{sec:MCMC}

Although the results of the previous section are theoretically appealing, a brute force
optimization of $\x_{MAP}^U$ is computationally intractable. Instead, we propose an
algorithmic approach based on MCMC
methods~\cite{Geman1984}. Our approach is reminiscent of the framework for lossy data
compression~\cite{Jalali2008,Jalali2012,BaronWeissman2012,Yang1997}.

\subsection{Universal compressor}\label{sec:compressor}

We propose a universal lossless compression formulation following
the conventions of Weissman and coauthors~\cite{Jalali2008,Jalali2012,BaronWeissman2012}.
We refer to the estimate as $\w$ in our algorithm.
Our goal is to characterize $-\ln(\mathbb{P}_U(\w))$, cf.~(\ref{eq:psidef}).
Although we are inspired by the Kolmogorov sampler approach~\cite{DonohoKolmogorov}, Kolmogorov complexity cannot be computed~\cite{Cover06,LiVitanyi2008}, and we instead use empirical entropy. For stationary ergodic sources, the empirical entropy
converges to the per-symbol entropy rate almost surely~\cite{Cover06}.

To define the empirical entropy, we first define the empirical symbol
counts:
\begin{equation}
n_q(\w,\alpha)[\beta] \triangleq \l| \{ i \in [q+1,N]: \w_{i-q}^{i-1}=\alpha, w_i=\beta \} \r|,
\label{eq:nq}
\end{equation}
where $q$ is the context depth~\cite{Rissanen1983,Willems1995CTW},
$\beta \in \replevels_F$, $\alpha\in(\replevels_F)^q$, $w_i$ is the $i$-th symbol of $\w$,
and $\w_i^j$ is the string comprising symbols $i$ through $j$ within $\w$.
We now define the order $q$ conditional empirical probability for the context
$\alpha$ as
\begin{equation}
\label{eq:def:Pcond}
\mathbb{P}_q(\w,\alpha)[\beta] \triangleq
 \frac{  n_q(\w,\alpha)[\beta] } { \sum_{\beta' \in \replevels_F} n_q(\w,\alpha)[\beta'] },
\end{equation}
and the order $q$ conditional empirical entropy,\footnote{Interested readers can refer to the definitions of entropy for thermodynamics and information theory in~\eqref{eq:def_entropy} and~\eqref{eq:entropy_continuous}, respectively.}
\begin{equation}
H_q(\w) \triangleq -\frac{1}{N} \sum_{\alpha \in (\replevels_F)^q,\beta \in \replevels_F} n_q(\w,\alpha)[\beta]
\log_2\left( \mathbb{P}_q(\w,\alpha)[\beta] \right),
\label{eq:def:H}
\end{equation}
where the sum is only over non-zero counts and probabilities.

Allowing the context depth $q \triangleq q_N=o(\log(N))$ to grow slowly with $N$,
various universal compression algorithms can achieve the empirical entropy
$H_{q}(\cdot)$ asymptotically~\cite{Rissanen1983,Willems1995CTW,LZ77}.
On the other hand, no compressor can outperform the entropy rate. Additionally,
for large $N$, the empirical symbol counts with context depth $q$ provide a
sufficiently precise characterization of the source statistics. Therefore, $H_q$
provides a concise approximation to the per-symbol coding length of a universal
compressor.

\subsection{Markov chain Monte Carlo}\label{sec:B-MCMC}

Having approximated the coding length, we now describe how to
optimize our objective function.
We define the energy $\Psi^{H_q}(\w)$
in an analogous manner to $\Psi^U(\w)$ (\ref{eq:psidef}),
using $H_q(\w)$ as our universal coding length:
\begin{equation}
\label{eq:MCMC_energy}
\Psi^{H_q}(\w) \triangleq NH_q(\w) + c_4 \|\y - \A \w\|^2,
\end{equation}
where $c_4=c_2 \log_2(\operatorname{e}) $.
The minimization of this energy is analogous to minimizing $\Psi^U(\w)$.

Ideally, our goal is to compute the globally minimum energy solution
$\displaystyle \x_{MAP}^{H_q} \triangleq \arg\min_{\w \in (\replevels_F)^N} \Psi^{H_q}(\w)$.
We use a stochastic MCMC
relaxation~\cite{Geman1984} to achieve the globally minimum solution
in the limit of infinite computation.
To assist the reader in appreciating how MCMC is used to compute $\x_{MAP}^{H_q}$, we include pseudocode for our approach in Algorithm~\ref{alg:MCMC}. The algorithm, called basic MCMC (B-MCMC), will be used as a building block for our latter Algorithms~\ref{alg:MCMCAL} and~4 in Section~\ref{sec:adaptive}.
The initial estimate $\w$ is obtained by quantizing the {\em initial point} $\x^*\in\mathbb{R}^N$ to $(\replevels_F)^N$. The initial point $\x^*$ could be the output of any CS signal estimation algorithm, and because $\x^*$ is a preliminary estimate of the signal that does not require high fidelity, we let $\x^*=\A^{\top} \y$ for simplicity, where $\{\cdot\}^{\top}$ denotes transpose.
We refer to the processing of a single entry of $\w$ as an iteration and group the
processing of all entries of $\w$, randomly permuted, into
super-iterations.

The Boltzmann PMF for a thermodynamic system was defined in~\eqref{eq:def_Boltzmann_background}. Similarly, we define the Boltzmann PMF for the energy $\Psi^{H_q}(\w)$~\eqref{eq:MCMC_energy} as
\begin{equation}
\label{eq:def_Boltzmann}
\mathbb{P}_s(\w) \triangleq \frac{1}{\zeta_s} \text{exp}\l(-s \Psi^{H_q}(\w)\r),
\end{equation}
where $s>0$ is inversely related to the temperature in simulated annealing and
$\zeta_s$ is a normalization constant.
MCMC samples from the Boltzmann PMF (\ref{eq:def_Boltzmann}) using a
{\em Gibbs sampler}: in each iteration, a single element $w_n$ is generated
while the rest of $\w$, $\w^{\backslash n} \triangleq \{ w_i:\ n \neq i\}$,
remains unchanged. We denote by $\w_1^{n-1} \beta \w_{n+1}^N$ the
concatenation of the initial portion of the output vector $\w_1^{n-1}$, the symbol
$\beta \in \replevels_F$, and the latter portion of the output $\w_{n+1}^N$. The
Gibbs sampler updates $w_n$ by resampling from the PMF:
\begin{eqnarray}
\mathbb{P}_s(w_n=a|\w^{\backslash n}) \label{eqn:Gibbs}&=&  \frac{ \text{exp}\left(-s\Psi^{H_q}(\w_1^{n-1}a\w_{n+1}^N) \right) }
{ \sum_{b \in \breplevels} \text{exp}\left( -s\Psi^{H_q}(\w_1^{n-1}b\w_{n+1}^N) \right) } \nonumber\\
&=& \frac{1}{\sum_{b \in \breplevels} \text{exp}\left[ -s\left( N\Delta H_q(\w,n,b,a) + c_4
\Delta d(\w,n,b,a) \right) \right] }\nonumber,
\end{eqnarray}
where
\begin{eqnarray*}
\Delta H_q(\w,n,b,a)  \triangleq
 H_q(\w_1^{n-1}b\w_{n+1}^N)  -  H_q(\w_1^{n-1}a\w_{n+1}^N)
\end{eqnarray*}
is the change in empirical entropy $H_q(\w)$ (\ref{eq:def:H}) when $w_n=a$
is replaced by $b$, and
\begin{equation}
\Delta d(\w,n,b,a) \triangleq  \|\y-\A(\w_1^{n-1}b\w_{n+1}^N)\|^2 - \|\y-\A(\w_1^{n-1}a\w_{n+1}^N)\|^2 \label{eq:Deltad}
\end{equation}
is the change in $\|\y-\A \w\|^2$ when $w_n=a$ is replaced by $b$. The maximum change in the energy within an iteration of Algorithm~\ref{alg:MCMC} is then bounded by
\begin{equation}
\Delta_q = \max_{1\le n \le N} \max_{\w \in (\breplevels)^N} \max_{a,b \in \replevels_F} \l|N\Delta H_q(\w,n,b,a)+c_4\Delta d(\w,n,b,a)\r|.
\label{eq:Deltaq}
\end{equation}
Note that $\x$ is assumed bounded (cf.\ Condition~\ref{cond:tech1}) so that (\ref{eq:Deltad}--\ref{eq:Deltaq}) are bounded as well.

\begin{algorithm}[!t]
\caption{Basic MCMC for universal CS -- Fixed alphabet} \label{alg:MCMC}
\begin{algorithmic}[1]
\State {\bf Inputs}: Initial estimate $\w$, reproduction alphabet $\replevels_F$, noise variance $\sigma_Z^2$, number of super-iterations $r$, temperature constant $c>1$, and context depth $q$
\State Compute $n_q(\w,\alpha)[\beta],~\forall~\alpha \in (\replevels_F)^q$, $\beta \in \replevels_F$
\For{$t=1$ to $r$} \Comment{super-iteration}
\State $s \leftarrow \ln(t)/(cN\Delta_q)$ \Comment{$s=s_t$, cf.~(\ref{eq:st})}
\State Draw permutation $\{1,\cdots,N\}$ at random
\For{$t'=1$ to $N$} \Comment{iteration}
\State Let~$n$ be component $t'$ in permutation
\For{all $\beta$ in $\replevels_F$} \Comment{possible new $w_n$}
\State Compute $\Delta H_q(\w,n,\beta,w_n)$ \label{sudo:fixedR_DeltaH}
\State Compute $\Delta d(\w,n,\beta,w_n)$  \label{sudo:fixedR_Deltad}
\State Compute $\mathbb{P}_s(w_n=\beta|\w^{\backslash n})$   \label{sudo:fixedR:fs}
\EndFor
\State Generate $w_n$ using $\mathbb{P}_s(\cdot|\w^{\backslash n})$ \Comment{Gibbs}
\State Update $n_q(\w,\alpha)[\beta],~\forall~\alpha \in (\replevels_F)^q$, $\beta \in \replevels_F$ \label{sudo:fixedR:update}
\EndFor
\EndFor
\State {\bf Output:}\ Return approximation $\w$ of $\x^U_{MAP}$
\end{algorithmic}
\end{algorithm}

In MCMC, the space $\w\in(\replevels_F)^N$
is analogous to a thermodynamic system,
and at low temperatures the system tends toward low energies. Therefore, during the execution of the algorithm, we set a sequence of decreasing temperatures that takes into account the maximum change given in (\ref{eq:Deltaq}):
\begin{align}
s_t \triangleq \ln(t+r_0)/(cN\Delta_q)~\textrm{for some}~c>1,
\label{eq:st}
\end{align}where $r_0$ is a temperature offset. At low temperatures, i.e., large $s_t$, a small difference in energy
$\Psi^{H_q}(\w)$ drives a big difference in probability, cf.~(\ref{eq:def_Boltzmann}). Therefore, we begin at a high
temperature where the Gibbs sampler can freely move around $(\replevels_F)^N$. As the
temperature is reduced, the PMF becomes more sensitive to changes in energy
(\ref{eq:def_Boltzmann}), and the trend toward $\w$ with lower energy grows stronger.
In each iteration, the Gibbs sampler modifies $w_n$ in a
random manner that resembles heat bath concepts in thermodynamics. Although
MCMC could sink into a local minimum, Geman and Geman~\cite{Geman1984} proved that if we decrease the temperature according to (\ref{eq:st}), then the randomness of Gibbs sampling will eventually drive MCMC out of the locally
minimum energy and it will converge to the globally optimal energy w.r.t.\ $\x_{MAP}^U$. Note that Geman and Geman proved that MCMC will converge, although the proof states that it will take infinitely long to do so. In order to help B-MCMC approach the global minimum with reasonable runtime, we will refine B-MCMC in Section~\ref{sec:adaptive}.

The following theorem is proven in Appendix~\ref{ap:th:conv}, following the framework established by Jalali and Weissman~\cite{Jalali2008,Jalali2012}.
\begin{myTheorem}
Let $X$ be a stationary ergodic source that obeys Condition~\ref{cond:tech1}. Then the outcome $\w^r$ of
Algorithm~\ref{alg:MCMC}
in the limit of an infinite number of super-iterations $r$ obeys
\begin{equation*}
\lim_{r \to \infty} \Psi^{H_q}(\w^r) = \min_{\widetilde{\w} \in (\replevels_F)^N} \Psi^{H_q}(\widetilde{\w}) = \Psi^{H_q}\left(\x_{MAP}^{H_q}\right).
\end{equation*}
\label{th:conv}
\end{myTheorem}

Theorem~\ref{th:conv} shows that Algorithm~\ref{alg:MCMC} matches the
best-possible performance of the universal MAP estimator as measured by the objective function $\Psi^{H_q}$, which should yield an MSE that is twice the MMSE (cf.\ Conjecture~\ref{conj:double_MMSE}). We want to remind the reader that Theorem~\ref{th:conv} is based on the stationarity and ergodicity of the source, which could have memory.
To gain some insight about the convergence process of MCMC, we focus on a fixed arbitrary sub-optimal sequence $\w\in(\replevels_F)^N$. Suppose that at super-iteration $t$ the energy for the algorithm's output $\Psi^{H_q}(\w)$ has converged to the steady state (see Appendix~\ref{ap:th:conv} for details on convergence). We can then focus on the probability ratio
$\displaystyle\rho_t=\mathbb{P}_{s_t}(\w)/\mathbb{P}_{s_t}\l(\x^{H_q}_{MAP}\r)$;
$\rho_t<1$ because $\x^{H_q}_{MAP}$ is the global minimum and has the largest Boltzmann probability over all $\w\in(\replevels_F)^N$, whereas $\w$ is sub-optimal. We then consider the same sequence $\w$ at super-iteration $t^2$; the inverse temperature is $2s_t$ and the corresponding ratio at super-iteration $t^2$ is (cf.~(\ref{eq:def_Boltzmann}))
\begin{equation*}
\frac{\mathbb{P}_{2s_t}(\w)}{\mathbb{P}_{2s_t}\l(\x^{H_q}_{MAP}\r)} = \frac{\text{exp}\l(-2s_t
\Psi^{H_q}(\w)\r)}{\text{exp}\l(-2s_t \Psi^{H_q}\l(\x^{H_q}_{MAP}\r)\r)} =
\left(  \frac{ \mathbb{P}_{s_t}(\w)}{\mathbb{P}_{s_t}\l(\x^{H_q}_{MAP}\r)}\right)^2.
\end{equation*}
That is, between super-iterations $t$ and $t^2$ the probability ratio $\rho_t$ is also squared, and the Gibbs sampler is less likely to generate samples whose energy differs significantly from the minimum energy w.r.t.\
$\x^{H_q}_{MAP}$. We infer from this argument that the probability concentration of our algorithm around the globally optimal energy w.r.t.\ $\x^{H_q}_{MAP}$ is linear in the number of super-iterations.

\subsection{Computational challenges}\label{sec:B-MCMC_complexity}

Studying the pseudocode of Algorithm~\ref{alg:MCMC}, we recognize that
Lines~\ref{sudo:fixedR_DeltaH}--\ref{sudo:fixedR:fs} must be implemented efficiently,
as they run $rN|\replevels_F|$ times. Lines~\ref{sudo:fixedR_DeltaH}
and~\ref{sudo:fixedR_Deltad} are especially challenging.

For Line~\ref{sudo:fixedR_DeltaH}, a naive update of $H_q(\w)$ has
complexity $O(|\replevels_F|^{q+1})$, cf.~(\ref{eq:def:H}). To address this problem,
Jalali and Weissman~\cite{Jalali2008,Jalali2012} recompute the empirical conditional
entropy in $O(q|\replevels_F|)$ time only for the $O(q)$ contexts whose corresponding
counts are modified~\cite{Jalali2008,Jalali2012}. The same approach can be used in
Line~\ref{sudo:fixedR:update}, again reducing computation from $O(|\replevels_F|^{q+1})$ to $O(q|\replevels_F|)$.
Some straightforward algebra allows us to convert Line~\ref{sudo:fixedR_Deltad} to a form that requires aggregate runtime of $O(Nr(M+|\replevels_F|))$.
Combined with the
computation for Line~\ref{sudo:fixedR_DeltaH}, and since $M \gg q |\replevels_F|^2$ (because $|\replevels_F|=\gamma^2, \gamma=\lceil\ln(N)\rceil, q=o(\log(N))$, and $M=O(N)$)
in practice, the entire runtime of our algorithm is $O(rMN)$.

The practical value of Algorithm~\ref{alg:MCMC} may be reduced due to its high
computational cost, dictated by the number of super-iterations $r$ required for convergence
to $\x^{H_q}_{MAP}$ and the large size of the reproduction alphabet. Nonetheless, Algorithm~\ref{alg:MCMC} provides a starting point toward further performance gains of more practical algorithms for computing $\x^{H_q}_{MAP}$, which are presented in Section~\ref{sec:adaptive}. Furthermore, our experiments in Section~\ref{sec:numerical} will show that the performance of the algorithm of Section~\ref{sec:adaptive} is comparable to and in many cases better than existing algorithms.

\section{Adaptive Reproduction Alphabet}
\label{sec:adaptive}
While Algorithm~\ref{alg:MCMC} is a first step toward universal signal estimation in CS, $N$ must be large enough to ensure that $\replevels_F$
quantizes a broad enough range of values of $\mathbb{R}$ finely enough to represent the
estimate $\x^{H_q}_{MAP}$ well. For large $N$, the estimation performance using the
reproduction alphabet~(\ref{eq:def:replevels2}) could suffer from high computational complexity.
On the other hand, for small $N$ the number of reproduction levels employed is insufficient to obtain acceptable performance. Nevertheless, using an excessive number of levels will slow down the convergence. Therefore, in this section, we explore techniques that tailor the reproduction alphabet adaptively to the signal being observed.

\subsection{Adaptivity in reproduction levels}\label{sec:L-MCMC}
To estimate better with finite $N$, we utilize reproduction levels that are
{\em adaptive} instead of the fixed levels in $\replevels_F$. To do so, instead of
$\w \in (\replevels_F)^N$, we optimize over a sequence $\u \in \Z^N$, where $|\Z| < |\breplevels|$ and $|\cdot|$ denotes the size.
The new reproduction alphabet $\Z$ does
not directly correspond to real numbers. Instead, there is an adaptive mapping
$\map: \Z \rightarrow \mathbb{R}$, and the reproduction levels are $\map(\Z)$. Therefore, we call $\Z$ the {\em adaptive} reproduction alphabet. Since the mapping $\map$ is one-to-one, we also refer to $\Z$ as reproduction levels.
Considering the energy function (\ref{eq:MCMC_energy}),
we now compute the empirical symbol counts $n_q(\u,\alpha)[\beta]$, order $q$
conditional empirical probabilities $\mathbb{P}_q(\u,\alpha)[\beta]$, and order $q$
conditional empirical entropy $H_q(\u)$ using $\u \in \Z^N$, $\alpha \in \Z^q$, and
$\beta \in \Z$, cf.~(\ref{eq:nq}), (\ref{eq:def:Pcond}), and~(\ref{eq:def:H}). Similarly, we
use $\|\y - \A \map(\u)\|^2$ instead of $\|\y - \A \w\|^2$, where $\map(\u)$ is the
straightforward vector extension of $\map$. These modifications yield an adaptive
energy function $\displaystyle \Psi^{H_q}_a(\u) \triangleq NH_q(\u) + c_4  \|\y - \A \map(\u)\|^2$.

We choose $\map_{opt}$ to optimize for minimum
squared error,
\begin{equation*}
\map_{opt} \triangleq \arg \min_{\map}\|\y-\A \map(\u)\|^2 = \arg \min_{\map}\left[\sum_{m=1}^M(y_m-[\A \map(\u)]_m)^2\right],
\end{equation*}
where $[\A \map(\u)]_m$ denotes the $m$-th entry of the vector $\A \map(\u)$.
The optimal mapping depends entirely on $\y$, $\A$, and $\u$. From a coding
perspective, describing $\map_{opt}(\u)$ requires $H_q(\u)$ bits for $\u$ and
$|\Z| b\log \log(N)$ bits for $\map_{opt}$ to match the resolution of the non-adaptive $\replevels_F$, with $b > 1$ an arbitrary constant~\cite{BaronWeissman2012}. The resulting coding length
defines our universal prior.

\textbf{Optimization of reproduction levels: }
We now describe the optimization procedure for $\map_{opt}$, which must be
computationally efficient. Write
\begin{align*}
\Upsilon(\map) \triangleq \|\y-\A \map(\u)\|^2 = \sum_{m=1}^M\left(y_m-\sum_{n=1}^N A_{mn}\map(u_n)\right)^2,
\end{align*}
where $A_{mn}$ is the entry of $\A$ at the $m$-th row, $n$-th column.
For $\Upsilon(\map)$ to be minimum, we need zero-valued derivatives as follows,
\begin{equation*}
\frac{d\Upsilon(\map)}{d\map(\beta)} = -2 \sum_{m=1}^M \left(y_m - \sum_{n=1}^N A_{mn} \map(u_n) \right)
\left( \sum_{n=1}^N A_{mn} \mathbbm{1}_{u_n=\beta} \right) = 0,~\forall~\beta \in \Z, \label{eq:derzero}
\end{equation*}
where the indicator function $\mathbbm{1}_{A}$ is 1 if the condition $A$ is met, else 0.
Define the location sets
$\displaystyle \N_\beta \triangleq \{ n: 1\le n \le N, u_n = \beta\}$
for each $\beta \in \Z$, and rewrite the derivatives of $\Upsilon(\map)$,
\begin{align}
\frac{d\Upsilon(\map)}{d\map(\beta)} =
-2 \sum_{m=1}^M \left( y_m - \sum_{\lambda \in \Z} \sum_{n \in \N_{\lambda}} A_{mn} \map(\lambda) \right)
\left( \sum_{n\in \N_\beta} A_{mn} \right).
\label{eq:derivative2}
\end{align}
Let the per-character sum column values be
\begin{equation}
\label{eq:mu_def}
\mu_{m\beta} \triangleq \sum_{n \in \N_{\beta}} A_{mn},
\end{equation}
for each $m\in\{1,\cdots,M\}$ and $\beta\in\Z$.
We desire the derivatives to be zero, cf.~(\ref{eq:derivative2}):
\begin{align*}
0 = \sum_{m=1}^M \left(y_m-\sum_{\lambda \in \Z}\map(\lambda) \mu_{m\lambda}\right) \mu_{m\beta}.
\end{align*}
Thus, the system of equations must be satisfied,
\begin{align}\label{eq:system_equation}
\sum_{m=1}^M y_m \mu_{m\beta}
=
\sum_{m=1}^M \left(\sum_{\lambda \in \Z} \map(\lambda) \mu_{m\lambda}\right) \mu_{m\beta}
\end{align}
for each $\beta \in \Z$. Consider now the right hand side,
\begin{align*}
\sum_{m=1}^M \left(\sum_{\lambda \in \alphabet} \map(\lambda) \mu_{m\lambda}\right)
\mu_{m\beta}
= \sum_{\lambda \in \alphabet} \map(\lambda) \sum_{m=1}^M \mu_{m\lambda} \mu_{m\beta},
\end{align*}
for each $\beta \in \Z$. The system of equations can be described in matrix form as follows,
\begin{equation*}
\label{eq:matrix_map_opt}
\overbrace{
\left[\begin{array}{ccc}
\sum_{m=1}^M \mu_{m\beta_1} \mu_{m\beta_1} & \cdots & \sum_{m=1}^M \mu_{m\beta_{|\alphabet|}} \mu_{m\beta_1}\\
\vdots & \ddots & \vdots\\
\sum_{m=1}^M \mu_{m\beta_1} \mu_{m\beta_{|\alphabet|}} & \cdots & \sum_{m=1}^M \mu_{m\beta_{|\alphabet|}} \mu_{m\beta_{|\alphabet|}}\\
\end{array}
\right]
}^{\Omega}
\overbrace{
\left[\begin{array}{c}
\map(\beta_1)\\
\vdots\\
\map(\beta_{|\alphabet|})
\end{array}\right]}
^{\map(\Z)}=
\overbrace{
\left[\begin{array}{c}
\sum_{m=1}^M y_m \mu_{m\beta_1}\\
\vdots\\
\sum_{m=1}^M y_m \mu_{m\beta_{|\alphabet|}}
\end{array}
\right].}
^{\Theta}
\end{equation*}
Note that by writing $\mu$ as a matrix with entries indexed by row $m$ and
column $\beta$ given by (\ref{eq:mu_def}), we can write $\Omega$ as a Gram
matrix, $\Omega=\mu^{\top} \mu$, and we also have $\Theta = \mu^{\top} \y$, cf.~(\ref{eq:system_equation}).
The optimal $\map$ can be computed as a $|\Z|\times 1$ vector
$\displaystyle \map_{opt} = \Omega^{-1} \Theta = (\mu^{\top}\mu)^{-1}\mu^{\top}\y$ if
$\Omega \in \mathbb{R}^{|\Z|\times |\Z|}$ is invertible. We note in passing that numerical
stability can be improved by regularizing $\Omega$. Note also that
\begin{equation}
\label{eqn:compute_ell2}
\|\y-\A \map(u)\|^2  =
\sum_{m=1}^M \left(
y_m - \sum_{\beta \in \Z} \mu_{m\beta} \map_{opt}(\beta)
\right)^2,
\end{equation}
which can be computed in $O(M|\Z|)$ time instead of $O(MN)$.

\begin{algorithm}[!t]
\caption{Level-adaptive MCMC} \label{alg:MCMCAL}
\begin{algorithmic}[1]
\State *{\bf Inputs}: Initial mapping $\map$, sequence $\u$, adaptive alphabet $\Z$, noise variance $\sigma_Z^2$, number of super-iterations $r$, temperature constant $c>1$, context depth $q$, and temperature offset $r_0$
\State Compute $n_q(\u,\alpha)[\beta],~\forall~\alpha \in \Z^q$, $\beta \in \Z$
\State *Initialize $\Omega$\label{Algo2:init_omega}
\For{$t=1$ to $r$} \Comment{super-iteration}
\State $s \leftarrow \ln(t+r_0)/(cN\Delta_q)$ \Comment{$s=s_t$, cf.~(\ref{eq:st})}
\State Draw permutation $\{1,\cdots,N\}$ at random
\For{$t'=1$ to $N$} \Comment{iteration}
\State Let~$n$ be component $t'$ in permutation
\For{all $\beta$ in $\Z$} \Comment{possible new $u_n$}
\State Compute $\Delta H_q(\u,n,\beta,u_n)$ \label{Algo2:DeltaH}
\State *Compute $\mu_{m\beta},\forall\ m \in\{1,\cdots,M\}$ \label{Algo2:mu}
\State *Update $\Omega$ \label{Algo2:Omega} \Comment{$O(1)$ rows and columns}
\State *Compute $\map_{opt}$ \Comment{invert $\Omega$} \label{Algo2:map_opt}
\State Compute $\|\y-\A\map(\u_1^{n-1}\beta \u_{n+1}^N)\|^2$ \label{Algo2:ell2}
\State Compute $\mathbb{P}_s(u_n=\beta|\u^{\backslash n})$ \label{Algo2:distribution}
\EndFor
\State *$\widetilde{u}_n \leftarrow u_n$  \Comment{save previous value}\label{Algo2:save_sequence}
\State Generate $u_n$ using $\mathbb{P}_s(\cdot|\u^{\backslash n})$ \Comment{Gibbs}
\State Update $n_q(\cdot)[\cdot]$ at $O(q)$ relevant locations
\State *Update $\mu_{m\beta},~\forall~m$, $\beta\in \{u_n,\widetilde{u}_n\}$\label{Algo2:update_mu}
\State *Update $\Omega$ \Comment{$O(1)$ rows and columns}\label{Algo2:update_omega}
\EndFor
\EndFor
\State *{\bf Outputs}: Return approximation $\map(\u)$ of $\x^U_{MAP},~\Z$, and temperature offset $r_0+r$
\end{algorithmic}
\end{algorithm}

\textbf{Computational complexity: }
Pseudocode for level-adaptive MCMC (L-MCMC) appears in Algorithm~\ref{alg:MCMCAL}, which resembles Algorithm~\ref{alg:MCMC}.
The initial mapping $\map$ is inherited from a quantization of the initial point $\x^*$, $r_0=0$ ($r_0$ takes different values in Section~\ref{sec:adaptive_size}), and other minor differences between B-MCMC and L-MCMC appear in lines marked by asterisks.

We discuss computational requirements for each line of the pseudocode that is run
within the inner loop.
\begin{itemize}
\item
Line~\ref{Algo2:DeltaH} can be computed in $O(q|\Z|)$ time (see discussion of Line~\ref{sudo:fixedR_DeltaH} of B-MCMC in Section~\ref{sec:B-MCMC_complexity}).
\item
Line~\ref{Algo2:mu} updates $\mu_{m\beta}$ for $m\in\{1,\cdots,M\}$ in $O(M)$ time.
\item
Line~\ref{Algo2:Omega}
updates $\Omega$. Because we only need to update
$O(1)$ columns and $O(1)$ rows, each such column and
row contains $O(|\Z|)$ entries, and each entry is a sum over $O(M)$
terms, we need $O(M|\Z|)$ time.
\item
Line~\ref{Algo2:map_opt}
requires inverting $\Omega$ in $O(|\Z|^3)$ time.
\item
Line~\ref{Algo2:ell2}
requires $O(M|\Z|)$ time, cf.~(\ref{eqn:compute_ell2}).
\item
 Line~\ref{Algo2:distribution} requires $O(|\Z|)$ time.
\end{itemize}
In practice we typically have $M \gg |\Z|^2$,
and so the aggregate complexity is $O(rMN|\Z|)$, which is
greater than the computational complexity of Algorithm~\ref{alg:MCMC} by a factor of $O(|\Z|)$.

\subsection{Adaptivity in reproduction alphabet size}\label{sec:adaptive_size}

While Algorithm~\ref{alg:MCMCAL} adaptively maps $\u$ to $\mathbb{R}^N$, the signal estimation quality heavily depends on $|\Z|$. Denote the true alphabet of the signal by $\X,~\x\in {\X}^N$; if the signal is continuous-valued, then $|\X|$ is infinite. Ideally we want to
employ as many levels as the runtime allows for continuous-valued signals,
whereas for discrete-valued signals we want $|\Z|=|\X|$.
Inspired by this observation, we propose to
begin with some initial $|\Z|$, and then adaptively adjust $|\Z|$ hoping to match $|\X|$. Hence, we propose the size- and level-adaptive MCMC algorithm (Algorithm~5.3), which invokes L-MCMC (Algorithm~\ref{alg:MCMCAL}) several times.

\textbf{Three basic procedures: }
In order to describe the size- and level-adaptive MCMC (SLA-MCMC) algorithm in detail, we introduce three alphabet adaptation procedures as follows.
\begin{itemize}
  \item {\em MERGE}: First, find the closest adjacent levels $\beta_1,\beta_2\in \Z$. Create a new level $\beta_3$ and add it to $\Z$. Let $\map(\beta_3)=(\map(\beta_1)+\map(\beta_2))/2$.
  Replace $u_i$ by $\beta_3$ whenever $u_i\in\{\beta_1,\beta_2\}$. Next, remove $\beta_1$ and $\beta_2$ from $\Z$.
\item {\em ADD-out}: Define the range $R_{\map}=[\min\map(\Z),$ $\max\map(\Z)]$, and $\mathcal{I}_{R_{\map}}=\max\map(\Z)-\min\map(\Z)$. Add a {\em lower } level $\beta_3$ and/or {\em upper level} $\beta_4$ to $\Z$ with
    \begin{eqnarray*}
    \map(\beta_3)=\min\map(\Z)-\frac{\mathcal{I}_{R_{\map}}}{|\Z|-1},\\ \map(\beta_4)=\max\map(\Z)+\frac{\mathcal{I}_{R_{\map}}}{|\Z|-1}.
    \end{eqnarray*}Note that $\l|\{u_i: u_i=\beta_3 \mbox{ or } \beta_4, i=1,\cdots,N\}\r|=0$, i.e., the new levels are empty.
\item {\em ADD-in}: First, find the most distant adjacent levels, $\beta_1$ and $\beta_2$. Then, add a level $\beta_3$ to $\Z$ with $\map(\beta_3)=(\map(\beta_1)+\map(\beta_2))/2$.
    For $i\in\{1,\cdots,|\Z|\}$ s.t. $u_i=\beta_1$, replace $u_i$
    by $\beta_3$ with probability
    \begin{equation*}
    \frac{\mathbb{P}_s(u_i=\beta_2)}{\mathbb{P}_s(u_i=\beta_1)+\mathbb{P}_s(u_i=\beta_2)},
    \end{equation*}where $\mathbb{P}_s(\cdot)$ is given in~(\ref{eq:def_Boltzmann}); for $i\in\{1,\cdots,|\Z|\}$ s.t. $u_i=\beta_2$, replace $u_i$ by $\beta_3$ with probability
    \begin{equation*}
    \frac{\mathbb{P}_s(u_i=\beta_1)}{\mathbb{P}_s(u_i=\beta_1)+\mathbb{P}_s(u_i=\beta_2)}.
    \end{equation*}Note that $\l|\{u_i: u_i=\beta_3, i=1,\cdots,N\}\r|$ is typically non-zero, i.e., $\beta_3$ tends not to be empty.
\end{itemize}
We call the process of running one of these procedures followed by running L-MCMC a {\em round}.

\begin{figure*}[t]
\begin{center}
\includegraphics[width=\textwidth]{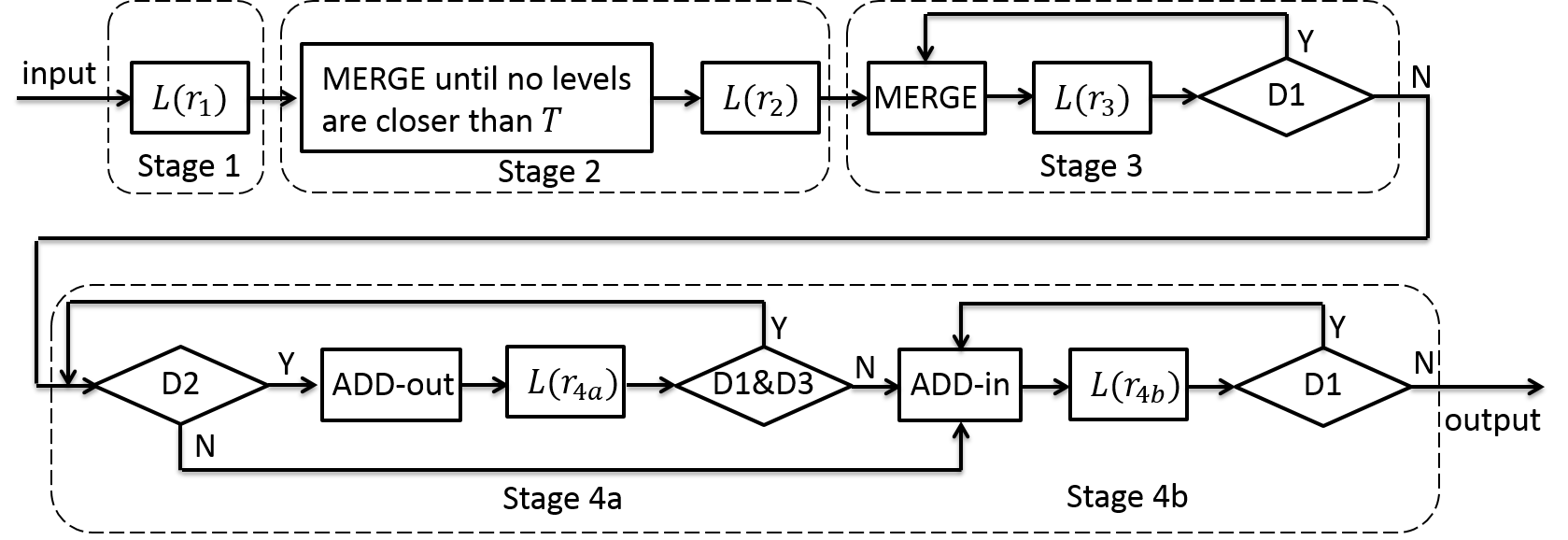}
\end{center}
\caption{Flowchart of Algorithm~5.3 (size- and level-adaptive MCMC).
L($r$) denotes running L-MCMC for $r$ super-iterations. The parameters $r_1$,$r_2$,$r_3$,$r_{4a}$, and $r_{4b}$ are the number of super-iterations used in Stages~1 through~4, respectively. Criteria $D1-D3$ are described in the text.}\label{fig:SLA-MCMC}
\end{figure*}

\textbf{Size- and level-adaptive MCMC: }
SLA-MCMC is conceptually illustrated in the flowchart in Figure~\ref{fig:SLA-MCMC}. It has four stages, and in each stage we will run L-MCMC for several super-iterations; we denote the execution of L-MCMC for $r$ super-iterations by L($r$). The parameters $r_1,r_2,r_3,r_{4a}$, and $r_{4b}$ are the number of super-iterations used in Stages~1 through~4, respectively. The choice of these parameters reflects a trade-off between runtime and estimation quality.

In Stage~1, SLA-MCMC uses a fixed-size adaptive reproduction alphabet $\Z$ to tentatively estimate the signal.
The initial point of Stage~1 is obtained in the same way as L-MCMC.
After Stage~1, the initial point and temperature offset for each instance of L-MCMC correspond to the respective outputs of the previous instance of L-MCMC.
If the source is discrete-valued and $|\Z|>|\X|$ in Stage~1, then multiple levels in the output $\Z$ of Stage~1 may correspond to a single level in $\X$.
To alleviate this problem,  in Stage~2 we merge levels closer than $T=\mathcal{I}_{R_{\map}}/\left(K_1\times(|\Z|-1)\right)$, where $K_1$ is a parameter.

However, $|\Z|$ might still be larger than needed; hence in Stage~3 we tentatively merge the closest adjacent levels. The criterion $D1$ evaluates whether the current objective function is lower (better) than in the previous round;
we do not leave Stage~3 until $D1$ is violated. Note that if $|\X|>|\Z|$ (this always holds for continuous-valued signals), then ideally SLA-MCMC should not merge any levels in Stage~3, because the objective function would increase if we merge any levels.

Define the outlier set $S=\{x_i: x_i\notin R_{\map}, i=1,\cdots,N\}$. Under Condition~\ref{cond:tech1}, $S$ might be small or even empty.
When $S$ is small, L-MCMC might not assign levels to represent the entries of $S$.
To make SLA-MCMC more robust to outliers, in Stage~4a we add empty levels outside the range $R_{\map}$ and then allow L-MCMC to change entries of $\u$ to the new levels during Gibbs sampling; we call this
{\em populating} the new levels.
If a newly added outside level is not populated, then we remove it from $\Z$. Seeing that the optimal mapping $\map_{opt}$ in L-MCMC tends not to map symbols to levels with low population, we
consider a criterion $D2$ where we
will add an outside upper (lower) level if the population of the current upper (lower) level is smaller than $N/(K_2|\Z|)$, where $K_2$ is a parameter.
That is, the criterion $D2$ is violated if both populations of the current upper and lower levels are sufficient (at least $N/(K_2|\Z|)$); in this case we do not need to add outside levels because $\map_{opt}$ will map some of the current levels to represent the entries in $S$.
The criterion $D3$ is violated if all levels added outside are not populated by the end of the round. SLA-MCMC keeps adding levels outside $R_{\map}$ until it is wide enough to cover most of the entries of $\x$.

Next, SLA-MCMC considers adding levels inside $R_{\map}$ (Stage~4b). If the signal is discrete-valued, this stage should stop when $|\Z|=|\X|$.
Else, for continuous-valued signals SLA-MCMC can add levels until the runtime expires.

In practice, SLA-MCMC runs L-MCMC at most a constant number of times, and the computational complexity is in the same order of L-MCMC, i.e., $O(rMN|\Z|)$. On the other hand, SLA-MCMC allows varying
$|\Z|$, which often improves the estimation quality.

\subsection{Mixing}\label{sec:mix}

Donoho proved for the scalar channel setting that $\x_{KS}$ is sampled from the posterior $\mathbb{P}_{X|Y}(\x|\y)$~\cite{DonohoKolmogorov}.
Seeing that the Gibbs sampler used by MCMC (cf.\ Section~\ref{sec:B-MCMC}) generates random samples, and the
outputs of our algorithm will be different if its random number generator is initialized with different {\em random seeds}, we speculate that running SLA-MCMC several times
will also yield independent samples from the posterior, where we note that
the runtime grows linearly in the number of times that we run SLA-MCMC.
By mixing (averaging over) several outputs of SLA-MCMC, we obtain $\widehat{\x}_{\mbox{avg}}$, which may
have lower squared error w.r.t.\ the true $\x$ than the average squared error obtained by a single SLA-MCMC output.
Numerical results suggest that mixing indeed reduces the MSE (cf.\ Figure~\ref{fig:MUnif_algos});
this observation suggests that mixing the outputs of multiple algorithms, including running a
random signal estimation algorithm several times, may reduce the squared error.

\section{Numerical Results} \label{sec:numerical}

In this section, we demonstrate that SLA-MCMC is comparable and in many cases better than existing algorithms in estimation quality,
and that SLA-MCMC is applicable when $M>N$. Additionally, some numerical evidence is provided to justify Conjecture~\ref{conj:double_MMSE} in Section~\ref{sec:conjecture}.
Then, the advantage of SLA-MCMC in estimating low-complexity signals is demonstrated. Finally, we compare B-MCMC, L-MCMC, and SLA-MCMC performance.

We implemented SLA-MCMC in Matlab\footnote{A toolbox that runs the simulations in this chapter is available at http://people.engr.ncsu.edu/dzbaron/software/UCS
\_BaronDuarte/} and tested it using several stationary ergodic sources.
Except when noted, for each source, signals $\x$ of length $N=10000$ were generated.
Each such $\x$ was multiplied by a Gaussian random matrix $\A$ with normalized columns and corrupted by i.i.d. Gaussian measurement noise $\z$.
Except when noted, the number of measurements $M$ varied between 2000 and 7000.
The noise variance $\sigma_Z^2$ was selected to ensure that the signal-to-noise ratio (SNR) was $5$ or $10$~dB; SNR was defined as
$\displaystyle \mbox{SNR}=10\log_{10}\left[(N \mathbb{E}[x^2])/(M\sigma_Z^2)\right]$.
According to Section~\ref{sec:compressor}, the context depth $q=o(\log(N))$, where the base of the logarithm is the alphabet size; using typical values such as $N=10000$ and $|\Z|=10$, we have $\log(N)=4$ and set $q=2$. While larger $q$ will slow down the algorithm, it might be necessary to increase $q$ when $N$ is larger. The numbers of super-iterations in different stages of SLA-MCMC are $r_{1}=50$ and $r_2=r_3=r_{4a}=r_{4b}=10$, the maximum total number of super-iterations is set to $240$, the initial number of levels is $|\Z|=7$, and the tuning parameters from Section~\ref{sec:adaptive_size} are $K_1,K_2=10$; these parameters seem to work well on an extensive set of numerical experiments.
SLA-MCMC was not given the true alphabet $\X$ for any of the sources presented in this chapter; our expectation is that it should adaptively adjust $|\Z|$ to match $|\X|$.
The final estimate $\widehat{\x}_{\mbox{avg}}$ of each signal was obtained by averaging over the outputs $\widehat{\x}$ of $5$
runs of SLA-MCMC, where in each run we initialized the random number generator with another random seed, cf.\ Section~\ref{sec:mix}. These choices of parameters seemed to provide a reasonable compromise between runtime and estimation quality.

We chose our performance metric as the mean signal-to-distortion ratio (MSDR) defined as
$\displaystyle \mbox{MSDR}=10\log_{10}\left(\mathbb{E}[x^2]/\mbox{MSE}\right)$. For each $M$ and SNR, the MSE was obtained after averaging over the squared errors of $\widehat{\x}_{\mbox{avg}}$ for 50 draws of $\x$, $\A$, and $\z$.
We compared the performance of SLA-MCMC to
that of ({\em i}) compressive sensing matching pursuit (CoSaMP)~\cite{Cosamp08}, a greedy method;
({\em ii}) gradient projection for sparse reconstruction (GPSR)~\cite{GPSR2007}, an optimization-based method;
({\em iii}) message passing approaches (for each source, we chose best-matched algorithms between EM-GM-AMP-MOS (EGAM for short)~\cite{EMGMTSP} and turboGAMP (tG for short)~\cite{turboGAMP}); and
({\em iv}) Bayesian compressive sensing~\cite{BCS2008} (BCS).
Note that EGAM~\cite{EMGMTSP} places a Gaussian mixture (GM) prior on the signal, and tG~\cite{turboGAMP} builds a prior set including the priors for the signal, the support set of the signal, the channel, and the amplitude structure. Both algorithms learn the parameters of their assumed priors online from the measurements. We compare the computational complexities of the algorithms above in Table~\ref{table:complexity}, where
$L$ bounds the cost of a matrix-vector multiply with $\A$ or the Hermitian transpose of $\A$, and $\epsilon$ is a given precision parameter~\cite{Cosamp08}; $r_P,r_E, r_G, r_M$ are the number of GPSR~\cite{GPSR2007}, expectation maximization (EM), GAMP~\cite{RanganGAMP2011ISIT}, and model selection~\cite{EMGMTSP} iterations, respectively; $T_1$ and $T_2$ are the average complexities for the EM algorithm and the turbo updating scheme~\cite{turboGAMP}.
\begin{table}[t!]
\caption{Computational complexity} 
\centering 
\begin{tabular}{c c} 
\hline\hline 
Algorithms & Complexity\\ [0.5ex] 
\hline 
SLA-MCMC & $O(rMN|\Z|)$    \\ 
CoSaMP & $O(L\log\frac{\|\x\|}{\epsilon})$\\
GPSR & $O(r_P MN)$ \\
EGAM & $O(r_M r_E T_1+r_M r_E r_G MN)$ \\
tG & $O(r_E T_2+r_E r_G MN)$ \\ [1ex] 
\hline 
\end{tabular}
\label{table:complexity} 
\end{table}
Because all these algorithms are iterative algorithms and require different number of iterations to converge or reach a satisfactory estimation quality, we also report their typical runtimes here. Typical runtimes are $1$ hour (for continuous-valued signals) and $15$ minutes (discrete-valued)
per random seed for SLA-MCMC, $30$ minutes for EGAM~\cite{EMGMTSP} and tG~\cite{turboGAMP}, and $10$ minutes for CoSaMP~\cite{Cosamp08} and GPSR~\cite{GPSR2007} on an Intel(R) Core(TM) i7 CPU 860 @ 2.8GHz with 16.0GB RAM running 64 bit Windows 7. The performance of BCS was roughly 5~dB below SLA-MCMC results. Hence, BCS results are not shown in the sequel. We emphasize that algorithms that use training data (such as dictionary learning)~\cite{Ramirez2011,AharoEB_KSVD,Mairal2008,Zhoul2011} will find our problem size $N=10000$ too large, because they need a training set that has more than $N$ signals. On the other hand, SLA-MCMC does not need to train itself on any training set, and hence is advantageous.

Among these baseline algorithms designed for i.i.d. signals, GPSR~\cite{GPSR2007} and EGAM~\cite{EMGMTSP} only need $\y$ and $\A$, and CoSaMP~\cite{Cosamp08} also needs the number of non-zeros in $\x$. Only tG~\cite{turboGAMP} is designed for non-i.i.d. signals; however, it must be aware of the probabilistic model of the source. Finally, GPSR~\cite{GPSR2007} performance was similar to that of CoSaMP~\cite{Cosamp08} for all sources considered in this section, and thus is not plotted.

\begin{figure}[t]
\begin{center}
\includegraphics[width=80mm]{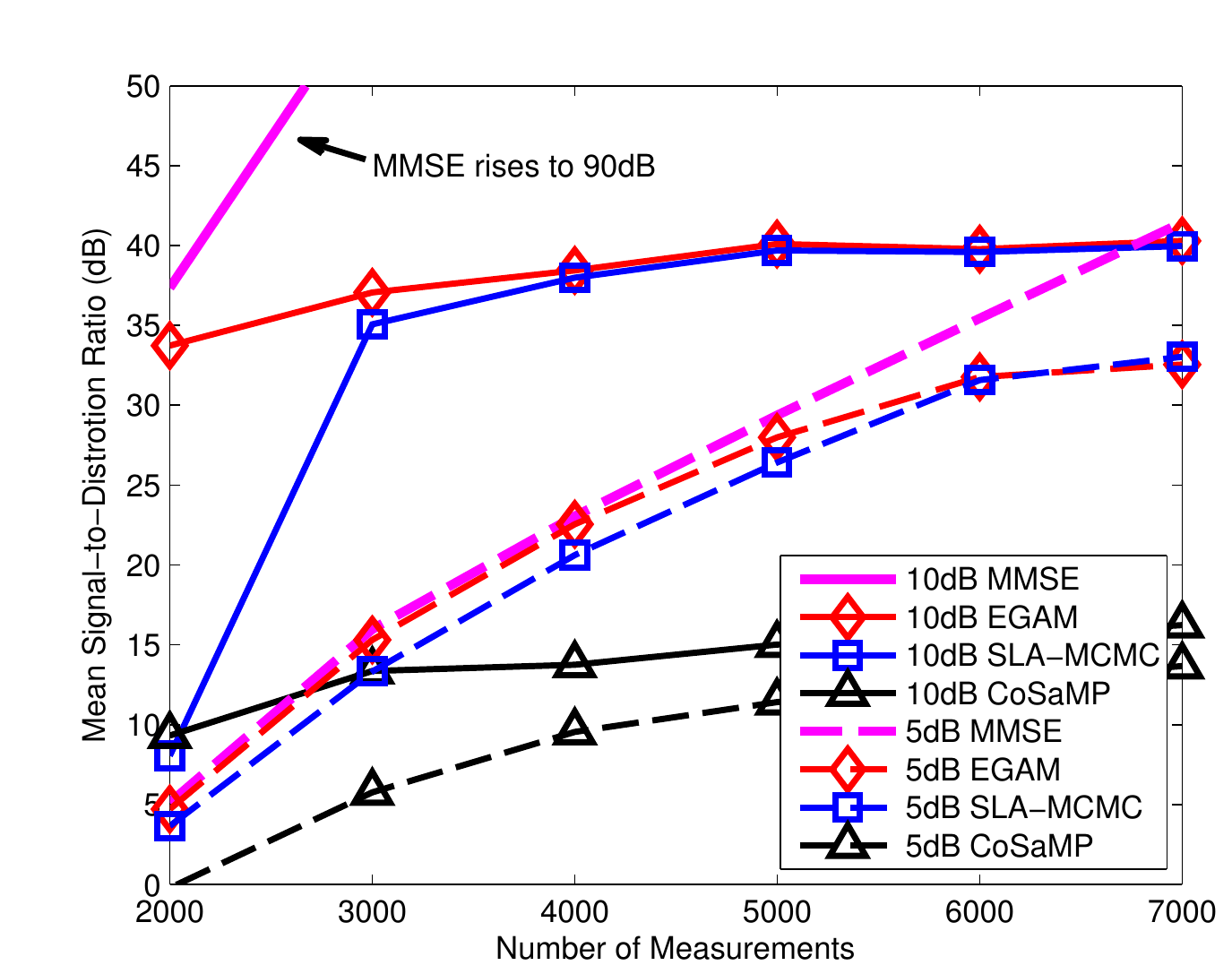}
\end{center}
\caption{SLA-MCMC, EGAM, and CoSaMP estimation results for a source with i.i.d. Bernoulli entries with non-zero probability of $3\%$ as
a function of the number of Gaussian random measurements $M$ for different SNR values ($N=10000$).}\label{fig:Ber}
\end{figure}

\subsection{Performance on discrete-valued sources}
{\bf Bernoulli source:} We first present results for an i.i.d. Bernoulli source. The Bernoulli source followed the distribution $f_X(x)=0.03\delta(x-1)+0.97\delta(x)$, where $\delta(\cdot)$ is the Dirac delta function. Note that SLA-MCMC did not know $\X=\{0,1\}$ and had to estimate it on the fly.
We chose EGAM~\cite{EMGMTSP} for message passing algorithms because it fits the signal with GM's,
which can accurately characterize signals from an i.i.d. Bernoulli source.
The resulting MSDR's for SLA-MCMC, EGAM~\cite{EMGMTSP}, and CoSaMP~\cite{Cosamp08} are plotted in Figure~\ref{fig:Ber}.
We can see that when $\mbox{SNR}=5~\mbox{dB}$, EGAM~\cite{EMGMTSP} approaches the MMSE~\cite{ZhuBaronCISS2013} performance for low to medium $M$; although SLA-MCMC is often worse than EGAM~\cite{EMGMTSP}, it is within $3$~dB of the MMSE performance.
This observation that SLA-MCMC approaches the MMSE for $\mbox{SNR}=5~\mbox{dB}$ partially substantiates Conjecture~\ref{conj:double_MMSE} in Section~\ref{sec:conjecture}. When $\mbox{SNR}=10~\mbox{dB}$, SLA-MCMC is comparable to EGAM~\cite{EMGMTSP} when $M\geq 3000$. CoSaMP~\cite{Cosamp08} has worse MSDR.

\begin{figure}[t]
\begin{center}
\includegraphics[width=80mm]{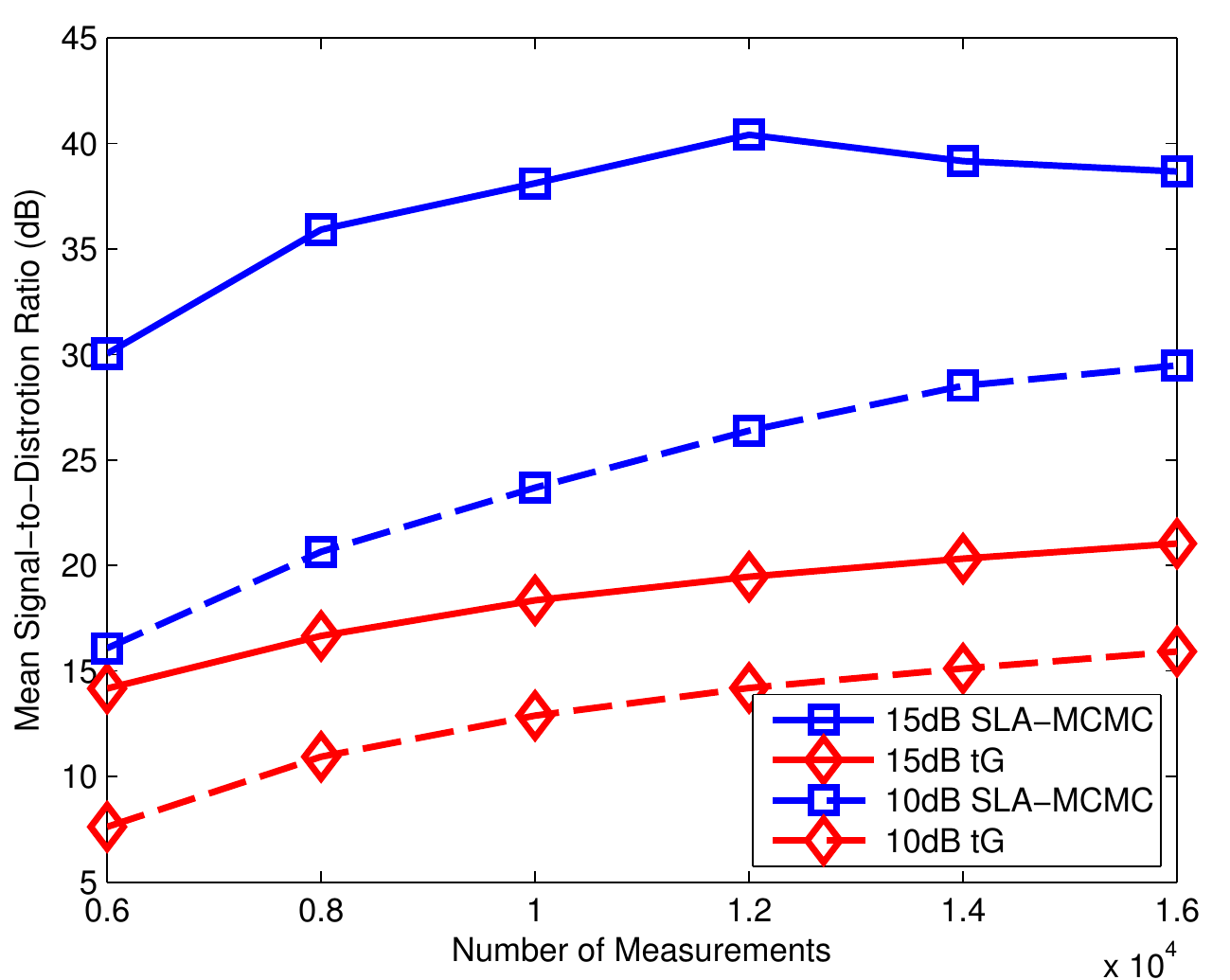}
\end{center}
\caption{SLA-MCMC and tG estimation results for a dense two-state Markov source with non-zero entries drawn from a Rademacher ($\pm 1$) distribution as
a function of the number of Gaussian random measurements $M$ for different SNR values ($N=10000$).}\label{fig:denseMRad}
\end{figure}

{\bf Dense Markov-Rademacher source:}
Considering that most algorithms are designed for i.i.d. sources, we now illustrate the performance of SLA-MCMC on non-i.i.d. sources by simulating a dense Markov-Rademacher (MRad for short) source. The non-zero entries of the dense MRad signal were generated by a two-state Markov state machine (non-zero and zero states). The transition from zero to non-zero state for adjacent entries had probability $\mathbb{P}_{01}=\frac{3}{70}$, while the transition from non-zero to zero state for adjacent entries had probability $\mathbb{P}_{10}=0.10$; these parameters yielded $30\%$ non-zero entries on average. The non-zeros were drawn from a Rademacher distribution, which took values $\pm1$ with equal probability.
With such denser signals, we may need to take more measurements and/or require higher SNR's to achieve similar performance to previous examples.
The number of measurements varied from $6000$ to $16000$, with $\mbox{SNR}=10~\mbox{and}~15~\mbox{dB}$.
Although tG~\cite{turboGAMP} does not provide an option that accurately characterize the MRad source, we still chose
to compare against its performance because it is applicable to non-i.i.d. signals.
The MSDR's for SLA-MCMC and tG~\cite{turboGAMP} are plotted in Figure~\ref{fig:denseMRad}. CoSaMP~\cite{Cosamp08} performs poorly as it is designed for sparse signal estimation, and its results are not shown. Although tG~\cite{turboGAMP} is designed for non-i.i.d. sources, it is nonetheless outperformed by SLA-MCMC.
This example shows that SLA-MCMC estimates non-i.i.d. signals well and is applicable to general linear inverse problems. However, recall that the computational complexity of SLA-MCMC is $O(rMN|\Z|)$. Hence, despite the appealing performance of SLA-MCMC shown in this example, we will suffer from high computational time when we have to apply SLA-MCMC in the case when $M>N$.

\subsection{Performance on continuous sources}\label{sec:bvsub}

We now discuss the performance of SLA-MCMC in estimating continuous sources.

{\bf Sparse Laplace (i.i.d.) source:} For unbounded continuous-valued signals, which do not adhere to Condition~\ref{cond:tech1}, we simulated an i.i.d. sparse Laplace source following the random variable $X=X_B X_L$,
where $X_B\sim \text{Ber}(0.03)$ is a Bernoulli random variable and $X_L$ follows a Laplace distribution with mean zero and variance one.
We chose EGAM~\cite{EMGMTSP} for message passing algorithms because it fits the signal with GM,
which can accurately characterize signals from an i.i.d. sparse Laplace source.
The MSDR's for SLA-MCMC, EGAM~\cite{EMGMTSP}, and CoSaMP~\cite{Cosamp08} are plotted in Figure~\ref{fig:sparseL}.
We can see that EGAM~\cite{EMGMTSP} approaches the MMSE~\cite{ZhuBaronCISS2013} performance in all settings;
SLA-MCMC outperforms CoSaMP~\cite{Cosamp08}, while it is approximately $2$~dB worse than the MMSE.
Recall from Conjecture~\ref{conj:double_MMSE} that we expect to achieve twice the MMSE, which is approximately $3$~dB below the signal-to-distortion ratio of MMSE, and thus SLA-MCMC performance is reasonable.
This example of SLA-MCMC performance approaching the MMSE further substantiates Conjecture~\ref{conj:double_MMSE}.

\begin{figure}[t]
\begin{center}
\includegraphics[width=80mm]{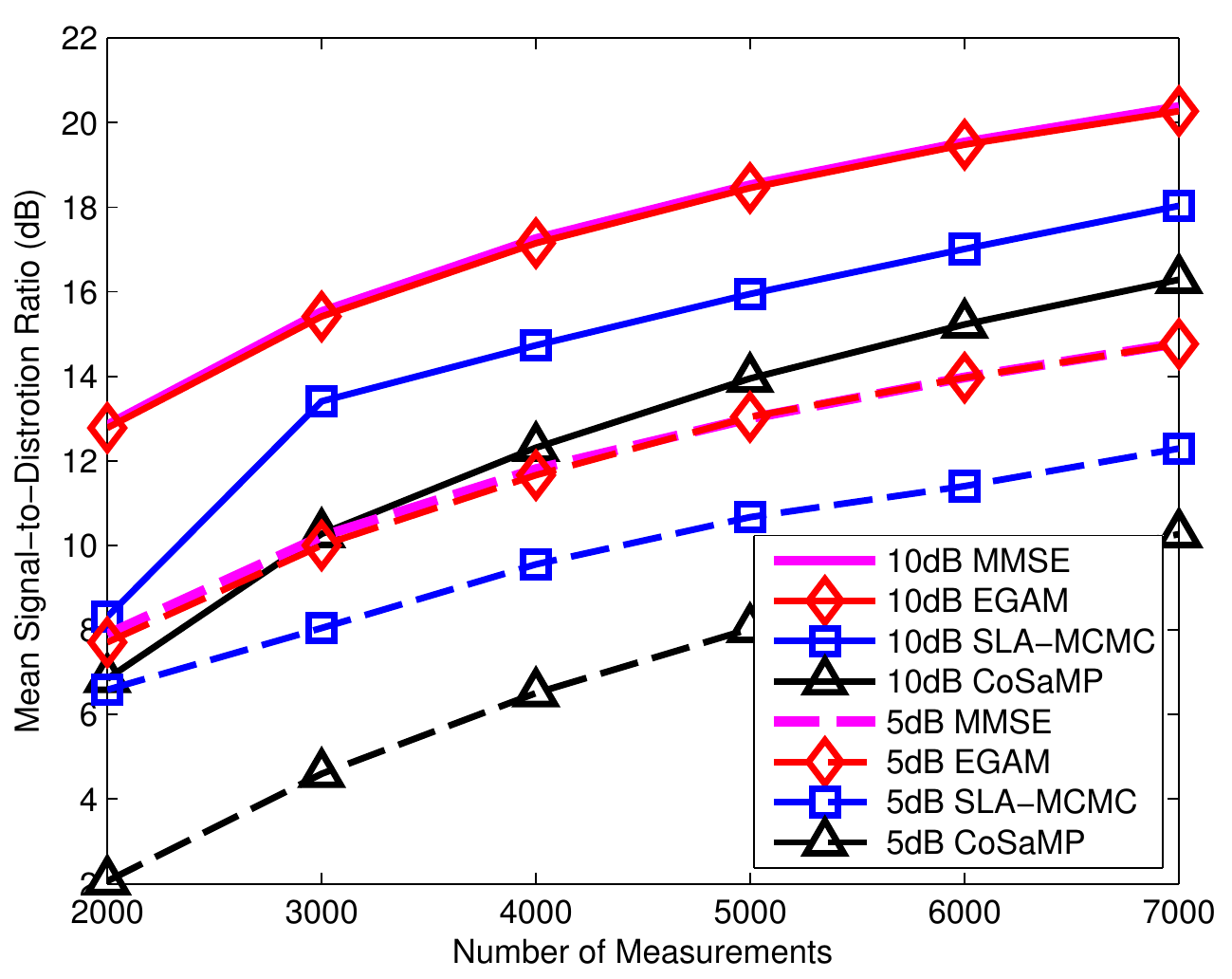}
\end{center}
\caption{SLA-MCMC, EGAM, and CoSaMP estimation results for an i.i.d. sparse Laplace source as
a function of the number of Gaussian random measurements $M$ for different SNR values ($N=10000$).}\label{fig:sparseL}
\end{figure}

{\bf Markov-Uniform source:} For bounded continuous-valued signals, which adhere to Condition~\ref{cond:tech1}, we simulated a Markov-Uniform (MUnif for short) source, whose non-zero entries were generated
by a two-state Markov state machine (non-zero and zero states) with $\mathbb{P}_{01}=\frac{3}{970}$ and $\mathbb{P}_{10}=0.10$; these parameters yielded $3\%$ non-zero entries on average. The non-zero entries were drawn from
a uniform distribution between $0$ and $1$.
We chose tG with Markov support and GM model options~\cite{turboGAMP} for message passing algorithms.
We plot the resulting MSDR's for SLA-MCMC, tG~\cite{turboGAMP}, and CoSaMP~\cite{Cosamp08} in Figure~\ref{fig:MUnif}.
We can see that the CoSaMP~\cite{Cosamp08} lags behind in MSDR. The SLA-MCMC curve is close to that of tG~\cite{turboGAMP} when $\mbox{SNR}=10~\mbox{dB}$, and it is slightly better than tG~\cite{turboGAMP} when $\mbox{SNR}=5~\mbox{dB}$.

When the signal model is known, the message passing approaches EGAM~\cite{EMGMTSP} and tG~\cite{turboGAMP} achieve quite low MSE's, because they can get close to the Bayesian MMSE. Sometimes the model is only known imprecisely, and SLA-MCMC can improve over message passing; for example, it is better than tG~\cite{turboGAMP} in estimating MUnif signals (Figure~\ref{fig:MUnif}), because tG~\cite{turboGAMP} approximates the uniformly distributed non-zeros by GM.

\begin{figure}[t]
\begin{center}
\includegraphics[width=80mm]{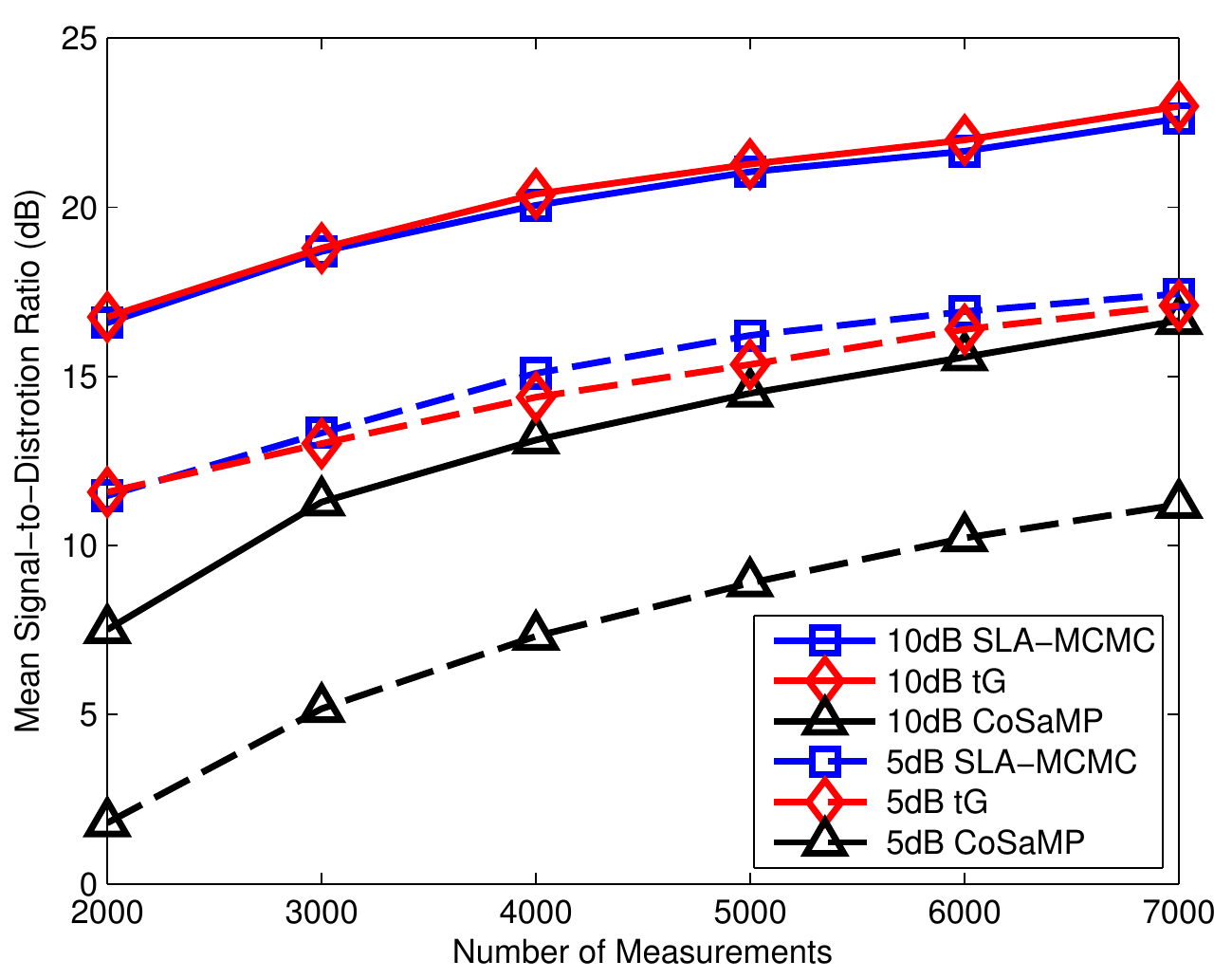}
\end{center}
\caption{SLA-MCMC, tG, and CoSaMP estimation results for a two-state Markov source with non-zero entries drawn from a uniform distribution $U[0,1]$ as a function of the number of Gaussian random measurements $M$ for different SNR values ($N=10000$).\label{fig:MUnif}}
\end{figure}

\subsection{Comparison between discrete and continuous sources}

When the source is continuous (Figures~\ref{fig:sparseL} and~\ref{fig:MUnif}), SLA-MCMC might be worse than the existing message passing approaches (EGAM~\cite{EMGMTSP} and tG~\cite{turboGAMP}).
One reason for the under-performance of SLA-MCMC is the $3$~dB gap of Conjecture~\ref{conj:double_MMSE}. The second reason is that SLA-MCMC can only assign finitely many levels to approximate continuous-valued signals, leading to under-representation of the signal. However, when it comes to discrete-valued signals that have finite size alphabets (Figures~\ref{fig:Ber} and~\ref{fig:denseMRad}), SLA-MCMC is comparable to and in many cases better than existing algorithms. Nonetheless, we observe in the figures that SLA-MCMC is far from the state-of-the-art when the SNR is high and measurement rate is low. Additionally, the dense MRad source in Figure~\ref{fig:denseMRad} has only a limited number of discrete levels and may not provide a general enough example.

\begin{figure}[t]
\begin{center}
\includegraphics[width=80mm]{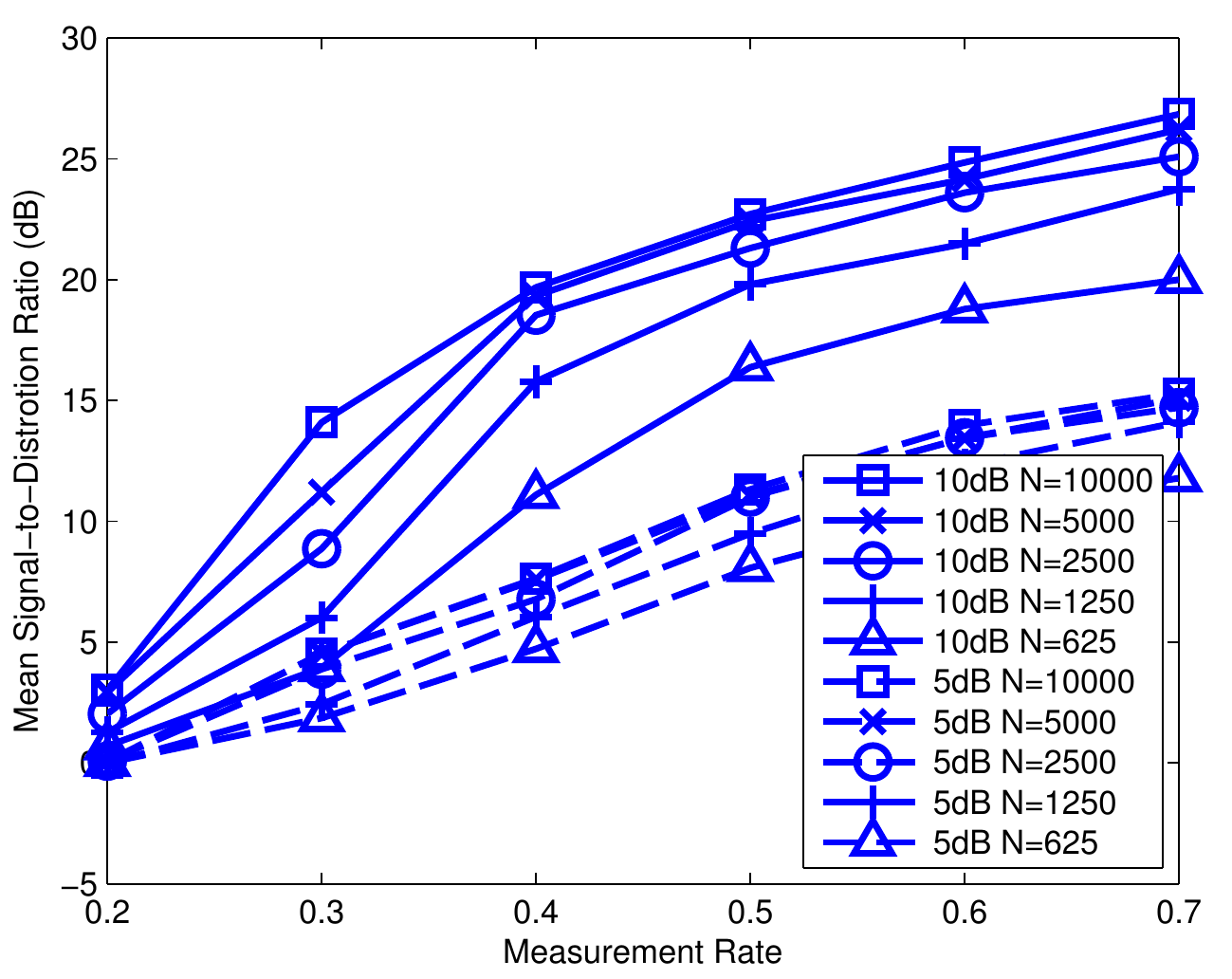}
\end{center}
\caption{SLA-MCMC estimation results for a four-state Markov switching source as
a function of the measurement rate $\kappa$ for different SNR values and signal lengths. Existing CS algorithms fail at estimating this signal, because this source is not sparse.}\label{fig:M4}
\end{figure}

\subsection{Performance on low-complexity signals}
SLA-MCMC promotes low complexity due to the complexity-penalized term in the objective function~(\ref{eq:MCMC_energy}). Hence, it tends to perform well for signals with low complexity such as the signals in
Figures~\ref{fig:Ber} and~\ref{fig:denseMRad} (note that the Bernoulli signal is sparse while the MRad signal is denser).
In this subsection, we simulated a non-sparse low-complexity signal. We show that complexity-penalized approaches such as SLA-MCMC might estimate low-complexity signals well.

{\bf Four-state Markov source:} To evaluate the performance of SLA-MCMC for discrete-valued non-i.i.d.\ and non-sparse signals, we examined a four-state Markov source (Markov4 for short) that generated the pattern $+1,+1,-1,-1,+1,+1,-1,-1\cdots$
with 3\% errors in state transitions, resulting in the signal switching from $-1$ to $+1$ or
vice versa either too early or too late. Note that the estimation algorithm did not know that this source is a binary source. While it is well known that sparsity-promoting
CS signal estimation algorithms~\cite{turboGAMP,Cosamp08,GPSR2007}
can estimate sparse sources from linear measurements,
the aforementioned switching source is not sparse in conventional sparsifying bases (e.g., Fourier, wavelet, and discrete cosine transforms), rendering such sparsifying transforms not applicable. Signals generated by this Markov source can be sparsified using an averaging analysis matrix~\cite{CandesCSdictonary2011} whose diagonal and first three lower sub-diagonals are filled with $+1$, and all other entries are $0$; this transform yields $6\%$ non-zeros in the sparse coefficient vector.
However, even if this matrix had been known {\em a priori}, existing algorithms based on analysis sparsity~\cite{CandesCSdictonary2011} did not perform satisfactorily, yielding MSDR's below $5$~dB.
Thus, we did not include the results for these baseline algorithms in Figure~\ref{fig:M4}.
On the other hand, Markov4 signals have low complexity in the time domain, and hence, SLA-MCMC successfully estimated Markov4 signals with reasonable quality even when $M$ was relatively small.
This Markov4 source highlights the special advantage of our approach in estimating low-complexity signals.

The MSDR's for shorter Markov4 signals are also plotted in Figure~\ref{fig:M4}. We can see that SLA-MCMC performs better when the signal to be estimated is longer. Indeed, SLA-MCMC needs a signal that is long enough to learn the statistics of the signal.

\begin{figure}[t]
\begin{center}
\includegraphics[width=80mm]{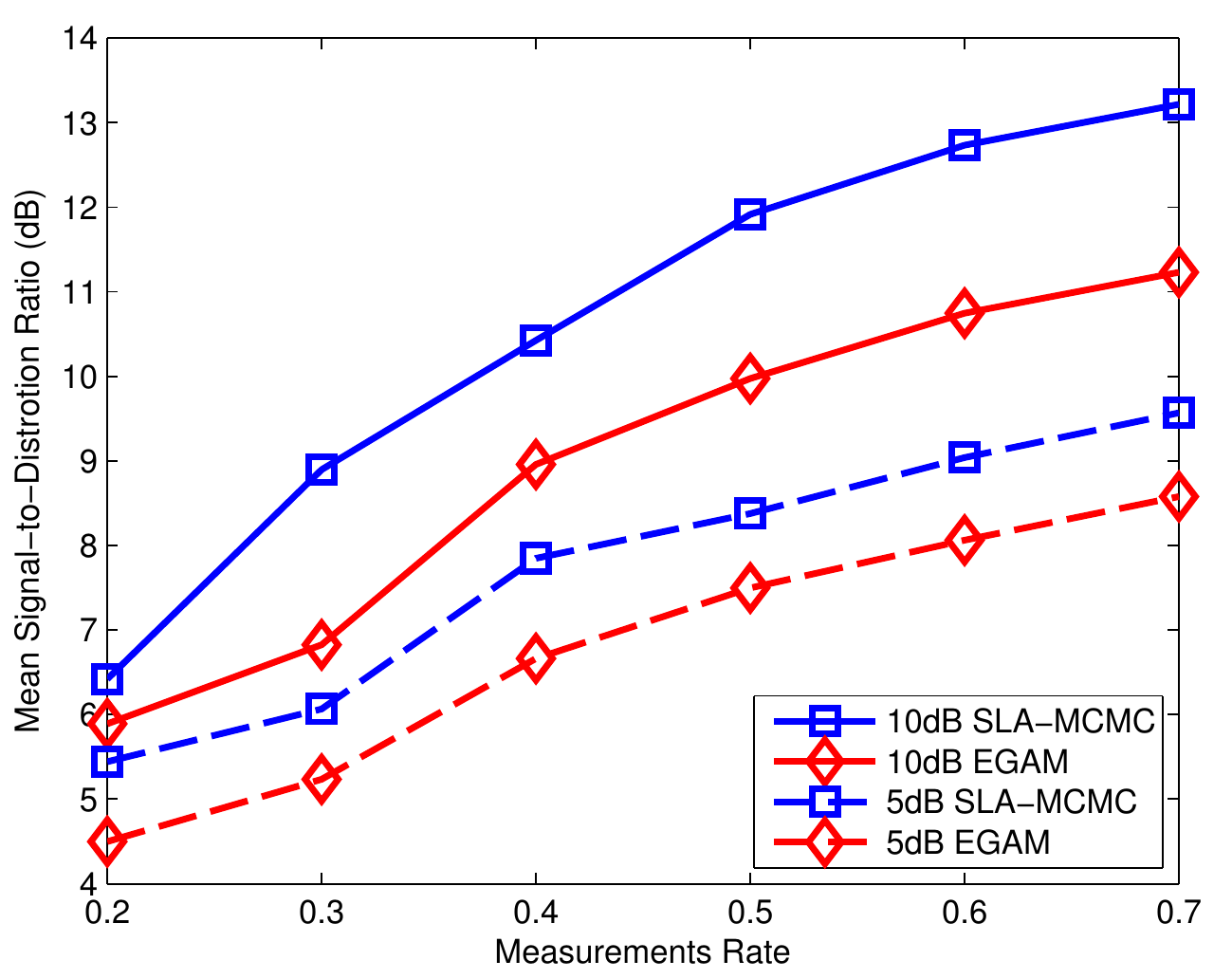}
\end{center}
\caption{SLA-MCMC and EGAM estimation results for a Chirp signal as a function of the measurement rate $\kappa$ for different SNR values ($N=9600$).}\label{fig:Chirp}
\end{figure}

\subsection{Performance on real world signals}
Our experiments up to this point use synthetic signals, where SLA-MCMC has shown comparable and in many cases better results than existing algorithms. This subsection evaluates how well SLA-MCMC estimates a real world signal. We use the ``Chirp'' sound clip from Matlab: we cut a consecutive part with length 9600 out of the ``Chirp'' (denoted by $\x$) and performed a short-time discrete cosine transform (DCT) with window size, number of DCT points, and hop size all being 32.
Then we vectorized the resulting short-time DCT coefficients matrix to form a coefficient vector $\theta$ of length 9600. By denoting the short-time DCT matrix by $\W^{-1}$, we have $\theta=\W^{-1}\x$. Therefore, we can rewrite~\eqref{eq:def_y} as $\y=\widetilde{\A}\theta +\z$,
where $\widetilde{\A}=\A \W$. We want to estimate $\theta$ from the measurements $\y$ and the matrix $\widetilde{\A}$. After we obtain the estimate $\widehat{\theta}$, we obtain the estimated signal by $\widehat{\x}=\W\widehat{\theta}$.
Although the coefficient vector $\theta$ may exhibit some type of memory, it is not readily modeled in closed form, and so we cannot provide a valid model for tG~\cite{turboGAMP}. Instead, we use EGAM~\cite{EMGMTSP} as our benchmark algorithm. We do not compare to CoSaMP~\cite{Cosamp08} because it falls behind in performance as we have seen from other examples. The MSDR's for SLA-MCMC and EGAM~\cite{EMGMTSP} are plotted in Figure~\ref{fig:Chirp}, where SLA-MCMC outperforms EGAM by 1--2 dB.

\subsection{Comparison of B-MCMC, L-MCMC, and SLA-MCMC}

We compare the performance of B-MCMC, L-MCMC, and SLA-MCMC with different numbers of seeds (cf.\ Section~\ref{sec:mix}) by examining the MUnif source (cf.\ Section~\ref{sec:bvsub}). We ran B-MCMC with the fixed uniform alphabet $\replevels_F$ in~(\ref{eq:def:replevels2}) with $|\replevels_F|=10$ levels. L-MCMC was initialized in the same way as Stage~1 of SLA-MCMC. B-MCMC and L-MCMC ran for $100$ super-iterations before outputting the estimates; this number of super-iterations was sufficient because it was greater than $r_1=50$ in Stage~1 of SLA-MCMC. The results are plotted in Figure~\ref{fig:MUnif_algos}.
B-MCMC did not perform well given the $\replevels_F$ in~(\ref{eq:def:replevels2}) and is not plotted.
We can see that SLA-MCMC outperforms L-MCMC. Averaging over more seeds provides an increase of $1$~dB in MSDR.\footnote{For other sources, we observed an increase in MSDR of up to $2$~dB.} It is likely that averaging over more seeds with each seed running fewer super-iterations will decrease the squared error. We leave the optimization of the number of seeds and the number of super-iterations in each seed for future work.
Finally, we tried a ``good'' reproduction alphabet in B-MCMC, $\displaystyle\widetilde{\replevels}_F=\frac{1}{|\replevels_F|-1/2}\{0,\cdots,|\replevels_F|-1\}$, and the results were close to those of SLA-MCMC. Indeed, B-MCMC is quite sensitive to the reproduction alphabet, and Stages~$2$--$4$ of SLA-MCMC find a good set of levels.
Example output levels $\map(\Z)$ of SLA-MCMC were: $\{-0.001,0.993\}$ for Bernoulli signals, $\{-0.998,0.004,1.004\}$ for dense MRad signals, $21$ levels spread in the range $[-3.283,4.733]$ for i.i.d. sparse Laplace signals, $22$ levels spread in the range $[-0.000,0.955]$ for MUnif signals, and $\{-1.010,0.996\}$ for Markov4 signals; we can see that SLA-MCMC adaptively adjusted $|\Z|$ to match $|\X|$ so that these levels represented each signal well. Also, we can see from Figures~\ref{fig:Ber}--\ref{fig:sparseL} that SLA-MCMC did not perform well in the low measurements and high SNR setting, which was due to mismatch between $|\Z|$ and $|\X|$.

\begin{figure}[t]
\begin{center}
\includegraphics[width=80mm]{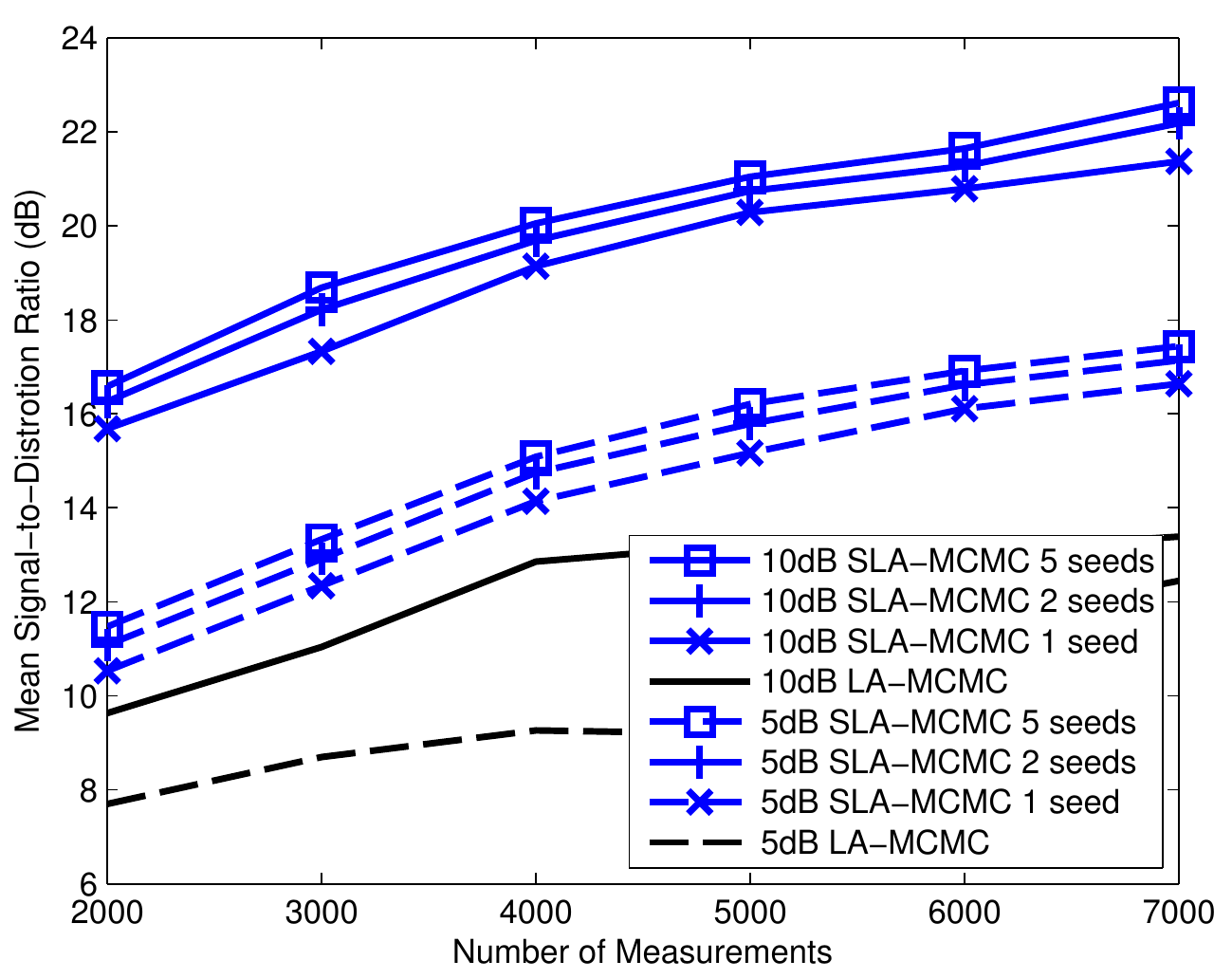}
\end{center}
\caption{SLA-MCMC with different number of random seeds and L-MCMC  estimation results for the Markov-Uniform source described in Figure~\ref{fig:MUnif} as a function of the number of Gaussian random measurements $M$ for different SNR values ($N=10000$).}\label{fig:MUnif_algos}
\end{figure}

\sectionmark{Numerical Results}
\section{Approximate Message Passing with Universal Denoising}\label{sec:AMP-UD}
\sectionmark{AMP-UD}

We note in passing another universal algorithm, approximate message passing with universal denoising (AMP-UD)~\cite{MaZhuBaronAllerton2014,MaZhuBaron2016TSP}, for CS signal estimation, of which the author of this dissertation is a coauthor. The signal $\x$ is assumed to be stationary and ergodic, but the input statistics are unknown. AMP-UD is a novel algorithmic framework that combines: ({\em i}) the approximate message passing CS signal estimation framework~\cite{DMM2009,Montanari2012,Bayati2011,Krzakala2012probabilistic,krzakala2012statistical,Barbier2015}, which solves the CS signal estimation problem by iterative scalar channel denoising; ({\em ii})  a universal denoising scheme based on context quantization~\cite{Sivaramakrishnan2008,SW_Context2009}, which partitions the stationary ergodic signal denoising into i.i.d. sub-sequence denoising; and ({\em iii})  a density estimation approach that approximates the probability distribution of an i.i.d. sequence by fitting a GM model~\cite{FigueiredoJain2002}. In addition to the algorithmic framework, Ma et al.~\cite{MaZhuBaronAllerton2014,MaZhuBaron2016TSP} provide three contributions: ({\em i})  numerical results showing that state evolution~\cite{DMM2011,Bayati2011,JavanmardMontanari2012,Donoho2013,Bayati2015} holds for non-separable Bayesian sliding-window denoisers; ({\em ii})  an i.i.d. denoiser based on a modified GM learning algorithm; and ({\em iii})  a universal denoiser that does not need information about the range where the input takes values from or require the input signal to be bounded. Ma et al.~\cite{MaZhuBaronAllerton2014,MaZhuBaron2016TSP} provide two implementations of AMP-UD with one being faster and the other being more accurate. The two implementations compare favorably with existing universal signal estimation algorithms (including the SLA-MCMC algorithm discussed in this chapter) in terms of both estimation quality and runtime.

\begin{figure}[t]
\begin{center}
\includegraphics[width=80mm]{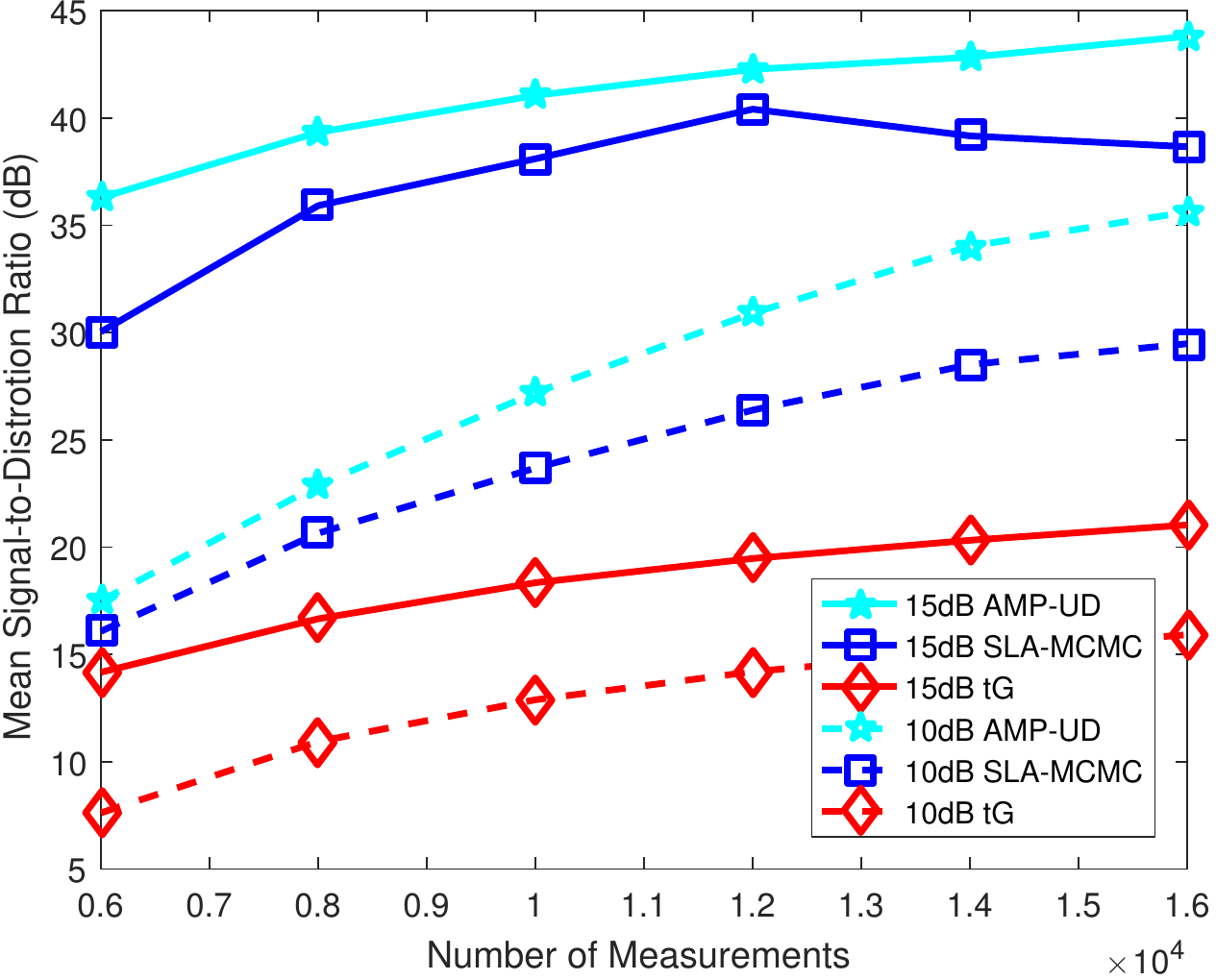}
\end{center}
\caption{AMP-UD~\cite{MaZhuBaronAllerton2014,MaZhuBaron2016TSP}, SLA-MCMC, and tG estimation results for a dense two-state Markov source with non-zero entries drawn from a Rademacher ($\pm 1$) distribution as a function of the number of Gaussian random measurements $M$ for different SNR values ($N=10000$).}\label{fig:AMP_UD_sim}
\end{figure}

To highlight the advantages of AMP-UD relative to SLA-MCMC, Figure~\ref{fig:AMP_UD_sim} compares the AMP-UD simulation results to the SLA-MCMC and tG~\cite{turboGAMP} results for the setting in Figure~\ref{fig:denseMRad}. We see that AMP-UD outperforms both algorithms. Moreover, the runtime of AMP-UD is around 5 minutes to estimate this MRad signal of length $10000$, while it usually takes SLA-MCMC an hour and tG~\cite{turboGAMP} 30 minutes to estimate this signal. Therefore, we see that AMP-UD is indeed promising.

\sectionmark{AMP-UD}
\section{Conclusion} \label{sec:conclusions}
\sectionmark{Conclusion}

This chapter provided universal algorithms for signal estimation from linear measurements. Here, universality
denotes the property that the algorithm need not be informed of the probability
distribution for the recorded signal prior to acquisition; rather, the algorithm
simultaneously builds estimates both of the observed signal and its distribution.
Inspired by the Kolmogorov sampler~\cite{DonohoKolmogorov} and motivated by the
need for a computationally tractable framework, our contribution
focused on stationary ergodic signal sources and relied on a maximum a posteriori estimation algorithm. The algorithm was then implemented via a Markov chain Monte Carlo formulation that is proven to be convergent in the limit of infinite computation.
We reduced the computational complexity and improve the estimation quality of the proposed algorithm by adapting the reproduction alphabet to match
the complexity of the input signal. Our experiments have shown that the performance of the proposed algorithm
is comparable to and in many cases better than existing
algorithms,
particularly for low-complexity sources that do not exhibit standard sparsity or compressibility.

As we were finishing this work, Jalali and Poor~\cite{JalaliPoor2014} have independently shown that our formulation~(\ref{eq:MCMC_energy}) also provides an implementable version of R{\'e}nyi entropy minimization. Their theoretical findings further motivated our proposed universal MCMC formulation. We noted in passing another universal algorithm that often achieves better estimation quality than the SLA-MCMC algorithm discussed in this chapter.

\chapter{Discussion}\label{chap-discuss}

This chapter concludes the dissertation. We begin by summarizing the previous chapters, and then we list our contributions. Finally, we propose some possible future directions.

\section{Summary and Contributions}
Linear models find wide applications in the real world, and the problem of estimating the underlying signal(s) from a linear model is called a linear inverse problem. Depending on the number of underlying signals, we have the single measurement vector problem (SMV) and the multi-measurement vector problem (MMV); depending on how the measurement matrix and the measurements are stored, we have the centralized linear model and the multi-processor linear model. Prior art includes algorithms for linear inverse problems and their corresponding performance characterizations. There are many remaining issues in the prior art. First, there is little work discussing the performance characterization for the linear inverse problems themselves. Second, when dealing with the distributed setting, there is little work studying the relations of different costs. At last, the existing algorithms for linear inverse problems require the prior knowledge of the unknown signal to some extent. These issues are important to practitioners. In this dissertation, we took advantage of the tools in statistical physics and information theory to address these issues in the large system limit, i.e., the length of the signal and the number of measurements go to infinity while the measurement rate (ratio between the number of measurements and the length of the signal) stays constant.

We started with providing background materials on statistical physics and information theory in Chapter~\ref{chap-basics}, and we also discussed the link between statistical physics and information theory.
Then, we studied the minimum mean squared error (MMSE) for MMV problem in Chapter~\ref{chap-MMV} by using the replica analysis from statistical physics.
We analyzed the MMSE for two settings of MMV problems, where the entries in the signal vectors are independent and identically distributed (i.i.d.), and share the same support. One MMV setting has i.i.d. Gaussian measurement matrices, while the other MMV setting has identical i.i.d. Gaussian measurement matrices. Replica analysis yields identical free energy expressions for these two settings in the large system limit. Because of the identical free energy expressions, the MMSE's for both MMV settings are identical. By numerically evaluating the free energy expression, we identified different performance regions for MMV where the MMSE as a function of the channel noise variance and the measurement rate behaves differently. We also identified a phase transition for belief propagation algorithms (BP) that separates regions where BP achieves the MMSE asymptotically and where it is sub-optimal. Simulation results of an approximated version of BP matched with the mean squared error (MSE) predicted by replica analysis. As a special case of MMV, we extended our replica analysis to complex SMV, so that we can calculate the MMSE for complex SMV with real or complex measurement matrices. Chapter~\ref{chap-MMV} is based on our work with Baron~\cite{ZhuBaronCISS2013} and with Baron and Krzakala~\cite{ZhuBaronKrzakala2016}.

In Chapter~\ref{chap-MP-AMP}, we studied the optimization of different costs in running a distributed algorithm; these costs include (but are not limited to) the computation cost, the communication cost, and the quality of the estimate. We focused our discussion on a certain distributed algorithm, multi-processor approximate message passing (MP-AMP). Our results might be extended to some other distributed and iterative algorithms. We proposed to use lossy compression (from information theory) on the messages being transmitted across the network, and we allowed the coding rate to vary from iteration to iteration for MP-AMP. Also, we proposed an algorithmic method to find the optimal coding rate for the messages being transmitted in the network for MP-AMP, so that we can achieve the smallest combined cost of computation and communication. In addition, we theoretically analyzed the optimal coding rate sequence in the limit of low
excess mean squared error (EMSE=MSE-MMSE) and it turns out that the optimal coding rate sequence is approximately linear when the EMSE is low. At last, we proved the existence of trade-offs among these different costs for MP-AMP. Chapter~\ref{chap-MP-AMP} is based on our work with Han et al.~\cite{HanZhuNiuBaron2016ICASSP} and with Baron and Beirami~\cite{ZhuBeiramiBaron2016ISIT,ZhuBaronMPAMP2016ArXiv}.

In Chapter~\ref{chap-SLAM}, we proposed a universal algorithm, size- and level-adaptive Markov chain Monte Carlo (SLA-MCMC), to solve the linear inverse problem.
Inspired by the Kolmogorov sampler~\cite{DonohoKolmogorov} and motivated by the need for a computationally tractable framework, our contribution focused on stationary ergodic signal sources and relied on a maximum a posteriori estimation algorithm. The algorithm was then implemented via a Markov chain Monte Carlo formulation (motivated from thermodynamics) that is proven to be convergent in the limit of infinite computation.
We reduced the computational complexity and improved the estimation quality of the proposed algorithm by adapting the reproduction alphabet to match the complexity of the input signal. Our experiments have shown that the performance of the proposed algorithm is comparable to and in many cases better than existing algorithms, particularly for low-complexity signals that do not exhibit standard sparsity or compressibility. Chapter~\ref{chap-SLAM} is based on our work with Baron and Duarte~\cite{JZ2014SSP,ZhuBaronDuarte2014_SLAM}.

\section{Future Directions}

Along the line of this dissertation, we list some possible future directions.
\begin{enumerate}
\item Our replica analysis in Chapter~\ref{chap-MMV} assumes that the non-zero entries of the jointly sparse signals are i.i.d. However, in real-world application, sometimes the non-zero entries that share the same support are dependent. Our derivation could possibly be generalized to such settings. When the non-zero entries of the signals are dependent, we suspect that the MMV setting with different matrices will yield lower MMSE than the MMV setting with identical matrices.
\item  As is discussed in Chapter~\ref{chap-MMV}, studying other error metrics than the MSE could also be of interest. We could extend the work of Tan and coauthors~\cite{Tan2014,Tan2014Infty}, so that we can both study the theoretic optimal performance for user-defined additive error metric and design algorithms that can achieve the theoretic optimal performance.
\item In Chapter~\ref{chap-MP-AMP}, our study of different costs is within the MP-AMP algorithm. One possible future direction could be to find a generic class of algorithms to which our analyses can apply. Another possible direction is to incorporate such ideas in a real-world software package design, which could be of great interest to industry.
\item Although both SLA-MCMC and AMP-UD from Chapter~\ref{chap-SLAM} seem promising, they are not so resilient to measurement matrices that are far from i.i.d. In order to make a larger impact, we need to design universal algorithms that are more resilient to non-i.i.d. matrices.

\end{enumerate}


\begin{spacing}{1}
 \setlength\bibitemsep{11pt} 
 \addcontentsline{toc}{chapter}{{\uppercase{\bibname}}} 
\titleformat{\chapter}[display]{\bf\filcenter
}{\chaptertitlename\ \thechapter}{11pt}{\bf\filcenter}
\titlespacing*{\chapter}{0pt}{-0.5in-9pt}{22pt}
\newgeometry{margin=1in,lmargin=1.25in,footskip=\chapterfootskip, includefoot}

\printbibliography[heading=myheading]
\end{spacing}

\restoregeometry
\appendix
\newgeometry{margin=1in,lmargin=1.25in,footskip=\chapterfootskip, includefoot}

\chapter{Appendix for Chapter~3}\label{chap:append-A}
This appendix follows the derivation of Barbier and Krzakala~\cite{Barbier2015}, except for some nuances.
Our compressed derivation makes the presentation self-contained.

Plugging~\eqref{eq:XmuNew} and the following identity~\cite{Barbier2015,Krzakala2012probabilistic},
\begin{equation*}
\begin{split}
  &1=\int \text{exp}\Bigg\{-\sij{a=1}{n} \l[\widehat{m}_a\l(m_a NJ-\sij{l=1}{N}(\widehat{\x}_l^a)^{\top}\x_l\r)\r]+\sij{a=1}{n}\Bigg[
  \widehat{Q}_a\l(Q_a\frac{NJ}{2}-\frac{1}{2}\sij{l=1}{N}(\widehat{\x}_l^a)^{\top}\widehat{\x}_l^a\r)\Bigg]-\\
  &\sij{1\leq a< b\leq n}{}\Bigg[\widehat{q}_{ab}\Bigg(q_{ab} NJ-\sij{l=1}{N}(\widehat{\x}_l^a)^{\top}\widehat{\x}_l^b\Bigg)\Bigg]\Bigg\}\pij{a=1}{n}dQ_a\  d\widehat{Q}_{a}\  dm_a \ d\widehat{m}_{a} \pij{1\leq a<b\leq n}{} dq_{ab}\ d\widehat{q}_{ab},
\end{split}
\end{equation*}
into \eqref{eq:EZn1}, we obtain
\begin{equation}\label{eq:EZn3}
\begin{split}
  \mathbb{E}_{\A,\x,\z}[Z^n]=&(2\pi\sigma_Z^2)^{-\frac{nMJ}{2}}\bigintsss \text{exp}\Bigg[NJ\Bigg(\frac{1}{2}\sij{a=1}{n}\widehat{Q}_aQ_a
  -\frac{1}{2}\sij{\substack{1\leq a,b\leq n\\a\neq b}}{}\widehat{q}_{ab}q_{ab}-\sij{a=1}{n}\widehat{m}_am_a\Bigg)\Bigg]\l[\pij{\mu=1}{M} \mathbb{X}_{\mu}\r]\times\\
  &\Gamma^N\pij{a=1}{n}dQ_a\  d\widehat{Q}_{a}\  dm_a\  d\widehat{m}_{a}\pij{\substack{1\leq a,b\leq n\\a\neq b}}{} dq_{ab}\ d\widehat{q}_{ab},
\end{split}
\end{equation}
where
\begin{equation}\label{eq:Gamma_original}
\Gamma=\!\!\!\int\! f(\x_1)\! \l[\pij{a=1}{n}f(\widehat{\x}^a_1)\r]\!\text{exp}\!\Bigg[\!-\frac{1}{2}\sij{a=1}{n}\widehat{Q}_a(\widehat{\x}^a_1)^{\top}\widehat{\x}^a_1+
\frac{1}{2}\sij{\substack{1\leq a,b\leq n\\a\neq b}}{}\widehat{q}_{ab}(\widehat{\x}^a_1)^{\top}\widehat{\x}^b_1+\sij{a=1}{n}\widehat{m}_a(\widehat{\x}^a_1)^{\top}\x_1\Bigg]d\x_1 \pij{a=1}{n}d\widehat{\x}^a_1.
\end{equation}

\textbf{Further simplification of~\eqref{eq:EZn1}}: The Stratanovitch transform~\cite{Stratanovitch-Wiki} in $J$ dimensions is given by
\begin{equation}\label{eq:strat}
\begin{split}
  \text{exp}\l[\frac{\widehat{q}}{2}\sij{\substack{1\leq a,b \leq n\\a\neq b}}{}(\widehat{\x}^a_1)^{\top}\widehat{\x}^b_1\r]&=\pij{j=1}{J}\text{exp}\l[\frac{\widehat{q}}{2}\sij{\substack{1\leq a,b\leq n\\a\neq b}}{}\widehat{x}^a_{1,j} \widehat{x}^b_{1,j}\r]\\
  &=\pij{j=1}{J}\int\text{exp}\l[\sqrt{\widehat{q}}h_j\sij{a=1}{n}\widehat{x}^a_{1,j}-\frac{\widehat{q}}{2}\sij{a=1}{n}\l(\widehat{x}^a_{1,j}\r)^2\r]\mathcal{D}h_j\\
  &=\int\text{exp}\l[\sqrt{\widehat{q}}\h^{\top}\sij{a=1}{n}\widehat{\x}^a_1-\frac{\widehat{q}}{2}\sij{a=1}{n}\l(\widehat{\x}^a_1\r)^{\top}\widehat{\x}^a_1\r]\mathcal{D}\h,
\end{split}
\end{equation}
where $\h=[h_1,...,h_J]^{\top}$, and the differential $\mathcal{D}h_j=\frac{1}{\sqrt{2\pi}}\operatorname{e}^{-h_j^2/2}d h_j$. With the Stratanovitch transform~\eqref{eq:strat},
we simplify $\Gamma$~\eqref{eq:Gamma_original} as follows,
\begin{equation}\label{eq:Gamma}
\Gamma=\int f(\x_1)\int\l[f(\h)\r]^n \mathcal{D}\h\ d\x_1,
\end{equation}
where
$f(\h)=\int f(\x_1)\operatorname{e}^{-\frac{\widehat{Q}+\widehat{q}}{2}\widehat{\x}_1^{\top}\widehat{\x}_1+\widehat{m}\widehat{\x}_1^{\top}\x_1+\sqrt{\widehat{q}}\h^{\top}\widehat{\x}_1}d\widehat{\x}_1$,
and we drop the super-script $a$ of $\widehat{\x}_1^a$ owing to the replica symmetry assumption~\cite{Krzakala2012probabilistic,krzakala2012statistical}.
In the limit of $n\rightarrow 0$, using another Taylor series $[f(\h)]^n\approx 1+n\log [f(\h)]$, we have $\int [f(\h)]^n \mathcal{D}\h\approx 1+n\int \log [f(\h)] \mathcal{D}\h\approx\operatorname{e}^{n\int \log [f(\h)]\mathcal{D}\h}$, so that $\mathbb{E}\l\{\int[f(\h)]^n\mathcal{D}\h\r\}\approx\mathbb{E}\l\{1+n\int \log [f(\h)] \mathcal{D}\h\r\}\approx\operatorname{e}^{\mathbb{E}\l\{n\int \log [f(\h)]\mathcal{D}\h\r\}}$. Hence, we can approximate~\eqref{eq:Gamma} as
\begin{equation}\label{eq:Gamma1}
\Gamma=\text{exp}\l\{n\int f(\x_1) \int\log [f(\h)] \mathcal{D}\h \ d\x_1\r\}.
\end{equation}
Considering~\eqref{eq:Gamma1}, we rewrite \eqref{eq:EZn3} as
\begin{equation}\label{eq:EZn4}
    \mathbb{E}_{\A,\x,\z}[Z^n]=\int \operatorname{e}^{nN\widetilde{\Phi}_J(m,\widehat{m},q,\widehat{q},Q,\widehat{Q})}
      dm\ d\widehat{m}\ dq\ d\widehat{q}\ dQ\ d\widehat{Q},
\end{equation}
where $\widetilde{\Phi}_J(m,\widehat{m},q,\widehat{q},Q,\widehat{Q})$ is given below,
\begin{equation}\label{eq:PhiJ}
\begin{split}
 & \widetilde{\Phi}_J(m,\widehat{m},q,\widehat{q},Q,\widehat{Q})=\frac{J}{2}(Q\widehat{Q}+q\widehat{q}-2m\widehat{m})-\frac{MJ}{2N}\l[\frac{\rho-2m+\sigma_Z^2+q}{Q-q+\sigma_Z^2}+\log(Q-q+\sigma_Z^2)-\log(\sigma_Z^2)\r]+\\
  &\int f(\x_1) \l\{ \int \log \l\{ \int f(\widehat{\x}_1)\text{exp}\l[-\frac{1}{2}(\widehat{Q}+\widehat{q})\widehat{\x}_1^{\top}\widehat{\x}_1+\widehat{m}\widehat{\x}^{\top}_1\x_1+
  \sqrt{\widehat{q}}\h^{\top}\widehat{\x}_1\r]d\widehat{\x}_1\r\} \mathcal{D}\h \r\} d\x_1-\frac{MJ}{2N}\log(2\pi\sigma_Z^2).
\end{split}
\end{equation}

\textbf{Free energy expression}:
We now substitute~\eqref{eq:EZn4} into~\eqref{eq:replicaTrick}. Assuming that the limits in~\eqref{eq:replicaTrick} commute and that we only evaluate~\eqref{eq:replicaTrick} at optimum points of $\widetilde{\Phi}_J$~\eqref{eq:PhiJ}~\cite{Barbier2015,Krzakala2012probabilistic,krzakala2012statistical}, we have
$\mathcal{F}=\widetilde{\Phi}_J(m^*,\widehat{m}^*,q^*,\widehat{q}^*,Q^*,\widehat{Q}^*)$,
where the asterisks denote stationary points. Next, we calculate the stationary points:
\[\frac{\partial \widetilde{\Phi}_J }{\partial m}=0 \Rightarrow \widehat{m}^*=\frac{\kappa}{Q^*-q^*+\sigma_Z^2},\]
\[\frac{\partial \widetilde{\Phi}_J }{\partial q}=0 \Rightarrow \widehat{q}^*=\kappa\frac{\sigma_Z^2+\rho-2m^*+q^*}{(Q^*-q^*+\sigma_Z^2)^2},\]
\[\frac{\partial \widetilde{\Phi}_J }{\partial Q}=0 \Rightarrow \widehat{Q}^*=\kappa\frac{2m^*-\rho-2q^*+Q^*}{(Q^*-q^*+\sigma_Z^2)^2},\]
where $\kappa$~\eqref{eq:measurementRate} is the measurement rate. Because we are analyzing the MMSE, we must assume that the estimated prior matches the true underlying prior, which is a Bayesian setting. Thus, $q^*=m^*$ and $Q^*=\rho$~\eqref{eq:auxParamsSet1}. Let $E=q^*-2m^*+Q^*=Q^*-q^*$, then we obtain $\widehat{q}^*=\widehat{m}^*=\frac{\kappa}{E+\sigma_Z^2}$ and $\widehat{Q}^*=0$. Therefore, we solve for the free energy as a function of $E$ in~\eqref{eq:free_energy3}. Using a change of variables, we obtain~\eqref{eq:free_energy4}, which is a function of $E$. Using~\eqref{eq:auxParamsSet1}, the MSE is
\begin{equation}\label{eq:DandE}
D=E+Q-q=E+\frac{\rho}{N}\overset{N\rightarrow \infty}{\longrightarrow} E.
\end{equation}
Hence, in the large system limit, we can regard the free energy~\eqref{eq:free_energy4} as a function of the MSE, $D$.

\chapter{Appendices for Chapter~4}\label{chap:append-MP-AMP}

\section{Impact of the Quantization Error}\label{app:verifyIndpt}

This appendix provides numerical evidence that ({\em i}) the quantization error $\n_t$
is independent of the scalar channel noise $\w_t$~\eqref{eq:indpt_noises} in the fusion center and
({\em ii}) $\w_t+\n_t$ is independent of the signal $\x$. In the following, we simulate the AMP equivalent scalar channel in each processor node and in the fusion center.
In the interest of simple implementation, we use scalar quantization (SQ)
to quantize $\f_t^p$~\eqref{eq:slave2} (in each processor node) and
hypothesis testing to evaluate ({\em i}) whether $\w_t$ and $\n_t$ (in the fusion center) are independent and ({\em ii}) whether $\w_t+\n_t$ and $\x$ are independent. Both parts are necessary for lossy SE~\eqref{eq:SE_Q} to hold: part ({\em i}) ensures that we can predict the variance of $\w_t+\n_t$ by $\sigma_t^2+PD_t$ and part ({\em ii}) ensures that lossy MP-AMP falls within the general framework of Bayati and Montanari~\cite{Bayati2011} and Rush and Venkataramanan~\cite{Rush_ISIT2016_arxiv}, so that lossy SE~\eqref{eq:SE_Q} holds. Details about our simulation appear below.

Considering~\eqref{eq:equivalent_scalar_channel} and~\eqref{eq:slave2}, we obtain that the AMP equivalent scalar channel in each processor node can be expressed as
\begin{equation}\label{eq:noisePseudoNode}
\f_t^p=\frac{1}{P}\x+\w_t^p,
\end{equation}
where $\sum_{p=1}^P \w_t^p=\w_t$~\eqref{eq:equivalent_scalar_channel}, 
and the variances of $\w_t^p$ and $\w_t$ can be expressed as $(\sigma_t^p)^2$ and $\sigma_t^2$, respectively~\eqref{eq:SE_Q}. Hence, we obtain $\sigma_t^2=\sum_{p=1}^P (\sigma_t^p)^2$. The signal $\x$ follows~\eqref{eq:BG} with $\rho=0.1$. The entries of $\w_t^p$ are i.i.d. and follow $\mathcal{N}(0,(\sigma_t^p)^2)$. Next, we apply an SQ to $\f_t^p$~\eqref{eq:noisePseudoNode},
\begin{equation}\label{eq:quantPseudoNode}
Q(\f_t^p)=\frac{1}{P}\x+\w_t^p+\n_t^p,
\end{equation}
where $Q(\cdot)$ denotes the quantization process, $\n_t^p$ is the quantization error
in processor node $p$, and recall that the variance of $\n_t^p$ is $D_t$.
We simulate the fusion center by calculating
\begin{equation}\label{eq:processInfusionCenter}
\f_t=\sum_{p=1}^P Q(\f_t^p)=\x+\w_t+\n_t,
\end{equation}
where $\n_t=\sum_{p=1}^P \n_t^p$.
Note that $\w_t$ is Gaussian due to properties of AMP~\cite{DMM2009,Montanari2012,Bayati2011}.
The total quantization error at the fusion center, $\n_t$, is also Gaussian,
due to the central limit theorem. Hence, in order to test the independence
of $\w_t$ and $\n_t$~\eqref{eq:processInfusionCenter}, we need only test whether
$\w_t$ and $\n_t$ are uncorrelated. We also test whether $\w_t+ \n_t$ and $\x$ are uncorrelated.

\begin{figure*}[t]
\centering
\subfloat[]{\label{fig:PearsonWN}
\includegraphics[width=0.48\textwidth]{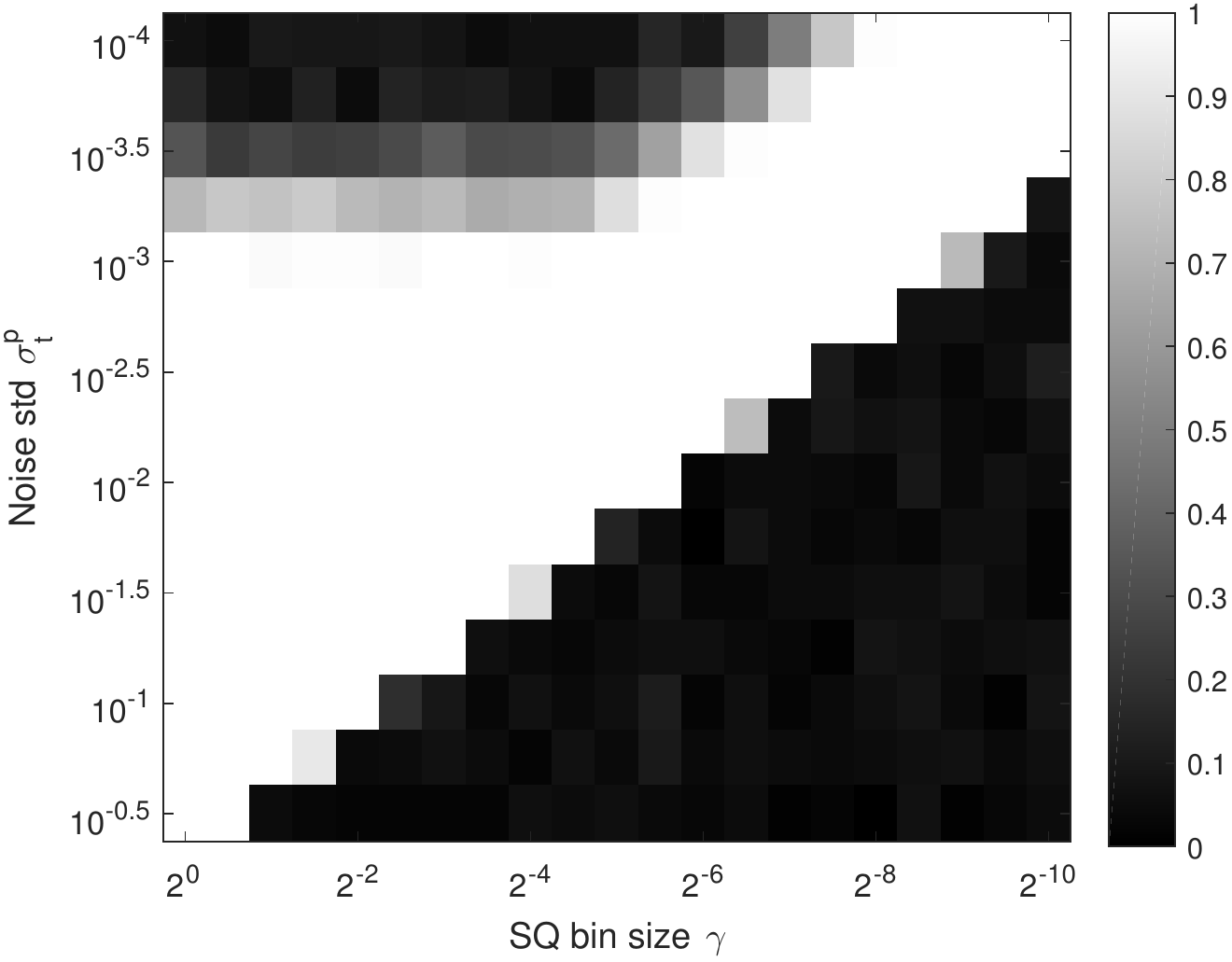}}
\subfloat[]{\label{fig:PearsonWpNX}
\includegraphics[width=0.48\textwidth]{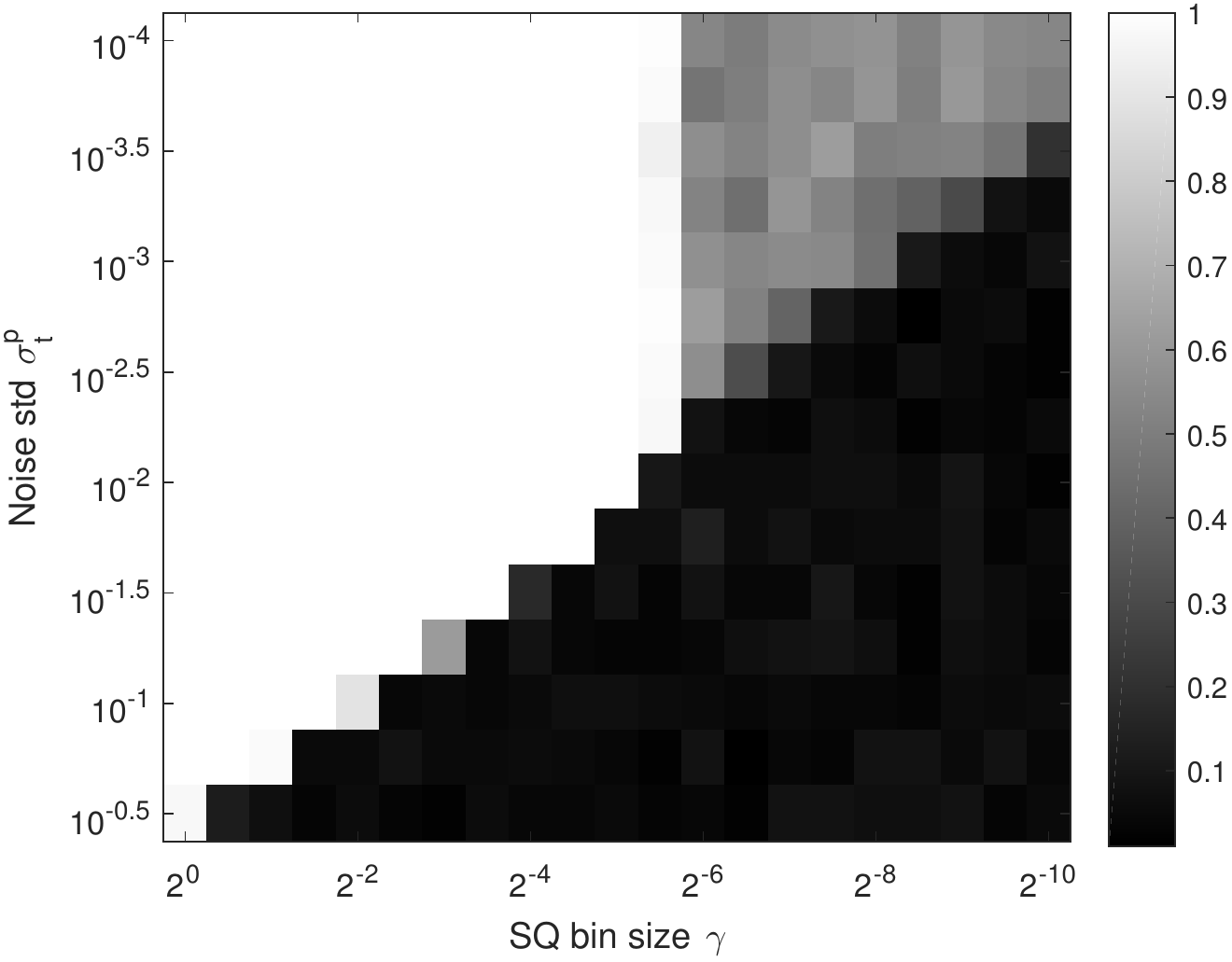}}
\caption{PCC test results. The darkness of the shades shows the fraction of
100 tests where we reject the null hypothesis (random variables being tested are uncorrelated)
with 5\% confidence. The horizontal and vertical axes represent the quantization bin size $\gamma$
of the SQ and the scalar channel noise standard deviation (std) $\sigma_t^p$ in each processor node,
respectively. Panel (a): Test the correlation between $\w_t$ and $\n_t$.
Panel (b): Test the correlation between $\w_t+\n_t$ and $\x$.}
\end{figure*}

We study the settings $\sigma_t^p\in\{10^{-0.5},\cdots,10^{-4}\}$ and $\gamma\in\{2^{0},\cdots,2^{-10}\}$,
where $\gamma$ denotes the SQ bin size. In each setting,
we simulate~\eqref{eq:noisePseudoNode}--\eqref{eq:processInfusionCenter} 100 times and
perform 100 Pearson correlation coefficient (PCC) tests~\cite{PCC} for $\w_t$ and $\n_t$,
respectively. The null hypothesis of the PCC tests~\cite{PCC}
is that $\w_t$ and $\n_t$ are uncorrelated. The null hypothesis is rejected
if the resulting $p$-value is smaller than 0.05.

For each setting, we record the fraction of 100 tests where
the null hypothesis is rejected, which is shown by the darkness of the shades
in Figure~\ref{fig:PearsonWN}. The horizontal and vertical axes represent the
quantization bin size $\gamma$ and the standard deviation (std) $\sigma_t^p$,
respectively. Similarly, we test $\w_t+\n_t$ and $\x$; results appear in
Figure~\ref{fig:PearsonWpNX}. We can see that when $\gamma \ll \sigma_t^p$
(bottom right corner),
({\em i}) $\w_t$ and $\n_t$ tend to be independent and
({\em ii}) $\w_t+\n_t$ and $\x$ tend to be independent.

Now consider Figure~\ref{fig:PearsonWpNX}, which provides PCC test results
evaluating possible correlations between $\w_t+\n_t$ and $\x$.
There appears to be a phase transition that separates regions where $\w_t+\n_t$
and $\x$ seem independent or dependent. We speculate that this phase transition
is related to the pdf of $\frac{1}{P}\x+\w_t^p$. To explain our hypothesis,
note that when the noise $\w_t^p$ is low (top part of~Figure~\ref{fig:PearsonWpNX}),
the phase transition is less affected by noise, and the role of $\gamma$
is smaller. By contrast, large noise (bottom) sharpens the phase transition.

In summary, it appears that when $\gamma<2\sigma_t^p=\frac{2\sigma_t}{\sqrt{P}}$,
we can regard ({\em i}) $\w_t$ and $\n_t$ to be independent and
({\em ii}) $\w_t+\n_t$ and $\x$ to be independent. The requirement
$\gamma<2\sigma_t^p=\frac{2\sigma_t}{\sqrt{P}}$ is motivated by Widrow and
Koll{\'a}r~\cite{widrow2008quantization};
we leave the study of this phase transition for future work.

\section{Numerical Evidence for Lossy SE}\label{app:verifyLossySE}
This appendix provides numerical evidence for lossy SE~\eqref{eq:lossySE}.
We simulate two signal types, one is the Bernoulli-Gaussian signal~\eqref{eq:BG}
and the other is a mixture Gaussian.

\begin{figure*}[t]
\centering
\subfloat[]{\label{fig:lossySEvsSim_BG}
\includegraphics[width=0.48\textwidth]{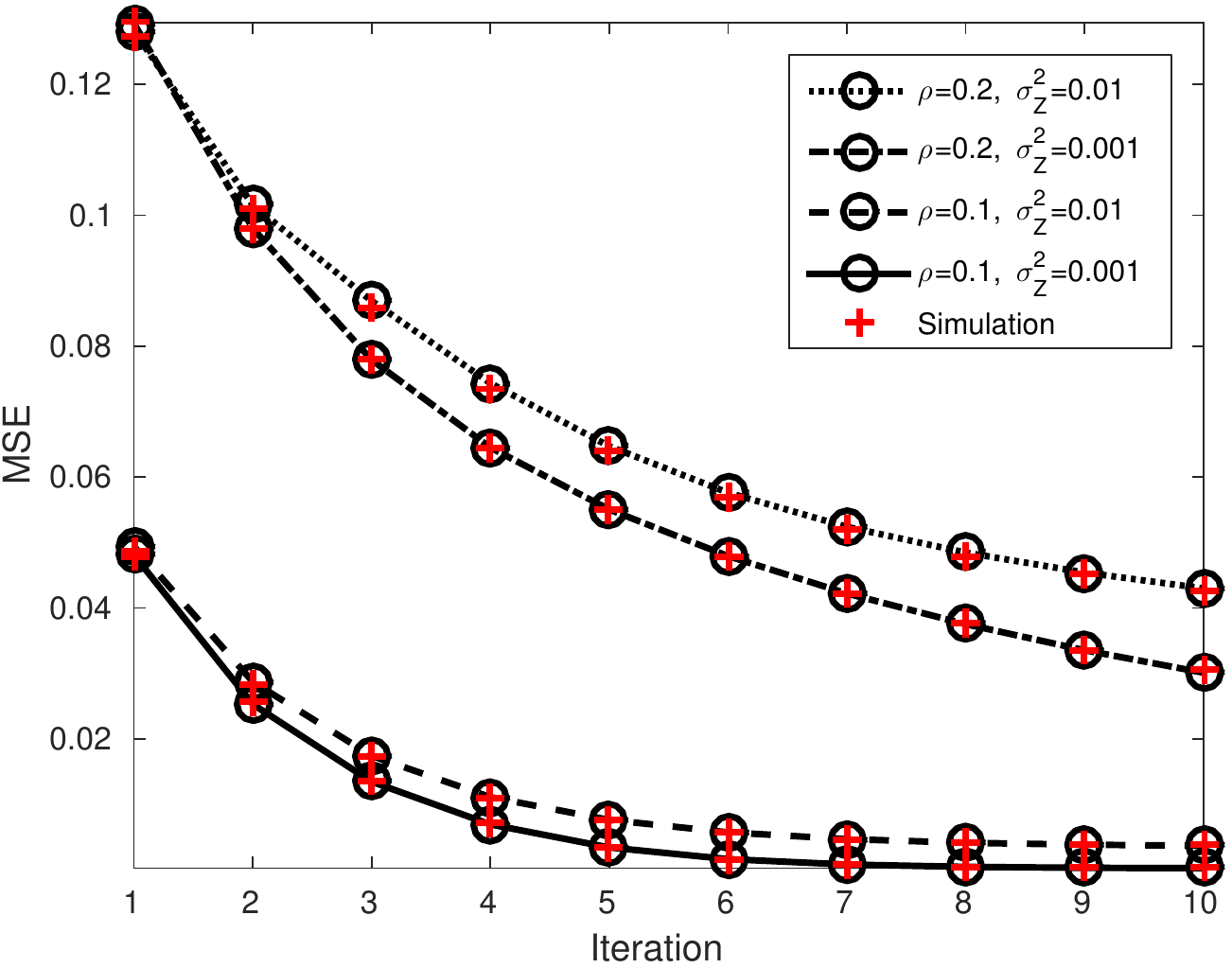}}
\subfloat[]{\label{fig:lossySEvsSim_mixG}
\includegraphics[width=0.48\textwidth]{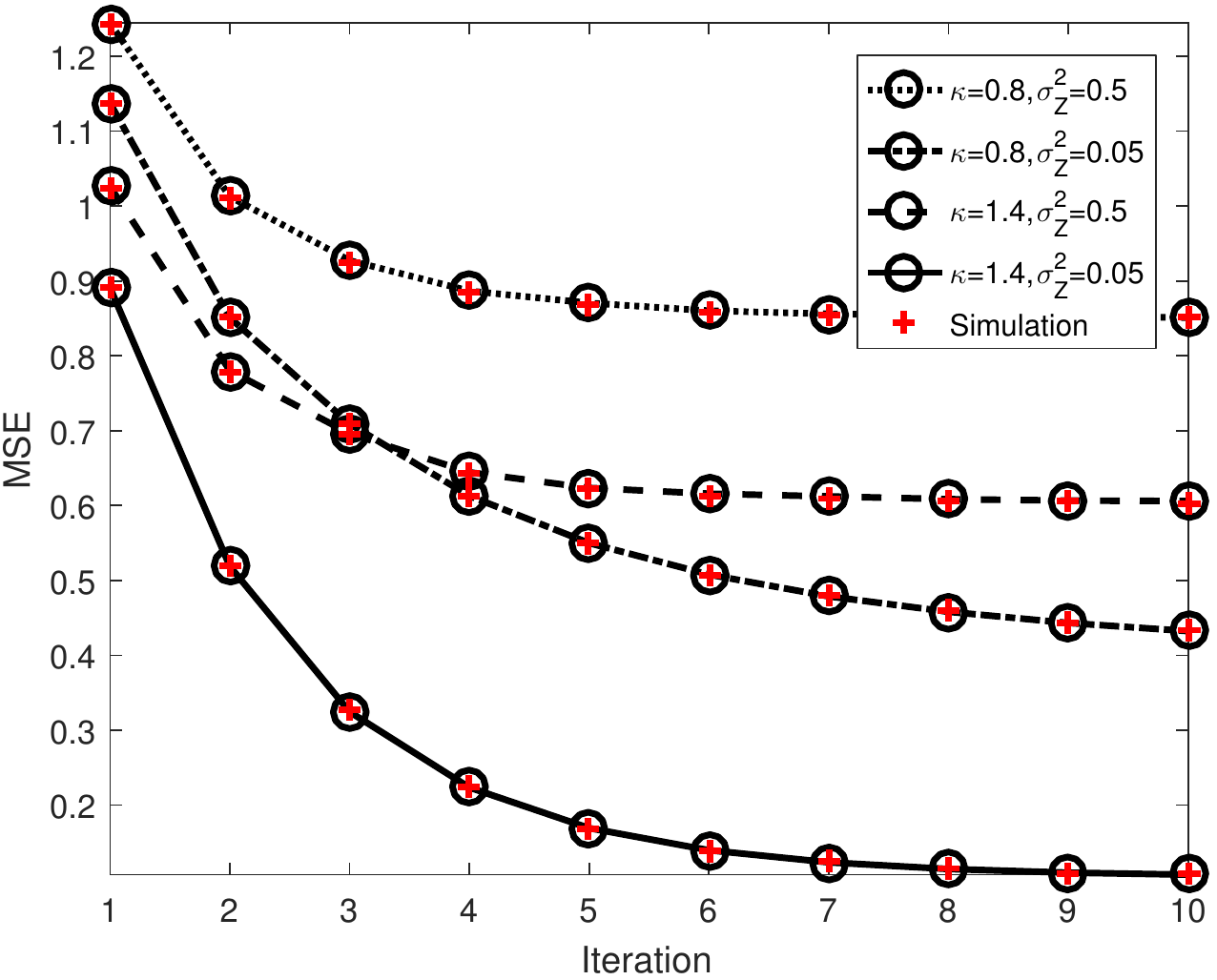}}
\caption{Comparison of the MSE predicted by lossy SE~\eqref{eq:SE_Q}
and the MSE of MP-AMP simulations for various settings.
The round markers represent MSE's predicted by lossy SE, and the (red) crosses represent
simulated MSE's. Panel (a): Bernoulli-Gaussian signal. Panel (b): Mixture Gaussian signal.}\label{fig:lossySEvsSim}
\end{figure*}

{\bf Bernoulli-Gaussian signals:} We generate 50 signals of length $10000$
according to~\eqref{eq:BG}. These signals are measured by $M=5000$ measurements
spread over $P=100$ distributed nodes. We estimate each of these signals by running
$T=10$ MP-AMP iterations. ECSQ is used to quantize $\f_t^p$~\eqref{eq:noisePseudoNode},
and $Q(\f_t^p)$~\eqref{eq:quantPseudoNode} is encoded at coding rate $R_t$.
We simulate settings with sparsity rate $\rho\in\{0.1,0.2\}$
and noise variance $\sigma_Z^2\in\{0.01,0.001\}$.
In each setting, we randomly generate the coding rate sequence $\mathbf{R}$,
s.t. the quantization bin size at each iteration satisfies $\gamma<\frac{2\sigma_t}{\sqrt{P}}$
(details in Appendix~\ref{app:verifyIndpt}).\footnote{Note that the constraint on $\gamma$
implies that $\mathbf{R}$ is likely monotone non-decreasing.}
A Bayesian denoiser,
$\eta_t(\cdot)=\mathbb{E}[\x|\f_t]$, is used in~\eqref{eq:master}.
The resulting MSE's from the MP-AMP simulation averaged over the 50 signals,
along with MSE's predicted by lossy SE~\eqref{eq:lossySE},
are plotted in Figure~\ref{fig:lossySEvsSim_BG}.
We can see that the simulated MSE's are close to the MSE's predicted by lossy SE.

{\bf Mixture Gaussian signals:} We independently generate 50 signals of length $10000$ according to
$X = \sum_{i\in \{0,1,2\}} \mathbbm{1}_{X_B = i} X_{G,i}$
where $X_B \sim \text{cat}(0.5, 0.3, 0.2)$ follows a categorical distribution on alphabet $\{0,1,2\}$, $X_{G,0} \sim \mathcal{N}(0,0.1)$, $X_{G,1} \sim \mathcal{N}(-1.5,0.8)$, and $X_{G,2} \sim \mathcal{N}(2,1)$.
We simulate settings with $T=10$, $P=100$, $\kappa=\frac{M}{N}\in\{0.8,1.6\}$,
and $\sigma_Z^2\in\{0.5,0.05\}$. In each setting, we randomly generate the coding rate  sequence $\mathbf{R}$, s.t. the quantization bin size at each iteration satisfies $\gamma<\frac{2\sigma_t}{\sqrt{P}}$. The results are plotted in Figure~\ref{fig:lossySEvsSim_mixG}. The simulation results match well with the lossy SE predictions.

\section{Integrity of Discretized Search Space}\label{app:Integrity}

\begin{figure}
\centering
 \includegraphics[width=0.48\textwidth]{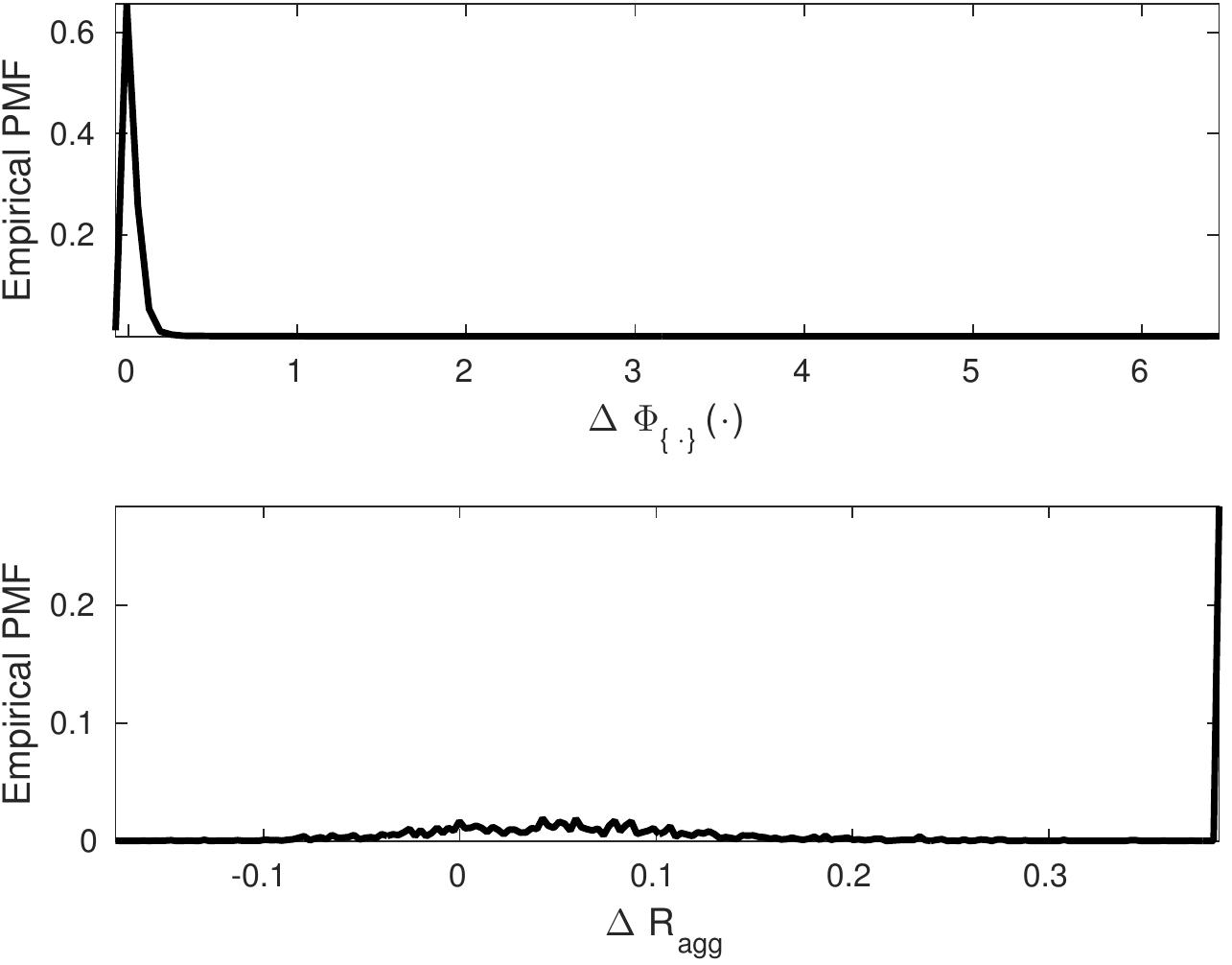}
  \caption{Justification of the discretized search space used in DP. Top panel:
Empirical PMF of the error in the cost function $\Delta \Phi_{\{\cdot\}}(\cdot)$ used to verify the integrity of the linear interpolation in the discretized search space of $\sigma^2$. Bottom panel:
Empirical PMF of $\Delta R_{agg}$; used to verify the integrity of the choice of $\Delta R=0.1$.}\label{fig:DPintegrity}
\end{figure}

When a coding rate $\widehat{R}$ is selected in MP-AMP iteration $t$, DP calculates the equivalent scalar channel noise variance $\sigma^2_{t+1}$~\eqref{eq:indpt_noises} for the next MP-AMP iteration according to~\eqref{eq:SE_Q}. The variance $\sigma^2_{t+1}$ is unlikely to lie on the discretized search space for $\sigma_t^2$, denoted by the {\em grid} $\mathcal{G}(\sigma^2)$. Therefore, $\Phi_{T-(t+1)}(\sigma^2_{t+1}(\widehat{R}))$ in~\eqref{eq:DPrecursion} does not reside in memory.
Instead of brute-force calculation of
$\Phi_{\{\cdot\}}(\cdot)$, we estimate it by fitting a function to the closest neighbors of $\sigma^2_{t+1}$ that lie on the grid $\mathcal{G}(\sigma^2)$  and finding $\Phi_{\{\cdot\}}(\cdot)$ according to the fit function. We evaluate a linear interpolation scheme.

{\bf Interpolation in $\mathcal{G}(\sigma^2)$:}
We run DP over the original coarse grid $\mathcal{G}^c(\sigma^2)$ with resolution $\Delta \sigma^2=0.01$ dB, and a 4$\times$ finer grid $\mathcal{G}^f(\sigma^2)$ with  $\Delta \sigma^2=0.0025$ dB. We obtain the cost function with the coarse grid $\Phi^c_{T-t} (( \sigma^2_t)_c)$ and the cost function with the fine grid $\Phi^f_{T-t} ((\sigma^2_t)_f),\ \forall t\in\{1,...,T\}, (\sigma^2_t)_c\in \mathcal{G}^c(\sigma^2), (\sigma^2_t)_f\in \mathcal{G}^f(\sigma^2)$. Next, we interpolate $\Phi^c_{T-t} (( \sigma^2_t)_c)$ over the fine grid $\mathcal{G}^f(\sigma^2)$ and obtain the interpolated $\Phi^i_{T-t} (( \sigma^2_t)_c)$.
In order to compare $\Phi^i_{T-t} (( \sigma^2_t)_c)$ with $\Phi^f_{T-t} (( \sigma^2_t)_c)$ in a comprehensive way, we consider the settings given by the Cartesian product of the following variables: ({\em i}) the number of distributed nodes $P\in\{50,100\}$, ({\em ii}) sparsity rate $\rho\in\{0.1,0.2\}$, ({\em iii}) measurement rate $\kappa=\frac{M}{N}\in\{3\rho,5\rho\}$, ({\em iv}) EMSE $\epsilon_T \in \{1,0.5\}$dB, ({\em v}) parameter $b\in\{0.5,2\}$, and ({\em vi}) noise variance $\sigma_Z^2\in \{0.01,0.001\}$.
In total, there are 64 different settings.
We calculate the error $\Delta\Phi_{T-t}\l(( \sigma^2_t)_c\r)=\Phi^i_{T-t} \l(( \sigma^2_t)_c\r)-\Phi^f_{T-t} \l(( \sigma^2_t)_c\r)$ and plot the empirical probability mass function (PMF) of $\Delta\Phi_{T-t}\l(( \sigma^2_t)_c\r)$
over all $t$, $(\sigma^2_t)_c$, and all 64 settings. The resulting empirical PMF
of $\Delta\Phi_{\{\cdot\}}(\cdot)$ is plotted in the top panel of Figure~\ref{fig:DPintegrity}.
We see that with 99\% probability, the error satisfies $\Delta\Phi_{\{\cdot\}}\l(\cdot\r)\leq 0.2$,
which corresponds to an inaccuracy of approximately 0.2 in the aggregate
coding rate $R_{agg}$.\footnote{Note that when calculating $\Phi^f$, we are still
using the corresponding interpolation scheme. Although this comparison is not ideal,
we believe it still provides the reader with enough insight.} In the simulation,
we used a resolution of $\Delta R=0.1$. Hence, the inaccuracy of 0.2 in $R_{agg}$
(over roughly 10 iterations) is negligible. Therefore, we use linear interpolation
with a coarse grid $\mathcal{G}^c(\sigma^2)$ with  $\Delta \sigma^2=0.01$ dB.

{\bf Integrity of choice of $\Delta R$:}
We tentatively select resolution $\Delta R=0.1$, and investigate the integrity of this $\Delta R$  over the 64 different settings above. After the coding rate sequence $\mathbf{R}^*=(R_1^*,\cdots,R_T^*)$ is obtained by DP for each setting, we randomly perturb $R_t^*$ by $R_p(t)=R_t^*+\beta_t,\ t=1,...,T$,
where $R_p(t)$ is the {\em perturbed coding rate},
the bias is $\beta_t\in \left[-\frac{\Delta R}{2},+\frac{\Delta R}{2}\right]$,
and $\mathbf{R}_p=(R_p(1),\cdots,R_p(T))$ is called the {\em perturbed coding rate sequence}.
After randomly generating 100 different perturbed coding rate sequences $\mathbf{R}_p$,
we calculate the aggregate coding rate~\eqref{eq:R_agg}, $R_{agg}^p$, of each $\mathbf{R}_p$;
we only consider the perturbed coding rate sequences that achieve EMSE no greater
than the optimal coding rate sequence $\mathbf{R}^*$ given by DP. The bottom panel
of Figure~\ref{fig:DPintegrity} plots the empirical PMF of $\Delta R_{agg}$, where
$\Delta R_{agg}=R_{agg}^p-R_{agg}^*$ and $R_{agg}^*=||\mathbf{R}^*||_1$.
Roughly 15\% of cases in our simulation yield $\Delta R_{agg}<0$
(meaning that the perturbed coding rate sequence has lower $R_{agg}$),
while for the other 85\% cases, $\mathbf{R}^*$ has lower $R_{agg}$.
Considering the resolution $\Delta R=0.1$, we can see that the perturbed
sequences are only marginally better than $\mathbf{R}^*$.
Hence, we verified the integrity of $\Delta R=0.1$.

\section[Proof of Lemma~4.1]{Proof of Lemma~\ref{lemma:DPoptimal}}\label{app:proofDPoptimal}
\begin{proof}
We show that our DP scheme~\eqref{eq:DPrecursion} fits into Bertsekas' formulation ~\cite{bertsekas1995},
which has been proved to be optimal.
Under Bertsekas' formulation, our decision variable is the coding rate $R_t$ and our state
is the scalar channel noise variance $\sigma_t^2$. Our next-state function is the lossy
SE~\eqref{eq:SE_Q} with the distortion $D_t$ being calculated from the RD function given the decision variable $R_t$.
Our additive cost associated with the dynamic system is $b\times \mathbbm{1}_{R_t\neq 0}+ R_t$.
Our control law maps the state $\sigma_t^2$ to a decision (the coding rate $R_t$).
Therefore, our DP formulation~\eqref{eq:DPrecursion} fits into the optimal DP formulation of
Bertsekas~\cite{bertsekas1995}. Hence, our DP formulation~\eqref{eq:DPrecursion}
is also optimal {\em for the discretized search spaces of $R_t$ and $\sigma_t^2$}.
\end{proof}

\begin{figure}[t]
\centering
\includegraphics[width=0.48\textwidth]{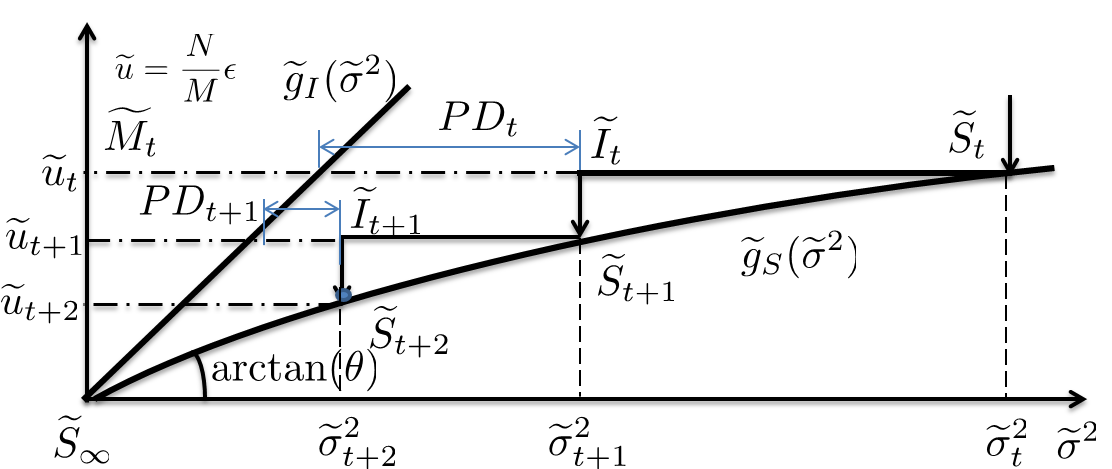}
\caption{Illustration of the evolution of $\widetilde{u}_t$. The vertical axis shows $\widetilde{u}_t=\frac{N}{M}\text{EMSE}=\frac{N}{M}\epsilon_t$.  The solid lines with arrows denote the lossy SE associated with a coding rate sequence and dashed-dotted lines are auxiliary lines.}\label{fig:evolveEMSE}
\end{figure}

\section[Proof of Theorem~4.1]{Proof of Theorem~\ref{th:optRateLinear}}\label{app:optRateLinear}
\begin{proof}
Our proof is based on the assumption that lossy SE~\eqref{eq:SE_Q} holds.
Consider the geometry of the SE incurred by $\mathbf{R}^*$ for arbitrary iterations $t$ and $t+1$, as shown in Figure~\ref{fig:evolveEMSE}.
Let $\widetilde{S}_t = (\widetilde{\sigma}_t^2,\widetilde{u}_t)$ and $R_t^*$ be the state and the optimal coding rate at iteration $t$, respectively.
We know that the slope of $\widetilde{g}_I(\cdot)$ is $\widetilde{g}_I'(\cdot)=1$. Hence,
the length of line segment $\widetilde{M}_t\widetilde{I}_t$  is $\widetilde{\sigma}_{t+1}^2=\widetilde{u}_t+PD_t$. That is
\begin{equation}\label{eq:iter_t}
PD_t=\widetilde{\sigma}_{t+1}^2-\widetilde{u}_t.
\end{equation}
Similarly, we obtain
\begin{equation}\label{eq:iter_tp1}
PD_{t+1}=\widetilde{\sigma}_{t+2}^2-\widetilde{u}_{t+1},
\end{equation}
where $\widetilde{u}_{t+1}$ and $\widetilde{\sigma}_{t+1}^2$ obey
\begin{equation}\label{eq:relate_2iters}
\widetilde{\sigma}_{t+1}^2=\widetilde{g}_S^{-1}(\widetilde{u}_{t+1}).
\end{equation}
Recall that, according to Taylor's theorem~\eqref{eq:Taylor_inverse}, we obtain that
\begin{equation}\label{eq:TaylorInverse}
\widetilde{g}_S^{-1}(\widetilde{u}_{t+1})=\frac{1}{\theta}\widetilde{u}_{t+1}+C \widetilde{u}_{t+1}^2,
\end{equation}
with $\theta$ defined in~\eqref{eq:theta}.
Although $C$ depends on $\widetilde{u}_{t+1}$, it is uniformly bounded, i.e., $C \in [-B,B]$ for some $0\leq B< \infty$.

Fixing $\widetilde{u}_{t}=\frac{N}{M}\epsilon_{t}^*$ and $\widetilde{u}_{t+2}=\frac{N}{M}\epsilon_{t+2}^*$, we  explore different distortions $D_{t}$ and $D_{t+1}$ that obey~\eqref{eq:iter_t}--\eqref{eq:relate_2iters}. According to Definition~\ref{def:optRate}, among distortions that obey~\eqref{eq:iter_t}--\eqref{eq:relate_2iters}, 
the optimal $D_{t}^*$ and $D_{t+1}^*$ correspond to the smallest
aggregate rate at iterations $t$ and $t+1$,  $R_t+R_{t+1}$.
Considering~\eqref{eq:DR}, we have
\begin{equation*}
R_{t}+R_{t+1}\!=\l[ \!\frac{1}{2}\log_2\l(\frac{C_1}{D_{t}}\r)+\frac{1}{2}\log_2\l(\frac{C_1}{D_{t+1}}\r) \r] (1+o_t(1)).
\end{equation*}
Therefore, in the large $t$ limit, minimizing $R_{t}+R_{t+1}$ is
identical to maximizing the product $D_{t}D_{t+1}$.
Considering~\eqref{eq:iter_t}--\eqref{eq:relate_2iters}, our optimization problem becomes
maximization over $F(\widetilde{u}_{t+1})$, where
\begin{equation}\label{eq:optProb}
F(\widetilde{u}_{t+1}) =
(\widetilde{\sigma}_{t+2}^2-\widetilde{u}_{t+1})(\widetilde{g}_S^{-1}(\widetilde{u}_{t+1})-\widetilde{u}_t).
\end{equation}
Invoking Taylor's theorem~\eqref{eq:TaylorInverse} and considering that $C \in [-B, B]$, we solve the optimization problem~\eqref{eq:optProb} in two extremes: one with $C = B$ and the other with  $C=-B$. 

In the case of $C=B$, we obtain
\begin{equation*}
F(\widetilde{u}_{t+1})=-\frac{1}{\theta}\widetilde{u}_{t+1}^2+\frac{1}{\theta}\widetilde{u}_{t+1}\widetilde{\sigma}_{t+2}^2+
B\widetilde{\sigma}_{t+2}^2\widetilde{u}_{t+1}^2-
B\widetilde{u}_{t+1}^3-\widetilde{u}_t\widetilde{\sigma}_{t+2}^2+
\widetilde{u}_t\widetilde{u}_{t+1}.
\end{equation*}
The maximum of $F(\widetilde{u}_{t+1})$ is achieved when $F'(\widetilde{u}_{t+1})=0$. That is,
\begin{equation}\label{eq:zeroDeriv}
F'(\widetilde{u}_{t+1})=-3B\widetilde{u}_{t+1}^2+\l(2B\widetilde{\sigma}_{t+2}^2-\frac{2}{\theta}\r)\widetilde{u}_{t+1}+\frac{\widetilde{\sigma}_{t+2}^2}{\theta}+\widetilde{u}_{t}=0.
\end{equation}
Considering  that $0<\widetilde{u}_{t+1}< \widetilde{u}_{t}$,
the root of the quadratic equation~\eqref{eq:zeroDeriv} is
\begin{equation}\label{eq:optSolution}
\widetilde{u}^*_{t+1}=\frac{1}{3B}\l[\l(B\widetilde{\sigma}_{t+2}^2-\frac{1}{\theta}\r)+A\r],
\end{equation}
where
\begin{equation}\label{eq:original_A}
A=\sqrt{\l(B\widetilde{\sigma}_{t+2}^2-\frac{1}{\theta}\r)^2+3B\l(\frac{\widetilde{\sigma}_{t+2}^2}{\theta}+\widetilde{u}_{t}\r)}.
\end{equation}
We can further simplify~\eqref{eq:original_A} as
\begin{equation}\label{eq:TaylorApprox}
\begin{split}
A=&\frac{1}{\theta}\sqrt{1+B(\theta \widetilde{\sigma}_{t+2}^2+B\theta^2 \widetilde{\sigma}_{t+2}^4+3\theta^2 \widetilde{u}_{t} )}\\
=&\frac{1}{\theta}\l[1+\frac{B}{2}(\theta \widetilde{\sigma}_{t+2}^2+B\theta^2\widetilde{\sigma}_{t+2}^4+
3\theta^2\widetilde{u}_{t})\r]+O(\widetilde{u}_t^2),
\end{split}
\end{equation}
Plugging~\eqref{eq:TaylorApprox} into~\eqref{eq:optSolution}, 
\begin{equation}\label{eq:opt_u_tp1}
\widetilde{u}_{t+1}^*=\frac{1}{2}(\widetilde{\sigma}_{t+2}^2+\theta \widetilde{u}_{t})+O(\widetilde{u}_{t}^2).
\end{equation}
Plugging~\eqref{eq:opt_u_tp1} into~\eqref{eq:iter_t} and~\eqref{eq:iter_tp1}, 
\begin{equation*}
\begin{split}
PD_t^*&=\frac{1}{2\theta}(\widetilde{\sigma}_{t+2}^2-\widetilde{u}_{t}\theta)+O(\widetilde{u}_t^2),\\
PD_{t+1}^*&=\frac{1}{2}(\widetilde{\sigma}_{t+2}^2-\widetilde{u}_t^2\theta )+O(\widetilde{u}_t^2),
\end{split}
\end{equation*}
which leads to
\begin{equation}\label{eq:optRatio_ori}
\frac{D_{t+1}^*}{D_t^*}=\theta(1+O(\widetilde{u}_t)).
\end{equation}

These steps provided the optimal relation between $D_{t}^*$ and $D_{t+1}^*$ when $C=B$.
For the other extreme case, $C=-B$,
similar steps will lead to~\eqref{eq:optRatio_ori},
where the differences between the results are higher order terms. Note that for any $C \in [-B, B]$ the higher order term is bounded between the two extremes.
Hence, the optimal $D_t^*$ and $D_{t+1}^*$
follow~\eqref{eq:optRatio_ori} leading to the first part of the claim~\eqref{eq:theorem1-1}.
Considering~\eqref{eq:DR} and~\eqref{eq:optRatio_ori},
\begin{equation*}
R_{t+1}^*-R_t^*=\frac{1}{2}\log_2 \l(\frac{1}{\theta}\r) (1+ o_t(1)).
\end{equation*}
Therefore, we obtain the second part of the claim~\eqref{eq:theorem1}.
\end{proof}

\section[Proof of Theorem~4.2]{Proof of Theorem~\ref{th:convergence}}\label{app:convergence}

\begin{proof}
Our proof is based on the assumption that lossy SE~\eqref{eq:SE_Q} holds.
Let us focus on an optimal coding rate sequence $\mathbf{R}^*=(R_1^*,\cdots,R_T^*)$.
Applying Taylor's theorem to calculate the ordinate of point $\widetilde{S}_{t+1}$ using its abscissa (Figure~\ref{fig:evolveEMSE}), we obtain
\begin{equation}\label{eq:Taylor_u1}
\widetilde{u}_{t+1}^*=\theta(\widetilde{u}_t^*+PD_t^*)+O((\widetilde{u}_t^*)^2).
\end{equation}
Therefore,
\begin{equation}\label{eq:ratio1}
\frac{\widetilde{u}_{t+1}^*}{\widetilde{u}_t^*}=\theta+\frac{\theta PD_t^*}{\widetilde{u}_t^*}+O(\widetilde{u}_t^*).
\end{equation}
Similarly, we obtain
\begin{equation}\label{eq:ratio2}
\frac{\widetilde{u}_{t+2}^*}{\widetilde{u}_{t+1}^*}=\theta+\frac{\theta PD_{t+1}^*}{\widetilde{u}_{t+1}^*}+O(\widetilde{u}_t^*).
\end{equation}
Plugging~\eqref{eq:optRatio_ori} and~\eqref{eq:Taylor_u1} into~\eqref{eq:ratio2}, we obtain
\begin{equation}\label{eq:ratio2_new}
\begin{split}
\frac{\widetilde{u}_{t+2}^*}{\widetilde{u}_{t+1}^*}&=\theta+\frac{\theta PD_t^*(1+O(u_t^*))}{\widetilde{u}_t^*+PD_t^*+O((\widetilde{u}_t^*)^2)}+O(\widetilde{u}_t^*)\\
&=\theta+\frac{\theta PD_t^*}{\widetilde{u}_t^*+PD_t^*}+O(\widetilde{u}_t^*).
\end{split}
\end{equation}
On the other hand,  $\lim_{t\rightarrow\infty}\frac{\widetilde{u}_{t+1}^*}{\widetilde{u}_{t}^*}=\lim_{t\rightarrow\infty} \frac{\widetilde{u}_{t+2}^*}{\widetilde{u}_{t+1}^*}$.
Therefore, considering~\eqref{eq:ratio1} and~\eqref{eq:ratio2_new}, we obtain
\begin{equation*}
\lim_{t\rightarrow\infty} \frac{\theta PD_t^*}{\widetilde{u}_t^*} = \lim_{t\rightarrow\infty} \frac{\theta PD_t^*}{\widetilde{u}_t^*+PD_t^*},
\end{equation*}
which leads to $\lim_{t\rightarrow}\frac{D_t^*}{\widetilde{u}_t^*} =0$. We obtain~\eqref{eq:theorem2_2} by noting that the optimal EMSE at iteration t is $\epsilon_t^*=\frac{M}{N}\widetilde{u}_t^*$.
Plugging~\eqref{eq:theorem2_2} into~\eqref{eq:ratio1}, we obtain~\eqref{eq:theorem2_1}.
\end{proof}

\chapter{Appendices for Chapter~5}\label{chap:append-SLAM}
\section{Proof of Theorem~\ref{th:conv}}
\label{ap:th:conv}
\sectionmark{Proof of Theorem~5.1}

Our proof mimics a very similar proof presented in~\cite{Jalali2008,Jalali2012} for lossy source coding; we include all details for completeness. The proof technique relies on mathematical properties of non-homogeneous (e.g., time-varying) Markov chains (MC's)~\cite{Bremaud1999}. Through the proof, $\ss \triangleq (\replevels_F)^N$ denotes the state space of the MC of codewords generated by Algorithm~\ref{alg:MCMC}, with size $|\ss| = |\replevels_F|^N$. We define a stochastic transition matrix $\P_{(t)}$ from $\ss$ to itself given by the Boltzmann distribution for super-iteration $t$ in Algorithm~\ref{alg:MCMC}. Similarly, $\pi_{(t)}$ defines the stable-state distribution on $\ss$ for $\P_{(t)}$, satisfying $\pi_{(t)} \P_{(t)} = \pi_{(t)}$.

\begin{myDef}~\cite{Bremaud1999}
{\em Dobrushin's ergodic coefficient} of an MC transition matrix $\P$ is denoted by $\xi(\P)$ and defined as
$\displaystyle \xi(\P) \triangleq \max_{1 \le i,j \le N} \frac{1}{2}\|\P_i-\P_j\|_1$,
where $\P_{i}$ denotes the $i$-th row of $\P$.
\end{myDef}
From the definition, $0 \le \xi(\P) \le 1$. Moreover, the ergodic coefficient can be rewritten as
\begin{align}
\xi(\P) = 1-\min_{1 \le i,j \le N} \sum_{k=1}^N \min(P_{ik},P_{jk}),
\label{eq:ergcoef}
\end{align}
where $P_{ij}$ denotes the entry of $\P$ at the $i$-th row, $j$-th column.

We group the product of transition matrices
across super-iterations as $\P_{(t_1\to t_2)} = \prod_{t=t_1}^{t_2}\P_{(t)}$.
There are two common characterizations for the stable-state behavior of a
non-homogeneous MC.
\begin{myDef}~\cite{Bremaud1999}
A  non-homogeneous MC is called {\em weakly ergodic} if for any distributions
$\eta$ and $\nu$ over the state space $\ss$, and any $t_1 \in \mathbb{N}$,
$\displaystyle {\lim\sup}_{t_2 \to \infty} \|\eta \P_{(t_1\to t_2)}-\nu \P_{(t_1\to t_2)}\|_1 = 0$,
where $\|\cdot\|_1$ denotes the $\ell_1$ norm.
Similarly, a non-homogeneous MC is called {\em strongly ergodic} if there
exists a distribution $\pi$ over the state space $\ss$ such that for any
distribution $\eta$ over $\ss$, and any $t_1 \in \mathbb{N}$,
$\displaystyle {\lim\sup}_{t_2 \to \infty} \|\eta \P_{(t_1 \to t_2)}-\pi\|_1 = 0$.
We will use the following two theorems from~\cite{Bremaud1999} in our proof.
\end{myDef}
\begin{myTheorem}~\cite{Bremaud1999}
An MC is weakly ergodic if and only if there exists a sequence of integers
$0 \le t_1 \le t_2 \le \cdots$ such that $\displaystyle\sum_{i=1}^\infty \left(1-\xi\left(\P_{(t_i \to t_{i+1})}\right)\right) = \infty$.
\label{th:block}
\end{myTheorem}
\begin{myTheorem}~\cite{Bremaud1999}
Let an MC be weakly ergodic. Assume that there exists a sequence of
probability distributions $\{\pi_{(t)}\}_{i=1}^\infty$ on the state space $\ss$
such that $\pi_{(t)} \P_{(t)} = \pi_{(t)}$. Then the MC is strongly ergodic if $\displaystyle\sum_{t=1}^\infty \|\pi_{(t)}-\pi_{(t+1)}\|_1 < \infty$.
\label{th:weak}
\end{myTheorem}

The rest of proof is structured as follows. First, we show that the sequence
of stable-state distributions for the MC used by Algorithm~\ref{alg:MCMC}
converges to a uniform distribution over the set of sequences that minimize
the energy function as the iteration count $t$ increases. Then, we show using
Theorems~\ref{th:block} and~\ref{th:weak} that the non-homogeneous MC
used in Algorithm~\ref{alg:MCMC} is strongly ergodic, which by the definition
of strong ergodicity implies that Algorithm~\ref{alg:MCMC} always converges
to the stable distribution found above. This implies that the outcome of
Algorithm~\ref{alg:MCMC} converges to a minimum-energy solution as
$t \to \infty$, completing the proof of Theorem~\ref{th:conv}.

We therefore begin by finding the stable-state distribution for the non-homogeneous
MC used by Algorithm~\ref{alg:MCMC}. At each super-iteration $t$, the distribution
defined as
\begin{equation}
\pi_{(t)}(\w) \triangleq \frac{\text{exp}\l(-s_t\Psi^{H_q}(\w)\r)}{\sum_{\z \in \ss} \text{exp}\l(-s_t\Psi^{H_q}(\z)\r)}
= \frac{1}{\sum_{\z \in \ss} \text{exp}\l(-s_t\l(\Psi^{H_q}(\z)-\Psi^{H_q}(\w)\r)\r)} \label{eq:pidist}
\end{equation}
satisfies $\pi_{(t)}\P_{(t)} = \pi_{(t)}$, cf.~(\ref{eqn:Gibbs}). We can show that the
distribution $\pi_{(t)}$ converges to a uniform distribution over the set of sequences
that minimize the energy function, i.e.,
\begin{align}
\lim_{t \to \infty} \pi_{(t)}(\w) = \left\{\begin{array}{cl}0 & \w \notin \mathcal{H}, \\ \frac{1}{|\mathcal{H}|} & \w \in \mathcal{H},\end{array}\right.
\label{eq:stable}
\end{align}
where $\mathcal{H} = \{\w \in \ss~\textrm{subject to}~\Psi^{H_q}(\w) = \min_{\z \in \ss} \Psi^{H_q}(\z)\}$.
To show (\ref{eq:stable}), we will show that $\pi_{(t)}(\w)$ is increasing for
$\w \in \mathcal{H}$ and eventually decreasing for $\w \in \mathcal{H}^C$.
Since for $\w \in \mathcal{H}$ and $\widetilde{\w} \in \ss$ we have
$\Psi^{H_q}(\widetilde{\w})-\Psi^{H_q}(\w) \ge 0$,  for $t_1 < t_2$ we have
\begin{equation*}
\sum_{\widetilde{\w} \in \ss} \text{exp}\l[-s_{t_1}\l(\Psi^{H_q}(\widetilde{\w})-\Psi^{H_q}(\w)\r)\r]
\ge \sum_{\widetilde{\w} \in \ss} \text{exp}\l[-s_{t_2}\l(\Psi^{H_q}(\widetilde{\w})-\Psi^{H_q}(\w)\r)\r],
\end{equation*}
which together with (\ref{eq:pidist}) implies $\pi_{(t_1)}(\w) \le \pi_{(t_2)}(\w)$. On the other hand, if $\w \in \mathcal{H}^C$, then we obtain
\begin{align}
\pi_{(t)}(\w)\! =\! \left\{\!\sum_{\widetilde{\w}: \Psi^{H_q}(\widetilde{\w}) \ge \Psi^{H_q}(\w)}\! \text{exp}\l[-s_t\l(\Psi^{H_q}(\widetilde{\w})\!-\!\Psi^{H_q}(\w)\r)\r]\!+\!\sum_{\widetilde{\w}: \Psi^{H_q}(\widetilde{\w}) < \Psi^{H_q}(\w)} \!\text{exp}\l[-s_t\l(\Psi^{H_q}(\widetilde{\w})\!-\!\Psi^{H_q}(\w)\r)\r]\right\}^{-1}. \label{eq:hlower}
\end{align}
For sufficiently large $s_t$, the denominator of (\ref{eq:hlower}) is dominated
by the second term, which increases when $s_t$ increases, and therefore
$\pi_{(t)}(\w)$ decreases for $\w \in \mathcal{H}^C$ as $t$ increases. Finally,
since all sequences $\w \in \mathcal{H}$ have the same energy $\Psi^{H_q}(\w)$,
it follows that the distribution is uniform over the symbols in $\mathcal{H}$.

Having shown convergence of the non-homogenous MC's stable-state distributions,
we now show that the non-homogeneous MC is strongly ergodic. The transition
matrix $\P_{(t)}$ of the MC at iteration $t$ depends on the temperature $s_t$ in
(\ref{eq:st}) used within Algorithm~\ref{alg:MCMC}. We first show that the MC
used in Algorithm~\ref{alg:MCMC} is weakly ergodic via Theorem~\ref{th:block}; the
proof of the following Lemma is given in~\ref{app:lemm:ergcoefbound}.
\begin{myLemma}
The ergodic coefficient of $\P_{(t)}$ for any $t\ge 0$ is upper bounded by
$\displaystyle \xi\left(\P_{(t)}\right) \le 1-\text{exp}(-s_tN\Delta_q)$,
where $\Delta_q$ is defined in (\ref{eq:Deltaq}).
\label{lemm:ergcoefbound}
\end{myLemma}
We note in passing that Condition~\ref{cond:tech1} ensures that $\Delta_q$ is
finite. Using Lemma~\ref{lemm:ergcoefbound} and (\ref{eq:st}), we can evaluate
the sum given in Theorem~\ref{th:block} as
\begin{align*}
\sum_{j=1}^\infty\left[1-\xi\left(\P_{(j)}\right)\right] \ge \sum_{j=1}^\infty\text{exp}(-s_jN\Delta_q)
= \sum_{j=1}^\infty\frac{1}{j^{1/c}} = \infty,
\end{align*}
and so the non-homogeneous MC defined by $\{\P_{(t)}\}_{t=1}^\infty$ is
weakly ergodic. Now we use Theorem~\ref{th:weak} to show that the MC is
strongly ergodic by proving that
$\displaystyle \sum_{t=1}^\infty \|\pi_{(t)}-\pi_{(t+1)}\|_1 < \infty$.
Since we know from earlier in the proof that $\pi_{(t)}(\w)$ is increasing for
$\w \in \mathcal{H}$ and eventually decreasing for $\w \in \mathcal{H}^C$, there
exists a $t_0 \in \mathbb{N}$ such that for any $t_1 > t_0$, we have
\begin{equation*}
\begin{split}
\sum_{t=t_0}^{t_1} \|\pi_{(t)}-\pi_{(t+1)}\|_1=& \sum_{\w \in \mathcal{H}}\sum_{t=t_0}^{t_1} \left(\pi_{(t+1)}(\w)-\pi_{(t)}(\w)\right) + \sum_{\w \notin \mathcal{H}}\sum_{t=t_0}^{t_1} \left(\pi_{(t)}(\w)-\pi_{(t+1)}(\w)\right)\\
=& \sum_{\w \in \mathcal{H}}\left(\pi_{(t_1+1)}(\w)-\pi_{(t_0)}(\w)\right) + \sum_{\w \notin \mathcal{H}} \left(\pi_{(t_0)}(\w)-\pi_{(t_1+1)}(\w)\right)\\
=& \|\pi_{(t_1+1)}-\pi_{(t_0)}\|_1 \le \|\pi_{(t_1+1)}\|_1+\|\pi_{(t_0)}\|_1 = 2.
\end{split}
\end{equation*}
Since the right hand side does not depend on $t_1$, we have that
$\sum_{t=1}^\infty \|\pi_{(t)}-\pi_{(t+1)}\|_1 < \infty$. This implies that the
non-homogeneous MC used by Algorithm~\ref{alg:MCMC} is strongly ergodic, and
thus completes the proof of Theorem~\ref{th:conv}.

\section[Proof of Lemma~C.1]{Proof of Lemma~\ref{lemm:ergcoefbound}}\label{app:lemm:ergcoefbound}

Let $\w',\w''$ be two arbitrary sequences in $\ss$. The probability of transitioning
from a given state to a neighboring state within iteration $t'$ of
super-iteration $t$ of Algorithm~\ref{alg:MCMC} is given by (\ref{eqn:Gibbs}), and
can be rewritten as
\begin{equation*}
\begin{split}
\mathbb{P}_{(t,t')}(\w_1&^{t'-1}a\w_{t'+1}^N|\w_1^{t'-1}b\w_{t'+1}^N) = \mathbb{P}_{s_t}(w_{t'} = a|\w^{\backslash t'}) = \frac{ \text{exp}\left(-s_t\Psi^{H_q}\l(\w_1^{t'-1}a\w_{t'+1}^N\r) \right) }{ \sum_{b\in\replevels_F} \text{exp}\left( -s_t\Psi^{H_q}\l(\w_1^{t'-1}b\w_{t'+1}^N\r) \right) }\\
& = \frac{ \text{exp}\left[-s_t\left(\Psi^{H_q}\l(\w_1^{t'-1}a\w_{t'+1}^N\r) -\Psi^{H_q}_{\min,t'}\l(\w_1^{t'-1},\w_{t'+1}^N\r)\right) \right] }{ \sum_{b\in \replevels_F} \text{exp}\left[ -s_t\left(\Psi^{H_q}\l(\w_1^{t'-1}b\w_{t'+1}^N\r)-\Psi^{H_q}_{\min,t'}\l(\w_1^{t'-1},\w_{t'+1}^N\r)\right) \right] } \ge \frac{\text{exp}(-s_t\Delta_q)}{|\replevels_F|},
\end{split}
\end{equation*}
where $\Psi^{H_q}_{\min,t'}(\w_1^{t'-1},\w_{t'+1}^N) = \min_{\beta \in \replevels_F} \Psi^{H_q}(\w_1^{t'-1}\beta \w_{t'+1}^N)$.
Therefore, the smallest probability of transition from $\w'$ to $\w''$ within
super-iteration $t$ of Algorithm~\ref{alg:MCMC} is bounded by
\begin{equation*}
\min_{\w',\w'' \in \replevels_F}\mathbb{P}_{(t)}(\w''|\w')  \ge \prod_{t'=1}^{N} \frac{\text{exp}(-s_t\Delta_q)}{|\replevels_F|}= \frac{\text{exp}(-s_tN\Delta_q)}{|\replevels_F|^N} = \frac{\text{exp}(-s_tN\Delta_q)}{|\ss|}.
\end{equation*}
Using the alternative definition of the ergodic coefficient (\ref{eq:ergcoef}),
\begin{align*}
\xi\left(\P_{(t)}\right) &= 1-\min_{\w',\w'' \in \ss} \sum_{\widetilde{\w} \in \ss} \min(\mathbb{P}_{(t)}(\widetilde{\w}|\w'),\mathbb{P}_{(t)}(\widetilde{\w}|\w''))\\
&\le 1 - |\ss|\frac{\text{exp}(-s_tN\Delta_q)}{|\ss|} = 1-\text{exp}(-s_tN\Delta_q),
\end{align*}
proving the lemma.

\restoregeometry

\backmatter

\end{document}